\documentclass[12pt]{article}
\usepackage{comment}
\usepackage{amsmath}
\usepackage{amsthm}
\usepackage{amssymb}
\usepackage{appendix}
\usepackage{mathrsfs}
\usepackage{enumerate}
\usepackage{mathtools}
\usepackage{multirow}
\usepackage{dsfont}
\usepackage{mdframed}
\usepackage{algorithm}
\usepackage[noend]{algpseudocode}
\usepackage{bm}
\usepackage{bbm}
\usepackage{natbib}
\setcitestyle{aysep={}} 
\usepackage{url}
\usepackage[dvipsnames]{xcolor}
\usepackage[colorlinks,linkcolor=Magenta,citecolor=Magenta,urlcolor = Magenta]{hyperref}
\hypersetup{colorlinks,allcolors=blue}
\usepackage[labelsep=period, font=sc, skip=0pt,justification=centering]{caption}
\usepackage{graphicx}
\usepackage{pdfpages}
\usepackage{setspace}
\usepackage{float}
\usepackage[top=1in, bottom=1in, left=1.0in, right=1.0in]{geometry}
\usepackage{makecell}
\usepackage{subcaption}
\usepackage[hang,flushmargin]{footmisc}
\usepackage{titling}
\usepackage{tikz}
\usepackage{booktabs}
\newcommand{\doubleoverline}[1]{\overline{\vphantom{#1}\overline{#1}}}

\usepackage{tikz}
\usetikzlibrary{positioning,calc}

\usepackage{titletoc}

\usepackage{etoc}

\usepackage{algorithm}
\usepackage{algpseudocode}
\usepackage{comment}
\usepackage{amsmath, amssymb}
\usepackage{setspace}
\usepackage{hyperref}

\usepackage[capitalize]{cleveref}
\crefname{section}{Section}{Sections}
\Crefname{section}{Section}{Sections}
\crefname{subsection}{Section}{Sections}
\Crefname{subsection}{Section}{Sections}
\crefname{subsubsection}{Section}{Sections}
\Crefname{subsubsection}{Section}{Sections}
\usetikzlibrary{shapes,positioning}

\tikzset{ell/.style={rectangle,draw,minimum height=0.65cm,minimum width=1cm,inner sep=0.25cm}}

\usepackage{graphicx}

\usepackage{titlesec}
\titlelabel{\thetitle.\ }

\newtheorem{theorem}{Theorem}[section]

\newtheorem{condition}{Condition}

\newtheorem{lemma}{Lemma}[section]

\newtheorem{remark}{Remark}[section]

\crefname{condition}{Condition}{Conditions}
\Crefname{condition}{Condition}{Conditions}

\newcommand{\be}{\begin{eqnarray*}}
\newcommand{\ee}{\end{eqnarray*}}

 \usepackage{amsmath,amssymb}
\usepackage[top=1in, bottom=1in, left=1.0in, right=1.0in]{geometry}
\usepackage{natbib}
\setcitestyle{aysep={}}
\usepackage{xcolor}

\usepackage{float}
\floatstyle{ruled}
\newfloat{algorithm}{tbp}{loa}
\floatname{algorithm}{Algorithm}
\newcommand{\algline}[2]{\par\makebox[2em][r]{{\scriptsize #1}:\ }\parbox[t]{\dimexpr\linewidth-2.5em}{#2}}
\usepackage{tikz}
\usetikzlibrary{arrows.meta,positioning}
\numberwithin{equation}{section}
 
\newcommand{\E}{\mathbb{E}}

\newcommand{\R}{\mathbb{R}}

\graphicspath{{../code/R/graphs/paper/}}%

\usepackage{stackengine}
\usepackage{blindtext}

\usepackage{authblk}
\usepackage{tikz}
\usetikzlibrary{arrows.meta,positioning}

\providecommand{\argmax}{\operatorname*{arg\,max}}

\newcommand{\safeincludegraphics}[2][]{%
  \IfFileExists{#2}{\includegraphics[#1]{#2}}{%
    \fbox{\begin{minipage}[c][0.26\textheight][c]{0.9\linewidth}\centering Missing figure: \texttt{\detokenize{#2}}\end{minipage}}%
  }%
}
\begin{document}
\title{Fast Online Inference on Semiparametric Models\thanks{First version: arXiv:2603.08614v1. We acknowledge helpful comments from participants of the Econometrics Seminars at Brown, Renmin and USP Universities, the 2026 Cowles Summer Conference in New Haven and the 2026 CIREQ-HEC Women in Econometrics Conference in Montreal.}}

\author{Xiaohong Chen\thanks{Dept. of Economics, Yale University.} \ \ \ Elie Tamer\thanks{Dept. of Economics, Harvard University.} \ \ \ Qingsong Yao\thanks{Dept. of Economics, Louisiana State University. }}
\maketitle

\begin{abstract}
\singlespacing
\noindent
This paper develops a framework for fast online  inference on semiparametric models with large sample sizes and possibly many covariates. The computational algorithm itself is the object of statistical study: after a globally consistent warm start in the first phase, the path of averaged online iterates generated in the second phase automatically delivers estimators with optimal convergence rates and valid confidence sets. Both phases require only a single pass over the data stream and are well suited to streaming data or to settings with storage/privacy constraints. For semiparametric monotone index models, the averaged trajectory of the second phase lead to  estimators that are automatically orthogonalized and satisfy the laws of the iterated logarithms, and policy functionals are updated along the same trajectory at negligible additional cost. The averaged trajectories satisfy functional central limit theorems, which yield fast online inference via random scaling and bypass the explicit variance estimation that complicates inference for semiparametric models. Applied to a fixed large sample, our online algorithm achieves substantial computational gains over corresponding offline procedures without sacrificing statistical performance. Monte Carlo experiments show adequate behavior. Our methods are applied to 19 million traffic-stop online records from the North Carolina State Patrol \citep{pierson2020large} and to the international trade data of \citet{helpman2008estimating} with over 300 regressors. We also compare our estimator to its parametric benchmark in both empirical illustrations. 
\end{abstract}

\noindent

\newpage

\section{Introduction}

An extensive literature in econometrics examines inference on flexible semiparametric models \citep{powellhandbook}. Traditionally, this literature has treated the estimator as given: one posits that an estimator exists and then derives its statistical properties---global consistency, convergence rates, limiting distributions, and consistent standard errors---while the computational machinery for actually obtaining that estimator is treated as a separate concern, with no bearing on the large-sample theory. Computation, however, has become a first-order issue in modern economic research. Data relevant to applied economists increasingly arrive as streams---the Stanford Open Policing Project has recorded hundreds of millions of traffic stops and is updated frequently \citep{pierson2020large}; scanner systems \citep{broda2010product,atkin2018retail,dellavigna2019uniform}, credit registries \citep{jimenez2014hazardous,agarwal2018banks}-, and digital platforms \citep{cavallo2017online} generate observations at a pace that makes repeated full-sample re-estimation impractical and---when records cannot be retained for privacy or memory reasons--- not feasible. Even in conventional fixed-sample settings, semiparametric models with many covariates and large sample sizes are costly to compute and  pose significant challenges for inference. This has led some practitioners to fall back on convenient parametric specifications---such as Logit or Probit links in binary response models---even when those assumptions are unwarranted and can lead to bias in policy evaluations.

In light of these challenges, we revisit inference for an important class of semiparametric models and take a fundamentally different approach: the computational algorithm itself becomes the object of statistical study. We propose a novel online algorithm designed for these models and use the sequence of iterates it generates as the statistical scaffold for deriving inference results, rather than treating computation and statistics as separate enterprises. After a quick initial warm start phase, the algorithm treats the current estimates of the parametric and nonparametric components as summary statistics, uses a small batch of newly arrived data to update both components jointly, and outputs the trajectory of the Polyak--Ruppert (PR) averaged iterates \citep{ruppert1988efficient,polyak1992acceleration} together with the averaged policy-functional iterates. Coupling computation with statistics pays off on both fronts. Theoretically, the PR averaged estimates converge to the truth almost surely at optimal rates and are automatically orthogonalized, and the trajectory of the PR averages satisfies a functional central limit theorem (FCLT). The FCLT enables us to construct fast online confidence sets via self-normalization (or ``random scaling'' of \citet{lee2022fast}), bypassing the complicated variance estimation and/or bootstrapping for inference on semiparametric models using full sample data. Practically, the automatic orthogonalization guarantees negligible first-order bias, which, together with the FCLT, yields fast and valid confidence sets even when the number of covariates is large. Each data point is processed only once, so the method is naturally suited to streaming data and to storage- or privacy-constrained environments; and the same trajectory also delivers almost surely consistent estimates of, and valid inference on, policy-relevant functionals---such as average marginal effects---at negligible additional cost.

The above described online principle applies to general semiparametric models. For a large class of semiparametric nonlinear models with many covariates, however, it is difficult to design a warm start phase that can quickly enter the stable local neighborhood of the global true parameter values regardless of the initialization of the algorithm, as well as difficult to establish almost sure convergence with optimal rates and trajectory FCLTs under lower level sufficient conditions. In this paper, we focus on a leading class of semiparametric nonlinear models, the following semiparameric monotone index model:
\[
  Y = F_0(Z_0) + \varepsilon, \qquad Z_0:=x_0 + X^{\top}\theta_0, \qquad \E(\varepsilon \mid x_0, X) = 0,
\]
where $F_0$ is the unknown and monotone link function, $\theta\in \R^p$ is unknown parameter of interest with a fixed but possibly large $p$ (e.g., $p=50, 100, 500$ in our Monte Carlo simulations, and $p>300$ in our trade application), the normalization on $x_0$ gives an interpretable index direction, and some of the covariates $X$ could be discrete (as in our applications and Monte Carlo experiments). An important parameter of interest is the average marginal effects (AME) parameter $\tau_0:=\mathbb{E}\left[\partial \mathbb{E}(Y|x_0, X)/\partial X\right]\equiv  \mathbb{E}\left[\nabla_z F_0\!\left(Z_0\right)\theta_0\right]$. 

 This monotone index class contains many workhorse models used in applied microeconomics, including binary choice, censored regression, and duration models \citep{powellhandbook}. Despite many theoretical works on identification, estimation and inference for $\theta_0,F_0,\tau_0$ in the literature, due to the computational difficulties of the existing procedures, there are very few empirical applications in economics allowing for unknown link function $F_0$ with many covariates and large sample sizes. Our paper will propose a novel, computationally fast online estimation and inference method for these semiparametric nonlinear models, in which the unknown monotone link function $F_0$ is approximated by a sequence of B-spline sieves\footnote{B-spline sieve is known to preserve monotonic shape. \citep{chen2007large}. Other shape-preserving sieves could also be used to approximate unknown monotone link function $F_0$; see Section \autoref{section7.1}.}. The semiparametric online problem is theoretically challenging for two reasons: the joint sieve nonlinear least-squares (NLS) criterion for $\theta_0,F_0$ is not globally convex; and the sieve dimension needs to grow slowly as data accumulate. These features place our problem outside the existing literature on online M-estimation of $\theta_0,F_0$ and inference on $\theta_0,\tau_0$ (see the related literature below for more discussions).

\paragraph{Contributions}
This paper makes three main contributions. The first is algorithmic. We propose a two-phase online learner for semiparametric index models. The warm-start phase uses a kernelized rank-based update, inspired by the maximum rank correlation of \citet{han1987non} and the smoothed maximum score of \citet{horowitz1992smoothed}, to obtain an almost surely consistent estimate of the index parameter $\theta_0$ from an arbitrary initialization. Our warm-start algorithm for $\theta$ is globally convex and stable, and hence converges to the true parameter $\theta_0$ regardless of the choice of initialization. This warm start in turn generates a consistent index $Z_0=x_0 + X^{\top}\theta_0$ estimate, on which a B-spline sieve least-squares update provides an almost sure consistent warm start for the nonparametric link component $F_0$. Initialized at these warm starts for $(\theta,F)$, the second phase performs joint local stochastic-gradient updates for $(\theta_0,F_0)$ via a sieve NLS criterion. To make sieve learning genuinely online, we introduce a re-embedding operator that maps previously learned sieve coefficients into the newly expanded sieve space; this preserves the fitted function exactly when the sieve dimension grows and allows the online path to continue without restarting the optimization.

The second contribution is theoretical. For the warm-start phase, we establish an almost-sure global consistency with rate for the online warm-start estimator of $\theta_0$, together with an almost-sure sup-norm convergence rate for online sieve estimator of $F_0$ with generated regressors. For the joint local refinement phase, we prove an asymptotic linear representation for the PR averages. The representation shows that the leading score for $\theta_0$ is automatically orthogonalized: the joint local update removes the first-order effect of estimating the unknown link $F_0$, without requiring a researcher to estimate a separate high-dimensional Riesz representer in order to construct an orthogonal score manually. The PR-averaged estimator obeys an almost-sure rate-optimal law of the iterated logarithm, and satisfies an FCLT. We further establish a corresponding limit theory for the online plug-in estimator of the AME $\tau_0$, so the same path supports inference on objects directly tied to policy evaluations.

The third contribution is practical. After a quick warm start, the procedure uses each new small fixed batch once, updates the parameter, sieve coefficients, and policy-functional running sums, and avoids repeated full-sample optimization. This is  valuable for streaming data, in environments in which storage is costly or restricted, and for conventional large offline samples in which repeated full-sample nonconvex optimization is slow. In our Monte Carlo experiments, the method remains accurate in dimensions as large as $p = 500$. In online vs offline timing comparisons with $p = 100$, online local refinement is roughly one to two orders of magnitude faster than full-sample sieve NLS while delivering comparable or smaller estimation error. Because the construction of confidence sets for $\theta_0,\tau_0$ uses only the stored trajectory of averaged updates via random scaling, it completely avoids the computational bottlenecks of conventional semiparametric inferences that require either additional nonparametric estimation of Riesz representers (to compute semiparametric standard errors) or additional numerical methods such as bootstraps.

\paragraph{Empirical Applications}
We illustrate the method in two empirical settings. The first uses the bilateral trade panel of \citet{helpman2008estimating}, a binary-choice application with more than 300 regressors once exporter, importer, and year fixed effects are included. A calibrated Probit experiment first verifies that the online semiparametric estimator recovers the parametric truth when the Probit link is correctly specified. We then apply our method to the observed trade outcomes without imposing a parametric link. The online semiparametric estimates agree closely with the parametric benchmarks for some geographic determinants of trade such as shared land border and landlock, but they differ substantially for the institutional and informational gravity variables---colonial ties, common language, shared religion, currency union, and free trade agreements. The estimated coefficient on colonial ties, for example, is less than half of its Probit counterpart. On the probability scale, the average marginal effect of a free trade agreement is smaller than the Probit benchmark by about three percentage points---a modest difference in absolute size, but an economically meaningful one relative to the baseline probability of positive trade. This illustrates how the shape of the link function—typically fixed by assumption in parametric models—can  impact estimated policy magnitudes.

The second application uses North Carolina traffic-stop records from the Stanford Open Policing Project \citep{pierson2020large}. After cleaning, the data contain more than 19 million stops ordered in time---a genuinely online administrative stream. The outcome is whether a stop results in a search. We process the post-warm-start records in their recorded order, each exactly once, and update the average marginal effects of race and gender along the online path. The warm start takes about two minutes using the first half of the 19 million data points; the online learning then processes the remaining 9.6 million post-June-2009 records in roughly 17 seconds---about half a million  records per second---for a total of 2.3 minutes. The final estimates are statistically distinguishable from zero and economically sizable relative to the low unconditional search rate. Most importantly, as in the trade application, the semiparametric estimates deviate substantially from the Probit and Logit benchmarks: the parametric models overstate the Black--Hispanic search-probability gap by more than a quarter and understate the White--Hispanic gap by roughly a half. Because all specifications are estimated on the same sample with the same covariates, these discrepancies are attributable to the shape of the link function alone. Consistent with the infra-marginality concerns emphasized in the policing literature \citep{simoiu2017problem,pierson2020large}, we interpret the estimated effects as descriptive rather than causal. The point of the exercise is to show that semiparametric inference can be implemented at the scale, and in the order, in which policy-relevant administrative data are generated.

\paragraph{Extensions}
We note that the Phase II algorithm and its properties remain valid with any other consistent warm starts. Indeed, our Monte Carlo studies in \autoref{sec6.3} show that the online Phase II algorithms perform virtually the same regardless whether our online Phase I algorithm or the global optimization algorithm SMCO of \citep{chen2026optimization} is used. The only difference is computational speed in the warm start phase: our online Phase I exploits the semiparametric monotone index model structure so is faster to find a good warm start than the generic global optimizer SMCO. Indeed, the logic of our Phase II algorithm, coupled with any consistent warm start algorithm such as SMCO, easily extends beyond the baseline monotone index model. In \autoref{conclusion}, we discuss sample-selection settings in which observability is itself an economic choice \citep{heckman1974,abhausmankhan,khan2024inference}; multi-index and multi-alternative choice models; and monotone neural-network approximations to the unknown link. These extensions share the same architecture: a warm-start phase, such as SMCO \citep{chen2026optimization}, first locates a stable basin, after which a local joint sieve online learner produces averaged paths for estimation and inference.

\paragraph{Related Literature}
Our paper contributes to two literatures. The first is the large literature on identification, offline estimation and inference for semiparametric index models in economics. Well-known estimation methods include, among others, maximum rank correlation, semiparametric least squares, efficient binary-response estimation, censored and sample-selection estimators, and rank-based procedures \citep{han1987non,powell1989semiparametric,hardle1993optimal,ichimura1993semiparametric,klein1993efficient,sherman1993limiting,neweyetal,ahn2018simple,fan2020rank,KHAN2025105901}. These works are designed for fixed samples and often require optimization or nuisance-estimation steps that are computationally costly in massive-data settings: each evaluation of the maximum rank correlation criterion alone entails $n(n-1)$ pairwise comparisons (where $n$ is the sample size), and the computational cost grows fast in the dimension of $\theta$ (i.e., the dimension of covariates $X$). To our knowledge, the present paper is among the first to give a joint online sieve-estimation and trajectory-inference theory for this class while retaining an interpretable index structure with many covariates.

The second is the vast literature on stochastic approximation (SA) in the tradition of \citet{robbins1951stochastic} and \citet{kushner-yin}, including abstract infinite-dimensional Hilbert-space online learning problems \citep{yin1990,chen2002asymptotic}. An important subclass of SA algorithms, stochastic gradient descent (SGD), is widely used in finite-dimensional optimization in machine learning and statistics \citep{ghadimi2013stochastic,toulis2017asymptotic,bottou2018optimization} and in nonparametric (kernel or sieve) density and conditional-mean estimation \citep{huang2013recursive,zhang2022sieve}. Recent papers develop inference based on PR averaged SGD iterates in parametric strongly convex M-estimation problems \citep{chenlee2020,lee2022fast,lee2025} and in parametric over-identified GMM problems \citep{chen2023sgmm,chen2025slim} for large data sets. Closest to our setting, \cite{han-jmlr2024} study online inference for $\theta_0$ in high-dimensional single-index models with streaming data. However, they assume that (i) the covariate vector $(x_0,X^{\top})$ satisfies a so-called linear expectation condition (i.e., for all parameter vector $(b_0,b^{\top})$, $\E[x_0 b_0 + X^{\top}b \mid x_0 + X^{\top}\theta_0]$ is linear in the true index $x_0 + X^{\top}\theta_0$) and is statistically full independent of the latent error term); and (ii) there is a loss function $\E[\ell (Y,x_0 + X^{\top}\theta)]$ that is convex in the index, has a unique solution in $\theta_0$ and does not depend on the unknown link function $F_0$. Under these and other assumptions, they study online SGD inference for $\theta_0$ without estimating the unknown link $F_0$. In addition, their methods cannot perform online inference on policy functionals that depend on both $\theta_0$ and the derivative of $F_0$ - like standard partial effects. Relative to this literature, our paper develops a joint sieve online estimation and trajectory-inference theory for semiparametric index models with finite- and infinite-dimensional unknowns, some discrete covariates, and policy effects.

\paragraph{Organization}
\autoref{section: model and data} introduces the model and reviews relevant offline estimators. \autoref{section:algorithm} presents the two-phase online algorithm, random-scaling inference, and online estimation of policy functionals. \autoref{sec:empirical} reports the two empirical illustrations.  \autoref{section3} establishes the statistical properties of the estimators.  \autoref{sec:MC} reports Monte Carlo experiments.  \autoref{conclusion} presents conclusions and extensions. All proofs are collected in the Appendix.

\paragraph{Notation}
For nonrandom sequences $\{a_N\}$ and $\{b_N\}$ with $b_N > 0$, we write
$a_N = o(b_N)$ if $\lim_{N \to \infty} a_N/b_N = 0$, and $a_N = O(b_N)$ if
$\overline{\lim}_{N \to \infty} |a_N|/b_N < \infty$.  We write $a_N \asymp b_N$ if both $a_N = O(b_N)$ and $b_N = O(|a_N|)$. Consider a sequence of random
variables $\{a_N(w)\}$ defined on the underlying probability space
$(A, \mathcal{A}, P_A)$. We say event $E\in\mathcal{A}$ occurs \textit{almost surely} if $P_A(A\setminus E) =0$.  For any nonrandom sequence $b_N > 0$, we write
$a_N = o(b_N) $  a.s.\  if  
$\lim_{N \to \infty} a_N /b_N = 0$ almost surely, and we write 
$a_N = O (b_N)$     a.s.\ if
$\overline{\lim}_{N \to \infty} |a_N|/b_N < \infty$  almost surely. Specifically, if  $a_N  = o(b_N) $  a.s.\ with $b_N \equiv 1$,  we write
$a_N \rightarrow_{\mathrm{a.s.}} 0$. We use $\rightarrow_d$ to denote
convergence in distribution; we use $\Rightarrow$ to
denote  weak convergence in the metric space that will be clear from the
context. For any scalar function $f(z)$ with domain $\mathcal{Z}$, $\|f\|_\infty \equiv \sup_{z \in \mathcal{Z}} |f(z)|$ denotes the sup norm. Finally, for  any real symmetric matrix $M$, $\lambda_{\max}(M)$ and
$\lambda_{\min}(M)$ denote its largest and smallest eigenvalues.

\section{Model and Offline Estimation: A Selective Review}\label{section: model and data}

\subsection{Semiparametric Monotone Index Models}
We first consider the standard semiparametric monotone index model
\begin{align}
    Y = F_0(x_0 + X^{\top}\theta_0) + \varepsilon, \quad \mathbb E(\varepsilon|x_0, X) = 0, \label{model}
\end{align}
where $Y$ is an observed response, $F_0(\cdot)$ is an unknown monotonically increasing function, $(x_0, X^{\top})^{\top}$ is a $(p+1)\times 1$ vector of regressors, $\theta_0$ is a $p\times 1$ unknown parameter, and $\varepsilon$ is an unobserved error term. The coefficient of $x_0$ is normalized to one for identification. Monotone index models nest binary choice, censored, hazard, and Box--Cox transformation models. Two examples illustrate the scope.

\paragraph{Example 1 (Binary economic choice).} In a binary random-utility model \citep{mcfadden2001economic}, $Y=\mathbf{1}\{V_1(X)-V_0(X)-u\ge 0\}$ with systematic utility difference $V_1(X)-V_0(X)=x_0+X^{\top}\theta_0$, so that
$\mathrm{Pr}(Y=1\mid x_0,X)=F_0(x_0+X^{\top}\theta_0)$, where $F_0$ is the distribution of $u$. Model \eqref{model} is therefore a distribution-free representation of an economic choice probability: it recovers the utility index and the link $F_0$ without imposing Logit or Probit. Semiparametric analysis of this model goes back to \citet{manski1975maximum,cosslett1987efficiency, han1987non, ichimura1993semiparametric,klein1993efficient}.

\paragraph{Example 2 (Duration).} In the proportional hazards model $\lambda(t\mid X)=\lambda_0(t)\exp(X^{\top}\theta_0)$ \citep{cox1972regression,kalbfleisch1980statistical,lancaster1990econometric}, suppose the outcome records whether a spell terminates before a fixed horizon $c$ (for example, whether an unemployment spell ends within a given period), $Y=\mathbf{1}\{T\le c\}$. Then
$
\mathbb{P}(Y=1\mid X) = 1-\exp(-\varLambda_0(c)\exp(X^{\top}\theta_0)) = F_0(X^{\top}\theta_0),
$ 
where $\varLambda_0(t)=\int_0^t \lambda_0(s)\,\mathrm{d}s$ and $F_0(z)=1-\exp(-\varLambda_0(c)e^{z})$ is strictly increasing, with shape governed by the unknown cumulative baseline hazard. This is exactly a monotone index binary choice model.

\subsection{Offline Estimators}
In the classical offline setting, a sample $\mathcal{D}_n=\{(x_{0,i},X_i,Y_i)\}_{i=1}^n$ of fixed size $n$ is collected once, and different strategies apply depending on the parameter of interest. In the following we write $Z_i(\theta)=x_{0,i}+X_i^{\top}\theta$.

When $\theta_0$ alone is of interest, the maximum rank correlation (MRC) estimator of \citet{han1987non} maximizes
$\widehat L_{\mathrm{MRC},n}(\theta) = \frac{1}{n(n-1)}\sum_{i\neq j}\boldsymbol{1}\{Z_i(\theta)>Z_j(\theta)\}\boldsymbol{1}\{Y_i>Y_j\}$
over $\theta\in\varTheta$. The MRC estimator requires no estimate of $F_0$, but its criterion is discontinuous in $\theta$ and is evaluated over all $n(n-1)$ pairs; even the smoothed version of \citet{horowitz1992smoothed} must be re-optimized over the full set of pairs at every iteration.

When $F_0$ is also of interest, one can adopt the method of sieves \citep{chen2007large}. Let $J=J_n$ be a sample-size-dependent sieve dimension, let $\Psi_J(\cdot)=(\psi_{J,1}(\cdot),\ldots,\psi_{J,J}(\cdot))^{\top}$ collect the basis functions, and define the sieve space
$\mathcal{S}_J=\{\sum_{j=1}^J b_j\psi_{J,j}(\cdot): b_1,\ldots,b_J\in\mathbb{R}\}$.
With $\theta_0$ known (or replaced by a preliminary estimate), $F_0$ is estimated by least squares of $Y_i$ on $\Psi_J(Z_i(\theta_0))$. When both components are of interest, the sieve nonlinear least squares (NLS) estimator solves
\begin{align}\label{sieve-nls}
    \left(\widehat\theta_n,\widehat\beta_n\right)
    \in
    \arg\min_{\theta\in\Theta,\ \beta\in\mathbb{R}^{J_n}}
    \frac{1}{2n}
    \sum_{i=1}^n
    \left(
    Y_i-\Psi_{J_n}(x_{0,i}+X_i^{\top}\theta)^{\top}\beta
    \right)^2,
\end{align}
with $\widehat F_{\mathrm{sieve\text{-}NLS},n}=\Psi_J(\cdot)^{\top}\widehat\beta_n$ \citep{ai2003efficient,chen2007large}.

All of these estimators are offline by design. Each criterion evaluation requires a complete pass through the data (through all $n(n-1)$ pairs for MRC); the NLS criterion \eqref{sieve-nls} is nonconvex in $(\theta,\beta)$, so reliable computation requires many such passes from multiple starting points; and the arrival of new observations forces re-estimation from scratch. These features make the offline toolkit impractical for large samples and unusable for streaming data, and they motivate the online algorithms of \autoref{section:algorithm}.

\subsection{Sieve NLS as a Structured Neural Network}

The sieve learner in \eqref{sieve-nls} can be interpreted as a single-hidden-layer neural network with economically meaningful restrictions. Spline bases admit representations through truncated power functions, so a sieve approximation of   $J$ basis functions can be written as
\begin{equation}
  F_J(x_0+X^{\top}\theta)
  =
  \sum_{j=1}^J \beta_j\psi_j(x_0+X^{\top}\theta)
  =
  \sum_{\ell=1}^{M_J} a_\ell
  \sigma_q(x_0+X^{\top}\theta-\kappa_\ell),
\end{equation}
where $M_J$ is the number of hidden units, $\sigma_q(t)=t_+^q$ is a ReLU-power activation, and the $\kappa_\ell$ are spline knots. The conditional mean then takes the network form
$\mathbb{E}(Y|x_0,X)\approx \sum_{\ell=1}^{M_J} a_\ell\,\sigma_q(b_{\ell,-1}+b_{\ell,0}x_0+X^{\top}b_\ell)$
subject to $(b_{\ell,0},b_\ell^{\top})^{\top}\propto (1,\theta^{\top})^{\top}$: all hidden units share a common index direction. The hidden layer learns nonlinear transformations of the scalar economic index $x_0+X^{\top}\theta$, while the output layer approximates the unknown monotone response curve. The restrictions are economically meaningful: the common first-layer direction is the interpretable index parameter $\theta_0$; monotonicity of $F_0$ can be imposed through monotone spline bases \citep{chen2007large}; and inference concerns the low-dimensional structural parameter rather than unrestricted network weights.

This interpretation also clarifies the role of sieve-dimension growth in online learning: increasing the sieve dimension enlarges the hidden layer, so previously learned weights must be \emph{re-embedded} into the expanded network in a way that preserves the learned shape of the response curve. The next section develops an online algorithm for \eqref{sieve-nls} built around exactly this recursive re-embedding. Richer architectures with multiple index directions (multinomial choice) or selection gates (sample selection) extend the same design and are developed in \autoref{extension}.

\section{The Online Learning Algorithms}\label{section:algorithm}

In the online environment, data arrive sequentially in small batches. Let $k=1,2,\cdots$ index a discrete time sequence. In period $k$ the researcher observes a batch of $B$ i.i.d.\ realizations of $(x_0,X,Y)$ from model \eqref{model}, where the batch size $B$ is a fixed integer. Define $W_{i,k}=(x_{0,i,k},X_{i,k},Y_{i,k})$ for $i=1,\ldots,B$ and $\mathcal{W}_k=\{W_{i,k}\}_{i=1}^B$. Throughout, $\{\mathcal{W}_k\}_{k=1}^{\infty}$ is assumed i.i.d.\ across $k$; the analysis can be extended to dependent data under additional appropriate conditions.

Our online algorithm consists of two phases, summarized in \autoref{fig:semiparametric_online_learning}. The warm-start Phase~I produces consistent estimators of $(\theta_0,F_0)$ from \emph{arbitrary} initializations (\autoref{sec2.1} and \autoref{sec2.2}). The refinement Phase~II locally and jointly updates both components to attain optimal convergence rates (\autoref{sec2.3}). The refinement trajectory is then the sole input to random-scaling inference (\autoref{sec2.4-rs}) and to online estimation and inference for policy functionals (\autoref{sec2.5}).

A single averaging device is used throughout. For any sequence of iterates $a_1,a_2,\cdots$, the Polyak--Ruppert (PR) average \citep{ruppert1988efficient,polyak1992acceleration} is defined as $\overline a_N=N^{-1}\sum_{k=1}^N a_k$, computed by the recursion $\overline a_N=\frac{N-1}{N}\overline a_{N-1}+\frac{1}{N}a_N$. Using this recursion, only the current average requires storage, so PR averaging is itself fully online. All estimators reported by our algorithms are PR averages.

\begin{figure}[tbp]
\centering
\begin{tikzpicture}[
    node distance=0.85cm and 0.9cm,
    box/.style={
        rectangle,
        rounded corners,
        draw=black,
        align=center,
        text width=3.2cm,
        minimum height=1.15cm,
        inner sep=5pt,
        font=\footnotesize
    },
    mainbox/.style={box, fill=blue!8},
    phasebox/.style={box, fill=orange!10},
    outputbox/.style={box, fill=green!10},
    inferbox/.style={box, fill=purple!10},
    arrow/.style={-{Stealth[length=2.4mm]}, thick}
]
\node[mainbox] (init) {
\textbf{Cold Start}\\
Arbitrary starting values\\
for $\theta$ and $F$
};
\node[phasebox, below=of init] (warm) {
\textbf{Warm Start}\\
Phase I learner/SMCO \citep{chen2026optimization}/\\
other methods
};
\node[phasebox, right=of warm] (joint) {
\textbf{Local Joint Optimization}\\
Joint NLS over $(\theta,F)$\\
for $\theta_0$ and $F_0$
};
\node[outputbox, right=1.0cm of joint, yshift=1.05cm] (est) {
\textbf{Rate-Optimal Estimators}\\
$\doubleoverline{\theta}$ and $\doubleoverline{F}$
};
\node[outputbox, right=1.0cm of joint, yshift=-1.05cm] (policy) {
\textbf{Policy Functionals: $\doubleoverline{\tau} = \tau(\doubleoverline{\theta}, \doubleoverline{F})$}
};
\node[inferbox, right=1.0cm of est, yshift=-1.05cm] (infer) {
\textbf{Online Inference}\\
Random scaling for $\theta_0$\\
and policy functionals $\tau_0$
};
\draw[arrow] (init) -- (warm);
\draw[arrow] (warm) -- (joint);
\draw[arrow] (joint.east) -- (est.west);
\draw[arrow] (joint.east) -- (policy.west);
\draw[arrow] (est.east) -- (infer.west);
\draw[arrow] (policy.east) -- (infer.west);
\end{tikzpicture}
\bigskip{}
\caption{Flow Chart of Semiparametric Online Learning}
\label{fig:semiparametric_online_learning}
\end{figure}

\subsection{Online Warm Start for $\theta_0$}\label{sec2.1}

Phase~I estimates $\theta_0$ without estimating $F_0$, through an online learner that is \emph{globally stable}: it converges almost surely to $\theta_0$ regardless of the starting point. Let $\widehat\theta_0$ denote an arbitrary initial guess and, for any $\theta\in\mathbb{R}^p$, define
$
Z_{i,k}(\theta) = x_{0,i,k} + X_{i,k}^{\top}\theta, \quad i = 1,\cdots, B.
$
For $k\geq 1$, the update from $\widehat\theta_{k-1}$ to $\widehat\theta_k$ is
\begin{equation}\label{first-step algorithm}
\widehat \theta_{k} = \widehat \theta_{k-1} + \frac{\gamma_k}{h_k\cdot B(B-1)}\sum_{i_1\neq i_2}^B \mathcal{K}\left(\frac{Z_{i_1,k}(\widehat\theta_{k-1}) - Z_{i_2,k}(\widehat\theta_{k-1})}{h_k}\right)\left(Y_{i_1,k} - Y_{i_2,k}\right)\left(X_{i_1,k} - X_{i_2,k}\right),
\end{equation}
where $\gamma_k>0$ is the learning rate, $\mathcal{K}$ is a kernel function, and $h_k$ is a bandwidth sequence.

Update \eqref{first-step algorithm} is related to the stochastic gradient of the   MRC criterion of \citet{han1987non} with one crucial modification: the ranking indicator $\boldsymbol{1}\{Y_{i_1,k}>Y_{i_2,k}\}$ is replaced by the response difference $Y_{i_1,k}-Y_{i_2,k}$. One of the core features of the algorithm is   global stability: it almost surely converges to $\theta_0$ from any initialization---which we establish formally in \autoref{theorem1} in \autoref{appendixB}. The PR average estimator is defined as
\begin{align}\label{PR average}
\overline{\theta}_N = \frac{1}{N}\sum_{k=1}^N\widehat\theta_{k}.
\end{align}
The pseudo code of the warm-start estimation of $\theta_0$ is given in Algorithm~\ref{algo:phase1} in the Appendix.

\subsection{Online Warm Start for $F_0$}\label{sec2.2}

We approximate $F_0$ by B-splines and estimate $F_0$ by the online sieve Least Squares. Throughout, we require that as the B-spline sieve dimension grows from $J$ to 
$J+1$, one interior knot is inserted at the midpoint of the largest knot interval---
preserving quasi-uniformity---while all existing knots are retained. The knot vectors are 
therefore nested, so the sieve spaces satisfy $\mathcal{S}_J \subseteq \mathcal{S}_{J+1}$ and 
the bases obey the exact refinement relation
\begin{equation}\label{eq:reembedding}
\Psi_J(z)=R_{J,J+1}\Psi_{J+1}(z), \qquad \forall z,
\end{equation}
where $R_{J,J+1}\in\mathbb{R}^{J\times(J+1)}$ is the transpose of the Boehm knot-insertion 
matrix \citep{boehm1980inserting}; we adopt the convention 
$R_{J,J}=\mathbb{I}_J$.

We suppose that a sequence of index estimates $\check\theta_0,\check\theta_1,\cdots$ is available: $\check\theta_{k-1}=\theta_0$ if $\theta_0$ were known, and otherwise any consistent online estimator, such as the PR average $\overline\theta_{k-1}$ of \autoref{sec2.1}. Write $\check Z_{i,k}=Z_{i,k}(\check\theta_{k-1})$.

Given an initial coefficient $\widehat\beta_0$ of dimension $J_0$ and a pre-specified increasing sequence of sieve dimensions $\{J_k\}_{k\ge1}$, the online sieve update is
\begin{align}\label{online_sieve}
\widehat{\beta}_{k} =
R_{J_{k-1}, J_k}^{\top}\widehat{\beta}_{k-1}
+ \frac{\eta_k}{B}\sum_{i=1}^B
\Big(
Y_{i,k}
-
\Psi_{J_k}(\check{Z}_{i,k})^{\top}
R_{J_{k-1}, J_k}^{\top}\widehat{\beta}_{k-1}
\Big)
\Psi_{J_k}(\check{Z}_{i,k}),
\end{align}
where $\eta_k>0$ is a learning rate. When $J_k=J_{k-1}$, \eqref{online_sieve} is the mini-batch SGD step for the least squares projection of $Y$ on $\mathcal{S}_{J_k}$. As in offline estimation, however, consistency requires the sieve dimension to grow, and enlarging a spline basis reconstructs the entire basis system: the coefficient vector for $\mathcal{S}_{J_{k-1}}$ is neither comparable to nor embedded in that for $\mathcal{S}_{J_k}$, so the update cannot simply be continued. Our solution is the \emph{re-embedding} of $\widehat\beta_{k-1}$: the previous coefficient is re-initialized as $R_{J_{k-1},J_k}^{\top}\widehat\beta_{k-1}$, which, since $\Psi_{J_{k-1}}(z)=R_{J_{k-1},J_k}\Psi_{J_k}(z)$, satisfies
\[
\Psi_{J_{k-1}}(z)^{\top}\widehat\beta_{k-1}
=
\Psi_{J_k}(z)^{\top} R_{J_{k-1},J_k}^{\top}\widehat\beta_{k-1}.
\]
The previously fitted function is thus preserved \emph{exactly} when embedded into the expanded sieve space, providing a stable starting point in the new space. The online estimator of $F_0$ and its PR average are
\begin{equation}\label{online estimator of F}
\widehat{F}_k(\cdot) = \Psi_{J_k}(\cdot)^{\top}\widehat{\beta}_k, \qquad 
\overline{F}_N(\cdot) = \frac{1}{N}\sum_{k=1}^N \widehat{F}_k(\cdot).
\end{equation}
It is convenient to average at the coefficient level: setting $\overline\beta_0=\widehat\beta_0$ and
\begin{align}
\overline{\beta}_{N}
=
\frac{1}{N}\widehat{\beta}_N
+
\frac{N-1}{N}
R_{J_{N-1}, J_N}^{\top}\overline{\beta}_{N-1}
\end{align}
yields the closed form
\begin{align}
\overline{\beta}_N
=
\frac{1}{N}\sum_{k=1}^N
R_{J_k, J_N}^{\top}\widehat{\beta}_k,
\qquad
R_{J_k, J_N}
=R_{J_k, J_{k+1}} R_{J_{k+1}, J_{k+2}} \cdots R_{J_{N-1}, J_N},
\end{align}
so that $\overline F_N(\cdot)=\Psi_{J_N}(\cdot)^{\top}\overline\beta_N$. The procedure is summarized in Algorithm~\ref{alg:global-online-F} in Appendix.

\subsection{Phase II: Joint Local Refinement of $\theta_0$ and $F_0$}\label{sec2.3}

Starting from the Phase~I warm starts, Phase~II jointly refines both components so that the resulting estimators attain optimal convergence rates. For sieve dimension $J$, define the sieve NLS loss
\begin{equation}\label{nls_score}
\mathcal{L}_J(\theta, \beta,W) = \frac{1}{2}\left(Y - \Psi_{J}\left(x_0 + X^{\top}\theta\right)^{\top}\beta\right)^2,
\end{equation}
write $\omega=(\theta^{\top},\beta^{\top})^{\top}$ (whose dimension grows with the sieve dimension), and let $\mathcal{L}_J(\omega)\equiv\mathbb{E}[\mathcal{L}_J(\omega,W_{i,k})]$ denote the population loss. As in \autoref{sec2.2}, the sieve dimension $J_k$ grows along the path, and each update starts from the re-embedded previous iterate $\mathcal{R}_{J_{k-1},J_k}^{\top}\widetilde\omega_{k-1}$, where
\begin{equation}
    \mathcal{R}_{J_1, J_2} = \begin{pmatrix}
        \mathbb{I}_p & 0 \\
        0 & R_{J_1, J_2}
    \end{pmatrix}
\end{equation}
extends the re-embedding operator to the joint parameter. The natural stochastic gradient update is then
\begin{align}\label{joint_update}
\widetilde\omega_{k} =\mathcal{R}_{J_{k-1}, J_k}^{\top}\widetilde\omega_{k-1} - \frac{\xi_k \cdot \widehat G_k  }{B}\sum_{i=1}^B \nabla_{\omega}\mathcal{L}_{J_k}\left(\mathcal{R}_{J_{k-1}, J_k}^{\top}\widetilde\omega_{k-1}, W_{i,k}\right),
\end{align}
where $\xi_k>0$ is the learning rate and $\widehat G_k$ is an (almost surely) positive definite conditioning matrix. For specific choices of the condition matrix in applications, see \autoref{sec:empirical}.

The critical complication is that $\mathcal{L}_J(\omega)$ is \emph{not} globally convex in $\omega$: unrestricted updates need not converge to the truth, which motivates confining \eqref{joint_update} to a neighborhood over which the population loss has nondegenerate curvature. Let $Z_0=x_0+X^{\top}\theta_0$ and define
\begin{equation}
    \omega_{J, 0} = \left(\theta^{\top}_0, \beta_{J,0}^{\top}\right)^{\top}, \ \ \beta_{J,0} = \left( \mathbb{E}\left[ \Psi_J( Z_0)\Psi_J(Z_0)^{\top}\right]\right)^{-1}\mathbb{E}\left[ \Psi_J( Z_0)F_0(Z_0)\right],
\end{equation}
where $\beta_{J,0}$ is the $L_2(P_{Z_0})$ projection coefficient.  We call $\omega_{J,0}$ the \emph{sieve pseudo-true parameter} at dimension $J$. It does not need to be the exact joint finite-$J$ population minimizer because the $\theta$ score at $\omega_{J,0}$ can contain a small approximation-bias term. We require
\begin{equation}\label{local positive eigen}
   \inf_{J}\lambda_{\mathrm{min}}\left(\frac{\partial^2 \mathcal{L}_J(\omega_{J,0}) }{\partial\omega\partial\omega^{\top}}\right) > 0,
\end{equation}   
so the population criterion has positive curvature at $\omega_{J,0}$ uniformly in $J$. As we formally show in \autoref{discussion_on_basin} in \autoref{appendixD}, the curvature extends to the neighborhood
\begin{equation}
    \Omega_{J} = \left\{\omega\in\mathbb{R}^{p+J}: \Vert \omega -  \omega_{J,0} \Vert_2 \leq C_{\Omega}J^{-2}\right\}
\end{equation}
for a suitable constant $C_{\Omega}>0$: the exponent $2$ arises from the Lipschitz constant of the Hessian matrix within the local basin, so $J^{-2}$ is the radius over which the Hessian of $\mathcal{L}_J$ moves by at most a small $\epsilon$ in operator norm. Since $\omega_{J_k,0}$ is unknown, we center a feasible neighborhood at the Phase I estimate $\widehat\omega_k = 
 \left(\check\theta_{k-1}^\top,
 \{R_{J_{k-1},J_k}^\top\widehat\beta_{k-1}\}^\top\right)^\top$:
\begin{equation}
    \widehat\Omega_{k} = \left\{\omega\in\mathbb{R}^{p+J_k}: \Vert \omega - \widehat\omega_{k} \Vert_2 \leq \frac{C_{\Omega}}{2}J_k^{-2} \right\}.
\end{equation}
When $\widehat\omega_k$ converges to $\omega_{J_k,0}$ faster than $J_k^{-2}$, then $\widehat\Omega_k\subseteq\Omega_{J_k}$ eventually almost surely, so nondegenerate curvature holds on the feasible set.  The feasible set is updated from the same incoming batch only after the center for that batch has been formed, and is retained for the next batch; hence the construction remains single-pass and $\widehat\Omega_k$ is measurable with respect to the past.  

Letting $\Pi_A(a)=\arg\min_{a'\in A}\Vert a'-a\Vert_2$ denote the (well-defined) Euclidean projection onto a compact convex set $A$, the Phase~II update is
\begin{align} \label{local joint update}
\widetilde\omega_{k} = \Pi_{\widehat\Omega_{k}}\left[\mathcal{R}_{J_{k-1}, J_k}^{\top}\widetilde\omega_{k-1} - \frac{\xi_k\cdot \widehat G_k }{B}\sum_{i=1}^B\nabla_{\omega}\mathcal{L}_{J_k}\left(\mathcal{R}_{J_{k-1}, J_k}^{\top}\widetilde\omega_{k-1}, W_{i,k}\right)\right].
\end{align}
The PR averages are computed along the way:
\begin{equation}\label{pr_theta_local}
    \doubleoverline{\theta}_N = \frac{1}{N}\sum_{k=1}^N \widetilde\theta_k,\qquad 
    \doubleoverline{\beta}_N = \frac{1}{N}\sum_{k=1}^N R_{J_k, J_N}^{\top}\widetilde\beta_k,
\qquad 
     \doubleoverline{F}_N(\cdot) = \Psi_{J_N}(\cdot)^{\top}\doubleoverline{\beta}_N.
\end{equation}
The pseudo code of Phase II procedure is summarized in Algorithm~\ref{alg:local-online} in Appendix.

\subsection{Random-Scaling Inference}\label{sec2.4-rs}

Inference on $\theta_0$ uses the random-scaling method of \citet{lee2022fast}, whose only input is the trajectory of Phase~II PR averages $\doubleoverline{\theta}_1,\doubleoverline{\theta}_2,\cdots$. Define the self-normalizing matrix
\begin{equation}\label{rs_variance}
V_N = \frac{1}{N^2}\sum_{k=1}^{N}k^2\left(\doubleoverline{\theta}_k - \doubleoverline{\theta}_N\right)\left(\doubleoverline{\theta}_k - \doubleoverline{\theta}_N\right)^{\top},
\end{equation}
which is updated recursively from running sums. A $95\%$ confidence interval for the $j$-th component of $\theta_0$ is
\begin{align}\label{rs_band}
\doubleoverline{\theta}_{N,j} \pm 6.747\cdot\sqrt{[V_N]_{jj}/N},
\end{align}
where $6.747$ is the two-sided $95\%$ critical value of the associated mixed-normal limit. No variance component needs to be estimated, so inference adds essentially no computational cost beyond estimation itself. Although random scaling was proposed for parametric online least squares, its justification requires only the validity of an FCLT for the averaged trajectory; the FCLTs established in \autoref{sec3.3} therefore extend it to the present semiparametric setting, including the policy functionals of \autoref{sec2.5}.

\subsection{Online Estimation and Inference for Policy Interventions}\label{sec2.5}

The joint trajectory also supports online estimation and inference for policy-relevant functionals of $(\theta_0,F_0)$. As a leading example, consider the \textit{average marginal effect}
\begin{equation}\label{average_marginal_effect}
    \tau_{0} =\mathbb{E}\left[\frac{\partial \mathbb{E}(Y|x_0, X)}{\partial X}\right]\equiv  \mathbb{E}\left[\nabla_z F_0\!\left(Z_0\right)\theta_0\right],
\end{equation}
whose $j$-th component measures the average change in the outcome induced by a one-unit shift in the $j$-th feature. Given the Phase~II iterates $\{(\widetilde\theta_k,\widetilde\beta_k)\}$, we propose the plug-in estimator and its PR average
\begin{equation}\label{online_estimator_of_average_me}
    \widetilde\tau_{k}  = \frac{1}{B}\sum_{i=1}^B \left(\nabla_z\Psi_{J_{k-1}}(Z_{i,k}(\widetilde\theta_{k-1}))\right)^{\top}\widetilde\beta_{k-1}\widetilde\theta_{k-1}, \ \ \ \doubleoverline{\tau}_N = \frac{1}{N}\sum_{k=1}^N \widetilde\tau_{k}.
\end{equation}
Similar to the online estimators of $\theta_0$ and $F_0$, the sequence $\doubleoverline{\tau}_1,\doubleoverline{\tau}_2,\cdots$ is also evaluated online, and we also show in \autoref{sec3.4} that it attains $1/\sqrt{N}$-consistency  and satisfies an FCLT. Consequently, the random-scaling inference  applies to it directly: we can simply replace $\doubleoverline{\theta}$ with $\doubleoverline{\tau}$ in (\ref{rs_variance}) and (\ref{rs_band}). We point out that the FCLT-based fast inference is particularly attractive because the marginal effect is of the main concern in most applications, and most importantly, as we show in the appendix, the asymptotic variance of $\doubleoverline{\tau}_N$ is complicated, making plug-in inference computationally demanding.

\section{Two Empirical Illustrations}
\label{sec:empirical}

\subsection{Application to a Synthetic Trade Data}

We first illustrate our online estimation algorithm  using the bilateral-trade panel synthesized from the data set of \citet*{helpman2008estimating} (hereafter HMR). \autoref{sec:emp:data} describes the data and estimation design. \autoref{sec:emp:synth} presents a specification check under a calibrated Probit data-generating process. \autoref{sec:emp:real} applies the estimator to the resampled panel outcomes.

\subsubsection{Data, Variables, and Empirical Design}
\label{sec:emp:data}

The original data set in \citet{helpman2008estimating} covers bilateral trade flows among $158$ countries over $1980$--$1989$, for a total of $n = 248{,}060$ country-pair-year observations. Positive trade is recorded for $44.6\%$ of observations. The outcome is
\begin{equation}
Y_{ij,t} \;=\; \mathbf{1}\{\text{country } i \text{ records positive exports to country } j \text{ in year } t\},
\end{equation}
which is modeled as a binary choice process. 
The normalized regressor is $x_0 = -\ln(\mathrm{distance}_{ij})$, the negative natural logarithm of physical distance between  countries $i$ and $j$. The remaining covariates are nine standard gravity variables---shared land border, island status, landlocked status, common legal origin, common language, colonial ties, currency union (CU), free trade agreement (FTA), and shared religion---together with $157$ exporter, $157$ importer, and $9$ year fixed effects\footnote{We drop 1 exporter dummy, 1 importer dummy, and 1 year dummy for identification.}. The island and landlock indicators take the value one if at least one of the two countries in the pair is, respectively, an island or landlocked, and zero otherwise. 

We focus on the estimated coefficients of the nine gravity variables and the finite-difference average marginal effects on the conditional probability, which are defined as
\begin{equation}
\label{eq:ame}
\mathrm{AME}_j \;=\; \mathbb{E}\!\left[F_0\!\left(Z_{0,-j} + \theta_{0,j}\right) - F_0(Z_{0,-j})\right], \qquad j \in \{\mathrm{CU}, \mathrm{FTA}\},
\end{equation}
where $Z_{0,-j}$ is the index evaluated at $x_j = 0$, $\theta_{0,j}$ is the associated coefficient, and $F_0$ is the link: the Gaussian CDF for Probit, the logistic CDF for Logit, and the estimated link for the semiparametric estimator.   We focus on the marginal impacts of two regressors: CU and FTA. These two marginal impacts are particularly
interesting because they quantify how belonging to the same currency union or the
same free-trade agreement affects the probability of trade between two countries. 

The empirical design comprises two streaming exercises.  In the first exercise, synthetic outcomes are repeatedly generated from a Probit data-generating process whose parameters are calibrated to the full-sample Probit fit. In this setting,  the true coefficient vector and marginal effect are both known, so this exercise serves as a robustness check: when Probit is the true data generating process, we show that online learning precisely recovers all the targets. In the second exercise, the estimator is applied to the panel outcomes directly. We randomly resample data from the original panel, and no parametric assumption is imposed on the link. In this case, we show that the semiparametric estimates differ from Probit or Logit for some targets, indicating potential model misspecification by Probit or Logit as well as gains from semiparametric modeling. 

\paragraph{Implementation Details.} \textit{Number of updates.} For both exercises, we consider a total of $N = 10^7$ 
updates, with a Phase~I global learner of $N_1 = 10^6$ updates and a Phase~II rate-optimal 
stage of $N - N_1 = 9 \times 10^6$ updates. We consider batch size $B = 25$.

\textit{Choices of sieves.} To simplify the implementation, we use degree $3$ (i.e., order $4$, or cubic) B-splines with uniformly spaced knots, and fix the sieve 
dimension at its initial value $J_0 = 50$ throughout updates. Since the true unknown link function $F_0$ is assumed to have approximation error rate $s$ with $s>7/2$\footnote{By approximation error rate $s$, we mean that $\Vert F_0 - g\Vert_{\infty}\leq CJ^{-s}$ for some $g\in\mathcal{S}_J$, see \autoref{appendixC} and \autoref{appendixD}.}, with cubic B-spline sieve dimension fixed at $J_0 = 50$ in this application, the approximation error rate 
is of order  $50^{-4}\approx 1.6\times 10^{-7}$. This  is negligible compared with the width of the confidence intervals reported below, which are of order $10^{-3}$.

\textit{Warm-start updates.} The online warm-start learner of $\theta$ performs kernel-based updates, initialized at the 
origin $\boldsymbol{0}_p$. We use the sixth-order Epanechnikov kernel and choose bandwidth 
$h_k = c_k\cdot k^{-\alpha_h}$ with $\alpha_h = 0.1$ and $c_k = 1.414\,\widehat\sigma_{Z,k}$, 
where $\widehat\sigma_{Z,k}$ is the batch-based sample standard deviation of 
$\widehat Z_{i,k} = x_{0,i,k}+ X_{i,k}^{\top}\widehat\theta_{k-1}$. We choose a constant 
learning rate $\gamma_k = 0.02$ for aggressive search in the first 
stage\footnote{A constant learning rate allows the algorithm to quickly locate a 
neighborhood of the true parameter, at the cost of nonvanishing noise in the iterates. 
Since Phase~I only supplies the warm start in the application, this aggressive search does not affect the 
local refinement and inference in Phase~II.} and start tracking the PR average after 
$N_0 = 5\times 10^5$ updates. The warm-start of sieve coefficients are constructed 
  by
\begin{equation}\label{Y_sieve}
   \widehat{\beta}_{N_1} = \left[\sum_{k=N_0+1}^{N_1}\sum_{i=1}^B
   \widetilde\Psi_{J_0}(\overline Z_{i,k})\widetilde\Psi_{J_0}(\overline Z_{i,k})^{\top}
   \right]^{-1}\left[\sum_{k=N_0+1}^{N_1}\sum_{i=1}^B
   \widetilde\Psi_{J_0}(\overline Z_{i,k})Y_{i,k}\right],
\end{equation}
where $J_0$ is the initial sieve dimension, 
$\overline{Z}_{i,k} = x_{0,i,k}+X_{i,k}^{\top}\overline{\theta}_{k-1}$, and 
$\widetilde\Psi_J(\cdot) = \Psi_J(\frac{2}{\pi}\arctan(\cdot))$. Note that 
$\frac{2}{\pi}\arctan(\cdot)$ maps $\mathbb{R}$ to $(-1,1)$, so spline basis functions 
with bounded support can be used to approximate a function with unbounded support.

\textit{Local refinements by conservative  learning rates.} The local refinement  starts from the last-step PR average of the first stage, i.e. $k=N_1 + 1$.   In this stage, specifying the projection space would require an additional choice on the tuning parameter $C_{\Omega}$. To avoid explicit construction of projection space (and choosing tuning parameter $C_{\Omega}$), we consider Phase II learning rate    $\xi_k = 2\cdot (k + 10^5)^{-0.55}$. In this case, $\xi_1 = 0.0036$, and  consequently the initial updates of second stage are conservative. This makes the second stage refinement   local even   without explicit projection. We also point out that this construction of learning rate only shifts its scale without  changing  its decaying rate.  

\textit{Choice of conditioning matrix $\hat{G}_k$.} To stabilize the updates, we  precondition the gradient with an inverse Hessian 
estimate that is updated recursively across Phase~II, based on the cumulative sample 
Hessian evaluated along the estimated trajectory\footnote{We recursively update the Hessian matrix, update the inverse Hessian matrix every 1000 iterations.  We also impose a small ridge term $10^{-4}\mathbb{I}$ to ensure invertibility.}. Since the parameter estimates are 
consistent, this matrix gain converges almost surely to the inverse population Hessian at 
the population truth; by standard results for averaged stochastic 
approximation with convergent matrix gains \citep{polyak1992acceleration}, the 
first-order asymptotic distribution of the PR average is invariant to such 
preconditioning; in particular, it neither  affects the asymptotic properties of the PR average estimators nor the validity of the random-scaling-based inference.

\begin{figure}[t!]

    \safeincludegraphics[
       width=1.00\linewidth,
      trim= 0cm 0cm 0cm 0cm,
       clip
   ]{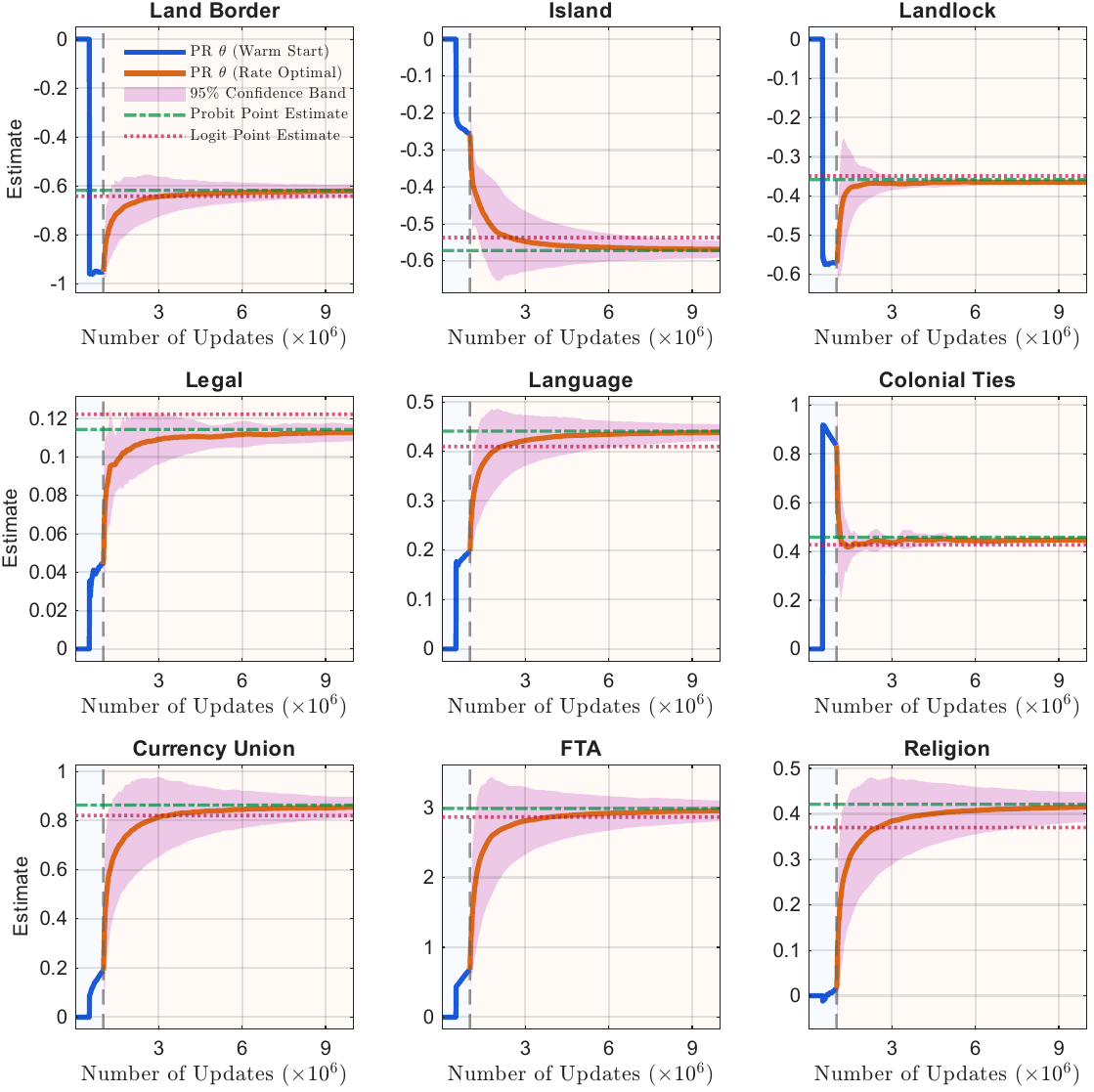}
    
  \vspace{0.4cm}
  
  \footnotesize{Note: The truth coefficients are fixed to the full-sample Probit estimators, so a valid semiparametric online estimation algorithm should recover the Probit benchmarks.}

  \vspace{0.6cm}
  \caption{Synthetic  Data: Parameters}\label{empfig1}
\end{figure}

\begin{figure}[t!]

    \safeincludegraphics[
       width=1.0\linewidth,
      trim=  0cm 0cm 0cm 0cm,
       clip
   ]{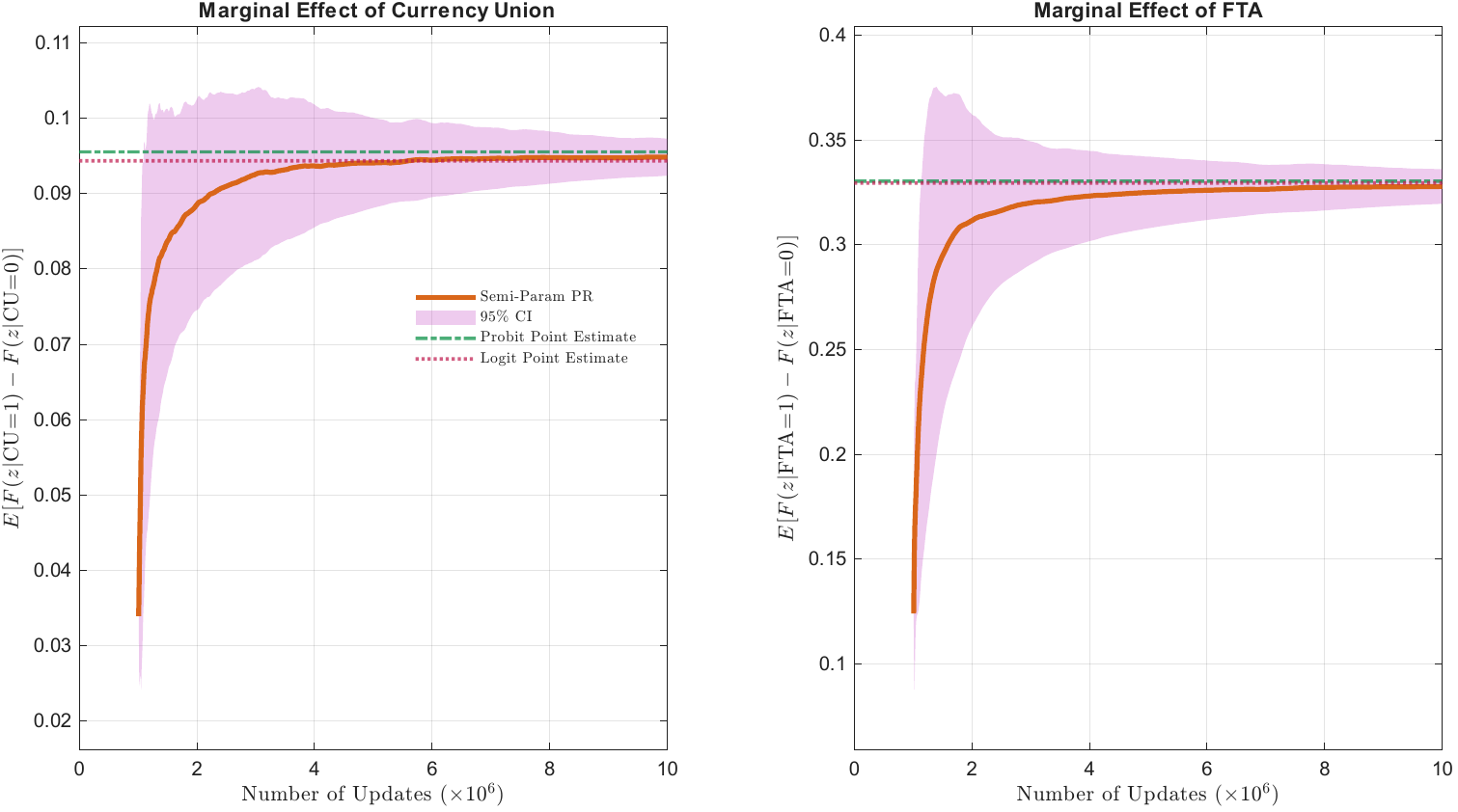}

   \vspace{0.4cm}

  \footnotesize{Note: The true AMEs are   given by (\ref{eq:ame}), where the link is Probit. A valid semiparametric online estimation algorithm should recover the Probit benchmarks.}

   \vspace{0.6cm}
    \caption{ Synthetic Data: Marginal Effects}
  \label{empfig2}
\end{figure}

\subsubsection{Semiparametric Online Learning under a Calibrated Probit DGP}
\label{sec:emp:synth}

Let $(\beta^*_{-1}, \beta_0^*, \beta^*)$ denote the Probit maximum-likelihood estimates on the original data set. At each update $k$, a batch of size $B = 25$ of covariate vectors $(x_{0,k}, X_k)$ is drawn with replacement from the empirical distribution, and the outcome is generated as
\begin{equation}
\label{eq:probit-dgp}
Y_{i,k} \;=\; \mathbf{1}\{\beta^*_{-1} + \beta_0^* x_{0,i,k} + X_{i,k}^\top \beta^* - u_{i,k} \geq 0\}, \qquad u_{i,k} \sim \mathcal{N}(0, 1),
\end{equation}
with $u_{i,k}$ independent across $i$ and $k$.   The true normalized coefficient is $\theta^* = \beta^*/\beta_0^*$, and a correctly specified estimator should recover it.

\autoref{empfig1} reports the PR average coefficient paths across both phases, together with $95\%$ random-scaling confidence bands and the Probit and Logit benchmarks. Phase~I brings the estimates close to their limits; Phase~II delivers smooth convergence with rapidly shrinking bands. Across all nine gravity coefficients, the semiparametric paths track the Probit benchmark, which is the truth, to within $3\%$ of the true value at the end of Phase~II. The agreement is uniformly tight: the semiparametric coefficient on FTA is $2.953$, against a Probit value of $2.988$; on currency union, $0.854$ against $0.864$; on shared land border, $-0.620$ against $-0.618$; and on shared religion, $0.415$ against $0.421$. \autoref{empfig2} reports the corresponding marginal-effect paths. The semiparametric AMEs ($0.0948$ for CU and $0.3278$ for FTA) coincide with the Probit benchmarks ($0.0955$ and $0.3303$) to within sampling noise. The above results together confirm that, when the parametric specification is correct given the covariate distribution, the online procedure recovers both the true coefficient vector and the true marginal effects.

\begin{figure}[t!]

    \safeincludegraphics[
       width=1.00\linewidth,
      trim= 0cm 0cm 0cm 0cm,
       clip
   ]{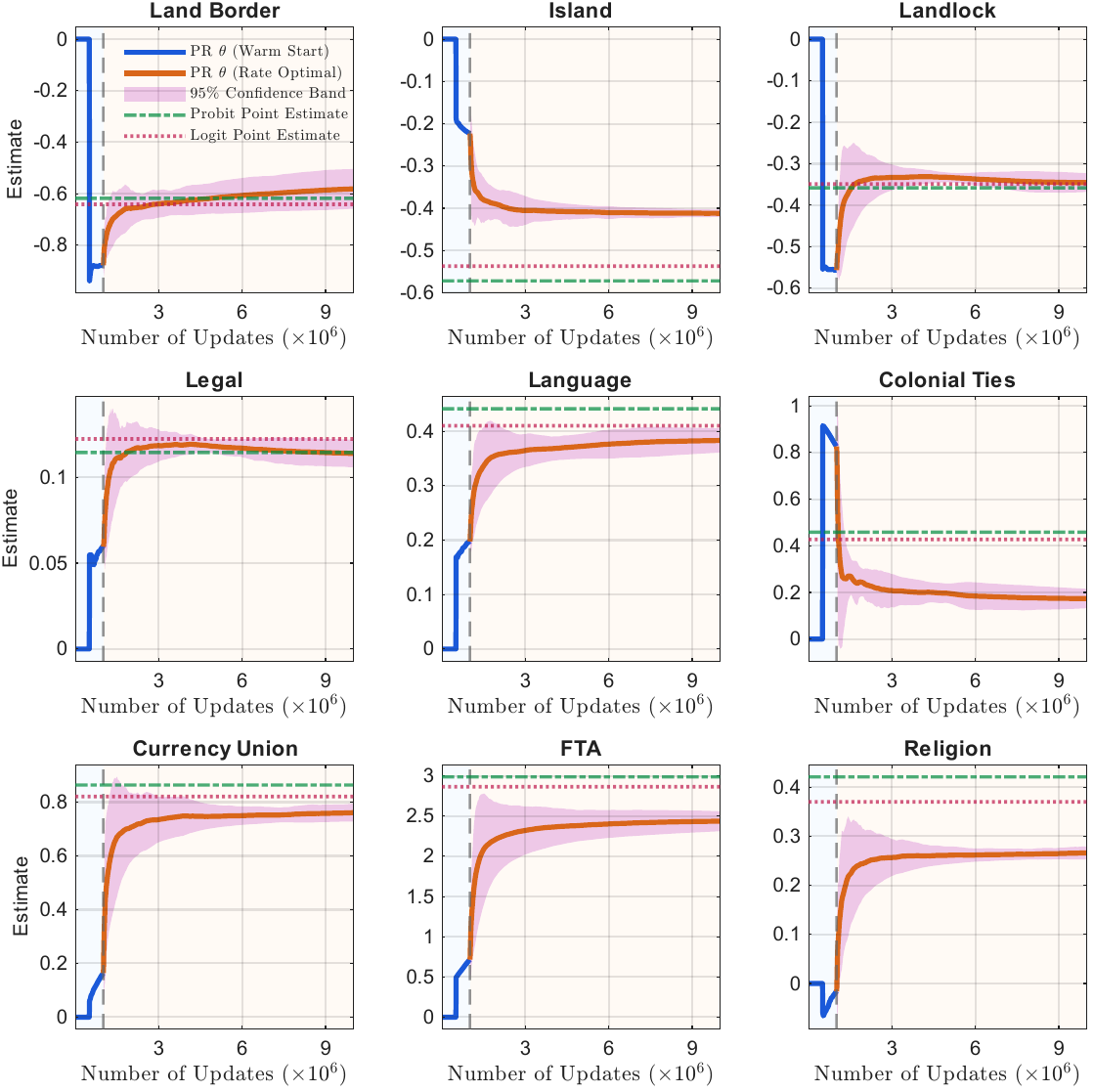}
    \caption{Real Data: Parameters}
  \label{empfig3}
\end{figure}

\begin{figure}[t!]

\centering
    \safeincludegraphics[
       width=1.0\linewidth,
      trim=  0cm 0cm 0cm 0cm,
       clip
   ]{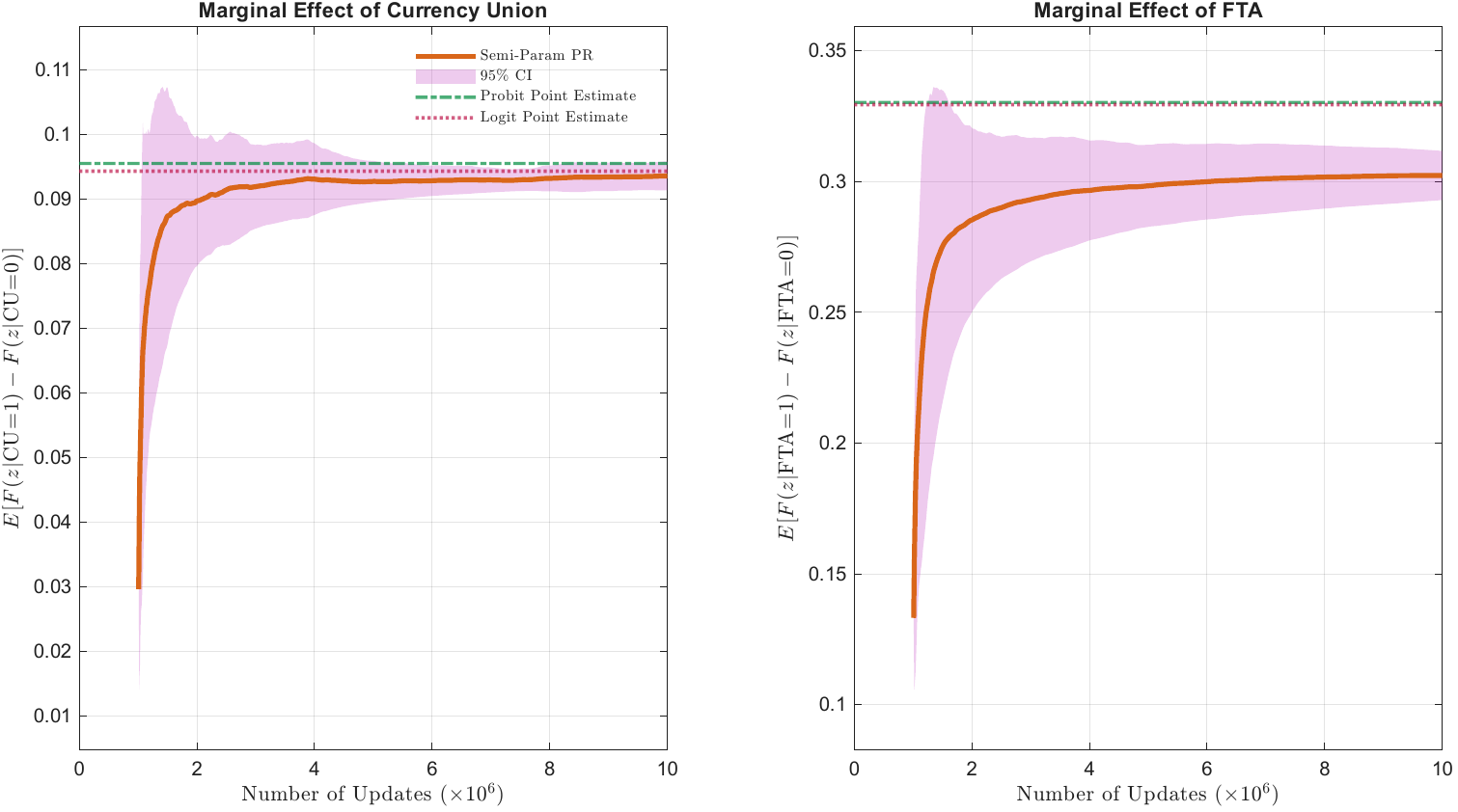}
    \caption{ Real Data: Marginal Effects}
  \label{empfig4}
\end{figure}
\subsubsection{Application to the HMR Panel}
\label{sec:emp:real}

In this section, we treat the 248,060 observation points as the sample space\footnote{In this case, the online confidence band should be interpreted as the quantification of the uncertainty from random subsampling conditioned on the original 248,060 data points. A narrower band indicates a more accurate estimate  of the full-sample counterpart.  }. At each update, a batch of size $B = 25$ is drawn with replacement from the joint empirical distribution of $(x_0, X, Y)$, yielding an i.i.d.\ stream with respect to the empirical measure. In this setting, no parametric assumption is imposed on the link.

\autoref{empfig3} reports the coefficient paths with $95\%$ random-scaling confidence bands. The convergence pattern mirrors that of~\autoref{empfig1}: Phase~I lifts the estimates away from zero, and Phase~II settles them at stable values with shrinking bands. The comparison with the parametric benchmarks, however, splits into two groups. For three of the nine coefficients---shared land border, landlocked status, and common legal origin---the semiparametric estimates agree closely with Probit and Logit, with deviations of less than $5\%$. For the remaining six---island status, common language, colonial ties, currency union, free trade agreement, and shared religion---the semiparametric estimates display sizeable departures from parametric ones. The largest proportional departure is on colonial ties, where the semiparametric estimate is $0.174$ against a Probit value of $0.458$---an attenuation of more than $60\%$. Sizeable departures also appear on shared religion ($0.266$ versus $0.421$, a $37\%$ attenuation), island status ($-0.412$ versus $-0.572$, a $28\%$ attenuation), and FTA ($2.437$ versus $2.988$, an $18\%$ attenuation). In all six cases, the semiparametric coefficient is smaller in magnitude than the Probit benchmark, with the Logit value lying between them.

\autoref{empfig4} reports the marginal-effect paths for currency union and FTA. For currency union, the three estimators agree closely on the probability scale: the semiparametric AME is $0.094$, compared with $0.096$ for Probit and $0.094$ for Logit. For FTA, however, the semiparametric AME of $0.302$ is significantly distinguishable from the Probit AME of $0.330$, a departure of roughly $3$ percentage points on the probability scale that lies well outside the bounds of sampling variability. While this gap may seem small in magnitude, it is economically meaningful: relative to the unconditional probability of a trade relationship of $44.6\%$, misspecification of this order can materially distort the implied policy conclusions.

Taken together, the two exercises demonstrate that the online semiparametric estimator is well-calibrated under correct parametric specification and is flexible enough to detect link misspecification when it arises. The second exercise reveals a systematic pattern: coefficient-level deviations from Probit are largely concentrated on the institutional and informational gravity variables---colonial ties, language, religion, currency union, FTA---rather than on the geographic ones---land border, landlock---although the island estimate by the Gaussian link is way off from the one estimated by our flexible sieve link as well. The probability-scale marginal effect of FTA inherits part of this distortion: a researcher relying on Probit to quantify the effect of FTA membership on the probability of trading would overstate the average partial derivative by roughly $3 $ percentage points, a discrepancy that is small in absolute terms but economically non-negligible at the trade-policy margins where such estimates are typically applied.

\subsection{Application to the Stream of Stanford Open Policing Data}
\label{sec:emp:nc}

This section illustrates our estimation approach using the North Carolina State Patrol  
records from the Stanford Open Policing Project \citep{pierson2020large}\footnote{The data are available at https://openpolicing.stanford.edu/data/.}, a
large administrative stream of traffic stops that is online in nature. The remainder of this section is arranged as follows: \autoref{sec:emp:nc:data}
describes the data and estimation design, and \autoref{sec:emp:nc:real} applies
the estimator to the stop outcomes directly, processing the records in their
recorded arrival/online order.

\subsubsection{Data, Variables, and Empirical Strategy}
\label{sec:emp:nc:data}

The Stanford Open Policing Project provides rich information on  nationwide traffic stops. We choose the data set from  North Carolina State Patrol
between December 1999 and December 2015. The original data contains 20,286,645 stop-level observations. After data cleaning\footnote{For data cleaning, we drop the observations with missing values and also stops in which the driver's race is  Asian or  Pacific islander/unknown/other. These account for 1.11\%, 1.20\%, and 0.79\% of full data set respectively. We further drop the observations 
recorded as checkpoints and as driving-while-impaired investigations, whose proportions are  1.18\% and 0.95\% respectively. We drop these small-proportion categories because they are too sparse to support separately estimated coefficients. },  we are left with
$n = 19{,}234{,}514$ data points, which essentially constitute a large online data stream.

The outcome of interest is whether a search is conducted following the stop by the officer, i.e.,
\begin{equation}
Y_{i} \;=\; \mathbf{1}\{\text{a search is conducted at stop } i\},
\end{equation}
which is modeled as a binary choice process. The regressors we consider include standardized  driver's age (mean zero and variance 1), indicator for female drivers (male driver as benchmark), indicators for White and Black drivers (Hispanic driver as benchmark), and the recorded reasons for stop including violations of speed limit, vehicle
regulatory, seat-belt, vehicle-equipment, investigation, safe-movement, and
other-motor-vehicle violations---with stop-light/sign violation as the omitted
category. This yields a total of $p=11$ regressors.
The normalized regressor $x_0$ is chosen as the negative standardized age; see the following rolling-window analysis for more discussion.

In this section, we mainly focus on the AME of  three demographic regressors, White, Black, and Female, on search probabilities. These
effects quantify how driver's demographics may shift the probability that a stop
results in a search, relative to  Hispanic, male, and stop-light violation baselines. The AME of these  three dummy regressors can be computed following \eqref{eq:ame}.  Also, following the previous section, when we evaluate the estimation results, we also benchmark against Probit and Logit link estimates based on the full sample.  However, we emphasize that our semiparametric estimation and the corresponding inference in this example are  online, so that each new recorded data point is processed only once.

\paragraph{Data Stability and  Estimation Sample} Traffic stops can be plausibly regarded as  independent across cases: individual stops are initiated by different officers encountering different 
drivers at dispersed locations across the state, and the search decision in one stop is not likely to affect the decision in another. However, what is far less 
plausible over  the 1999–2015 sample window is  distributional stability. The composition of 
stops and the propensity to search respond to slowly moving institutional forces including
changes in enforcement priorities and departmental policy, personnel turnover, budget 
cycles, and the documented statewide decline in search rates over the 2000s. As a result, the joint distribution of stop characteristics and search outcomes may drift  across the 
sample even while individual stops remain essentially independent draws from the 
prevailing regime.

\begin{figure}[h]
    \centering
    \safeincludegraphics[width=0.86\linewidth]{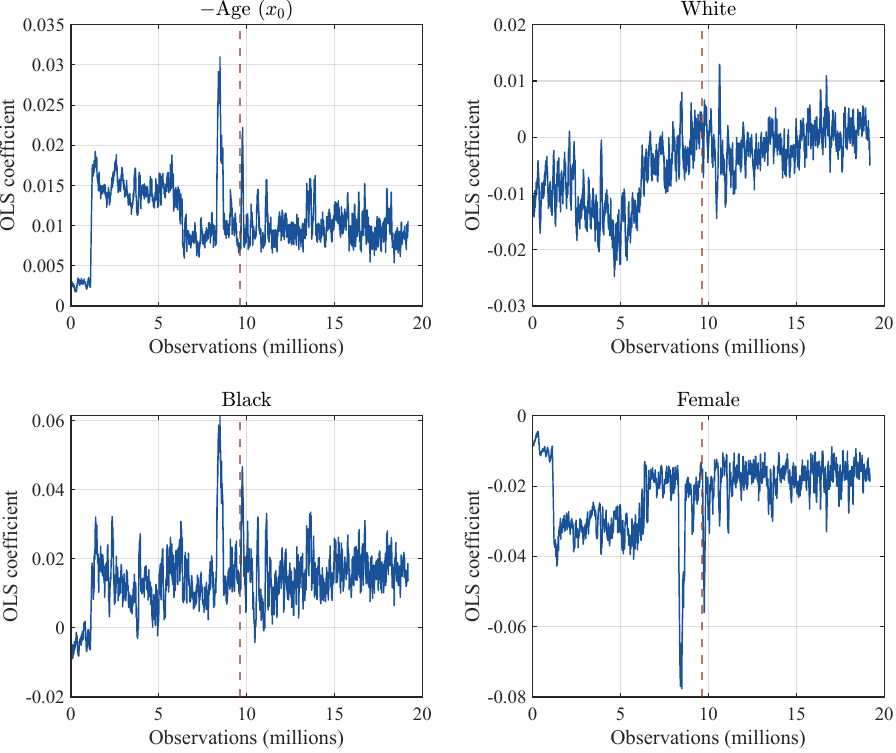}
    \caption{OLS Rolling-Window Estimation Results}
    \label{fig:rolling}
\end{figure}

To provide a more detailed diagnostic analysis for the potential structural instability associated with the original data, we  report in \autoref{fig:rolling}  the   rolling window least squares
estimation results for the North Carolina sample in recorded order.  Starting from the initial observation in 1999, we iteratively conduct OLS estimation based on regressors mentioned above including a constant term, using
a sequence of rolling windows each with $10^5$ observations. The window moves forward by $10^3$ observations each time, so that the adjacent windows overlap  heavily by 99\% observations. This allows us to   detect  structural shifts that are even smooth across time. In \autoref{fig:rolling}, we report the estimation results for four regressors: $-\text{age}$, white, black, and female.

\autoref{fig:rolling} clearly indicates that the first half of the sample exhibits pronounced structural instability. First of all, a  discrete structural change happens to coefficients of negative age, black and female at around 1.1 million-th data point. Such shift has strong persistence and lifts up the magnitude of the above three coefficients. Another surge occurs to the three coefficients at around 8.4 million-th point, again inflating the coefficient magnitude, although with weak persistence. The coefficient of white features smooth structural breaks: it drifts gradually downward to about -0.025  by the 4.6 million-th observation and then climbs back to fluctuate around zero. Given all the observations, if we use the data set including a large proportion of the first half, the identical distribution condition may be violated, and the corresponding estimators and the confidence intervals may lose their statistical properties. 

However, when we look at the second half of the data set, which starts from June~2009, we can 
clearly see that the rolling-window estimates behave differently: no coefficient path 
displays a sustained level shift. This result is  consistent with local heterogeneity 
in the composition of stops, rather than with regime change. 

In light of this evidence, and to align the estimation sample with the i.i.d.\ sampling 
framework we imposed before, we conduct online estimation and inference only on the 
second half of the sample. The first half 
serves solely as a training pool, supplying the kernel-based warm start and the bootstrap 
burn-in described below, which does not contribute to the reported estimators or confidence intervals. 
Note that in this case, the warm start is  computed under a mixture of regimes, so its population 
target may differ from the truth governing the post-2009 stream. This is immaterial 
for our procedure provided that the shift between the two regimes is moderate. All reported estimates and confidence intervals 
are driven exclusively by the post-June-2009 stream. The full-sample Probit/Logit estimators are also conducted based on the second half of the data. Finally, note that the rolling-window estimation results of the coefficient of $-\text{age}$ is strictly positive across all windows. This also supports the use of negative standardized age as $x_0$. 

\begin{figure}
    \centering
    \includegraphics[width=1.0\linewidth, trim=0cm 0cm 0cm 0cm,
        clip]{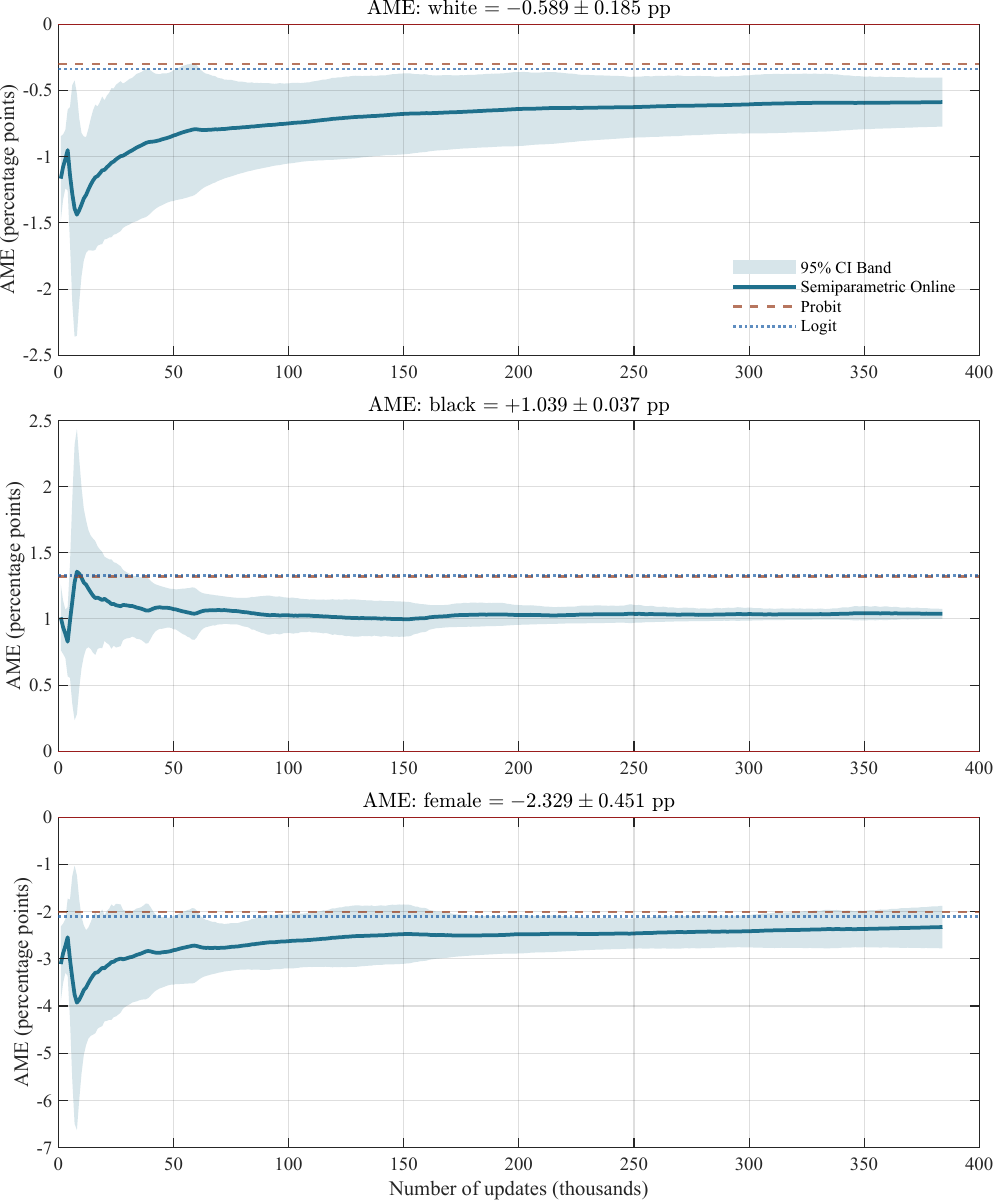}
    \vspace{0.3cm}

    \begin{minipage}{\linewidth}
    \raggedright\footnotesize
    Note: Warm start estimation takes around 2 minutes, while online
    local refinement takes 17 seconds. The applications were implemented
    in MATLAB R2024b and executed on a laptop with an Intel Core Ultra 9
    185H processor (16 cores, 2.5\,GHz) and 64\,GB of RAM.
    \end{minipage}

    \vspace{0.3cm}
    \caption{Semiparametric Online Estimation Results for Average
    Marginal Effects: Stanford Policing Data}
    \label{fig:stanford policing}
\end{figure}

\paragraph{Implementation of the Algorithm.} In this application, we also choose batch size $B = 25$. However, we slightly change the way of the warm-start construction. In particular, we fix the first
half of the records and iteratively draw random samples (of batch size 25) from this offline data set with replacement. We then apply the online learning algorithms to these randomly generated batches with $N = 3\times 10^6$, $N_1= 10^6$ and $N_0=5\times 10^5$.   After the construction of the warm start, we shift to the online estimation with the second half of the data, where the rate-optimal phase  proceeds  with the remaining $9.6$ million records strictly in their recorded order and each data point is processed only once. This roughly leads to 384 thousand updates. Given the smaller number of total updates, we choose a larger second phase learning rate $\xi_k = 5(k+10^5)^{-0.55}$.  Other choices of tuning parameters are identical to those in the previous section. 

\subsubsection{Online Estimation and Results for AME}
\label{sec:emp:nc:real}
The warm starts take roughly 2 minutes, while the subsequent online refinements take only 17 seconds -- processing roughly half a million of data points per second. 
\autoref{fig:stanford policing} reports the PR averages of the AME paths on the 
probability scale from Phase~II. The effects are economically meaningful against the low 
baseline unconditional search rate of $2.32\%$ (second half of the data): relative to a baseline (hispanic male) 
driver, black drivers face an average marginal effect of $+1.04$ percentage points, 
women $-2.33$ percentage points, and white drivers $-0.59$ percentage points; all three 
random-scaling 95\% confidence intervals exclude zero at the last update. Here we caution 
against a causal reading: differences in search rates alone do not directly identify 
discrimination. As emphasized in the policing literature 
\citep{simoiu2017problem,pierson2020large}, such benchmark comparisons are subject to 
infra-marginality and may reflect non-discriminatory factors. Consequently, we  should treat the 
estimated effects as descriptive.

The central finding of \autoref{fig:stanford policing} is that the semiparametric 
estimates differ systematically from the parametric benchmarks: for two of the three average 
marginal effects, the Probit and Logit point estimates lie outside the semiparametric 95\% 
confidence intervals. The parametric models place the black--hispanic search gap at 
roughly $1.32$ percentage points, overstating the semiparametric estimate of $1.04$ 
($\pm 0.037$) by over $25$ percent; for the white--hispanic gap the distortion runs the 
other way, with parametric estimates of about $-0.3$ percentage points amounting to 
roughly half of the semiparametric $-0.59$ ($\pm 0.19$).  

The above discussion implies that link misspecification   does 
not merely rescale the estimated disparities, it compresses one racial gap while 
inflating another.   Since 
all specifications are estimated based on the same post-June-2009 sample with the same 
covariates, these discrepancies should be attributed to the shape of the link function 
itself, the object that parametric practice fixes by assumption. This underscores the value of flexible link estimation for measuring disparities  and, at this scale of sample size, our online procedure delivers it at essentially no additional computational cost.

\section{Statistical Properties}\label{section3}

\subsection{Global Almost Sure Consistency of Phase I Algorithms}\label{theta_warm_start}

Define $\Delta\theta =\theta -\theta_0$ for any $\theta\in\mathbb{R}^p$. The following lemmas shows that our Phase I online algorithms generate paths that converge to the true values $(\theta_0,F_0)$ almost surely.

\begin{lemma}\label{F0_warm_start} Let $\overline{\theta}_N$ be the PR averages of the Phase I algorithm for $\theta_0$. If \autoref{condition1}--\ref{condition4} in \autoref{appendixB} hold, then: $\Vert \Delta\overline{\theta}_N\Vert_2\rightarrow_{a.s.} 0$, and see \autoref{theorem1}  in \autoref{appendixB} for its almost sure rate of convergence.
\end{lemma}

\begin{lemma}\label{new_lemma-F} Let $\overline{F}_N$ be the PR averages of the Phase I algorithm for $F_0$. If  \autoref{condition5}--\ref{rate} in \autoref{appendixC} hold, then: $\Vert \overline{F}_N  -  F_0 \Vert_{\infty} \rightarrow_{a.s.}0$, and see \autoref{theorem6} in \autoref{appendixC} for its almost sure rate of convergence.
\end{lemma}

Two points are worth mentioning. First, the above convergence results are global in the sense that the results are valid regardless of the choice of the initial starting points. It is known to be very difficult to find a global solution to non-convex optimization problem over a high-dimensional parameter space \citep{chen2026optimization}. Consequently, the above global consistency result is not only theoretically important but also practically crucial for semiparametric index models with many covariates. Second, the global convergence is almost surely so is path-wise. This guarantees the validity of the warm starts for almost all paths of data, which is crucial for online estimation with single passage of data.

\subsection{Rate-Optimal Statistical Properties of Phase II Algorithms}

\subsubsection{Automatic Orthogonalization and Rate-Optimality}\label{sec3.3}
This section formally studies the asymptotic properties of the local online learner $\doubleoverline{\omega}_N$.   Recall that $Z_0 = x_0 + X^{\top}\theta_0$.  
Define $\mu_0(\theta, z) = \mathbb{E}(X|x_0 +X^{\top}\theta =z)$, 
\begin{equation}
    \mathbb{M}_{\theta}  = \mathbb{E}\left[\left(\nabla_z F_0(Z_0)\right)^2(X- \mu_0(\theta_0, Z_0))(X- \mu_0(\theta_0, Z_0))^{\top}\right],
\end{equation}
\begin{equation}
    P_{J} = \left(\mathbb{E}\left[\Psi_J(Z_0)\nabla_z F_0(Z_0) X^{\top}\right]\right)\left(\mathbb{E}\left[\left(\nabla_z F_0(Z_0)\right)^2XX^{\top}\right]\right)^{-1},
\end{equation}
and 
\begin{equation}
     \mathbb{M}_{\beta, J}  = \mathbb{E}\left[\left(\Psi_J(Z_0)- P_J\nabla_z F_0(Z_0)X\right)\left(\Psi_J(Z_0)- P_J\nabla_z F_0(Z_0)X\right)^{\top}\right].
\end{equation}
We consider the learning rate $\xi_k$ and the sieve dimension $J_k$ to be chosen as \begin{equation}
    \xi_k = \xi_0k^{-\alpha_{\xi}}, \quad J_k =[J_0k^{\alpha_{\dagger}}], 
\end{equation} where $\xi_0,  \alpha_{\xi}$ and $\alpha_{\dagger}$ are all positive constants, and $J_0$ is a positive integer. Here $\frac{1}{2}<\alpha_{\xi}<1$ controls the decreasing rate of the learning rate, while $\alpha_{\dagger}$ controls the increasing speed of sieve dimension $J_k$ with the number of updates $k$. Detailed conditions on $\alpha_{\xi}$ and $\alpha_{\dagger}$ can be found in \autoref{appendixD}. 

Recall that $\Delta\theta = \theta - \theta_0$ for any $\theta\in\mathbb{R}^p$. Similarly, for any $\beta\in \mathbb{R}^J$, define $\Delta\beta = \beta - \beta_{J,0}$, where $\beta_{J,0}$ is the pseudo true sieve coefficient at dimension $J$. We have the following lemma that provides the asymptotic linear representation of the PR average estimators, with almost sure convergence rates for remainder terms.

\begin{lemma}\label{theorem_NLS_update}
Let  \autoref{condition5}--\ref{condition7} in  \autoref{appendixC} and
\autoref{cond11}--\ref{new_rate}  in \autoref{appendixD} hold.  Then: 
    \begin{align}
        \Delta\doubleoverline{\theta}_N & = \frac{1}{N}\sum_{k=1}^N \frac{1}{B}\sum_{i=1}^B \varepsilon_{i,k} \mathbb{M}_{\theta}^{-1} \nabla_z F_0(Z_{0,i,k})\left(X_{i,k} - \mu_0(\theta_0, Z_{0,i,k})\right) + o\left(N^{-\frac{1}{2}}\right), \qquad  \rm{a.s.}
    \end{align}
    and 
    \begin{align}
     \Delta\doubleoverline{\beta}_N  & =   
     \frac{1}{N}\sum_{k=1}^{N}R_{J_k, J_N}^{\top}\mathbb{M}^{-1}_{\beta, J_k}\frac{1}{B}\sum_{i=1}^B \varepsilon_{i,k}\left(\Psi_{J_k}\left(Z_{0,i,k}\right) - P_{J_k}\nabla_{z}F_0(Z_{0,i,k})X_{i,k}\right)\nonumber \\
     & + \frac{1}{N}\sum_{k=1}^{N}\left(R_{J_k, J_N}^{\top}\beta_{J_k,0}- \beta_{J_N, 0}\right) +o\left(N^{-\frac{1}{2}}\right), \qquad  \rm{a.s.} 
     \end{align}

\end{lemma}

\autoref{theorem_NLS_update} is a key result of the paper. Two features are worth emphasizing. 
First, the influence functions in the leading terms of both $\Delta\doubleoverline{\theta}_N$ and $\Delta\doubleoverline{\beta}_N$ are \emph{automatically orthogonalized}. That is, orthogonality emerges from the joint optimization over $(\theta,\beta)$, without requiring any explicit orthogonalization step. This is important both conceptually and computationally. In particular, constructing an orthogonal score for $\theta$ would require estimating $\mu_0(\theta_0,z)$, a vector of $p$ unknown functions, which becomes increasingly costly as $p$ grows. The joint estimation procedure therefore avoids a potentially high-dimensional nuisance estimation step.

The second feature is that  the  sieve coefficient $\Delta\doubleoverline{\beta}_N$ displays a clear online learning structure. For example, $\Delta\doubleoverline{\beta}_N$ contains the bias term $
\frac{1}{N}\sum_{k=1}^{N}\left(R_{J_k,J_N}^{\top}\beta_{J_k,0}-\beta_{J_N,0}\right),
$
which  reflects the evolution of the re-embedded sieve target as the dimension $J_k$ increases. Moreover, the presence of the transformation $R_{J_k,J_N}^{\top}$ in the leading term of $\Delta\doubleoverline{\beta}_N$ has a useful implication for function estimation. For the pointwise estimator $\doubleoverline{F}_N(z)=\Psi_{J_N}(z)^{\top}\doubleoverline{\beta}_N$, the leading stochastic term can be rewritten as 
\[
\frac{1}{N}\sum_{k=1}^{N}\Psi_{J_k}(z)^{\top}\mathbb{M}^{-1}_{\beta,J_k}\frac{1}{B}\sum_{i=1}^B \varepsilon_{i,k}
\left(\Psi_{J_k}(Z_{0,i,k}) - P_{J_k}\nabla_z F_0(Z_{0,i,k})X_{i,k}\right),
\]
where we use the fact that $\Psi_{J_N}(z)^{\top}R_{J_k,J_N}^{\top}=\Psi_{J_k}(z)^{\top}$. This shows that each update contributes to the variance through the \emph{local sieve representation} at level $J_k$, rather than the final sieve $J_N$. This highlights that the variance accumulation is naturally aligned with the evolving sieve complexity.

\begin{remark}
    Note that when $\sigma^2(x_0, X) = \sigma^2_0$, where $\sigma_0^2$ is a positive constant, the online estimator $\doubleoverline{\theta}_N$ attains semiparametric efficiency \citep{ai2003efficient}. However, when there is heteroscedasticity, semiparametric efficiency no longer holds. In this case, we can reweight the update to attain efficiency. 
\end{remark}

\autoref{theorem_NLS_update} yields a root-$N$ law for $\doubleoverline{\theta}_N$ and the explicit moving-sieve upper bound for $\doubleoverline{F}_N$ stated next. Define 
\[
\mathbb{V}_{\theta} := \frac{1}{B}\mathbb{E}\left[\varepsilon^2(\nabla_z F_0(Z_0))^2 (X - \mu_0(\theta_0, Z_0))(X - \mu_0(\theta_0, Z_0))^{\top}\right],~~~\Omega_{\theta}:=\mathbb{M}_{\theta}^{-1}\mathbb{V}_{\theta}\mathbb{M}_{\theta}^{-1}.
\]
\begin{theorem}\label{corollary1}
    Let all conditions in \autoref{theorem_NLS_update} hold and
    $\Omega_\theta$ be positive definite. Then we have:
    \begin{enumerate}
        \item Almost surely, for every
        fixed $a\in\mathbb R^p\setminus\{0\}$,
        \[
        \limsup_{N\to\infty}\sqrt{\frac{N}{2\log\log N}}
        \frac{a^\top\Delta\doubleoverline\theta_N}
             {\sqrt{a^\top\Omega_\theta a}}=1,
        \qquad
        \liminf_{N\to\infty}\sqrt{\frac{N}{2\log\log N}}
        \frac{a^\top\Delta\doubleoverline\theta_N}
             {\sqrt{a^\top\Omega_\theta a}}=-1.
        \]
        \item Almost surely,
        \begin{align}\label{phaseII_uniform_link_rate}
        \Vert\doubleoverline F_N-F_0\Vert_{\infty}
        =O\!\left(
        J_N^{1/2}N^{-1/2}(\log N)^{1/2}
        +\frac1N\sum_{k=1}^N J_k^{-s}
        \right),
        \end{align}
     where $J_k^{-s}$, $s>0$, denotes the sup-norm approximation error rate of using $J_k$ dimensional sieve for $F_0$ (see \autoref{appendixD}). When $J_k\asymp k^{\alpha_\dagger}$, the approximation error 
        term is $O(N^{-s\alpha_\dagger})$ if
        $s\alpha_\dagger<1$, $O(N^{-1}\log N)$ if
        $s\alpha_\dagger=1$, and $O(N^{-1})$ if
        $s\alpha_\dagger>1$.
    \end{enumerate} 
\end{theorem}

\subsubsection{FCLTs for $\theta$ and Policy Functionals}\label{sec3.4}

\autoref{theorem_NLS_update} also gives the following FCLT of $\doubleoverline{\theta}_N$. 

\begin{theorem}\label{corollary2}
    Let all conditions in \autoref{theorem_NLS_update} hold. Then:
  \begin{align}          \left(N\Omega_{\theta}\right)^{-\frac{1}{2}} [Nr]\Delta\doubleoverline{\theta}_{[Nr]}\Longrightarrow \mathbb{W}_p(r), \ r\in[0,1],
        \end{align} 
        where $\mathbb{W}_p(r)$ is a $p$-dimensional standard Wiener process. Consequently,
      \begin{align}           \Omega_{\theta}^{-\frac{1}{2}}  N^{\frac{1}{2}}\Delta\doubleoverline{\theta}_{N}\rightarrow_d \mathcal{N}(0,\mathbb{I}_p).
        \end{align} 
\end{theorem}

Based on the results of \autoref{theorem_NLS_update}, we can also develop an FCLT result for the average marginal effect $\doubleoverline{\tau}_{N} $ proposed in \autoref{sec2.5}. To formally state the theorem, define   \begin{align*}
    \Upsilon_{i,k}   &  =  \varepsilon_{i,k} \mathbb{E}[\nabla_{zz} F_0(Z_0)\theta_0X^{\top}] \mathbb{M}_{\theta}^{-1} \nabla_z F_0(Z_{0,i,k})\left(X_{i,k} - \mu_0(\theta_0, Z_{0,i,k})\right)  \\
    & +  \varepsilon_{i,k}\mathbb{E}[\nabla_z \Psi_{J_k}(Z_0)^{\top}]\mathbb{M}^{-1}_{\beta, J_k}\left(\Psi_{J_k}\left(Z_{0,i,k}\right) - P_{J_k}\nabla_{z}F_0(Z_{0,i,k})X_{i,k}\right)\theta_0\\
    & +  \varepsilon_{i,k}\mathbb{E}[\nabla_z F_0(Z_0)]\mathbb{M}_{\theta}^{-1} \nabla_z F_0(Z_{0,i,k})\left(X_{i,k} - \mu_0(\theta_0, Z_{0,i,k})\right)\\
     & +  \left(\nabla_z F_0(Z_{0,i,k}) - \mathbb{E}[\nabla_z F_0(Z_0)]\right)\theta_0.
\end{align*} In \autoref{appendixE}, we show that $\frac{1}{N}\sum_{k=1}^N\frac{1}{B}\sum_{i=1}^B\Upsilon_{i,k}$ is the influence function of $\Delta\doubleoverline{\tau}_N$. There we also impose   a mild condition  that $\mathbb{E}[\Upsilon_{k}\Upsilon_{k}^{\top}]\rightarrow V_{\Upsilon}$ for some positive definite matrix $V_{\Upsilon}$. We have the following theorem. 

\begin{theorem}\label{FCLT_average_marginal_effect}
    Let all conditions in \autoref{theorem_NLS_update} and \autoref{cond_ame} in \autoref{appendixE} hold. Then 
    \begin{align}
        (NV_{\Upsilon})^{-\frac{1}{2}}[Nr]\Delta\doubleoverline{\tau}_{[Nr]}   \Longrightarrow \mathbb{W}_p(r), \ \ r\in[0, 1].
    \end{align}
\end{theorem}

\begin{remark}
  By using FCLT-based random-scaling inference, we avoid estimating complicated asymptotic variance matrix $V_{\tau}$, which is also empirically attractive. 
    We also point that, to save space and illustrate the main idea, we only report the FCLT for AME of continuous regressors; for formal results of FCLT of AME for binary regressors as we used in applications, see \autoref{ame_binary}.
\end{remark}

\section{Monte Carlo Experiments}\label{sec:MC}

This section presents Monte Carlo studies to evaluate the performances of our online learning algorithm. 
 The simulation data is generated according to the following binary-choice DGP:
\begin{align}\label{monte_carlo_model}
Y = \boldsymbol{1}\!\left(x_{0} + X^{\top}\theta_0 - 0.5\Vert \theta_0 \Vert_2 \cdot u \ge 0\right),
\end{align}
where $Y$ is the observed response, $(x_0, X^{\top})^{\top}$ is the regressor vector. Unless otherwise specified,  we consider the following setups. $x_0\sim \mathcal{N}(0,1)$, $x_j\sim \mathcal{N}(0,1)$ or $x_j\sim \frac{1}{\sqrt{5}}(\frac{1}{2}\mathcal{N}(2,1) + \frac{1}{2}\mathcal{N}(-2,1))$ for $1\leq j \leq 0.4p$, and $x_j \sim \text{Rademacher}$ for $0.4p + 1\leq j \leq p$. 
The latent shock $u$ is independent of $(x_0,X)$. We consider two shock distributions: (i) $u\sim \text{Cauchy}$ and (ii) a skewed-normal--type distribution generated as $u=(v_1+|v_2|)/\sqrt{2}, \  v_1,v_2 \sim \text{i.i.d. }\mathcal{N}(0,1).$ We specify $\theta_0$ as $\theta_0 = (\theta_{10}^{\top},-\theta_{20}^{\top},\theta_{30}^{\top},-\theta_{40}^{\top})^{\top}$, where $\theta_{10},\theta_{20}\in \mathbb{R}^{0.2p}$,  $\theta_{30},\theta_{40}\in \mathbb{R}^{0.3p}$, and  their arguments decrease linearly from 1 to 0.  
We consider two dimensions $p\in\{50, 100\}$, two batch sizes $B\in\{25,50\}$, and fixing the total number of online updates at $N=  10^7$. Unless otherwise specified,  for all Monte Carlo results reported below, the details of the implementations of the online algorithm are identical to that in \autoref{sec:empirical}. 

In \autoref{sec6.1} and \autoref{sec6.2} below, the reported simulation results on our online estimators are averaged over 200 independent Monte Carlo replications. Since it is very time-consuming to estimate a high-dimensional monotone index model using any offline procedure with a large sample size, the simulation results in \autoref{sec6.3} below are averaged over 20 independent Monte Carlo replications only. 

\begin{table}[t]
\centering
\caption{Simulation Results (Normal + Rademacher Design)}
\label{table:normal_parameter}

\footnotesize
\renewcommand{\arraystretch}{1.12}
\setlength{\tabcolsep}{4pt}

\begin{tabular}{l c *{8}{c}}
\toprule
Statistic & Target
& \multicolumn{4}{c}{Cauchy Errors} 
& \multicolumn{4}{c}{Skewed Normal Errors} \\
\cmidrule(lr){3-6} \cmidrule(lr){7-10}

& 
& \multicolumn{2}{c}{$p=50$} & \multicolumn{2}{c}{$p=100$}
& \multicolumn{2}{c}{$p=50$} & \multicolumn{2}{c}{$p=100$} \\
\cmidrule(lr){3-4} \cmidrule(lr){5-6}
\cmidrule(lr){7-8} \cmidrule(lr){9-10}

& 
& $B=25$ & $B=50$ & $B=25$ & $B=50$
& $B=25$ & $B=50$ & $B=25$ & $B=50$ \\

\midrule

\multirow{2}{*}{Bias}
& $\theta_0(1)$
& 0.0001 & 0.0000 & 0.0003 & 0.0000
& 0.0000 & 0.0000 & 0.0001 & 0.0000 \\
& Avg
& 0.0000 & 0.0000 & 0.0001 & 0.0000
& 0.0000 & 0.0000 & 0.0000 & 0.0000 \\[12pt]

\multirow{2}{*}{RMSE}
& $\theta_0(1)$
& 0.0007 & 0.0004 & 0.0014 & 0.0006
& 0.0004 & 0.0003 & 0.0006 & 0.0003 \\
& Avg
& 0.0005 & 0.0003 & 0.0010 & 0.0005
& 0.0003 & 0.0002 & 0.0005 & 0.0003 \\[12pt]

\multirow{2}{*}{CR}
& $\theta_0(1)$
& 0.9400 & 0.9450 & 0.9250 & 0.9200
& 0.9550 & 0.9200 & 0.9550 & 0.9550 \\
& Avg
& 0.9461 & 0.9424 & 0.9482 & 0.9403
& 0.9454 & 0.9407 & 0.9499 & 0.9412 \\[12pt]

\multirow{2}{*}{CI Length}
& $\theta_0(1)$
& 0.0034 & 0.0020 & 0.0060 & 0.0030
& 0.0019 & 0.0013 & 0.0030 & 0.0018 \\
& Avg
& 0.0028 & 0.0017 & 0.0050 & 0.0025
& 0.0016 & 0.0010 & 0.0025 & 0.0015 \\

\bottomrule
\end{tabular}

\vspace{4pt}
\footnotesize
\parbox{\linewidth}{
\noindent
\textit{Note:}  Total number of updates is $N = 10^7$, and the effective sample size is  $n = B\times N$. This also applies to \autoref{table:mixture_parameter} and \autoref{table_ave_combined}.  Results are Monte Carlo averages over 200 independent simulation replications. 
All online estimators/algorithms start from zero. The nominal coverage rate is 0.95. 
CR is constructed via random scaling. 
$\theta_0(1)$ refers to the first component of $\theta_0$, whose value is 1. 
Avg reports the average of each summary statistic across all parameters.
}
\normalsize

\end{table}

\begin{table}[t]
\centering
\caption{Simulation Results (Mixture Normal + Rademacher Features)}
\label{table:mixture_parameter}

\footnotesize
\renewcommand{\arraystretch}{1.12}
\setlength{\tabcolsep}{4pt}

\begin{tabular}{l c *{8}{c}}
\toprule
Statistic & Target
& \multicolumn{4}{c}{Cauchy Errors} 
& \multicolumn{4}{c}{Skewed Normal Errors} \\
\cmidrule(lr){3-6} \cmidrule(lr){7-10}

& 
& \multicolumn{2}{c}{$p=50$} & \multicolumn{2}{c}{$p=100$}
& \multicolumn{2}{c}{$p=50$} & \multicolumn{2}{c}{$p=100$} \\
\cmidrule(lr){3-4} \cmidrule(lr){5-6}
\cmidrule(lr){7-8} \cmidrule(lr){9-10}

& 
& $B=25$ & $B=50$ & $B=25$ & $B=50$
& $B=25$ & $B=50$ & $B=25$ & $B=50$ \\

\midrule

\multirow{2}{*}{Bias}
& $\theta_0(1)$
& 0.0000 & 0.0000 & 0.0003 & 0.0000
& 0.0000 & 0.0000 & 0.0000 & 0.0001 \\
& Avg
& 0.0000 & 0.0000 & 0.0002 & 0.0000
& 0.0000 & 0.0000 & 0.0000 & 0.0000 \\[12pt]

\multirow{2}{*}{RMSE}
& $\theta_0(1)$
& 0.0007 & 0.0004 & 0.0013 & 0.0006
& 0.0004 & 0.0003 & 0.0006 & 0.0004 \\
& Avg
& 0.0006 & 0.0003 & 0.0010 & 0.0005
& 0.0003 & 0.0002 & 0.0005 & 0.0003 \\[12pt]

\multirow{2}{*}{CR}
& $\theta_0(1)$
& 0.9550 & 0.9350 & 0.9500 & 0.9350
& 0.9550 & 0.9350 & 0.9600 & 0.9300 \\
& Avg
& 0.9483 & 0.9388 & 0.9515 & 0.9458
& 0.9481 & 0.9379 & 0.9529 & 0.9406 \\[12pt]

\multirow{2}{*}{CI Length}
& $\theta_0(1)$
& 0.0036 & 0.0021 & 0.0063 & 0.0030
& 0.0020 & 0.0013 & 0.0030 & 0.0019 \\
& Avg
& 0.0028 & 0.0017 & 0.0051 & 0.0025
& 0.0016 & 0.0010 & 0.0025 & 0.0015 \\

\bottomrule
\end{tabular}

\end{table}

\subsection{Online Estimation and Inference of $\theta_0$}\label{sec6.1}

This section focuses on Phase II's estimation and inference of parameter $\theta_0$. We report the bias, root mean squared error (RMSE), confidence interval coverage rate (CR), and average confidence interval length in \autoref{table:normal_parameter} and \autoref{table:mixture_parameter}.    The results demonstrate that the procedure achieves negligible bias, small RMSE, and reliable inference uniformly across all configurations, with coverage rates tightly centered around the nominal $0.95$ level.
One clear pattern is the effect of the mini-batch size $B$. Increasing $B$ from 25 to 50 leads to uniform improvements in precision--- the RMSE declines, and confidence interval  becomes noticeably shorter. This is consistent with the effective sample size interpretation: larger batches reduce the variance of the stochastic updates and stabilize the online trajectory.

The distribution of the shock $u$ affects the level of difficulty but not the qualitative conclusions. Heavy-tailed Cauchy errors lead to larger RMSEs and longer confidence intervals relative to skewed-normal designs, reflecting the presence of extreme realizations. Nevertheless, the estimator remains stable, and coverage stays close to nominal even in these challenging settings.

Increasing the dimension from $p=50$ to $p=100$ results in only modest increases in RMSE and interval length, indicating that the procedure scales well with dimensionality. Importantly, inference remains well-calibrated across dimensions, with no systematic deterioration in coverage.

Finally, the results are robust to the distribution of regressors. Whether regressors are Gaussian, mixture-normal, or discrete, the same qualitative patterns emerge: accurate estimation, stable inference, and systematic gains from larger batch sizes.

\begin{table}[tbp]
\centering
\caption{Average Marginal Effects: Normal and Rademacher Regressors}
\label{table_ave_combined}

\renewcommand{\arraystretch}{1.25}
\footnotesize
\setlength{\tabcolsep}{4pt}

\begin{tabular}{l c *{8}{c}}
\toprule
& 
& \multicolumn{4}{c}{Cauchy Errors}
& \multicolumn{4}{c}{Skewed Normal Errors} \\
\cmidrule(lr){3-6} \cmidrule(lr){7-10}

& 
& \multicolumn{2}{c}{$p = 50$} & \multicolumn{2}{c}{$p = 100$}
& \multicolumn{2}{c}{$p = 50$} & \multicolumn{2}{c}{$p = 100$} \\
\cmidrule(lr){3-4} \cmidrule(lr){5-6}
\cmidrule(lr){7-8} \cmidrule(lr){9-10}

&
& $B = 25$ & $B = 50$
& $B = 25$ & $B = 50$
& $B = 25$ & $B = 50$
& $B = 25$ & $B = 50$ \\

\midrule

\multirow{2}{*}{Bias}
& $\tau_0(1)$  
& 0.0000 & 0.0000 & 0.0000 & 0.0000
& 0.0000 & 0.0000 & 0.0000 & 0.0000 \\
& Avg        
& 0.0000 & 0.0000 & 0.0000 & 0.0000
& 0.0000 & 0.0000 & 0.0000 & 0.0000 \\[12pt]

\multirow{2}{*}{RMSE}
& $\tau_0(1)$  
& 0.0000 & 0.0000 & 0.0000 & 0.0000
& 0.0000 & 0.0000 & 0.0000 & 0.0000 \\
& Avg       
& 0.0000 & 0.0000 & 0.0000 & 0.0000
& 0.0000 & 0.0000 & 0.0000 & 0.0000 \\[12pt]

\multirow{2}{*}{CR}
& $\tau_0(1)$   
& 0.9550 & 0.9300 & 0.9500 & 0.9350
& 0.9550 & 0.9300 & 0.9500 & 0.9350 \\
& Avg     
& 0.9486 & 0.9393 & 0.9514 & 0.9428
& 0.9486 & 0.9393 & 0.9514 & 0.9428 \\[12pt]

\multirow{2}{*}{CI Length}
& $\tau_0(1)$  
& 0.0002 & 0.0001 & 0.0002 & 0.0001
& 0.0002 & 0.0001 & 0.0002 & 0.0001 \\
& Avg        
& 0.0002 & 0.0001 & 0.0002 & 0.0001
& 0.0002 & 0.0001 & 0.0002 & 0.0001 \\

\bottomrule
\end{tabular}

\vspace{4pt}
\footnotesize
\parbox{\linewidth}{
\noindent
\textit{Note:} Results are Monte Carlo averages over 200 simulation replications. 
$\tau_0(1)$ denotes the first component of the average marginal effect. 
Avg reports averages across all components.
}
\normalsize

\end{table}

\subsection{Estimation and Inference of Marginal Effects}\label{sec6.2}

We now examine the finite-sample performance of the online estimators for average marginal effects defined in (\ref{average_marginal_effect}). To save space, we only report the simulation results for Normal and Rademacher regressors in \autoref{table_ave_combined}.   We can see that the bias and RMSE of $\doubleoverline{\tau}_N$ are uniformly small and close to 0. Moreover, the coverage rate of the random-scaling confidence interval is close to the nominal $0.95$ level uniformly across  dimensions, batch sizes, and shock distributions, confirming the validity of online inference for the average marginal effect.

In \autoref{appendixA}, we also report additional results on estimation and inference for point-wise marginal effect.

\subsection{Computation Efficiency:  Insensitivity to Choice of Warm Start}\label{sec6.3}

 Although our algorithm is designed for online estimation, it can also be applied to offline estimation with a fixed sample. In particular, we can draw random subsample from the fixed data set with replacement, and process these subsamples using our algorithm as if they were online. This section provides simulation results that demonstrate significant computational gains of our online Phase II estimation over the offline estimation using the same sieve NLS criterion function and the same warm start.
 
As an illustration, consider the DGP \eqref{monte_carlo_model} with $u\sim \mathrm{Cauchy}$ and $p = 100$. We consider two sample sizes: $n = 5\times 10^5$ and $n = 10^6$, and correspondingly, two sieve dimensions $J_n=30$ and $J_n=40$. Two estimation methods are compared: full-sample estimation using MATLAB optimization package \texttt{lsqnonlin} and local online NLS with multiple passes. Precisely, the full-sample estimator solves the joint sieve NLS criterion via \texttt{lsqnonlin} using the trust-region-reflective 
algorithm with an analytical sparse Jacobian for the joint $(\theta, \beta)$ block; the solver is capped at 100 iterations and 300 function evaluations, and the function, step, and first-order optimality tolerances are all set to $10^{-6}$. For the local online sieve NLS, we randomly reshuffle the full sample, and sequentially draw data points from the reshuffled data with fixed batch size $B$ and feed them to the online algorithm. A single pass is defined as one complete traversal of the dataset, consisting of $n/B$ sequential updates with $B=25$; \textit{multiple passes} feature multiple data reshuffling and complete traversal. In the simulation, we consider 10 passes, with PR averaging over the last 5 passes only.

\begin{table}[!h]
\centering
\caption{Full-sample versus online estimation, Online Warm Start}
\footnotesize
\label{tab:fullsample-vs-online-large}
\begin{tabular}{lcccccc}
\hline\hline
& \multicolumn{3}{c}{$n = 5 \times 10^5$, $J_n = 30$} & \multicolumn{3}{c}{$n = 10^6$, $J_n = 40$} \\
\cline{2-4} \cline{5-7}
Estimator & Bias & RMSE & Time (min) & Bias & RMSE & Time (min) \\
\hline
Online Phase I (warm start)          & 0.1864 & 0.2180 & 0.1256 & 0.1868 & 0.2180 & 0.1349 \\
\\
Full-sample NLS (\texttt{lsqnonlin}) & 0.0064 & 0.0193 & 9.5290  & 0.0043 & 0.0140 & 24.623 \\
\\
Online Phase II (multi-pass)         & 0.0022 & 0.0144 & 0.1217 & 0.0017 & 0.0109 & 0.2676 \\
\hline
\multicolumn{7}{l}{\textit{Speedup ratio (Full-sample time / Online time)}} \\
  &  &  & 78.3  &  &  & 92.0  \\
\hline\hline
\multicolumn{7}{p{0.95\textwidth}}{\footnotesize \textit{Notes.} Results are averaged 
over 20 Monte Carlo replications with $p = 100$ for $\theta_0$, $J_n$ is sieve order for $F_0$. The online Phase I updates with $N_0 = 1\times 10^5$ and $N_1 = 2\times 10^5$ (see \autoref{sec:emp:data} for definitions of $N_0$ and $N_1$). Both methods are warm-started 
from the same Phase~I output, so reported times exclude the shared Phase~I cost. Full-sample estimator solves 
the joint sieve NLS criterion via \texttt{lsqnonlin}. The online 
estimator uses local joint sieve NLS refinement, multi-pass uses 10 passes through the data, with PR averaging 
over the last 5 passes only.}
\end{tabular}
\end{table}

\begin{table}[!h]
\centering
\caption{Full-sample versus online estimation, SMCO warm start}
\footnotesize
\label{tab:fullsample-vs-online-smco}
\begin{tabular}{lcccccc}
\hline\hline
& \multicolumn{3}{c}{$n = 5 \times 10^5$, $J_n = 30$} & \multicolumn{3}{c}{$n = 10^6$, $J_n = 40$} \\
\cline{2-4} \cline{5-7}
Estimator & Bias & RMSE & Time (min) & Bias & RMSE & Time (min) \\
\hline
SMCO Phase I (warm start)            & 0.1907 & 0.2324 & 3.7990  & 0.2021 & 0.2470 & 3.4410  \\
\\
Full-sample NLS (\texttt{lsqnonlin}) & 0.0033 & 0.0192 & 9.6050  & 0.0031 & 0.0138 & 23.897 \\
\\
Online Phase II (multi-pass)         & 0.0023 & 0.0146 & 0.1440  & 0.0021 & 0.0104 & 0.2620  \\
\hline
\multicolumn{7}{l}{\textit{Speedup ratio (Full-sample time / Online time)}} \\
  &  &  & 66.5 &  &  & 91.3 \\
\hline\hline
\multicolumn{7}{p{0.95\textwidth}}{\footnotesize \textit{Notes.} Results are averaged 
over 20 Monte Carlo replications with $p = 100$ for $\theta_0$, $J_n$ is sieve order for $F_0$. Both methods are warm-started 
from the same SMCO Phase~I output, so reported times exclude the shared Phase~I cost. 
SMCO Phase~I uses a subsample of size $10^4$ to maximize the negative sieve NLS criteria jointly over $(\theta, \beta)$ on the 
search box $[-2,2]^{p} \times [-3,3]^{J_n +1}$. SMCO is run with 10 Sobol-initialized 
starts; each start uses \texttt{iter\_max}~=~100 refinement iterations preceded by 
\texttt{iter\_boost}~=~200 boost iterations, forward partial differences, a fixed 
$5\%$ bounds buffer, the run-max accumulator turned on, and convergence tolerance 
$10^{-8}$. The best objective value across the 10 starts is taken as the warm start.}
\end{tabular}
\end{table}

Two warm-start algorithms are considered: the online learner proposed in \autoref{sec2.1} and \ref{sec2.2}, and a generic warm start SMCO \citep{chen2026optimization}.   Simulation results are reported in \autoref{tab:fullsample-vs-online-large} and \autoref{tab:fullsample-vs-online-smco}. Regardless of the choice of warm start, an attractive pattern emerges: initiating from the same warm start, full-sample estimation takes around 10 minutes to solve the optimization problem when sample size is $5\times 10^5$ and 24 minutes when sample size is $1\times 10^6$, whereas our online Phase II requires only about 8 and 16 seconds.  The computational advantage is more compelling for larger sample sizes. More importantly, the bias and RMSE of package-based full-sample optimization are slightly larger than those of online estimators. This suggests that the full-sample objective function may be \textit{under-optimized}  given the tuning parameters chosen for \texttt{lsqnonlin} package. To close this gap, additional time will be necessary, making full-sample estimation even less favorable.  

\section{Conclusion and Extension}\label{conclusion}

This paper develops a two-phase online estimation and inference algorithm for semiparametric monotone index models with many covariates (some of which are discrete).
The online warm-start phase uses a globally stable update rule that consistently learns the finite-dimensional index parameter from arbitrary initialization, and also learn a sup-norm consistent link function via online sieve least squares using the warm start index as a consistent generated regressor. It then performs locally joint sieve nonlinear least squares updates in the second online phase. The PR averages of the Phase II updates attain the optimal convergence rates, are automatically orthogonalized and satisfy the FCLT. 
A key byproduct of the procedure is a sequence of parameter updates---learning trajectories---that can be used for online inference via random scaling with essentially no additional nonparametric estimation burden. The same trajectories also support online estimation and inference for policy-relevant functionals that depend on both the parametric and nonparametric components. 
 Monte Carlo experiments show good finite-sample performance with coverage rates close to nominal. Two empirical illustrations demonstrate feasibility in a high-dimensional setting while leaving the link function unspecified. The procedure is ideal for streaming environments where continuously re-estimating on the full sample (or even storing it) is infeasible, but is also a computationally fast alternative to full-sample offline sieve M estimation and inference methods. Our method delivers a practical semiparametric toolkit for real-time estimation.

\subsection{Extensions with Generic Warm Starts}\label{extension}

The framework naturally invites many extensions, and we discuss a few here. \autoref{section7.1} shows that our online estimation and inference framework can be extended to the case where the unknown function $F_0$ is approximated using deep neural network (DNN).  \autoref{sec7.1} and \ref{sec7.2} show that our framework can be extended to various semiparametric models. In particular, \autoref{sec7.1} sketches an online sample-selection model in which both the selection and  outcome equations are learned semiparametrically. \autoref{sec7.2} outlines an extension to multi-alternative choice models, connecting the framework to interpretable multi-index neural architectures. In the extensions, the Phase I update can be replaced with any online updates that deliver consistent warm starts, such as SMCO \citep{chen2026optimization}. 

\subsubsection{Approximating Unknown Link Using Monotonic DNN}\label{section7.1}

This section considers the extension where the monotonic link function $F_0$ is approximated using DNN instead of linear sieves. We follow \citet*{sartor2025advancing}'s construction so that the DNN-based estimator of $F_0$ is monotonic, which preserves the shape. Then we replace the spline-based approximation of $F_0$ with the DNN-based estimator in the NLS criterion, and use Phase II update to jointly estimate the DNN weights and unknown parameter $\theta_0$, based on which the average marginal effects can be estimated and online inference using random scaling can be conducted.

\begin{figure}
    
   \includegraphics[width=1\linewidth, trim={10 10 10 10}, clip]{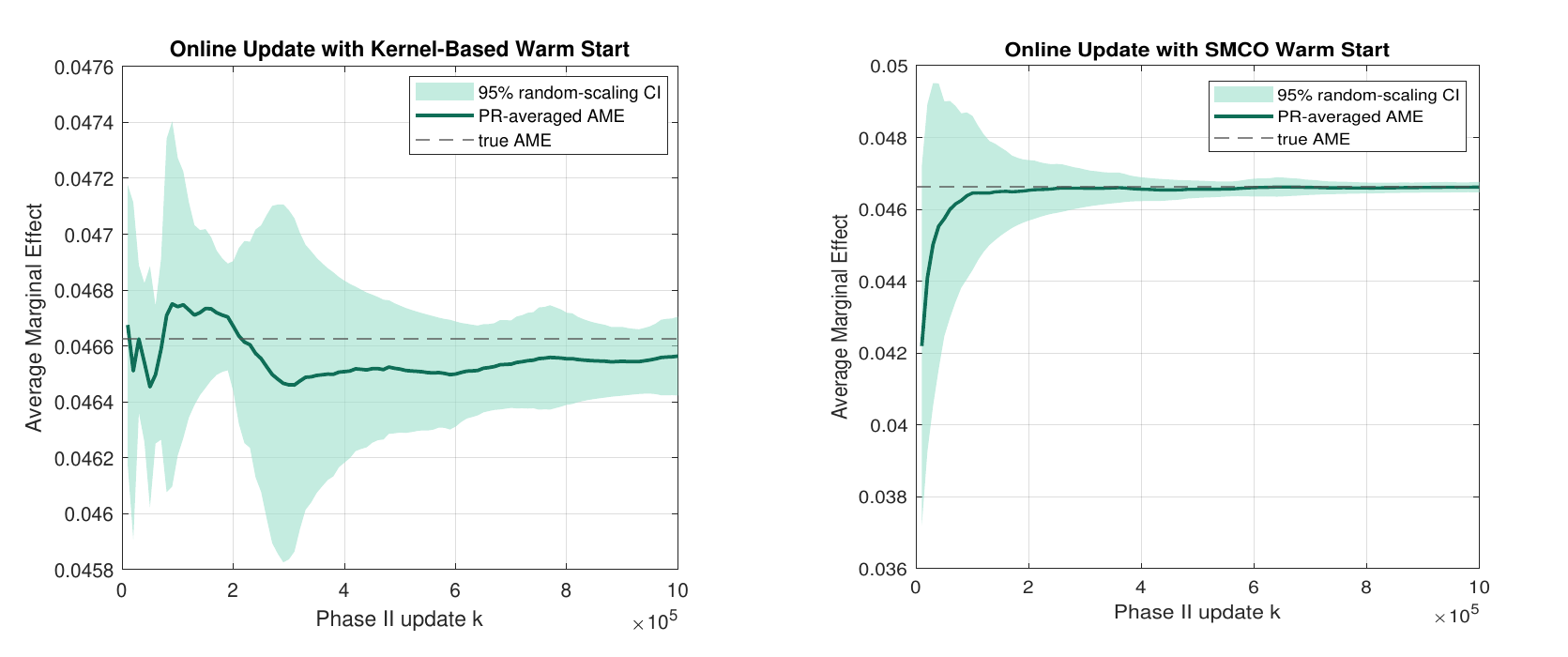}

    \caption{The PR Average Trajectory of Average Marginal Effect from One Monte Carlo Replicate, $p=500$, DNN sieve}
    \label{fig:dnn-link}
\end{figure}

\begin{table}[t!]
    \centering
    \caption{Online Phase II AME estimator under different warm starts, $p=500$, Monotone DNN sieve.}
    \label{tab:warmstart}
    \begin{tabular}{lcccc}
        \hline\hline
        Warm start & Bias & RMSE & Coverage & Avg.\ CI length \\
        \hline
        Online Kernel & 0.000006 & 0.000077 & 0.94 & 0.000371 \\
        SMCO         & 0.000021 & 0.000109 & 0.97 & 0.000701 \\
        \hline\hline
    \end{tabular}

    \begin{minipage}{0.9\linewidth}
    \footnotesize
    \vspace{0.5em}
    \footnotesize{\textit{Note:} Note: Reported online Phase II results are averaged across 200 Monte Carlo replications, on the bias, root mean squared error (RMSE), empirical coverage rate, the average length of the 95\% confidence interval for the average marginal effect (AME). Confidence intervals are constructed using random-scaling inference. ``Online Kernel'' and ``SMCO'' denote the two Phase I warm-start procedures used to initialize the online Phase II updates. The table shows that online Phase II random-scaling inference attains coverage close to the nominal level regardless of which warm-start procedure is used. The table is not to compare the two warm starts against each other since their effective warm-start sample sizes differ and thus their bias and RMSE are not directly comparable.} 
    \end{minipage}
\end{table}

The DGP follows (\ref{monte_carlo_model}) in \autoref{sec:MC}, except that here we consider a higher dimension $p=500$. $x_0, x_1, \cdots, x_{40}$ are standard normal, and $x_{41},\cdots, x_{100}$ are Rademacher. Their corresponding coefficients are constructed similarly as in $p=100$ case of \autoref{sec:MC}. The remaining 400 regressors are all standard normally distributed, and each of their coefficients is equal to $0.05$ so that the total contribution of the additional 400 regressors to index variance is equal to 1. A deep neural network (DNN) with depth of 4 and width of 16 for each layer is used to approximate unknown monotone link function $F_0$.\footnote{Under \citet{sartor2025advancing}'s construction, 4 is the minimum depth of a DNN to guarantee its shape-preserving.}
The estimation and inference target is the average marginal effect. For illustration, we focus on the average marginal effect of the first regressor, which is given by $\mathbb{E}(\nabla_z F(Z_0))$ due to the fact that $\theta_{0,1} = 1$.

We consider two methods to construct the warm starts of $\theta_0$ and DNN weights. The first method resembles our spline-based two-phase online learning algorithm: we use the kernel-based online algorithm to search for a warm start of $\theta_0$, then use the estimated $\theta_0$ to construct the index estimator, based on which the weights can be updated based on the NLS criterion. Finally, we jointly refine the weights and $\theta$ based on NLS criterion. For this method, we choose $N_0 = 1\times 10^5$ and $N_1 = 2\times 10^5$.  The second method uses SMCO  \citep{chen2026optimization} to search for joint warm starts of $\theta_0$ and DNN weights, where the warm-start sample size is chosen as $10^4$. For both methods, we choose batch size $B = 25$, and consider $10^6$ joint updates. Finally, 200 independent Monte Carlo replicates are performed.

\autoref{tab:warmstart} reports the simulation results across the 200 Monte Carlo replicates, and \autoref{fig:dnn-link} plots the PR average path of the AME and its random-scaling-based confidence band for one replicate. We can see that the confidence band constructed from the random scaling approach has good coverage rate close to the normal one, regardless of the choices of the warm start.

\subsubsection{Sample Selection}\label{sec7.1}

Let \(D\in\{0,1\}\) indicate whether the outcome is observed, and suppose that the
researcher observes \((D,DY^*,x_{10},X_1,x_{20}, X_2)\).  A semiparametric online sample-selection model is
\begin{align}
  D &= 1\{S^*\ge 0\}, \quad
  \mathrm{Pr}(D=1|x_{10},X_1)=F_{10}(x_{10}+X_1^{\top}\theta_{10}), \label{eq:selection}\\
  Y^* &= F_{20}(x_{20}+X_2^{\top}\theta_{20})+\varepsilon, \label{eq:outcome}\\
  \E[\varepsilon\mid x_{10}, X_1, x_{20}, X_2,D=1] &= \lambda_0\{F_{10}(x_{10}+X_1^{\top}\theta_{10})\}. \label{eq:control}
\end{align}
Here \(F_{10}\) is an unknown monotone selection link, \(F_{20}\) is the unknown outcome link,
and \(\lambda_0\) is a control function capturing selection on unobservables.  The
selected conditional mean is therefore
\[
  \E[Y^*\mid x_{10},X_1, x_{20}, X_2, D=1]
  =
  F_{20}(x_{20}+X_2^{\top}\theta_{20})+\lambda_0\{F_{10}(x_{10}+X_1^{\top}\theta_{10})\}.
\]
This nests the no-selection-bias case by setting \(\lambda_0\equiv0\).  It also separates
the economics of outcome production, through \(F_{20}(x_{20}+X_2^{\top}\theta_{20})\), from the economics
of observability or participation, through \(F_{10}(x_{10}+X_1^{\top}\theta_{10})\). For identification we need the usual economic exclusion or support condition: some component of \((x_{10},X_1)\) should move selection without entering the outcome index except
through the control function.

The sample-selection model has a direct choice interpretation.  The selection
indicator \(D\) may represent labor-force participation, market entry, survey response,
purchase, attention, or consideration.  The outcome equation is observed only after that
economic choice.  The control function \(\lambda_0\{F_{10}(x_{10}+X_1^{\top}\theta_{10})\}\) captures the
unobserved factors that link the first-stage selection choice with the selected outcome.
Thus selection correction is not merely a missing-data device.  It is an online two-stage
model of economic choice and realized outcomes.

\paragraph{Online Algorithm}

At batch \(k\), with observations \(W_{i,k}=(D_{i,k},D_{i,k}Y^*_{i,k},x_{10,i,k},X_{1,i,k},x_{20,i,k}, X_{2,i,k})\),
we propose to use the following three online learners.

\textbf{Step 1: selection learner}.
Estimate \((\theta_{10},F_{10})\) from the binary outcome \(D\).  One can use the same
two-phase monotone-index learners as in the baseline paper.

\textbf{Step 2: selected outcome learner}.
For selected observations, define the generated control of $W_{i,k}$ as 
$
  \widehat q_{i,k-1}=\widehat F_{10, k-1}(x_{10,i,k}+X_{1,i,k}^{\top}\widehat\theta_{1,k-1})
$, where $\widehat F_{10, k-1}$ and $\widehat\theta_{1,k-1}$ are online estimators of $F_{10}$ and $\theta_{10}$ in period $k-1$ from the first step algorithm. Then estimate \((\theta_{20},F_{20},\lambda_0)\)
from the selected outcome loss
\[
  \mathcal{L}^{y}_{J,L}(\theta,\beta,\delta;W_{i,k})
  =
  D_{i,k}
  \Big[
    Y^*_{i,k}
    -\Psi_J(x_{20,i,k}+X_{2,i,k}^{\top}\theta_2)^{\top}\beta
    -\Psi_L(\widehat q_{i,k-1})^{\top}\delta
  \Big]^2.
\]
The online stochastic-gradient update is
\[
  \omega_k^y
  =
  \mathcal{R}_{k}^{\top}\omega_{k-1}^y
  -
  \frac{\xi_k}{B}\sum_{i=1}^B
  \nabla_{\omega^y}\mathcal{L}^{y}_{J_k,L_k}
  (\mathcal{R}_{k}^{\top}\omega_{k-1}^y;W_{i,k}),
\]
where \(\omega^y=(\theta^{\top},\beta^{\top},\delta^{\top})^{\top}\), and \(\mathcal{R}_k\) embeds the old coefficient
vector into the enlarged sieve spaces for \(F_0\) and \(\lambda_0\), which is similarly defined as that in the baseline algorithm.

\textbf{Step 3: random-scaling inference}.
The trajectory of PR averages can be used for random-scaling inference for
\(\theta_{10}\), \(\theta_{20}\), average marginal effects, and selection-corrected means.  

\subsubsection{Extension to Multi-Alternative Choice}\label{sec7.2}

Let an agent choose \(l\in\{0,1,\ldots,L\}\) from utilities
\[
  U_l = V_l(x_{0l}, X_l;\theta_j)+u_l,
  \qquad
  V_l(x_{0l}+X_l;\theta_l)=x_{0l}+X_l^{\top}\theta_{l0}.
\]
Let $X$ collects all features vectors from the $L$ utilities, the observable object is the vector of choice probabilities
\[
  P_l(X)=\mathrm{Pr}\{U_l\ge U_{l^{\prime}}\ \hbox{for all }0\leq l^{\prime}\leq L\mid X\},
  \qquad \sum_{l=0}^{L} P_l(X)=1.
\]
A direct extension of the monotone index model is the vector monotone-index model
\[
  P_l(X)=G_{l,0}(V_0,\ldots,V_L),
  \qquad l=0,\ldots,L,
\]
where for each $l$, \(G_{l,0}:\R^{L+1}\to\Delta^L\) is an unknown vector-valued link.  The loss is
 quasi-likelihood:
\[
  \mathcal{L}(\theta,\beta;W)
  =
  -\sum_{l=0}^L 1\{Y=l\}\log \left[\Psi_J(x_{00}+X_0^{\top}\theta_0, \cdots, x_{L0}+X_L^{\top}\theta_0)^{\top}\beta_l \right].
\]
To preserve the economic restrictions of random utility, one can represent
\(G_{0l}\) as the gradient of a convex surplus or choice-potential function:
\[
  G_{l,0}(v)=\frac{\partial \mathbf{U}_L(v)}{\partial v_j},
  \qquad
  \mathbf{U}_L(v)
  =
  \sum_{r=1}^R a_r\phi(c_r+b_r'v),
  \qquad a_r\ge0.
\]
Convexity of $\mathbf{U}_L$ helps encode the shape restrictions associated with stochastic
choice, while the gradient automatically gives nonnegative own-utility responses after
normalization.  The online update is again a single-hidden-layer neural network, now
with the network output constrained to be a choice-probability vector.

A further extension that allows for attention or consideration can be written as
\[
  A_l = 1\{V_{l}(x_{0l}+X_l^{\top}\theta_{l0})+u_l\ge0\},
  \qquad
  Y=\argmax_{l:A_l=1}\{x_{0l}+X_l^{\top}\theta_{0l}+u_l\}.
\]
The first layer learns which alternatives enter the consideration set; the second layer
learns choice probabilities conditional on consideration.  This connects the sample
selection model above with McFadden-style economic choice: observability, attention,
and final choice are all economic decisions that can be learned online with interpretable
semiparametric neural-index components.


\bibliography{overall_1_combined}

\newpage 

\appendix

\crefname{section}{Appendix}{Appendices}
\Crefname{section}{Appendix}{Appendices}
\crefname{subsection}{Appendix}{Appendices}
\Crefname{subsection}{Appendix}{Appendices}
\crefname{subsubsection}{Appendix}{Appendices}
\Crefname{subsubsection}{Appendix}{Appendices}

\begin{center}
    \textbf{\LARGE Online Appendix}
\end{center}

Throughout this Online Appendix,  unless otherwise specified,  $C$  denotes a generic positive constant whose value may change even within a line. When we append ``a.s.'' to an inequality, we mean that the inequality holds eventually almost surely; that is, for almost all paths of the data, the inequality holds for all sufficiently large indices. For any vector $A = (a_1, \ldots, a_n)^{\top}$, $\|A\|_2 = \sqrt{\sum_{i=1}^n a_i^2}$ denotes its Euclidean norm. For any matrix $A = (a_{ij})_{m \times n}$, $\|A\|_F = \sqrt{\sum_{i=1}^m \sum_{j=1}^n a_{ij}^2}$ denotes its Frobenius norm.  When $f(z)$ is a vector of functions, $\|f\|_\infty = \sup_{z\in\mathcal{Z}}\|f(z)\|_2 $.

The following provides a clear table of contents as the guidelines of the appendix.

\section*{Appendix Table of Contents}
\startcontents[appendix]
\printcontents[appendix]{}{1}{\setcounter{tocdepth}{2}}
\bigskip 

\section{More Simulation Results for Point-Wise Marginal Effects}\label{appendixA}

\begin{table}[tbp]
\centering
\caption{Pointwise Marginal Effects: Normal and Rademacher Regressors}
\label{pwme_combined}

\renewcommand{\arraystretch}{1.25}
\footnotesize
\setlength{\tabcolsep}{4pt}

\begin{tabular}{l *{8}{c}}
\toprule
& \multicolumn{4}{c}{Cauchy Errors} & \multicolumn{4}{c}{Skewed Normal Errors} \\
\cmidrule(lr){2-5} \cmidrule(lr){6-9}

& \multicolumn{2}{c}{$p = 50$} & \multicolumn{2}{c}{$p = 100$}
& \multicolumn{2}{c}{$p = 50$} & \multicolumn{2}{c}{$p = 100$} \\
\cmidrule(lr){2-3} \cmidrule(lr){4-5}
\cmidrule(lr){6-7} \cmidrule(lr){8-9}

& $B=25$ & $B=50$ & $B=25$ & $B=50$
& $B=25$ & $B=50$ & $B=25$ & $B=50$ \\

\midrule

Bias
& 0.0004 & 0.0002 & 0.0004 & 0.0003
& 0.0004 & 0.0001 & 0.0003 & 0.0003 \\[12pt]

RMSE
& 0.0076 & 0.0053 & 0.0098 & 0.0063
& 0.0067 & 0.0045 & 0.0083 & 0.0055 \\[12pt]

CR
& 0.9437 & 0.9373 & 0.9453 & 0.9403
& 0.9425 & 0.9380 & 0.9413 & 0.9425 \\[12pt]

CI Length
& 0.0303 & 0.0206 & 0.0379 & 0.0244
& 0.0258 & 0.0170 & 0.0321 & 0.0207 \\

\bottomrule
\end{tabular}

\vspace{4pt}
\footnotesize
\parbox{\linewidth}{
\noindent
\textit{Note:} Results are Monte Carlo averages over 200 simulation replications. 
All statistics for pointwise marginal effects are averaged across the evaluation grid points. 
CR denotes coverage rate, and CI Length is the average confidence interval length.
}
\normalsize

\end{table}

\begin{figure}[t!]
    \centering
    \safeincludegraphics[width=0.9\linewidth]{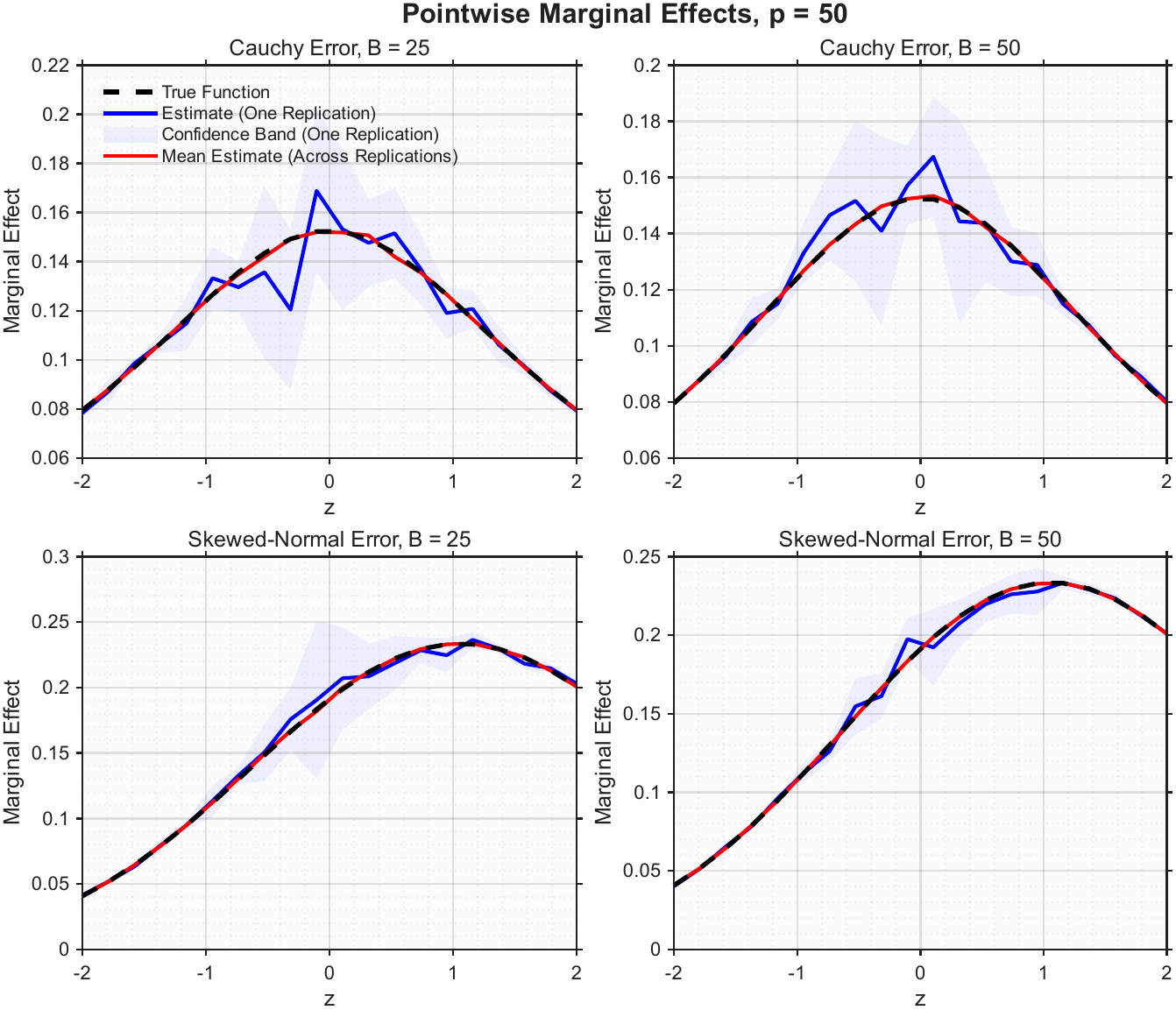}
    \caption{Simulation Results of Point-Wise Marginal Effect: $p=50$ Case}
    \label{fig:pwme2}
\end{figure}

\begin{figure}[t!]
    \centering
    \safeincludegraphics[width=0.9\linewidth]{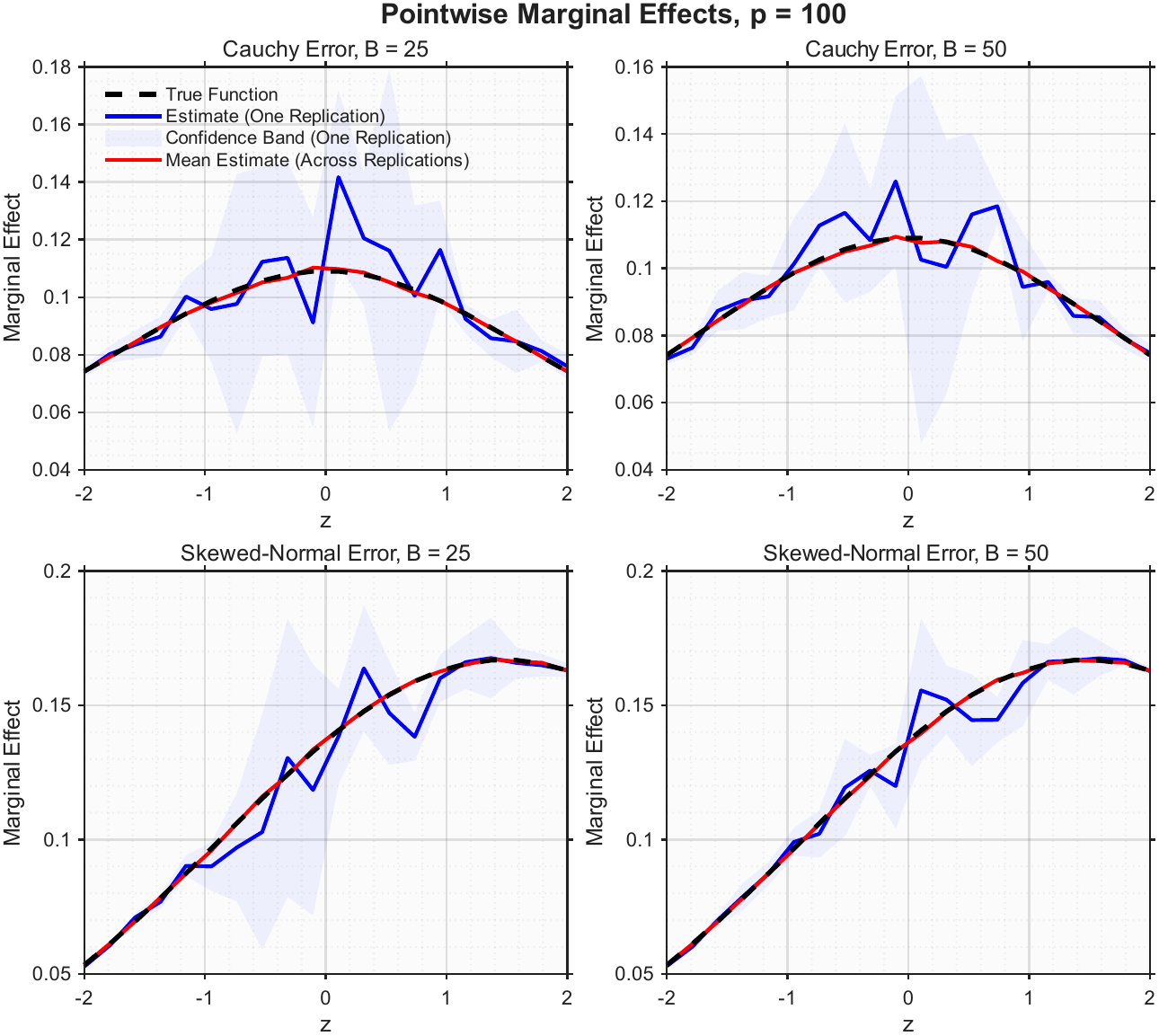}
    \caption{Simulation Results of Point-Wise Marginal Effect: $p=100$ Case}
    \label{fig:pwme1}
\end{figure}

  We can also consider point-wise marginal effect $
\tau_0(x_0, X) \;=\; \nabla_z F_0(z)\theta_0(1) = \nabla_z F_0(z)$, which captures the effect of marginal change of $x_1$ when index value $Z_0$ is equal to $z$.   We consider the same simulation design as in the main text, and consider  $z\in[-2,2]$, which is discretized into 20 grid points.  \autoref{pwme_combined} reports the corresponding statistics for the pointwise estimator.  The pointwise estimator follows a non-standard convergence rate that depends on the rate at which the sieve dimension grows. However, since in the simulation we fix the sieve dimension, we use standard FCLT to provide confidence intervals. The point-wise estimator exhibits modestly larger bias and RMSE than the average estimator, reflecting the slower nonparametric rate, but coverage remains close to the nominal level across all configurations. 

\autoref{fig:pwme2} and \autoref{fig:pwme1} display the estimated point-wise marginal effect trajectory across the grid of values of the index $z$, together with the $95\%$ random-scaling confidence band from a single representative replication and the average estimate across replications. The semiparametric estimator tracks the true point-wise effect closely throughout the support of $z$, with confidence bands that are tight in the interior of the support and slightly wider near the boundaries, consistent with standard nonparametric behavior.

Across both functionals, the simulation evidence confirms that the online procedure delivers reliable estimation and valid inference for policy-relevant marginal effects with no additional nonparametric estimation burden beyond the trajectories already produced by the joint updating algorithm.

\section{Phase I Kernel-Based Warm Start of
\texorpdfstring{$\theta_0$}{theta0}}\label{appendixB}

\subsection{Pseudo Code}
The pseudo code for constructing the warm-start estimator of $\theta_0$ is
provided in Algorithm~\ref{algo:phase1}.

\begin{algorithm}
\caption{Online Warm-Start Learner for $\theta_0$ (Phase I)}
\label{algo:phase1}
\textbf{Require:} streaming batches $\mathcal{W}_k$, $k = 1, \ldots, N$; initial guess $\widehat\theta_0 \in \mathbb{R}^p$; kernel $\mathcal{K}$; learning rates $\{\gamma_k\}$; bandwidths $\{h_k\}$; batch size $B$.\quad
\textbf{Ensure:} PR average $\overline{\theta}_N$.
\algline{1}{Initialize $\widehat\theta_0$ as supplied and set
$\overline{\theta}_0 \gets \widehat\theta_0$}
\algline{2}{\textbf{for} $k = 1, 2, \ldots, N$ \textbf{do}}
\algline{3}{\quad Receive batch $\mathcal{W}_k$ and compute $Z_{i,k}(\widehat\theta_{k-1})$, $i=1,\ldots,B$}
\algline{4}{\quad $\widehat\theta_k \gets \widehat\theta_{k-1} +  \dfrac{ \gamma_k }{h_k \cdot B(B-1)}
            \displaystyle\sum_{i_1 \neq i_2}^B
            \mathcal{K}\!\left(\dfrac{Z_{i_1,k}(\widehat\theta_{k-1}) - Z_{i_2,k}(\widehat\theta_{k-1})}{h_k}\right)
            (Y_{i_1,k} - Y_{i_2,k})(X_{i_1,k} - X_{i_2,k})$
\hfill $\triangleright$ online update \eqref{first-step algorithm}}
\algline{5}{\quad $\overline{\theta}_k \gets \dfrac{k-1}{k}\overline{\theta}_{k-1}
            + \dfrac{1}{k}\widehat\theta_k$ \hfill $\triangleright$ PR average}
\algline{6}{\textbf{return} $\overline{\theta}_N$}
\end{algorithm}

\subsection{Sufficient Conditions and Discussions}

\begin{condition}\label{condition1}
Let $s_{\mathcal K}\in\mathbb N$ and let the batch size satisfy $B\geq2$.  (i) The link
$F_0$ is bounded and has uniformly bounded derivatives through order $s_{\mathcal K}$;
(ii) The pair $(x_0,X)$ has a joint density $f(x_0,X)$ with partial derivatives
through order $s_{\mathcal K}$ in $x_0$; (iii) There are nonnegative envelopes
$g_0,\ldots,g_{s_{\mathcal K}}$ such that, for $0\leq j\leq s_{\mathcal K}$,
\[
 \sup_{x_0\in\mathbb R}\left|\frac{\partial^j f(x_0,X)}{\partial x_0^j}\right|
 \leq g_j(X),\qquad
 \int(1+\Vert X\Vert_2^4)g_j(X)dX<\infty,
\]
and $\mathbb E\Vert X\Vert_2^4<\infty$;
(iv) $\sup_{x_0,X}\mathbb E(\varepsilon^4\mid x_0,X)<\infty$, and
$\sigma^2(x_0,X)=\mathbb E(\varepsilon^2\mid x_0,X)$ has a uniformly bounded
derivative in $x_0$;  (v) The matrix $\mathcal V_{\mathcal K,0}$ defined below
is positive definite.
\end{condition}

\begin{condition}\label{condition2}
The kernel $\mathcal K$ is even and satisfies: (i)
$\int\mathcal K(t)dt=1$, $\int t^j\mathcal K(t)dt=0$ for
$1\leq j\leq s_{\mathcal K}-1$, and $\int|\mathcal K(t)t^j|dt<\infty$ for every
$j\geq s_{\mathcal K}$; (ii) $\int\mathcal K^4(t)|t|^jdt<\infty$ for every $j\geq0$;
and (iii) $\int(\nabla_t\mathcal K(t))^2|t|^jdt<\infty$ for every $j\geq0$.

\end{condition}

\begin{condition}\label{condition3}
   There exist $-\infty<\underline{z}<\overline{z}<\infty$ and $r_X>0$ such that for any $z\in [\underline{z},\overline{z}]$ and $X$ with $\Vert X\Vert_2 \leq r_X$, there holds $\nabla_z F_0(z)\geq \underline{c}_F$ and $f(z, X)\geq \underline{c}_f$, where $\underline{c}_F$ and $\underline{c}_f$ are two positive constants.
\end{condition}

\begin{condition}\label{condition4} $\gamma_k = \gamma_0 k ^{-\alpha_{\gamma}},  h_k = h_0 k ^{-\alpha_{h}}$, where $\gamma_0,\alpha_{\gamma}, h_0,\alpha_h>0$ and $\alpha_{\gamma}<1$. Moreover,     
$s_{\mathcal K}\alpha_h>1/2$, and $2\alpha_{\gamma} - 3\alpha_h >1$.
\end{condition}

\begin{remark}
    \autoref{condition1} and \autoref{condition2} impose regular conditions on data generating process and kernel functions. \autoref{condition3} guarantees global stability of the algorithm.   Note that for \autoref{condition3},  $X$ can be normalized to have zero expectation, in which case such condition  requires that $X$ has nonvanishing density around its expectation. Moreover, such condition also requires that $x_0$ has sufficiently large support, which is required for point identification. Finally, \autoref{condition3} requires that all regressors are continuous. For the case where there are discrete regressors, see \autoref{discrete_regressors}.  \autoref{condition4} regulates the bandwidth and learning rate. Specifically, there holds  $\sum_{k=1}^{\infty}\gamma_k = \infty$, $\sum_{k=1}^{\infty}\gamma_k^2h_k^{-3}<\infty$ and  $\sum_{k=1}^{\infty}h_k^{2s_{\mathcal K}}<\infty$,
\end{remark}

\subsection{Main Results: \autoref{theorem1}}
This section describes the global convergence properties of the  online estimator $\widehat\theta_N$    and its PR average $\overline\theta_N$  proposed in Section 3.1 of main text.   We will formally show that both $\widehat\theta_N$ and  $\overline\theta_N$ almost surely converge to $\theta_0$  regardless of the choice of the starting point, so for almost every path of the data stream, the algorithm locates a neighborhood of $\theta_0$ with arbitrary precision. 

To ease our exposition, in this section we continue with the assumption that $x_0$ and $X$ are all continuous with   joint density  $f(x_0, X)$. Our results can also be extended to  the  case where there are discrete regressors, see \autoref{discrete_regressors}.

Define 
  \begin{align*}
  \mathbb{H}_0 =  & \int \nabla_z F_0(z)f(z - X_1^{\top}\theta_0, X_1)f(z-X_2^{\top}\theta_0, X_2) \left(X_1 - X_2\right)\left(X_1 - X_2\right)^{\top}dzdX_1dX_2,
  \end{align*}
  and
      \begin{align*}
        \mathcal{V}_{\mathcal{K},0} = & \int \mathcal{K}^2(t)dt \int \left(\sigma^2(z - X_1^{\top}\theta_0, X_1) + \sigma^2(z - X_2^{\top}\theta_0, X_2)\right) \times  \\
        & \ \ \ \ \ \ \ \ \ \ \ \ \ \ \ \ \ \ \ \ \ \ \ \ (X_1-X_2)(X_1-X_2)^{\top} f(z- X_1^{\top} \theta_0, X_1)f(z- X_2^{\top}\theta_0, X_2)dzdX_1dX_2,
    \end{align*}   
    where $\sigma^2(x_0, X) = \mathbb{E}(\varepsilon^2|x_0, X)$ is the conditional variance. Further, let learning rate $\gamma_k$ and bandwidth parameter $h_k$ be chosen as  \begin{equation}
    \gamma_k = \gamma_0 k ^{-\alpha_{\gamma}}, \quad   h_k = h_0 k ^{-\alpha_{h}}, 
    \end{equation}
     with $\gamma_0,h_0$, $\alpha_{h}$, and $ \alpha_{\gamma} $   chosen as in \autoref{condition4}. 
   For any $\theta\in\mathbb{R}^p$, define $\Delta\theta \equiv \theta - \theta_0$. The following theorem states the theoretical properties of the online estimator $\widehat\theta_N$ and $\overline\theta_N$, which also proves the first part of Lemma 5.1 in them main text.

 \begin{theorem}\label{theorem1}
      Let \autoref{condition1}--\ref{condition4} hold. Then  $\widehat\theta_N \rightarrow_{a.s.} \theta_0$ and  $\overline\theta_N \rightarrow_{a.s.} \theta_0$. Further,
for every fixed coordinate $j=1,\ldots,p$,we have that 
\begin{align}
 & \limsup_{N\to\infty}e_j^{\top}\left( \frac{4N^{1+\alpha_h}\log(\log (N))\,
\mathcal{V}_{\mathcal{K},0}}
{B(B-1)h_0(1+\alpha_h)}\right)^{-1/2}
\mathbb H_0N\Delta\overline\theta_N=1,
 \qquad \text{a.s.}\\
 &\liminf_{N\to\infty}e_j^{\top}\left( \frac{4N^{1+\alpha_h}\log(\log (N))\,
\mathcal{V}_{\mathcal{K},0}}
{B(B-1)h_0(1+\alpha_h)}\right)^{-1/2}
\mathbb H_0N\Delta\overline\theta_N=-1,
 \qquad \mathrm{a.s.}
\end{align}
 Finally, there holds 
\begin{equation}
\left( \frac{2N^{1+\alpha_h}\mathcal{V}_{\mathcal{K},0}}{B(B-1)h_0(1+\alpha_{h})}\right)^{-\frac{1}{2}}\mathbb{H}_0N\Delta\overline\theta_N\rightarrow_{d} \mathcal N(0, \mathbb I_p).
   \end{equation} 
  \end{theorem}

 \autoref{theorem1} first states the a.s. convergence of $\widehat\theta_N$ and $\overline\theta_N$. Hence for almost \textit{all paths} of the data stream, our algorithm will lead to consistent estimators for the unknown parameter $\theta_0$ as long as the number of updates is sufficiently large, regardless of the choice of the initial guess $\widehat\theta_0$.
Leveraging such global stability property, we can quickly locate a small neighborhood around the true parameter.   
\autoref{theorem1} also shows that \[\Vert \Delta\overline\theta_N\Vert_2 = O(\sqrt{N^{-1 +\alpha_h}\log(\log(N))}), \qquad \mathrm{a.s.}\] 
\autoref{theorem1} is the main result for the warm start estimator of $\theta_0$. In the following, we further obtain the convergence rate of the non-averaged estimator $\widehat\theta_N$.  In particular, we will show that $\Vert \Delta\widehat\theta_N\Vert_2 = O(\sqrt{N^{-\alpha_{\gamma} +\alpha_h}\log(\log ( N))}, \ \mathrm{a.s.}$. Since $\alpha_{\gamma}<1$, the  PR average  $\overline\theta_N$ converges at faster rate than $\widehat\theta_N$. As a result, we   recommend using the PR average $\overline\theta_N$ to construct the projection space when one needs to be constructed.  
The proof of \autoref{theorem1} is lengthy and hence postponed to \autoref{appendixF}.

\section{Phase I Sieve LS Warm Start Estimator of $F_0$}\label{appendixC}
\subsection{Pseudo Code}
The pseudo code is provided in Algorithm \autoref{alg:global-online-F}.
\begin{algorithm}
\caption{Online Warm-Start Learner for $F_0$ (Phase I)}
\label{alg:global-online-F}
\textbf{Require:} streaming batches $\mathcal{W}_k$, $k = 1, \ldots, N$; sieve dimensions $\{J_k\}$; learning rates $\{\eta_k\}$; initial coefficient $\widehat{\beta}_0 \in \mathbb{R}^{J_0}$; index estimates $\{\check{\theta}_{k-1}\}$.
\algline{1}{$\overline{\beta}_0 \gets \widehat{\beta}_0$}
\algline{2}{\textbf{for} $k = 1, 2, \ldots, N$ \textbf{do}}
\algline{3}{\quad Receive batch $\mathcal{W}_k$; \ $\check{Z}_{i,k} \gets Z_{i,k}(\check{\theta}_{k-1})$, $i = 1, \ldots, B$}
\algline{4}{\quad \textbf{Update sieve coefficient} \hfill $\triangleright$ embed via $R_{J_{k-1}, J_k}^{\top}$, then online update \eqref{online_sieve}
    \[
        \widehat{\beta}_k \gets R_{J_{k-1}, J_k}^{\top}\widehat{\beta}_{k-1}
        + \frac{\eta_k}{B}\sum_{i=1}^B
        \Big(Y_{i,k} - \Psi_{J_k}(\check{Z}_{i,k})^{\top} R_{J_{k-1}, J_k}^{\top}\widehat{\beta}_{k-1}\Big)
        \Psi_{J_k}(\check{Z}_{i,k})
    \]}
\algline{5}{\quad \textbf{Update PR-averaged coefficient} \hfill $\triangleright$ recursive averaging in common sieve space
    \[
        \overline{\beta}_k \gets \frac{1}{k}\widehat{\beta}_k
        + \frac{k-1}{k} R_{J_{k-1}, J_k}^{\top}\,\overline{\beta}_{k-1}
    \]}
\algline{6}{\textbf{return} $\overline{F}_N(\cdot) = \Psi_{J_N}(\cdot)^{\top}\overline{\beta}_N$}
\end{algorithm}

\subsection{Sufficient Conditions and Discussions}

Following \citet{chen2015optimal}, we define the $L_2$-projection of function $F_0$ onto $\mathcal{S}_J$   as
\begin{align}\label{L2-projection}
\mathbb{P}_J(F_0)(z) = \Psi_J(z)^{\top} \Gamma_{J}^{-1}\mathbb{E}\left[ \Psi_J( Z_0)F_0(Z_0)\right]\equiv \Psi_J(z)^{\top}\beta_{J,0},
\end{align}
where recall that $\mathcal{S}_J=\{\sum_{j=1}^J b_j\psi_{J,j}(\cdot): b_1,\ldots,b_J\in\mathbb{R}\}$ and $\Gamma_J = \mathbb{E}[\Psi_J(Z_0)\Psi_J(Z_0)^{\top}]$. By definition, $\beta_{J,0}$ is the $L^2(P_{Z_0})$ projection coefficient. Then 
for any function $g\in\mathcal{S}_J$, we have that $\Vert  \mathbb{P}_J(F_0) - F_0\Vert_{\infty} = \Vert  \mathbb{P}_J(F_0 - g) - (F_0- g)\Vert_{\infty},$ so the sup-norm sieve approximation error of $L_2$-projection can be bounded by 
\begin{equation}
\Vert  \mathbb{P}_J(F_0) - F_0\Vert_{\infty}\leq \inf_{g\in\mathcal{S}_J}\Vert F_0 - g\Vert _{\infty}\left( 1+\frac{\Vert \mathbb{P}_J(F_0 - g)\Vert_{\infty}}{\Vert F_0 - g\Vert _{\infty}}\right).
\end{equation}
One implicit condition we impose is that 
$
 \sup_J \sup_{g\in \mathcal{S}_J: 0<\Vert F_0 - g \Vert _{\infty}<\infty} \Vert F_0 - g\Vert _{\infty}^{-1}\Vert \mathbb{P}_J(F_0 - g)\Vert_{\infty}  <\infty,
$ which can be guaranteed when the operator norm of $\mathbb{P}_J$ is uniformly bounded for all $J$. This condition holds, by \citet{chen2015optimal},  for compactly supported wavelets and splines. We also impose sup-norm approximation error rate as 
$
\inf_{g\in\mathcal{S}_J}\Vert F_0 - g\Vert _{\infty}\leq C_{F_0}J^{-s}
$
where $s>0$ is approximation error rate of $F_0$. Finally, we impose that the derivative function also guarantees uniform approximation error of order $1-s$. The above is summarized in the following condition. 

\begin{condition}\label{condition5} There is some positive integer $s$ and constant $C_{F_0}$ such that for all $J$,  $\Vert F_0(\cdot) - \mathbb{P}_J(F_0)(\cdot)\Vert _{\infty}\leq C_{F_0}J^{ -s}$ and  $\Vert \nabla_z [F_0(\cdot) - \mathbb{P}_J(F_0)(\cdot)]\Vert _{\infty}\leq C_{F_0}J^{ 1 -s}$.  
\end{condition}

\begin{remark}\label{rem:unbounded-support-transform}
\autoref{condition5} does not impose any restriction on the support of $Z_0\equiv x_0 + X^{\top}\theta_0$. 
However, when spline basis functions are used, the uniform approximation rates in \autoref{condition5}(ii)  
are typically established for functions defined on compact domains. 
To accommodate unbounded support, we can simply apply a transformation that maps 
$\mathbb{R}$ into a bounded interval \citep{bierens2014consistency}. For instance, the transformation 
$\widetilde Z_0 = \arctan(Z_0)$ ensures $\widetilde Z_0 \in (-\pi/2, \pi/2)$ regardless of the support of $Z_0$.
Let $\{\psi_j(\cdot)\}_{j \geq 1}$ be spline basis functions defined on $(-\pi/2, \pi/2)$. 
We can approximate the transformed function $
\widetilde F_0(u) := F_0(\tan(u)), \quad u \in (-\pi/2, \pi/2),
$ 
using $\{\psi_j(u)\}_{j\geq 1}$. Equivalently, this corresponds to using the sieve 
$\{\psi_j(\arctan(\cdot))\}$ to approximate $F_0(\cdot)$ on $\mathbb{R}$.
Assume in addition that $\widetilde F_0$ extends to the closed transformed
interval with the smoothness and boundary behavior required in
\autoref{condition5}. Under this reparameterization, \autoref{condition5} can be imposed on $\widetilde F_0$, 
allowing the framework to accommodate unbounded support of $Z_0$.
\end{remark}
Recall that $\mu_0(\theta_0, z) = \mathbb{E}(X|Z_0 = z)$.  We choose 
\[
 \eta_k=\eta_0k^{-\alpha_\eta},
 \qquad J_k=[J_0k^{\alpha_\dagger}],
\]
where $\eta_0,\alpha_\eta,\alpha_\dagger>0$ and $J_0$ is a positive integer.
We further impose the following conditions on the data-generating process,
the sieve, and the predictable first-stage estimator of $\theta_0$.

\begin{condition} \label{condition6}
(i) $\mathbb{E}\Vert X\Vert_2^4<\infty$, $\max\{\Vert F_0\Vert_{\infty}, \Vert \nabla_z F_0\Vert_{\infty} ,  \Vert \nabla_{zz} F_0\Vert_{\infty}\}<\infty $, $\sup_{x_0,X}\mathbb{E}(\varepsilon^2|x_0,X)<\infty$ and
$\mathbb E|\varepsilon|^{\kappa} <\infty$ for some $\kappa> 2 +  \max\{\frac{2}{\alpha_{\eta}}, \frac{2(1+\alpha_{\dagger})}{1-\alpha_{\dagger}}\}$; (ii) There holds
$\sup_z\Vert (\nabla_z F_0(z))\mu_0(\theta_0,z)\Vert_2<\infty$, and moreover,
$
\sup_z\left\Vert
\mathbb P_J\!\left[(\nabla_zF_0(\cdot))\mu_0(\theta_0,\cdot)\right](z)
-(\nabla_zF_0(z))\mu_0(\theta_0,z)
\right\Vert_2\leq C_\mu J^{-s_\mu}
$
for some $C_{\mu},s_{\mu}>0$ and all positive integers $J$, with the projection
applied componentwise.
\end{condition}

\begin{remark}
   \autoref{condition6}  regulates the tail behaviors of $X$ and $\varepsilon$. The moment condition on $\varepsilon$ is used to build almost sure convergence and sup-norm rate for the sieve estimator, and for showing the FCLT of the online marginal effect estimator in \autoref{FCLT_average_marginal_effect}.
    We mention that if we only pursue the rate of convergence in probability for online sieve estimator,  then $\kappa > 2 + \frac{2\alpha_{\dagger}}{1-\alpha_{\dagger}}$ alone as in \citet{chen2015optimal} will suffice (note that under optimal sieve rate $1/(2s+1)$, the condition reduces to $2+1/s$, which is exactly the dimension/smoothness condition in \citet{chen2015optimal}). When we further pursue optimal a.s. convergence rates of sieve estimator, we additionally need  $\kappa > 2   + \max\{\frac{2}{\alpha_{\eta}}, \frac{2(1+\alpha_{\dagger})}{1-\alpha_{\dagger}}\}$. 
\end{remark}

\begin{condition}\label{condition7} (i) $\Vert \Psi_J\Vert_{\infty}\leq CJ^{1/2}$ and $\Vert \nabla_z\Psi_J\Vert_{\infty}\leq CJ^{3/2}$; (ii)  $0<\inf_{J}\lambda_{\mathrm{min}}(\Gamma_J)\leq \sup_{J}\lambda_{\text{max}}(\Gamma_J)<\infty$; (iii) Let $\mathcal{Z}_0$ be the support of $Z_0$. For each positive integer $L$, there exists $z_1,\cdots, z_L\in\mathcal{Z}_0$ such that $\sup_{z\in\mathcal{Z}_0}\min_{1\leq l\leq L}\Vert \Psi_J(z)- \Psi_J(z_l) \Vert  \leq C_{\Psi}J^{3/2}L^{-1}$ for all $J$, where $C_{\Psi}>0$ is a  constant.  
    
\end{condition}

\begin{remark}
     \autoref{condition7} (i) and (ii) are natural for normalized B-spline basis. Note that if $\mathcal{Z}_0$ is bounded, the \autoref{condition7} (iii) automatically holds given \autoref{condition7} (i). Otherwise, when $\mathcal{Z}_0$ is unbounded, we consider sieve functions of form   
$\{\psi_j(\arctan(\cdot))\}$ discussed in
\autoref{rem:unbounded-support-transform}.
\end{remark}

\begin{condition}\label{new_theta_condition_for_sieve}
Let $\mathcal F_k=\sigma(\mathcal W_1,\ldots,\mathcal W_k)$. For specified constants $0<\alpha_\theta\leq1/2$ and $c_\theta\geq0$,
$\Vert\Delta\check\theta_k\Vert_2=
O_{\mathrm{a.s.}}(k^{-\alpha_\theta}(\log k)^{c_\theta})$.
For each $k$, $\check\theta_k$ is $\mathcal{F}_k$-measurable.
\end{condition}

\begin{condition}\label{rate}
The exponents lie in the strict interior
\[
 \frac12<\alpha_\eta,\qquad
 \max\{1-2s\alpha_\dagger,
        1+5\alpha_\dagger-2\alpha_\theta\}
 <\alpha_\eta<
 \min\{2\alpha_\theta,1-3\alpha_\dagger,
        (2s+1)\alpha_\dagger\}.
\]
\end{condition}

Throughout the rest of the Appendix, to avid complicated discussion on $s\alpha_{\dagger}$,  we always assume that $s\alpha_{\dagger}<1$ and $\alpha_{\theta}+s_{\mu}\alpha_{\dagger}<1$. We note that violations of either assumptions will not affect any of the following results. 

\subsection{Main Results: \autoref{theorem6}}

This section develops the sup-norm almost-sure convergence rate of the warm-start estimator $\overline{F}_N$ defined in Section 3.2 of main text.  The argument also applies to generated-regressor online sieve estimation with predictable first-stage estimators satisfying \autoref{new_theta_condition_for_sieve}.

 Recall that $Z_0 = x_0 + X^{\top}\theta_0$ and $\Gamma_J = \mathbb{E}\left[\Psi_J(Z_{0})\Psi_J(Z_{0})^{\top}\right]$.    
For any $N\geq 1$, define $\Delta\overline{\beta}_{N} = \overline{\beta}_{N} - \beta_{J_N, 0}$,
where $\beta_{J_N,0}$ is the pseudo-true sieve coefficient and $\eta_k,J_k$
are as defined above.
\autoref{new_lemma-F} is implied by the following results.  
 
\begin{theorem}\label{theorem6}
    Let \autoref{condition5}--\ref{rate} hold.  Assume also that
    $\widehat\beta_0$ is $\mathcal F_0$-measurable and
    $\mathbb E\Vert\widehat\beta_0\Vert_2^2<\infty$.  Then
\begin{align*}
\Delta\overline{\beta}_N & = \frac{1}{N}\sum_{k=1}^N R_{J_k, J_N}^{\top}\Gamma_{J_k}^{-1}\frac{1}{B}\sum_{i=1}^B \varepsilon_{i,k}\Psi_{J_k}(Z_{0,i,k})- \frac{1}{N}\sum_{k=1}^N R_{J_k, J_N}^{\top}\Gamma_{J_k}^{-1}\mathbb{E}[(\nabla_z F_0(Z_0))\Psi_{J_k}(Z_0)X^{\top}]\Delta\check\theta_{k-1} \\
 & + \frac{1}{N}\sum_{k=1}^N\left(R_{J_k, J_N}^{\top} \beta_{J_k, 0} - \beta_{J_N,0}\right) + r_N, \qquad \Vert r_N\Vert_2=o_{\rm a.s.}(N^{-1/2})
\end{align*}
   Moreover,
    defining
    \[
      A_{\mu,N}:=\frac1N\sum_{k=1}^N
      J_k^{-s_\mu}\Vert\Delta\check\theta_{k-1}\Vert_2,
    \]
    we have
    \begin{align}\label{eq:C1-supnorm-corrected}
    \left\Vert \overline{F}_N-F_0 \right\Vert_{\infty}
    =O_{\mathrm{a.s.}}\!\left(
       \sqrt{\frac{J_N\log N}{N}}
       +\frac1N\sum_{k=1}^N J_k^{-s}
       +\left\Vert\frac1N\sum_{k=1}^N
          \Delta\check\theta_{k-1}\right\Vert_2
       +A_{\mu,N}\right).
    \end{align}
    where $s>0$ is the approximation error rate of $F_0$. 
    In particular, with $J_k\asymp k^{\alpha_\dagger}$, $ 
    \frac1N\sum_{k=1}^N J_k^{-s}
    =
      O(N^{-s\alpha_\dagger}) $ 
    and  $A_{\mu,N}=
      O_{\mathrm{a.s.}}(N^{-\alpha_\theta-s_\mu\alpha_\dagger}
          (\log N)^{c_\theta})$. 
\end{theorem}
\autoref{theorem6} is the key result of online sieve estimation with possibly generated regressor.  It provides the first-order expansion of the average online estimator $\overline{\beta}_N$, which contains four terms. The first term is comparable to the variance term of the full-sample sieve estimator (e.g., \citet{chen2015optimal}), but since in the $k$-th round of update we only use $J_k$ sieve functions, the effective variance component in the $k$-th round is $\varepsilon_{i,k}R_{J_{k}, J_N}^{\top}\Gamma_{J_k}^{-1}\Psi_{J_k}(Z_{0,i,k})$ instead of $\varepsilon_{i,k}\Gamma_{J_{N}}^{-1}\Psi_{J_N}(Z_{0,i,k})$. This highlights the online learning structure of the algorithm. The second term describes the first-order impacts of using the generated regressor $\check Z_{i,k}$ for sieve estimation. The third term describes the impacts of changing pseudo true sieve coefficients and re-embedding when the sieve dimension    increases. The last term collects all the high-order remainders.

When we consider the sup-norm estimation error of $F_0$, apart from the last two terms capturing the average
index error and the projection error $A_{\mu,N}$,     the stochastic and approximation components match the familiar
sieve benchmark under a compatible tuning choice.  The usual minimax choice
$\alpha_\dagger=1/(2s+1)$ is admissible here when it lies in the strict
region of \autoref{rate}.  With the fastest permitted first stage,
$\alpha_\theta=1/2$, this requires
\[
 s>\frac72,
 \qquad
 \max\{1/2,5\alpha_\dagger\}<\alpha_\eta<1-3\alpha_\dagger.
\]
Under these restrictions, we can  choose $J_N\asymp N^{1/(2s+1)}$, then
$
 J_N^{\frac{1}{2}}N^{-\frac{1}{2}}\log^{\frac{1}{2}}(N)+ J_N^{-s} \asymp N^{-\frac{s}{2s+1}}\log^{\frac{1}{2}}(N),
$ 
and this error rate matches the usual minimax benchmark up to poly-log  terms \citep{stone1982optimal,belloni2015some,chen2015optimal}. 

\subsection{Proof of \autoref{theorem6}}

\subsubsection{Additional Lemma}
We first state a useful lemma. 
\begin{lemma}\label{as converence of sums}
    Let  nonnegative random sequence $Q_0,Q_1,\cdots$ satisfy $\mathbb{E}Q_{0}^2<\infty$ and 
    \[
  Q_{N}  \leq (1 - CN^{-\alpha_1}) Q_{N-1}  + CN^{\alpha_2}T_N
    \]
    where $\frac{1}{2}<\alpha_1<1$ and $\alpha_2\in \mathbb{R}$, and $T_N\geq 0$ satisfies  $\mathbb{E}_{N-1}T_N^2\leq C_T<\infty$   for all $N$ with $\mathbb{E}_{N-1}$ being the expectation conditioned on $Q_0, T_1,\cdots, T_{N-1}$.  Then \[
      Q_N  = O(N^{\alpha_2 + \alpha_1}), \qquad \text{a.s}
    \]
    Under the same conditions, if instead   
    \[
  Q_{N}  \leq  Q_{N-1}  + CN^{\alpha_2}T_N
    \]
    then for any $r>1$, 
    \[
    Q_N=\begin{cases}
      O_{\mathrm{a.s.}}(N^{1+\alpha_2}),&\alpha_2>-1,\\
      O_{\mathrm{a.s.}}(\log^{r}( N)),&\alpha_2=-1,\\
      O_{\mathrm{a.s.}}(1),&\alpha_2<-1.
    \end{cases}
    \]
\end{lemma}

\begin{proof}
Assume that $1-CN^{-\alpha_1}\in[0,1]$; otherwise we start from some sufficiently large $k_0$.  Define 
$\varepsilon_{T,N}=T_N-\mathbb E_{N-1}T_N$ and let $\overline Q_N$ be defined as $\overline{Q}_N = (1-CN^{-\alpha_1})\overline{Q}_{N-1} +CN^{\alpha_2}T_N$ with $\overline Q_0=Q_0$.  Recursion gives $0\leq Q_N\leq\overline Q_N$.  Decompose
$\overline Q_N=Q_{1,N}+Q_{2,N}$, where
\begin{align*}
 Q_{1,N}&=(1-CN^{-\alpha_1})Q_{1,N-1}
          +CN^{\alpha_2}\mathbb E_{N-1}T_N,\\
 Q_{2,N}&=(1-CN^{-\alpha_1})Q_{2,N-1}
          +CN^{\alpha_2}\varepsilon_{T,N},
\end{align*}
with $Q_{1,0}=\mathbb EQ_0$ and $Q_{2,0}=Q_0-\mathbb EQ_0$.
Let $r=\alpha_1+\alpha_2$.  Since
$\mathbb E_{N-1}T_N\leq C_T^{1/2}$, iteration gives
\[
 |Q_{1,N}|
 \leq C e^{-CN^{1-\alpha_1}}
 +C\sum_{k=1}^N
   e^{-C(N^{1-\alpha_1}-k^{1-\alpha_1})}k^{\alpha_2}.
\]
By \autoref{lemma4}, we have that $|Q_{1,N}| = O(N^{\alpha_1 +\alpha_2}) = O(N^r)$. 
For the martingale part, note that 
\[
 N^{-r}Q_{2,N}=(1-CN^{-\alpha_1})\left(\frac{N-1}{N}\right)^r(N-1)^{-r}Q_{2,N-1}+CN^{-\alpha_1}\varepsilon_{T,N},
\]
Because $\alpha_1<1$, the $N^{-\alpha_1}$ contraction dominates the
$O(N^{-1})$   drift, so  
\[
 \mathbb E_{N-1}[N^{-r}Q_{2,N}]^2
 \leq(1-CN^{-\alpha_1})[(N-1)^{-r}Q_{2,N-1}]^2+CN^{-2\alpha_1}.
\]
Now $\sum_NN^{-2\alpha_1}<\infty$ and
$\sum_NN^{-\alpha_1}=\infty$.  Robbins--Siegmund implies that $(N^{-r}Q_{2,N})^2$
converges and that $\sum_NN^{-\alpha_1}[N^{-r}Q_{2,N}]^2<\infty$; hence its limit
is zero.  Therefore $Q_{2,N}=o_{\rm a.s.}(N^r)$.

For the recursion without contraction, use the same domination and write
\[
 Q_{1,N}=Q_{1,0}+C\sum_{k=1}^Nk^{\alpha_2}
                 \mathbb E_{k-1}T_k,
 \qquad
 Q_{2,N}=Q_{2,0}+C\sum_{k=1}^Nk^{\alpha_2}\varepsilon_{T,k}.
\]
The deterministic part has the exact three orders
\[
 Q_{1,N}=\begin{cases}
 O(N^{1+\alpha_2}),&\alpha_2>-1,\\
 O(\log N),&\alpha_2=-1,\\
 O(1),&\alpha_2<-1.
 \end{cases}
\]
The martingale quadratic variation is bounded by
$C\sum_{k\leq N}k^{2\alpha_2}$.  A dyadic martingale maximal inequality and
Borel--Cantelli give, for every fixed $r_0>1$,
\[
 |Q_{2,N}|=\begin{cases}
 O_{\rm a.s.}\{N^{1/2+\alpha_2}(\log N)^{r_0/2}\},
       &\alpha_2>-1/2,\\
 O_{\rm a.s.}\{(\log N)^{(r_0+1)/2}\},&\alpha_2=-1/2,\\
 O_{\rm a.s.}(1),&\alpha_2<-1/2,
 \end{cases}
\]
where in the last case the martingale converges in $L_2$ and almost surely.
   Combining $Q_N\leq Q_{1,N}+|Q_{2,N}|$ proves the
second assertion, including the logarithmic boundary $\alpha_2=-1$.
\end{proof}

The following proof of \autoref{theorem6} will be divided into three steps. 
    
\subsubsection{ Obtaining Initial Convergence Rate for $\Vert \Delta\widehat\beta_k\Vert_2$}\label{C1.part1}

\vspace{0.2cm}
\textit{Summary:}
\textit{This part will develop the following result:}  
\[
\Vert \Delta\widehat\beta_k \Vert_2 = O\left(k^{\frac{1-\alpha_{\eta} - \alpha_{\dagger}\underline{\alpha}_{F}}{2}}\log^{\frac{r}{2}}(k)\right), \qquad \text{a.s.}
\]
\textit{where $\underline{\alpha}_F = \min\{2s-1, 2\alpha_{\theta}/\alpha_{\dagger} - 3, \alpha_{\eta}/\alpha_{\dagger}-1\}>0$ and  $r>1+2c_{\theta}$}.
\vspace{0.5cm}

To derive an initial a.s. convergence rate for $\Vert \Delta\widehat\beta_k\Vert_2$, we first  provide a bound for $\mathbb E_{k-1}\Vert\Delta\widehat{\beta}_{k}\Vert^2_2$ based on $ \Vert\Delta\widehat{\beta}_{k-1}\Vert_2^2$. Such bound will allow us to use the convergence result in \citet{robbins1971convergence}.
Note that according to the definition of $R_{J_1, J_2}$, we have that 
\[
\widehat\beta_k = R_{J_{k-1}, J_k}^{\top}\widehat{\beta}_{k-1}
+ \frac{\eta_k}{B}\sum_{i=1}^B
\Big(
Y_{i,k}
-
\Psi_{J_k}(\check{Z}_{i,k})^{\top}
R_{J_{k-1}, J_k}^{\top}\widehat{\beta}_{k-1}
\Big)
\Psi_{J_k}(\check{Z}_{i,k}),
\]
for all $k$. So 
\begin{align*}
    \widehat\beta_k - \beta_{J_k, 0} &  = R_{J_{k-1},J_k}^{\top}\left(\widehat\beta_{k-1} - \beta_{J_{k-1},0}\right) + \left(R_{J_{k-1},J_k}^{\top}\beta_{J_{k-1}, 0} - \beta_{J_k, 0}\right) \\
    & + \frac{\eta_k}{B}\sum_{i=1}^B \varepsilon_{i,k}\Psi_{J_k}\left(\check{Z}_{i,k}\right) + \frac{\eta_k}{B}\sum_{i=1}^B\left(F_0\left(Z_{0,i,k}\right) - F_0\left(\check{Z}_{i,k}\right)\right)\Psi_{J_{k}} \left(\check{Z}_{i,k}\right)\\
    & + \frac{\eta_k}{B}\sum_{i=1}^B \left( F_0(\check{Z}_{i,k}) - \mathbb{P}_{J_{k}}(F_0)(\check{Z}_{i,k})\right)\Psi_{J_k} \left(\check{Z}_{i,k}\right) \\
    & + \frac{\eta_k}{B}\sum_{i=1}^B  \Psi_{J_k} \left(\check{Z}_{i,k}\right)\Psi_{J_k} \left(\check{Z}_{i,k}\right)^{\top}\left(\beta_{J_k,0} -R _{J_{k-1},J_k}^{\top}\beta_{J_{k-1}, 0} \right) \\
    & + \frac{\eta_k}{B}\sum_{i=1}^B  \Psi_{J_k} \left(\check{Z}_{i,k}\right)\Psi_{J_k} \left(\check{Z}_{i,k}\right)^{\top}R _{J_{k-1},J_k}^{\top}\left(\beta_{J_{k-1}, 0} - \widehat\beta_{k-1} \right)
\end{align*}
Define $\check{\Gamma}_{J, k} =\frac{1}{B}\sum_{i=1}^B \Psi_{J}(\check{Z}_{i,k})\Psi_{J}(\check{Z}_{i,k})^{\top}$, we have that 
\begin{align*}
    \Delta\widehat\beta_k  & = \left(\mathbb{I}_{J_k} - \eta_k \check{\Gamma}_{J_k, k}\right)R^{\top}_{J_{k-1},J_k}\Delta\widehat\beta_{k-1} + \left(\mathbb{I}_{J_k} - \eta_k \check{\Gamma}_{J_k, k}\right)\left(R_{J_{k-1},J_k}^{\top}\beta_{J_{k-1}, 0} - \beta_{J_k, 0}\right)\\
    &  + \frac{\eta_k}{B}\sum_{i=1}^B\varepsilon_{i,k} \Psi_{J_{k}} \left(Z_{0,i,k}\right) + \frac{\eta_k}{B}\sum_{i=1}^B \left( F_0({Z}_{0,i,k}) - \mathbb{P}_{J_{k}}(F_0)(Z_{0,i,k})\right)\Psi_{J_{k}} \left( Z_{0,i,k}\right)\\
   & +  \eta_k \left[\underset{\zeta_{1,k}}{\underbrace{\frac{1}{B}\sum_{i=1}^B\left(F_0\left(Z_{0,i,k}\right) - F_0\left(\check{Z}_{i,k}\right)\right)\Psi_{J_{k}} \left(\check{Z}_{i,k}\right)}} + \underset{\zeta_{2,k}}{\underbrace{ \frac{1}{B}\sum_{i=1}^B \varepsilon_{i,k} (\Psi_{J_{k}} \left(\check{Z}_{i,k}\right) - \Psi_{J_{k}}\left(Z_{0,i,k}\right) )}}\right.\\
   & \left.  +  \underset{ \zeta_{3,k}}{\underbrace{\frac{1}{B}\sum_{i=1}^B \left(\left( F_0(\check{Z}_{i,k}) - \mathbb{P}_{J_{k}}(F_0)(\check{Z}_{i,k})\right)\Psi_{J_{k}} \left(\check{Z}_{i,k}\right) -\left( F_0(Z_{0,i,k}) - \mathbb{P}_{J_{k}}(F_0)(Z_{0,i,k})\right)\Psi_{J_{k}} \left( Z_{0,i,k}\right)\right) }}\right].
\end{align*}
Further define  $\widehat{\Gamma}_{J, k} = \frac{1}{B}\sum_{i=1}^B\Psi_J(Z_{0,i,k})\Psi_J(Z_{0,i,k})^{\top}$. Since $\Vert \Psi_J\Vert_{\infty}\leq C_{\Psi}J^{\frac{1}{2}}$ and $\Vert \nabla_z \Psi_J\Vert_{\infty}\leq C_{\Psi}J^{\frac{3}{2}}$, we have that 
\[
\left\Vert \widehat{\Gamma}_{J_{k}, k} - \check{\Gamma}_{J_{k}, k}\right\Vert_F \leq \frac{CJ_{k}^2 }{B}\sum_{i=1}^B\Vert X_{i,k}\Vert_2 \Vert \Delta\check{\theta}_{k-1}\Vert_2,  
\]
\[
\Vert \zeta_{1,k}\Vert_2  \leq \frac{CJ_{k}^{\frac{1}{2}}}{B}\sum_{i=1}^B\left|F_0\left(Z_{0,i,k}\right) - F_0\left(\check{Z}_{i,k} \right)\right|\leq CJ_{k}^{\frac{1}{2}}\frac{1}{B}\sum_{i=1}^B \Vert X_{i,k}\Vert_2 \Vert \Delta\check{\theta}_{k-1} \Vert_2, 
\]
\[
\Vert \zeta_{2,k}\Vert_2 \leq  \frac{CJ_{k}^{\frac{3}{2}}}{B}\sum_{i=1}^B |\varepsilon_{i,k}|\Vert X_{i,k}\Vert_2  \Vert \Delta\check{\theta}_{k-1} \Vert_2. 
\]
Moreover, since
\begin{align*}
\zeta_{3,k} & = \frac{1}{B}\sum_{i=1}^B \int_{0}^1
\left.\left(\nabla_z [F_0(z) - \mathbb{P}_{J_k}(F_0)(z)]\Psi_{J_k}(z)\right.\right.\\
& \ \ \ \ \ \ \ \ \ \ \ \ \ \ \ \left. \left. + [F_0(z) - \mathbb{P}_{J_k}(F_0)(z)]\nabla_z\Psi_{J_k}(z)\right)\right|_{z = Z_{0,i,k}+\tau(\check Z_{i,k} - Z_{0,i,k})}
d\tau\,(\check Z_{i,k}-Z_{0,i,k}),
\end{align*}
and \autoref{condition5} requires that  $\Vert \nabla_z [F_0(z) - \mathbb{P}_J(F_0)(z)]\Vert _{\infty}\leq C_{F_0}J^{ 1 -s}$, 
we have that 
\[
\Vert \zeta_{3,k}\Vert_2 \leq    \frac{ CJ_{k}^{\frac{3}{2}-s }}{B}\sum_{i=1}^B \Vert X_{i,k}\Vert_2 \Vert \Delta\check{\theta}_{k-1} \Vert_2. 
\]
Now we obtain several useful results.  First
note that
\begin{align*}
 &  \mathbb{E}_{k-1}\left[\left(\mathbb I_{J_{k}} - \eta_k \widehat{\Gamma}_{J_{k}, k} + \eta_k \left( \widehat{\Gamma}_{J_{k}, k} - \check{\Gamma}_{J_{k}, k}\right)\right)^{\top}\left(\mathbb I_{J_{k}} - \eta_k \widehat{\Gamma}_{J_{k}, k} + \eta_k \left( \widehat{\Gamma}_{J_{k}, k} - \check{\Gamma}_{J_{k}, k}\right)\right) \right] \\
 & = \mathbb I_{J_{k}} - 2\eta_k \Gamma_{J_{k}} + 2\eta_{k} \mathbb{E}_{k-1} \left( \widehat{\Gamma}_{J_{k}, k} - \check{\Gamma}_{J_{k}, k}\right) + \eta_k^2\mathbb{E}_{k-1} \left(\check{\Gamma}_{J_{k}, k}^2\right).
\end{align*}
We also show that 
\[
\Vert R_{J_1, J_2}^{\top}\Vert_{\mathrm{op}}\leq
\sqrt{\frac{\sup_J\overline\lambda(\Gamma_J)}
           {\inf_J\underline\lambda(\Gamma_J)}}
\] holds for any $J_1<J_2$. This is because, note that $\Psi_{J_1}(Z_0) = R_{J_1, J_2}\Psi_{J_2}(Z_0)$. Then $\Gamma_{J_1} = \mathbb{E}[\Psi_{J_1}(Z_0)\Psi_{J_1}(Z_0)^{\top}] = R_{J_1, J_2}\mathbb{E}[\Psi_{J_2}(Z_0)\Psi_{J_2}(Z_0)^{\top}]R_{J_1, J_2}^{\top} = R_{J_1, J_2}\Gamma_{J_2}R_{J_1, J_2}^{\top}$. For any vector $a$, due  to the lower and upper boundedness of the eigenvalues of $\Gamma_J$ with respect to $J$,  we have that 
\[
\sup_{J}\lambda_{\text{max}}(\Gamma_J)\Vert a\Vert_2^2 \geq  a^{\top}\Gamma_{J_1}a = (R_{J_1, J_2}^{\top} a)^{\top}\Gamma_{J_2}(R_{J_1, J_2}^{\top} a)\geq \inf_{J}\lambda_{\mathrm{min}}(\Gamma_J)\Vert R_{J_1, J_2}^{\top} a \Vert_2^2 
\]
So 
\[
\Vert R_{J_1, J_2}^{\top} a \Vert_2^2 \leq \frac{\sup_{J}\lambda_{\text{max}}(\Gamma_J)}{\inf_{J}\lambda_{\mathrm{min}}(\Gamma_J)}\Vert a\Vert_2^2
\]
This proves the uniform boundedness of $\Vert R_{J_1, J_2}\Vert_{\mathrm{op}}$. 

Given the above results, due to the lower and upper boundedness of the eigenvalues of $\Gamma_J$ with respect to $J$ and $\mathbb{E}\Vert X\Vert_2<\infty$,   for $k$ sufficiently large we have that 
\begin{align*}
& \mathbb{E}_{k-1} \left\Vert \left(\mathbb I_{J_{k}} - \eta_k \check{\Gamma}_{J_{k}, k}\right)R_{J_{k-1},J_k}^{\top}\Delta\widehat{\beta}_{k-1}\right\Vert_2 ^2\\
& \leq  \Vert \mathbb I_{J_{k}} - 2\eta_k  \Gamma_{J_{k}} \Vert_{\mathrm{op}} \Vert R_{J_{k-1},J_k}^{\top}\Delta\widehat{\beta}_{k-1}\Vert_2^2  +  2\eta_{k}  \left(\mathbb{E}_{k-1} \Vert \widehat{\Gamma}_{J_{k}, k} - \check{\Gamma}_{J_{k}, k}\Vert_F\right)\Vert R_{J_{k-1},J_k}^{\top}\Delta\widehat{\beta}_{k-1}\Vert_2^2 \\
& + \eta_k^2\left(\mathbb{E}_{k-1} \Vert \check{\Gamma}_{J_{k}, k}\Vert_F^2\right)   \Vert R_{J_{k-1},J_k}^{\top}\Delta\widehat{\beta}_{k-1}\Vert_2^2  \\
& \leq \left(1 - C \eta_k + C\eta_k J_{k}^2 \Vert \Delta\check{\theta}_{k-1} \Vert_2 + C  \eta_k^2 J_{k}^{2} + C\boldsymbol{1}_k\right) \Vert \Delta\widehat{\beta}_{k-1}\Vert_2^2,
\end{align*}
where $\boldsymbol{1}_k = \boldsymbol{1}(J_k > J_{k-1})$. When
$\alpha_{\theta}>2\alpha_{\dagger}$ and
$\alpha_{\eta}>2\alpha_{\dagger}$, as implied by the strict rate region,
we have that \begin{align*}
& \mathbb{E}_{k-1} \left\Vert \left(\mathbb I_{J_{k}} - \eta_k \check{\Gamma}_{J_{k}, k}\right)R_{J_{k-1},J_k}^{\top}\Delta\widehat{\beta}_{k-1}\right\Vert_2 ^2  \leq \left(1 - C\eta_k + C\boldsymbol{1}_k\right)\Vert \Delta\widehat{\beta}_{k-1}\Vert_2^2, \qquad \text{a.s.}
\end{align*}
Moreover,
\begin{align*}
& \mathbb{E}_{k-1} \left\Vert  \frac{\eta_k}{B}\sum_{i=1}^B \left( F_0(Z_{0,i,k}) - \mathbb{P}_{J_{k}}(F_0)(Z_{0,i,k})\right)\Psi_{J_{k}} \left( Z_{0,i,k}\right)+ \frac{\eta_k}{B}\sum_{i=1}^B\varepsilon_{i,k} \Psi_{J_{k}} \left(Z_{0,i,k}\right) \right \Vert_2^2\leq C\eta_k^2 J_{k} ,
\end{align*}
and, because $\mathbb{E}\Vert X_{i,k}\Vert_2^2<\infty$,   
\[
\mathbb{E}_{k-1}\left\Vert \eta_k (\zeta_{1,k} + \zeta_{2,k}+ \zeta_{3,k})\right\Vert^2_2 \leq C\eta_k^2 J^{3}_{k-1}\Vert \Delta\check{\theta}_{k-1}\Vert _2^2.
\]
Note that  for any constants $a,b$, we have that $|ab|\leq \frac{a^2}{2c} + \frac{c b^2}{2}$ for arbitrary $c>0$. Using this result, we have that  
\begin{align*}
   & \left| \mathbb{E}_{k-1}\left[2\left(\left(\mathbb I_{J_{k}} - \eta_k \check{\Gamma}_{J_{k}, k}\right)R_{J_{k-1},J_k}^{\top}\Delta\widehat{\beta}_{k-1}\right)^{\top} \times \right. \right.\\
    & \ \ \ \ \ \ \ \ \left.\left.\frac{\eta_k}{B}\sum_{i=1}^B \left( F_0(Z_{0,i,k}) - \mathbb{P}_{J_{k}}(F_0)(Z_{0,i,k})\right)\Psi_{J_{k}} \left( Z_{0,i,k}\right)+ \frac{\eta_k}{B}\sum_{i=1}^B\varepsilon_{i,k} \Psi_{J_{k}} \left(Z_{0,i,k}\right) \right]\right|\\
    & \leq c \eta_k \left(1 - C \eta_k  + C\boldsymbol{1}_k\right) \Vert \Delta\widehat{\beta}_{k-1}\Vert^2_2  + c^{-1}C\eta_k J^{ -2s + 1 }_{k},
\end{align*}
where $c$ is a positive constant and can be arbitrarily chosen. We also have that, because, $\mathbb{E}\Vert X_{i,k}\Vert_2^2<\infty$,   
\begin{align*}
   & \left| \mathbb{E}_{k-1}\left[2\left(\left(\mathbb I_{J_{k}} - \eta_k \check{\Gamma}_{J_{k}, k}\right)R_{J_{k-1},J_k}^{\top}\Delta\widehat{\beta}_{k-1}\right)^{\top} \eta_k \left(\zeta_{1,k} + \zeta_{2,k}+ \zeta_{3,k}\right)\right]\right|\\
    & \leq c \eta_k \left(1 - C  \eta_k + C \boldsymbol{1}_k\right)  \Vert \Delta\widehat{\beta}_{k-1} \Vert_2^2  + c^{-1}C\eta_k J^{3}_{k}\Vert \Delta\check{\theta}_{k-1}\Vert_2^2  ,
\end{align*}
and 
\begin{align*}
& \left|\mathbb{E}_{k-1}\left[ \left(\frac{\eta_k}{B}\sum_{i=1}^B \left( F_0({Z}_{0,i,k}) - \mathbb{P}_{J_{k}}(F_0)(Z_{0,i,k})\right)\Psi_{J_{k}} \left( Z_{0,i,k}\right)+ \frac{\eta_k}{B}\sum_{i=1}^B\varepsilon_{i,k} \Psi_{J_{k}} \left(Z_{0,i,k}\right)\right)^{\top}\right.\right.\\
& \times \left.\left. \eta_k \left(\zeta_{1,k} + \zeta_{2,k}+\zeta_{3,k}\right)\right]\right| 
 \leq C\eta_k^2 \left(J_{k-1}  + J^{3}_{k-1}\Vert \Delta\check{\theta}_{k-1}\Vert^2 \right).
\end{align*}
Since $\eta_k\rightarrow 0$ according to our choice, then  we can choose $c$ sufficiently small such that for $k$ large, $(1+ 2c\eta_k )\left(1 - C  \eta_k \right)
   \leq 1-C'\eta_k$,
where $0<c<C/2$ is fixed and $C'>0$.
Finally, we show that  
 \[\Vert \beta_{J_k, 0} -R_{J_{k-1},J_k}^{\top}\beta_{J_{k-1},0}\Vert_2 \leq C\boldsymbol{1}_k J_k^{- s},\]where $C$ does not depend on $J_k$. 
To show this, note that 
\begin{align*}
& \left|\Psi_{J_k}(Z_0)^{\top}(\beta_{J_k,0} - R_{J_{k-1},J_k}^{\top}\beta_{J_{k-1},0})\right| \\
& \leq  \left|\Psi_{J_k}(Z_0)^{\top}\beta_{J_k, 0} - F_0(Z_0)\right| + \left|F_0(Z_0) - \Psi_{J_{k-1}}(Z_0)^{\top}\beta_{J_{k-1}, 0}\right|\leq CJ_k^{-s}\end{align*} according to \autoref{condition5} and the fact that $\sup_kJ_k/J_{k-1}<\infty$. Moreover, using \autoref{condition7}, 
\begin{align*}
 \mathbb{E}\vert\Psi_{J_k}(Z_0)^{\top}(\beta_{J_k,0} - R_{J_{k-1},J_k}^{\top}\beta_{J_{k-1},0})\vert^2 &= (\beta_{J_k,0} - R_{J_{k-1},J_k}^{\top}\beta_{J_{k-1},0})^{\top}\Gamma_{J_k}(\beta_{J_k,0} - R_{J_{k-1},J_k}^{\top}\beta_{J_{k-1},0})\\
&\geq C\Vert \beta_{J_k,0} - R_{J_{k-1},J_k}^{\top}\beta_{J_{k-1},0} \Vert _2^2.
\end{align*}
This shows the desired result.

So together we have that 
\begin{align*}
   \mathbb E_{k-1}  \Vert \Delta\widehat{\beta}_{k}\Vert^2_2 & \leq    \left(1  - C\eta_k + C\boldsymbol{1}_k\right)  \Vert \Delta\widehat{\beta}_{k - 1}  \Vert_2^2  + C \left( \eta_kJ_k^{1-2s} + \eta_k J_k^3 \Vert \Delta\check\theta_{k-1}\Vert_2^2 + \eta_k^2 J_k  + \boldsymbol{1}_k J_k^{-2s}\right).
\end{align*}
Now define
\[
b(k)=k^{-\alpha_{\eta} -(2s-1)\alpha_{\dagger}}
+k^{-\alpha_{\eta}+3\alpha_{\dagger}-2\alpha_{\theta}}
 \log^{2c_\theta}(k)
+k^{-2\alpha_{\eta}+\alpha_{\dagger}}.
\]
The logarithm in the middle term is required by
 \autoref{new_theta_condition_for_sieve}.  The preceding inequality
then becomes
\begin{align*}
   \mathbb E_{k-1}  \Vert \Delta\widehat{\beta}_{k}\Vert^2_2 & \leq    \left(1  - Ck^{-\alpha_{\eta}} + C\boldsymbol{1}_k\right)  \Vert \Delta\widehat{\beta}_{k - 1}  \Vert_2^2   + C \left(  b(k) + \boldsymbol{1}_k k^{-2s\alpha_{\dagger}}\right), \qquad \text{a.s}
\end{align*}
Now using the above dynamics, we can analyze the behavior of $\Delta\widehat{\beta}_k$. Note that Theorem 1 of \citet{robbins1971convergence} can not be directly applied here due to non-contraction of $ \mathbb E_{k-1}  \Vert \Delta\widehat{\beta}_{k}\Vert^2_2$  when sieve dimension changes. To deal with this issue, for any $m\geq J_0$, we define $\mathcal{J}(m) = \min\{k: J_k \geq m\}$, which is the starting period of sieve dimension $m$. Then we look at dynamics between $\mathbb{E}_{\mathcal{J}(m-1)-1}\Vert \Delta\widehat\beta_{\mathcal{J}(m)-1}\Vert_2^2$ and $\Vert \Delta\widehat\beta_{\mathcal{J}(m-1)-1}\Vert_2^2$. Obviously, 
\begin{align*}
    \mathbb{E}_{\mathcal{J}(m-1)}\Vert \Delta\widehat\beta_{\mathcal{J}(m)-1}\Vert_2^2 & \leq \left(1  - C(\mathcal{J}(m)-1)^{-\alpha_{\eta}}\right)  \mathbb{E}_{\mathcal{J}(m-1)}\Vert \Delta\widehat{\beta}_{\mathcal{J}(m) - 2}  \Vert_2^2 + Cb(\mathcal{J}(m)-1) \qquad \text{a.s}\\
    & \leq \left(1  - C(\mathcal{J}(m)-1)^{-\alpha_{\eta}}\right) \left(1  - C(\mathcal{J}(m)-2)^{-\alpha_{\eta}}\right) \mathbb{E}_{\mathcal{J}(m-1)}\Vert \Delta\widehat{\beta}_{\mathcal{J}(m) - 3}  \Vert_2^2\\
    & + C\left(b(\mathcal{J}(m)-1) + \left(1  - C(\mathcal{J}(m)-1)^{-\alpha_{\eta}}\right)b(\mathcal{J}(m)-2)\right), \qquad \text{a.s}\\
    &  \leq \left[\prod_{j=\mathcal{J}(m-1)+1}^{\mathcal{J}(m)-1}(1- Cj^{-\alpha_{\eta}})\right]\Vert \Delta\widehat\beta_{\mathcal{J}(m-1)}\Vert_2^2 \\
& \quad + C \left[\prod_{j=\mathcal{J}(m-1)+1}^{\mathcal{J}(m)-1}(1- Cj^{-\alpha_{\eta}})\right]\sum_{j=\mathcal{J}(m-1)+1}^{\mathcal{J}(m)-1}b(j)\left[\prod_{l=\mathcal{J}(m-1)+1}^{j}(1- Cl^{-\alpha_{\eta}})\right]^{-1}.
\end{align*}
Using \autoref{lemma4}, when $\frac{1}{2}<\alpha_{\eta}<1$,  we have that for each $m$,  \[  \left[\prod_{j=\mathcal{J}(m-1) + 1}^{\mathcal{J}(m)-1}(1- Cj^{-\alpha_{\eta}})\right]\asymp  \exp\left(-\frac{C}{1-\alpha_{\eta}}\left(\mathcal{J}(m)^{1-\alpha_{\eta}} - \mathcal{J}(m-1)^{1-\alpha_{\eta}}\right)\right).\] 
We can also show that
\begin{align*}
& \sum_{j=\mathcal{J}(m-1)+1}^{\mathcal{J}(m) - 1}b(j)\left[\prod_{l=\mathcal{J}(m-1)+1}^{j}(1- Cl^{-\alpha_{\eta}})\right]^{-1}\\
& \leq C\,\mathcal{J}(m)^{\alpha_{\eta}}\,b(\mathcal{J}(m))
\exp\left(\tfrac{C}{1-\alpha_{\eta}}\left(\mathcal{J}(m)^{1-\alpha_{\eta}}
-\mathcal{J}(m-1)^{1-\alpha_{\eta}}\right)\right).
\end{align*}
To see this, write $Q_j = \prod_{l=1}^{j}(1-Cl^{-\alpha_{\eta}})$, so that the summand equals
$b(j)Q_{\mathcal{J}(m-1)}\,Q_{j}^{-1}$.
Extending the (positive) summation range and applying
\autoref{lemma4}(ii) and then \autoref{lemma4}(i),
\begin{align*}
\sum_{j=\mathcal{J}(m-1)+1}^{\mathcal{J}(m)-1} b(j)\,Q_{\mathcal{J}(m-1)}\,Q_{j}^{-1}
&\leq Q_{\mathcal{J}(m-1)}\sum_{j=1}^{\mathcal{J}(m)-1} b(j)\,Q_{j}^{-1}\\
&\leq C\, Q_{\mathcal{J}(m-1)}\cdot
\mathcal{J}(m)^{\alpha_{\eta}}\,b(\mathcal{J}(m))
\exp\left(\tfrac{C}{1-\alpha_{\eta}}\mathcal{J}(m)^{1-\alpha_{\eta}}\right)\\
&\leq C\,\mathcal{J}(m)^{\alpha_{\eta}}\,b(\mathcal{J}(m))
\exp\left(\tfrac{C}{1-\alpha_{\eta}}\left(\mathcal{J}(m)^{1-\alpha_{\eta}}
-\mathcal{J}(m-1)^{1-\alpha_{\eta}}\right)\right),
\end{align*}
where the last step uses \autoref{lemma4}(i) to bound
$Q_{\mathcal{J}(m-1)}\leq C\exp(-\tfrac{C}{1-\alpha_{\eta}}\mathcal{J}(m-1)^{1-\alpha_{\eta}})$.
By the mean value theorem and $\mathcal{J}(m)\asymp m^{1/\alpha_{\dagger}}$, when $1-\alpha_{\eta} - \alpha_{\dagger}>0$, we have that 
\[
\mathcal{J}(m)^{1-\alpha_{\eta}} - \mathcal{J}(m-1)^{1-\alpha_{\eta}}
\asymp \mathcal{J}(m)^{-\alpha_{\eta}}\left(\mathcal{J}(m)-\mathcal{J}(m-1)\right)
\asymp m^{-\frac{\alpha_{\eta}}{\alpha_{\dagger}}}\cdot m^{\frac{1-\alpha_{\dagger}}{\alpha_{\dagger}}}
= m^{\frac{1-\alpha_{\eta}-\alpha_{\dagger}}{\alpha_{\dagger}}},
\]
which yields the upper bound. This leads to that 
\begin{align*}
\mathbb{E}_{\mathcal{J}(m-1)}\Vert \Delta\widehat\beta_{\mathcal{J}(m)-1}\Vert_2^2 & \leq C\exp\left(-C_{\beta}m^{\frac{1-\alpha_{\eta}-\alpha_{\dagger}}{\alpha_{\dagger}}}\right) \Vert \Delta\widehat\beta_{\mathcal{J}(m-1)}\Vert_2^2 \\
& + C\left(m^{1-2s}
+m^{3-\frac{2\alpha_{\theta}}{\alpha_{\dagger}}}
 \log^{2c_\theta}(m)
+m^{1-\frac{\alpha_{\eta}}{\alpha_{\dagger}}}\right), \ \  \text{a.s.}
\end{align*}
 and because the changing pseudo-true sieve coefficients only has impacts of magnitude $m^{-2s}$ at sieve dimension $m$, which is effectively dominated by the $m^{1-2s}$ term,  we eventually have that 
 \begin{align*}
\mathbb{E}_{\mathcal{J}(m-1)-1}\Vert \Delta\widehat\beta_{\mathcal{J}(m)-1}\Vert_2^2 & \leq C\exp\left(-C_{\beta}m^{\frac{1-\alpha_{\eta}-\alpha_{\dagger}}{\alpha_{\dagger}}}\right) \Vert \Delta\widehat\beta_{\mathcal{J}(m-1)-1}\Vert_2^2 \\
& + C\left(m^{1-2s}
+m^{3-\frac{2\alpha_{\theta}}{\alpha_{\dagger}}}
 \log^{2c_\theta}(m)
+m^{1-\frac{\alpha_{\eta}}{\alpha_{\dagger}}}\right), \qquad  \text{a.s.}
\end{align*}
Define $\underline{\alpha}_F = \min\{2s-1, 2\alpha_{\theta}/\alpha_{\dagger} - 3, \alpha_{\eta}/\alpha_{\dagger}-1\}>0$, then 
\[
\mathbb{E}_{\mathcal{J}(m-1)-1}\Vert \Delta\widehat\beta_{\mathcal{J}(m)-1}\Vert_2^2\leq C\exp\left(-C_{\beta}m^{\frac{1-\alpha_{\eta}-\alpha_{\dagger}}{\alpha_{\dagger}}}\right) \Vert \Delta\widehat\beta_{\mathcal{J}(m-1)-1}\Vert_2^2 +Cm^{-\underline{\alpha}_F}\log^{2c_{\theta}}(m), \qquad \text{a.s.}
\]
Fix
$r>2c_\theta+1$.  Then
\begin{align*}
    & \mathbb{E}_{\mathcal{J}(m-1)-1}\left[m^{-1+\underline{\alpha}_F}\log^{-r}(m)\Vert \Delta\widehat\beta_{\mathcal{J}(m)-1}\Vert_2^2\right] \\
& \leq \frac{Cm^{-1+\underline{\alpha}_F}\log^{-r}(m)}{(m-1)^{-1+\underline{\alpha}_F}\log^{-r}(m-1)}\exp\left(-C_{\beta}m^{\frac{1-\alpha_{\eta}-\alpha_{\dagger}}{\alpha_{\dagger}}}\right) \left[(m-1)^{-1+\underline{\alpha}_F}\log^{-r}(m-1)\Vert \Delta\widehat\beta_{\mathcal{J}(m-1)-1}\Vert_2^2\right]\\
& + Cm^{-1}\log^{-r+2c_\theta}(m), \ \  \text{a.s.}\\
    & \leq C\exp\left(-C_{\beta}m^{\frac{1-\alpha_{\eta}-\alpha_{\dagger}}{\alpha_{\dagger}}}\right) \left[(m-1)^{-1+\underline{\alpha}_F}\log^{-r}(m-1)\Vert \Delta\widehat\beta_{\mathcal{J}(m-1)-1}\Vert_2^2\right] + Cm^{-1}\log^{-r+2c_\theta}(m), \ \  \text{a.s.}
\end{align*}
Since $r>2c_\theta+1$ makes
$\sum_{m\geq2}m^{-1}\log^{-r+2c_\theta}(m)<\infty$, Theorem~1 of
\citet{robbins1971convergence} shows that
$m^{-1+\underline{\alpha}_F}\log^{-r}(m)\Vert
\Delta\widehat\beta_{\mathcal{J}(m)-1}\Vert_2^2$ converges almost surely to a
finite random variable.  Thus
\[
\Vert \Delta\widehat\beta_{\mathcal{J}(m)-1}\Vert_2^2 = O\left(m^{1 - \underline{\alpha}_F}\log^{r}(m)\right), \qquad \text{a.s.}
\] 
Then we have that 
\begin{align*}
   \mathbb E_{k-1}  \Vert \Delta\widehat{\beta}_{k}\Vert^2_2 & \leq    \left(1  - Ck^{-\alpha_{\eta}} \right)   \Vert \Delta\widehat{\beta}_{k - 1}  \Vert_2^2 + + C k^{-\alpha_{\eta} - \alpha_{\dagger}\underline{\alpha}_{F}}\log^{2c_{\theta}}(k)\\
   & + C\boldsymbol{1}_k k^{\alpha_{\dagger}(1 - \underline{\alpha}_F)} \log^{r}(k)  + C\boldsymbol{1}_kk^{-2\alpha_{\dagger}s}, \qquad \text{a.s}
\end{align*}
where recall that $r>2c_{\theta}+1$. Note that   $\frac{k^{\alpha_{\eta} + \alpha_{\dagger}\underline{\alpha}_{F} - 1}\log^{-r}(k)}{(k-1)^{\alpha_{\eta} + \alpha_{\dagger}\underline{\alpha}_{F} - 1}\log^{-r}(k-1)} \leq 1 + Ck^{-1}$. Then
\begin{align*}
  & \mathbb E_{k-1} \left[k^{\alpha_{\eta} + \alpha_{\dagger}\underline{\alpha}_{F} - 1}\log^{-r}(k) \Vert \Delta\widehat{\beta}_{k}\Vert^2_2\right]\\
  &\leq    \left(1  - Ck^{-\alpha_{\eta}}\right)\left(1+Ck^{-1}\right)  \left[(k-1)^{\alpha_{\eta} + \alpha_{\dagger}\underline{\alpha}_{F} - 1}\log^{-r}(k-1)\Vert \Delta\widehat{\beta}_{k - 1}  \Vert_2^2\right]\\
   &  + C k^{-1}\log^{-r+2c_{\theta}}(k) + C\boldsymbol{1}_k k^{\alpha_{\eta} +\alpha_{\dagger} -1 }  + C\boldsymbol{1}_k k^{-2\alpha_{\dagger}s +\alpha_{\eta} +\alpha_{\dagger}\underline{\alpha}_F-1}\log^{-r}(k) , \qquad \text{a.s.}
\end{align*}
$\sum_{k=1}^{\infty}k^{-1}\log^{-r+2c_{\theta}}(k)<\infty$ for any $r>1+2c_{\theta}$. Moreover, $\sum_k \boldsymbol{1}_k k^{\alpha_{\eta}+\alpha_{\dagger}-1} = \sum_m m^{\frac{\alpha_{\eta}+\alpha_{\dagger}-1}{\alpha_{\dagger}}}<\infty$ because $(\alpha_{\eta}+\alpha_{\dagger}-1)<-2\alpha_{\dagger}$. Also, since $\underline{\alpha}_F\leq 2s-1$,  $-2\alpha_{\dagger}s +\alpha_{\eta} + \alpha_{\dagger}\underline{\alpha}_F  - 1\leq \alpha_{\eta} -\alpha_{\dagger}-1$, so   $\sum_k \boldsymbol{1}_k k^{-2\alpha_{\dagger}s +\alpha_{\eta} +\alpha_{\dagger}\underline{\alpha}_F-1}\log^{-r}(k)<\infty$.  This, by the  Theorem 1 of \citet{robbins1971convergence},  shows that $k^{\alpha_{\eta} + \alpha_{\dagger}\underline{\alpha}_{F} - 1}\log^{-r}(k) \Vert \Delta\widehat{\beta}_{k}\Vert^2_2$ converges almost surely to a finite random variable   and we prove that 
\[
\Vert \Delta\widehat\beta_k \Vert_2 = O\left(k^{\frac{1-\alpha_{\eta} - \alpha_{\dagger}\underline{\alpha}_{F}}{2}}\log^{\frac{r}{2}}(k)\right), \qquad \text{a.s}
\]

\subsubsection{ Refine the Convergence Rate of $\Vert \Delta\widehat\beta_k\Vert_2$}\label{C1.part2}

\vspace{0.2cm}
\textit{Summary: This part shows that}
\[
\Vert \Delta\widehat{\beta}_N\Vert_2  = O\left( N^{-\frac{\alpha_{\eta} - \alpha_{\dagger}}{2}}\log^{\frac{1}{2}}(N)\right), \qquad  \text{a.s.}
\]

 \vspace{0.2cm}
Recall that in the previous step,  we have derived the following dynamics of $\Delta\widehat\beta_{k}$:  \begin{align*}
    \Delta\widehat\beta_k  & = \left(\mathbb{I}_{J_k} - \eta_k \check{\Gamma}_{J_k, k}\right)R^{\top}_{J_{k-1},J_k}\Delta\widehat\beta_{k-1} + \left(\mathbb{I}_{J_k} - \eta_k \check{\Gamma}_{J_k, k}\right)\left(R_{J_{k-1},J_k}^{\top}\beta_{J_{k-1}, 0} - \beta_{J_k, 0}\right)\\
    &  + \frac{\eta_k}{B}\sum_{i=1}^B\varepsilon_{i,k} \Psi_{J_{k}} \left(Z_{0,i,k}\right) + \frac{\eta_k}{B}\sum_{i=1}^B \left( F_0({Z}_{0,i,k}) - \mathbb{P}_{J_{k}}(F_0)(Z_{0,i,k})\right)\Psi_{J_{k}} \left( Z_{0,i,k}\right)\\
   & +  \eta_k \left[\underset{\zeta_{1,k}}{\underbrace{\frac{1}{B}\sum_{i=1}^B\left(F_0\left(Z_{0,i,k}\right) - F_0\left(\check{Z}_{i,k}\right)\right)\Psi_{J_{k}} \left(\check{Z}_{i,k}\right)}} + \underset{\zeta_{2,k}}{\underbrace{ \frac{1}{B}\sum_{i=1}^B \varepsilon_{i,k} (\Psi_{J_{k}} \left(\check{Z}_{i,k}\right) - \Psi_{J_{k}}\left(Z_{0,i,k}\right) )}}\right.\\
   & \left.  +  \underset{ \zeta_{3,k}}{\underbrace{\frac{1}{B}\sum_{i=1}^B \left(\left( F_0(\check{Z}_{i,k}) - \mathbb{P}_{J_{k}}(F_0)(\check{Z}_{i,k})\right)\Psi_{J_{k}} \left(\check{Z}_{i,k}\right) -\left( F_0(Z_{0,i,k}) - \mathbb{P}_{J_{k}}(F_0)(Z_{0,i,k})\right)\Psi_{J_{k}} \left( Z_{0,i,k}\right)\right) }}\right].
\end{align*}
For notational simplicity, further denote 
\[
A_ k = \left(\mathbb{I}_{J_k} - \eta_k  {\Gamma}_{J_k}\right)R^{\top}_{J_{k-1},J_k}
\]
\[
\zeta_{4,k} =   \frac{1}{B}\sum_{i=1}^B\varepsilon_{i,k} \Psi_{J_{k}} \left(Z_{0,i,k}\right) + \frac{1}{B}\sum_{i=1}^B \left( F_0({Z}_{0,i,k}) - \mathbb{P}_{J_{k}}(F_0)(Z_{0,i,k})\right)\Psi_{J_{k}} \left( Z_{0,i,k}\right),
\]
\[
\zeta_{5,k} = \left(\Gamma_{J_k} - \check \Gamma_{J_k, k}\right)R_{J_{k-1}, J_k}^{\top}\Delta\widehat\beta_{k-1},
\]
and 
\[
\zeta_{6,k} = \left(\mathbb{I}_{J_k} - \eta_k \check{\Gamma}_{J_k, k}\right)\left(R_{J_{k-1},J_k}^{\top}\beta_{J_{k-1}, 0} - \beta_{J_k, 0}\right).
\]
We then have that 
\[
\Delta\widehat\beta_k  = A_k \Delta\widehat\beta_{k-1} + \eta_k\sum_{l=1}^5 \zeta_{l,k} + \zeta_{6,k}. 
\]
This leads to \begin{align*}
\Delta\widehat\beta_N
&=
\left(\prod_{k=1}^N A_k\right)\Delta\widehat\beta_0
+
\sum_{k=1}^N
\left(\prod_{j=k+1}^N A_j\right)
\left(
\eta_k\sum_{l=1}^5 \zeta_{l,k}+\zeta_{6,k}
\right),
\end{align*}
where $
\prod_{k=1}^N A_k = A_N A_{N-1}\cdots A_1$, 
$\prod_{j=k+1}^N A_j = A_N A_{N-1}\cdots A_{k+1}
$, and $\prod_{j=N+1}^N A_j = \mathbb I_{J_N}.$
Now we analyze the a.s. order of the above terms one by one. 

For the first term $\left(\prod_{j=1}^N A_j\right)\Delta\widehat\beta_0$, we have that $\Vert A_k \Vert_{\mathrm{op}}\leq 1- C\eta_k$ when $J_{k} = J_{k-1}$ and $\Vert A_k \Vert_{\mathrm{op}}\leq C$ when $J_k >J_{k-1}$. So 
\[
\left\Vert \left(\prod_{k=1}^N A_k\right)\Delta\widehat\beta_0 \right\Vert_{2}\leq \exp\left(-C\sum_{k=1}^N \eta_{k} + CJ_N\right)\Vert \Delta\widehat\beta_0\Vert _2 \leq \exp\left(-CN^{1-\alpha_{\eta}}\right)\Vert \Delta\widehat\beta_0\Vert _2
\]
when $1-\alpha_{\eta}>\alpha_{\dagger}$, which decreases exponentially and is thus negligible. In the following we only need to focus on $\zeta$-terms.   

For $\sum_{k=1}^N
\left(\prod_{j=k+1}^N A_j\right)
\eta_k\zeta_{1,k},
$ again using \autoref{lemma4}, when $1-\alpha_{\eta}>\alpha_{\dagger}$, we have that 
\begin{align*}
& \left\Vert \sum_{k=1}^N
\left(\prod_{j=k+1}^N A_j\right)
\eta_k\zeta_{1,k} \right\Vert_2\\
& \leq C\sum_{k=1}^N \exp\left( - C\sum_{j=k+1}^N\eta_j + C\left(J_N - J_k\right)\right)\eta_k J_{k}^{\frac{1}{2}}\frac{1}{B}\sum_{i=1}^B \Vert X_{i,k}\Vert_2 \Vert \Delta\check{\theta}_{k-1} \Vert_2 \\
& \leq C\exp\left(-CN^{1-\alpha_{\eta}}\right)\sum_{k=1}^N\exp\left(Ck^{1-\alpha_{\eta}}\right)k^{\frac{\alpha_{\dagger}}{2}-\alpha_{\eta}- \alpha_{\theta}}\log^{c_\theta}(k)\frac{1}{B}\sum_{i=1}^B \Vert X_{i,k}\Vert_2, \qquad \text{a.s}\\
& = O\left( N^{\frac{\alpha_{\dagger}}{2}-\alpha_{\theta}}\log^{c_\theta}(N)\right) \qquad  \text{a.s.}
\end{align*}
The last line order is confirmed by \autoref{as converence of sums}.

To show the rate of $\left\Vert \sum_{k=1}^N
\left(\prod_{j=k+1}^N A_j\right)
\eta_k\zeta_{2,k} \right\Vert_2 $, we denote   $Q_N =\sum_{k=1}^N
\left(\prod_{j=k+1}^N A_j\right)
\eta_k\zeta_{2,k}$.  Then  
\[
\mathbb E_{N-1}\Vert Q_{N}\Vert_2^2
\leq \left( 1- CN^{-\alpha_{\eta}} + C\boldsymbol{1}_k\right)
\Vert Q_{N-1}\Vert_2^2
+ CN^{-2\alpha_{\eta}+3\alpha_{\dagger} - 2\alpha_{\theta}}
  \log^{2c_\theta}(N), \qquad \text{a.s.}
  \]
Here the cross term vanishes conditionally because $\mathbb{E}_{k-1}\varepsilon
_{i,k}=0$.  The innovation exponent is
$2\alpha_\eta+2\alpha_\theta-3\alpha_\dagger$ because $\mathbb{E}\Vert X\Vert_2^2<\infty$.  Applying the proof of \autoref{C1.part1}, we easily have that for any $r>1+2c_{\theta}$, 
\[
\left\Vert \sum_{k=1}^N
\left(\prod_{j=k+1}^N A_j\right)
\eta_k\zeta_{2,k} \right\Vert_2
=O \!\left(
 N^{-\{\alpha_\theta+\alpha_\eta-(3\alpha_\dagger+1)/2\}}\log^{\frac{r}{2}}(N)
\right), \qquad \text{a.s.}
\]
Following the proof of $\zeta_{1,k}$ term, we can also verify that 
\begin{align*}
\left\Vert \sum_{k=1}^N
\left(\prod_{j=k+1}^N A_j\right)
\eta_k\zeta_{3,k} \right\Vert_2   \leq
CN^{(\frac{3}{2}-s)\alpha_{\dagger} - \alpha_{\theta}}
\log^{c_\theta}(N), \qquad \text{a.s}
\end{align*}

We next control  $\zeta_{5,k}$,  postponing the discussion of $\zeta_{4,k}$ to the end.    We show that this term does not determine the convergence rate of $\Vert \Delta\widehat\beta_k\Vert_2$. To show this,   we will show that for any $\alpha_{\beta}>0$, if $\Vert \Delta\widehat\beta_{k}\Vert_2  = O_{\text{a.s.}}(k^{-\alpha_{\beta}}\text{PolyLog})$, then \[\left\Vert\sum_{k=1}^N
\left(\prod_{j=k+1}^N A_j\right)
\eta_k\zeta_{5,k}\right\Vert_2 = O\left(N^{-\alpha_{\beta}-\delta}\text{PolyLog}\right), \qquad \text{a.s.}
\] for some $\delta>0$. To show this, first let \[Q_N = \sum_{k=1}^N
\left(\prod_{j=k+1}^N A_j\right)
\eta_k \left(\Gamma_{J_k} - \widehat  \Gamma_{J_k, k}\right)R_{J_{k-1}, J_k}^{\top}\Delta\widehat\beta_{k-1}\boldsymbol{1}\left(\Vert \Delta\widehat\beta_{k-1}\Vert_2\leq k^{-\alpha_{\beta}}\text{PolyLog}\right).\] 
Then \[\text{Var}(Q_N)\leq  CN^{-\alpha_{\eta} + 2\alpha_{\dagger} - 2\alpha_{\beta}}\text{PolyLog}.\]For $N$ sufficiently large, the largest increment of the above partial sum is bounded by $CN^{-\alpha_{\eta} + \alpha_{\dagger} -\alpha_{\beta} }\text{PolyLog}$. By \citet{freedman1975tail,troppfreedmans}, we have that 
\[
\Vert Q_N\Vert_2 = O\left(N^{\frac{-\alpha_{\eta}+2\alpha_{\dagger}}{2}-\alpha_{\beta}}\text{PolyLog} + N^{-\alpha_{\eta}+\alpha_{\dagger}-\alpha_{\beta}}\text{PolyLog}\right), \qquad \text{a.s.} 
\]
This shows that $\Vert Q_N\Vert_2 = O_{a.s.}(N^{\frac{-\alpha_{\eta}+2\alpha_{\dagger}}{2}-\alpha_{\beta}}\text{PolyLog})$. And by appropriately choosing the poly-log terms, we can guarantee that the difference between $Q_N$ and \[\sum_{k=1}^N
\left(\prod_{j=k+1}^N A_j\right)
\eta_k \left(\Gamma_{J_k} - \widehat  \Gamma_{J_k, k}\right)R_{J_{k-1}, J_k}^{\top}\Delta\widehat\beta_{k-1} \] decays almost surely exponentially quickly. This shows that \[
\left\Vert \sum_{k=1}^N
\left(\prod_{j=k+1}^N A_j\right)
\eta_k \left(\Gamma_{J_k} - \widehat  \Gamma_{J_k, k}\right)R_{J_{k-1}, J_k}^{\top}\Delta\widehat\beta_{k-1}\right\Vert_2 = O\left(N^{-\frac{\alpha_{\eta}-2\alpha_{\dagger}}{2}-\alpha_{\beta}}\text{PolyLog}\right), \qquad \text{a.s.}
\]
On the other side, let $Q_N = \sum_{k=1}^N
\left(\prod_{j=k+1}^N A_j\right)
\eta_k \left(\check\Gamma_{J_k,k} - \widehat  \Gamma_{J_k, k}\right)R_{J_{k-1}, J_k}^{\top}\Delta\widehat\beta_{k-1}$, it is direct to verify that 
\begin{align*}
    \Vert Q_N\Vert_2  \leq (1- CN^{-\alpha_{\eta}} + C\boldsymbol{1}_N)\Vert Q_{N-1}\Vert_2 + C\eta_N J_N^2 \Vert \Delta\check\theta_{N-1}\Vert_2\Vert \Delta\widehat\beta_{N-1}\Vert_2  \sum_{i=1}^B \Vert X_{i,k}\Vert_2
\end{align*}  By \autoref{as converence of sums}, we have 
\[
\Vert Q_N\Vert_2  =O\left(N^{2\alpha_{\dagger}-\alpha_{\theta}-\alpha_{\beta}}\text{PolyLog}\right), \qquad \text{a.s.}  
\]
Define $\delta = \min\{\frac{\alpha_{\eta}-2\alpha_{\dagger}}{2},  \alpha_{\theta} -2\alpha_{\dagger}\}>0$ under \autoref{rate}, we obviously have that  
\[\left\Vert\sum_{k=1}^N
\left(\prod_{j=k+1}^N A_j\right)
\eta_k\zeta_{5,k}\right\Vert_2 = O\left(N^{-\alpha_{\beta}-\delta}\text{PolyLog}\right), \qquad \text{a.s.}
\]
This shows the result. 

For $\sum_{k=1}^N
\left(\prod_{j=k+1}^N A_j\right)
\zeta_{6,k}$, we obviously have that 
\[
\left\Vert\sum_{k=1}^N
\left(\prod_{j=k+1}^N A_j\right)
\zeta_{6,k}\right\Vert_2 \leq C\exp\left(-CN^{1-\alpha_{\eta}}\right)\sum_{k=1}^N\exp\left( Ck^{1-\alpha_{\eta}}\right)\boldsymbol{1}_k J_k^{-s}\leq CN^{-s\alpha_{\dagger}}, 
\]
which is effectively the bias of sieve approximation.

We finally look at $\zeta_{4,k}$. We first look at the order of   the second term  \[ \frac{\eta_k}{B}\sum_{i=1}^B \left( F_0({Z}_{0,i,k}) - \mathbb{P}_{J_{k}}(F_0)(Z_{0,i,k})\right)\Psi_{J_{k}} \left( Z_{0,i,k}\right)\]  Denote the weighted sum as $Q_N$, then 
\[
\mathbb{E}_{N-1}\Vert Q_N\Vert_{2}^2
\leq (1-C\eta_N+C\boldsymbol{1}_N)\Vert Q_{N-1}\Vert_{2}^2
+ C\eta_N^2 J_N^{1 - 2s}\leq (1-CN^{-\alpha_{\eta}}+C\boldsymbol{1}_N)\Vert Q_{N-1}\Vert_{2}^2
+ CN^{-2\alpha_{\eta}+(1-2s)\alpha_{\dagger}}.
\]
Here $ 2\alpha_\eta+(2s-1)\alpha_\dagger>1$ under
\autoref{rate}.  Thus  
\[
 \Vert Q_N\Vert_2
 =O \!\left(
 N^{-\alpha_\eta - (s-\frac12)\alpha_\dagger+ \frac12}\text{PolyLog}\right), \qquad \text{a.s.}
\]

It remains to analyze the order of the following sum
\[
\sum_{k=1}^N\left(\prod_{j=k+1}^N A_j\right)
\frac{\eta_k}{B}\sum_{i=1}^B \varepsilon_{i,k}\Psi_{J_k}(Z_{0,i,k}) = \sum_{k=1}^N
\frac{1}{B}\sum_{i=1}^B \varepsilon_{i,k}\eta_k\left(\prod_{j=k+1}^N A_j\right)\Psi_{J_k}(Z_{0,i,k}).
\]
which is of the main focus here in this proof section. We will adopt the proof idea following \citet{chen2015optimal}. Note that in the following proofs, since $B$ is a fixed batch size, we sometimes absorb it into constant without explicit notification. 

First of all, 
it's direct to verify that when $1-\alpha_{\eta}>\alpha_{\dagger}$,  
\[
\left\Vert\eta_k\left(\prod_{j=k+1}^N A_j\right)\Psi_{J_k}(Z_{0,i,k}) \right\Vert_2 \leq C J_k^{\frac{1}{2}}\eta_k \exp\left(-C(N^{1-\alpha_{\eta}} - k^{1-\alpha_{\eta}})\right) 
\]
so 
\[
\max_{1\leq i\leq B}\max_{1\leq k\leq N}\left\Vert\eta_k\left(\prod_{j=k+1}^N A_j\right)\Psi_{J_k}(Z_{0,i,k}) \right\Vert_2 \leq CN^{-\alpha_{\eta} + \frac{\alpha_{\dagger}}{2}}.\] To analyze the behavior of $\sum_{k=1}^N
\frac{1}{B}\sum_{i=1}^B \varepsilon_{i,k}\eta_k\left(\prod_{j=k+1}^N A_j\right)\Psi_{J_k}(Z_{0,i,k})$, let $M_N$  be a positive number to be chosen later, and $\boldsymbol{1}_{i,k}^{M_N}$ as the indicator of whether $|\varepsilon_{i,k}| \leq M_N$ occurs and $\boldsymbol{1}_{i,k}^{M_N, c} = 1- \boldsymbol{1}_{i,k}^{M_N}$. Following \citet{chen2015optimal}, we define \begin{align}\label{truncated e}
\varepsilon_{1,N,i,k } = \varepsilon_{i,k}\boldsymbol{1}_{i,k}^{M_N}   - \mathbb E(\varepsilon_{i,k}\boldsymbol{1}_{i,k}^{M_N}|Z_{0,i,k}), \ \  \varepsilon_{2,N,i, k } = \varepsilon_{i,k} - \varepsilon_{1,N,i,k}.
\end{align}
We obviously have that   $|\varepsilon_{1,N,i,k}|\leq 2M_N$. Also, we have that, under \autoref{condition6} about the conditional variance of $\varepsilon$, there holds 
\begin{align*} \mathbb{E}\left(\left.\varepsilon_{1,N,i,k}^2\right|Z_{0,i,k}\right)& = \mathbb{E}\left(\left.\varepsilon_{i,k}^2 \boldsymbol{1}_{i,k}^{M_N} + \left(\mathbb E(\varepsilon_{i,k}\boldsymbol{1}_{i,k}^{M_N}|Z_{0,i,k})\right)^2 -2\varepsilon_{i,k}\boldsymbol{1}_{i,k}^{M_N} \mathbb E(\varepsilon_{i,k}\boldsymbol{1}_{i,k}^{M_N}|Z_{0,i,k})\right|Z_{0,i,k}\right)\\
& \leq \mathbb{E}\left(\left.\varepsilon_{i,k}^2 \boldsymbol{1}_{i,k}^{M_N}\right|Z_{0,i,k}\right) - \left(\mathbb E(\varepsilon_{i,k}\boldsymbol{1}_{i,k}^{M_N}|Z_{0,i,k})\right)^2 \leq \sup_{z} \mathbb{E}\left(\left.\varepsilon^2 \right|Z_{0} = z\right)<\infty,
\end{align*}and \begin{align*}\mathbb{E}\left|\varepsilon_{2,N,i,k}\right|& \leq \mathbb{E}\left(\left|\varepsilon_{i,k}\right|\boldsymbol{1}_{i,k}^{M_N,c}\right) + \mathbb{E}\left|\mathbb{E}\left(\left.\varepsilon_{i,k} \boldsymbol{1}_{i,k}^{M_N,c}\right|Z_{0,i,k}\right)\right|\\
& \leq 2\mathbb{E}\left(|\varepsilon_{i,k}|\boldsymbol{1}_{i,k}^{M_N,c}\right)\leq 2\mathbb{E}(|\varepsilon_{i,k}|^{\kappa}\boldsymbol{1}_{i,k}^{M_N,c})M_N^{1-\kappa}\leq CM_N^{1-\kappa}.\end{align*} Note that 
\begin{align*}
    \sum_{k=1}^N
\frac{1}{B}\sum_{i=1}^B \varepsilon_{i,k}\eta_k\left(\prod_{j=k+1}^N A_j\right)\Psi_{J_k}(Z_{0,i,k}) & = \sum_{k=1}^N
\frac{1}{B}\sum_{i=1}^B \varepsilon_{1,N,i,k}\eta_k\left(\prod_{j=k+1}^N A_j\right)\Psi_{J_k}(Z_{0,i,k})\\
    & + \sum_{k=1}^N
\frac{1}{B}\sum_{i=1}^B \varepsilon_{2, N,i,k}\eta_k\left(\prod_{j=k+1}^N A_j\right)\Psi_{J_k}(Z_{0,i,k}).
\end{align*} For any $\varrho>0$, Markov's inequality together \autoref{lemma4} lead to 
\begin{align*}
   & P\left(\left\Vert  \sum_{k=1}^N
\frac{1}{B}\sum_{i=1}^B \varepsilon_{2, N,i,k}\eta_k\left(\prod_{j=k+1}^N A_j\right)\Psi_{J_k}(Z_{0,i,k})  \right\Vert_2 >\varrho \right)\\
  & \leq \varrho^{-1}\sum_{k=1}^N \eta_{k}\exp\left(-C(N^{1-\alpha_{\eta}} - k^{1-\alpha_{\eta}})\right)J_k^{\frac{1}{2}}\mathbb E|\varepsilon_{2,N,i,k}|  \leq C\varrho^{-1} N^{\frac{\alpha_{\dagger}}{2}} M_N^{1-\kappa}.
\end{align*}On the other side, noting that
\begin{align*}
&  \sum_{k=1}^N\sum_{i=1}^B \mathbb{E} \left[\varepsilon_{1,N,i,k}^2\left\Vert \eta_k\left(\prod_{j=k+1}^N A_j\right)\Psi_{J_k}(Z_{0,i,k})\right\Vert_2^2 \right] \\
& = \sum_{k=1}^N\sum_{i=1}^B \mathbb{E} \left[\mathbb{E}(\varepsilon_{1,N,i,k}^2|Z_{0,i,k})\left\Vert \eta_k\left(\prod_{j=k+1}^N A_j\right)\Psi_{J_k}(Z_{0,i,k})\right\Vert_2^2 \right]\\  
& \leq C\sum_{k=1}^N\sum_{i=1}^B J_k\eta_k^2\exp\left(-C(N^{1-\alpha_{\eta}}-k^{1-\alpha_{\eta}})\right)\leq C N^{-\alpha_{\eta} + \alpha_{\dagger}}.
\end{align*}
Then by using the results from \citet{tropp2012user}, we have that 
\begin{align*}
  P\left(\left\Vert   \sum_{k=1}^N
\frac{1}{B}\sum_{i=1}^B \varepsilon_{1, N,i,k}\eta_k\left(\prod_{j=k+1}^N A_j\right)\Psi_{J_k}(Z_{0,i,k})  \right\Vert_2 >\varrho \right)  & \leq C\exp\left(C_1\log(N) -C_2\varrho^2N^{\alpha_{\eta} - \alpha_{\dagger}}\right)\\
  & + C\exp\left(C_3\log(N) -  C_4\varrho M_N^{-1} N^{\alpha_{\eta} - \frac{\alpha_{\dagger}}{2}}\right)
\end{align*}
 As a result, if we choose $\varrho_N = \sqrt{(C_1+2)N^{\alpha_{\dagger} - \alpha_{\eta}}\log(N)/C_2}$, we have that
\[
 \sum_{N=1}^{\infty}C\exp\left(C_1\log(N)
 -C_2\varrho_N^2N^{\alpha_{\eta} - \alpha_{\dagger}}\right)
 \leq C\sum_{N=1}^{\infty}N^{-2}<\infty.
\]
Furthermore, if we choose $M_N = \sqrt{C_4^2 (C_1+2)N^{\alpha_{\eta}}/C_2(C_3+2)^2\log(N)}$, we have that 
\[
\sum_{N=1}^{\infty}  C\exp\left(C_3\log(N) -  C_4\varrho_N M_N^{-1} N^{\alpha_{\eta} - \frac{\alpha_{\dagger}}{2}}\right)\leq C\sum_{N=1}^{\infty}N^{-2} < \infty,
\]
and finally when $\frac{\alpha_{\eta}}{2} - \frac{\alpha_{\eta}(\kappa - 1)}{2}<-1$, or equivalently, $\kappa > 2 + 2/\alpha_{\eta} $, which is imposed directly by \autoref{condition6}, we have that 
\[
\sum_{N=1}^{\infty} \varrho_N^{-1} N^{\frac{\alpha_{\dagger}}{2}}M_N^{1-\kappa} \leq \sum_{N=1}^{\infty}\frac{CN^{\frac{\alpha_{\eta}}{2}}\log^{-\frac{1}{2}}(N)}{N^{\frac{\alpha_{\eta} (\kappa - 1)}{2}}\log^{-\frac{\kappa-1}{2}}(N)}< \infty
\]
This shows that  
\[
\left\Vert  \sum_{k=1}^N
\frac{1}{B}\sum_{i=1}^B \varepsilon_{i,k}\eta_k\left(\prod_{j=k+1}^N A_j\right)\Psi_{J_k}(Z_{0,i,k})  \right\Vert_2   =O\left( N^{-\frac{\alpha_{\eta} - \alpha_{\dagger}}{2}}\log^{\frac{1}{2}}(N)\right), \qquad \text{a.s.}
\]
As a result, when $\alpha_{\eta}<2\alpha_{\theta}, 2\alpha_{\theta} + \alpha_{\eta} > 1 + 2\alpha_{\dagger}, \alpha_{\eta} +2s\alpha_{\dagger}>1, \alpha_{\eta}   < (2s+1)\alpha_{\dagger}$ all hold as implied by \autoref{rate}, we have that \[\frac{\alpha_{\eta}-\alpha_{\dagger}}{2}<\min\{\alpha_{\theta}-\frac{\alpha_{\dagger}}{2}, \frac{2\alpha_{\theta} +2\alpha_{\eta} -3\alpha_{\dagger} -1}{2}, \frac{2\alpha_{\eta}+(2s-1)\alpha_{\dagger}-1}{2},s\alpha_{\dagger} \}.\] Then 
\[
\Vert \Delta\widehat{\beta}_N\Vert_2  = O\left( N^{-\frac{\alpha_{\eta} - \alpha_{\dagger}}{2}}\log^{\frac{1}{2}}(N)\right), \qquad \text{a.s.}
\]
This shows the refined rate.

\subsubsection{ Obtain the Sharp Convergence Rate of PR Average Estimator}\label{C1.part3}

Note that the PR average of $\beta$ is given by 
\[
\Delta\overline{\beta}_N = \frac{1}{N}\sum_{k=1}^N R_{J_k, J_N}^{\top}\Delta\widehat\beta_{k} + \frac{1}{N}\sum_{k=1}^N\left(R_{J_k, J_N}^{\top} \beta_{J_k, 0} - \beta_{J_N,0}\right).
\]
So  \[
\Delta\widehat\beta_k  = A_k \Delta\widehat\beta_{k-1} + \eta_k\sum_{l=1}^5 \zeta_{l,k} + \zeta_{6,k}. 
\]
Because $\alpha_\dagger<1$, the integer sequence
$J_k=[J_0k^{\alpha_\dagger}]$ satisfies $J_k-J_{k-1}\in\{0,1\}$ for $k$ large.  Fix a deterministic
$m_0\geq J_0$ so large that such property holds from the block indexed by
$m_0$ onward and every such constant-dimension block has length at least two (because $\mathcal{J}(m+1) - \mathcal{J}(m)\rightarrow \infty$).
For $N$ with $J_N\geq m_0$, define the exact  remainder
$
 \mathcal R_N^{(0)}
 :=\frac1N\sum_{k<\mathcal J(m_0)}
 R_{J_k,J_N}^{\top}\Delta\widehat\beta_k.
$
The prefix contains a fixed number of terms and the nested re-embedding
operators are uniformly bounded, so
$\Vert\mathcal R_N^{(0)}\Vert_2=O_{\rm a.s.}(N^{-1})$.

For any $m_0\leq m\leq J_N$, we have that $A_k = \mathbb{I}_m - \eta_k \Gamma_m $ and $\zeta_{6,k} = 0$ when $\mathcal{J}(m)+1\leq k\leq \mathcal{J}(m+1)-1$ because the sieve dimension does not change across $k-1$ to $k$ at these $k$.  
Then 
\[
\eta_k\Gamma_m \Delta\widehat\beta_{k-1} = \Delta\widehat\beta_{k-1} - \Delta\widehat\beta_k  + \eta_k\sum_{l=1}^5 \zeta_{l,k}, \ \ \mathcal{J}(m)+1\leq k\leq \mathcal{J}(m+1)-1. 
\]
So \[
  \Delta\widehat\beta_{k-1} = \Gamma_m^{-1}\eta_k^{-1} (\Delta\widehat\beta_{k-1} - \Delta\widehat\beta_k )  + \Gamma_m^{-1}\sum_{l=1}^5 \zeta_{l,k}, \ \ \mathcal{J}(m)+1\leq k\leq \mathcal{J}(m+1)-1, 
\]
which, by Abel transformation,  leads to
\begin{align*}
\sum_{k=\mathcal{J}(m)}^{\mathcal{J}(m+1)-2}\Delta\widehat\beta_{k} &= \Gamma_m^{-1}\sum_{k=\mathcal{J}(m)+1}^{\mathcal{J}(m+1)-2}\Delta\widehat\beta_k\left( \eta_{k+1}^{-1} -  \eta_k^{-1}\right) + \Gamma_m^{-1}\Delta\widehat\beta_{\mathcal{J}(m)}\eta_{\mathcal{J}(m)+1}^{-1} \\
& - \Gamma_m^{-1}\Delta\widehat\beta_{\mathcal{J}(m+1)-1}\eta_{\mathcal{J}(m+1)-1}^{-1}+ \Gamma_m^{-1}\sum_{k=\mathcal{J}(m)+1}^{\mathcal{J}(m+1)-1}\sum_{l=1}^5 \zeta_{l,k}.
\end{align*}
We then have the following exact expansion
\begin{align*}
    \frac{1}{N}\sum_{k=1}^N R_{J_k, J_N}^{\top}\Delta\widehat\beta_{k} &  = \mathcal R_N^{(0)}+\frac{1}{N}\sum_{m = m_0}^{J_N - 1} R_{m, J_N}^{\top} \sum_{k = \mathcal{J}(m)}^{\mathcal{J}(m+1)-2}\Delta\widehat\beta_{k} + \frac{1}{N}\sum_{k=\mathcal{J}(J_N)}^{N-1}\Delta\widehat\beta_k \\
    & + \frac{1}{N}\sum_{m= m_0}^{J_N-1}R_{m,J_N}^{\top}\Delta\widehat\beta_{\mathcal{J}(m+1)-1} + \frac{1}{N}\Delta\widehat\beta_N. 
\end{align*}
For the prefix and the terms on the second line of the equality, the results
from \autoref{C1.part2} give
\begin{align*}
    &  \left\Vert \mathcal R_N^{(0)}+\frac{1}{N}\sum_{m= m_0}^{J_N-1}R_{m,J_N}^{\top}\Delta\widehat\beta_{\mathcal{J}(m+1)-1} + \frac{1}{N}\Delta\widehat\beta_N\right\Vert_2  \\
    & \leq \frac{C(\omega)}{N}+\frac{C}{N}\sum_{m=m_0}^{J_N-1} \Vert \Delta\widehat\beta_{\mathcal{J}(m+1)-1}\Vert_2 + \frac{C}{N}\Vert \Delta\widehat\beta_N \Vert_2 \\
    & \leq \frac{C(\omega)}{N}+\frac{C}{N}\sum_{m=m_0}^{J_N} m^{-\frac{\alpha_{\eta}- \alpha_{\dagger}}{2\alpha_{\dagger}}}\log^{\frac{1}{2}}(m) +  \frac{C}{N} J_N^{-\frac{\alpha_{\eta}- \alpha_{\dagger}}{2\alpha_{\dagger}}}\log^{\frac{1}{2}}(N)  = o\left(N^{-\frac{1}{2}}\right), \qquad \text{a.s.}
\end{align*}
when $1 + \alpha_{\eta}>3\alpha_{\dagger}$. So we only need to look at \[\frac{1}{N}\sum_{m = m_0}^{J_N - 1} R_{m, J_N}^{\top} \sum_{k = \mathcal{J}(m)}^{\mathcal{J}(m+1)-2}\Delta\widehat\beta_{k} + \frac{1}{N}\sum_{k=\mathcal{J}(J_N)}^{N-1}\Delta\widehat\beta_k.\]
Define 
$\iota_N=\boldsymbol 1\{\mathcal J(J_N)<N\}$.  This indicator is needed because the
terminal block is empty when $N$ itself is an expansion checkpoint.  Then
\begin{align*}
    &\frac{1}{N}\sum_{m = m_0}^{J_N - 1} R_{m, J_N}^{\top} \sum_{k = \mathcal{J}(m)}^{\mathcal{J}(m+1)-2}\Delta\widehat\beta_{k} + \frac{1}{N}\sum_{k=\mathcal{J}(J_N)}^{N-1}\Delta\widehat\beta_k\\
    & = \frac{1}{N}\sum_{m = m_0}^{J_N - 1} R_{m, J_N}^{\top}\Gamma_m^{-1}\sum_{k=\mathcal{J}(m)+1}^{\mathcal{J}(m+1)-2}\Delta\widehat\beta_k\left( \eta_{k+1}^{-1} -  \eta_k^{-1}\right) + \frac{\iota_N}{N}\Gamma_{J_N}^{-1}\sum_{k=a_N+1}^{N-1}\Delta\widehat\beta_k\left( \eta_{k+1}^{-1} -  \eta_k^{-1}\right) \\
    &+\frac{1}{N}\sum_{m = m_0}^{J_N - 1} R_{m, J_N}^{\top}\Gamma_m^{-1}\sum_{k=\mathcal{J}(m)+1}^{\mathcal{J}(m+1)-1}\sum_{l=1}^5 \zeta_{l,k} + \frac{\iota_N}{N}\Gamma_{J_N}^{-1}\sum_{k=a_N+1}^{N}\sum_{l=1}^5 \zeta_{l,k}\\
    &+ \frac{1}{N}\sum_{m = m_0}^{J_N - 1} R_{m, J_N}^{\top}\left(\Gamma_m^{-1}\Delta\widehat\beta_{\mathcal{J}(m)}\eta_{\mathcal{J}(m)+1}^{-1} - \Gamma_m^{-1}\Delta\widehat\beta_{\mathcal{J}(m+1)-1}\eta_{\mathcal{J}(m+1)-1}^{-1}\right)\\
    & + \frac{\iota_N}{N}\left(\Gamma_{J_N}^{-1}\Delta\widehat\beta_{\mathcal J(J_N)}\eta_{\mathcal J(J_N)+1}^{-1} - \Gamma_{J_N}^{-1}\Delta\widehat\beta_{N}\eta_{N}^{-1}\right).
\end{align*}
Since in \autoref{C1.part2} we have shown that $\Vert \Delta\widehat{\beta}_N\Vert_2  = O_{\rm a.s.}( N^{-\frac{\alpha_{\eta} - \alpha_{\dagger}}{2}}\log^{\frac{1}{2}}(N))$, for $N$ sufficiently large, the first line on RHS is bounded by
\begin{align*}
    \frac{C}{N}\sum_{k=1}^N k^{\alpha_{\eta}-1}\Vert \Delta\widehat\beta_k \Vert_2\leq\frac{C}{N}\sum_{k=1}^N k^{\alpha_{\eta}-1 -\frac{\alpha_{\eta} - \alpha_{\dagger}}{2}}\log^{\frac{1}{2}}(k))\leq CN^{\frac{\alpha_{\eta}+\alpha_{\dagger}-2}{2}}\log^{\frac{1}{2}}(N)), \qquad \text{a.s.}
\end{align*}
The last two lines are together bounded by 
\[
\frac{C}{N}\sum_{m=m_0}^{J_N} m^{\frac{\alpha_{\eta}}{\alpha_{\dagger}} - \frac{\alpha_{\eta}- \alpha_{\dagger}}{2\alpha_{\dagger}}}\log^{\frac{1}{2}}(m) + \frac{C}{N}  N^{\alpha_{\eta} - \frac{\alpha_{\eta}- \alpha_{\dagger}}{2}}\log^{\frac{1}{2}}(N)\leq CN^{\frac{\alpha_{\eta}+3\alpha_{\dagger} -2}{2}}\log^{\frac{1}{2}}(N), \qquad \text{a.s.}
\]
When  $\alpha_{\eta}+3\alpha_{\dagger}<1$, the above two terms are both of order $o(N^{-1/2})$. So we only need to look at the summations of $\zeta$-terms. Recall from the proof of \autoref{C1.part1} and \autoref{C1.part2}, we have $\Vert \zeta_{1,k}\Vert_2 \leq CJ_k^{\frac{1}{2}}\sum_{i=1}^B\Vert X_{i,k}\Vert_2\Vert \Delta\check\theta_{k-1}\Vert_2$, $\Vert \zeta_{2,k}\Vert_2 \leq CJ_k^{\frac{3}{2}}\sum_{i=1}^B|\varepsilon_{i,k}|\Vert X_{i,k}\Vert_2\Vert \Delta\check\theta_{k-1}\Vert_2$, $\Vert \zeta_{3,k}\Vert_2 \leq CJ_k^{\frac{3}{2}-s}\sum_{i=1}^B\Vert X_{i,k}\Vert_2\Vert \Delta\check\theta_{k-1}\Vert_2$, $\Vert \zeta_{4,k}\Vert_2 \leq CJ_k^{\frac{1}{2}}(\sum_{i=1}^B|\varepsilon
_{i,k}|+J_k^{-s})$, and $\Vert \zeta_{5,k}\Vert_2 \leq CJ_k\Vert\Delta\widehat\beta_{k-1}\Vert_2$. When $s>1$, applying  \autoref{as converence of sums}
\begin{align*}
    & \left\Vert \frac{1}{N}\sum_{m = m_0}^{J_N - 1} R_{m, J_N}^{\top}\Gamma_m^{-1}\sum_{k=\mathcal{J}(m)+1}^{\mathcal{J}(m+1)-1}\sum_{l=1}^5 \zeta_{l,k} + \frac{1}{N}\Gamma_{J_N}^{-1}\sum_{k=\mathcal{J}(J_N)+1}^{N}\sum_{l=1}^5 \zeta_{l,k}- \frac{1}{N}\sum_{k=1}^NR_{J_k, J_N}^{\top}\Gamma_{J_k}^{-1} \sum_{l=1}^5 \zeta_{l,k}  \right\Vert_2\\
    &\leq \frac{C(\omega)}{N}+\frac{C}{N}\sum_{m=m_0}^{J_N}\sum_{l=1}^5 \Vert \zeta_{l,\mathcal{J}(m) }\Vert_2 \leq \frac{C(\omega)}{N}+\frac{C}{N}\sum_{m=m_0}^{J_N} \left(m^{\frac{1}{2}}\sum_{i=1}^B\Vert X_{i,\mathcal{J}(m)}\Vert_2 \Vert \Delta\check\theta_{\mathcal{J}(m)-1}\Vert_2\right) \\
    & + \frac{C}{N}\sum_{m=J_0}^{J_N} \left(m^{\frac{3}{2}}\sum_{i=1}^B\Vert X_{i,\mathcal{J}(m)}\Vert_2 |\varepsilon_{i,\mathcal{J}(m)}|\Vert\Delta\check\theta_{\mathcal{J}(m)-1}\Vert_2\right) +\frac{C}{N}\sum_{m=J_0}^{J_N}\left(m^{\frac{1}{2}}\sum_{i=1}^B |\varepsilon_{i,\mathcal{J}(m)}| +m^{\frac{1}{2}-s}\right)\\
    & + \frac{C}{N}\sum_{m=J_0}^{J_N}m\Vert \Delta\widehat\beta_{\mathcal{J}(m)-1}\Vert_2 = O\left(N^{\frac{5\alpha_{\dagger}}{2} - \alpha_{\theta}-1} + N^{\frac{3\alpha_{\dagger}}{2}-1} + N^{\frac{5\alpha_{\dagger}}{2}-\frac{\alpha_{\eta}}{2} -1}\right), \qquad \text{a.s.} 
\end{align*}
up to some poly-log terms. When $1+2\alpha_{\theta}>5\alpha_{\dagger}$, $3\alpha_{\dagger}<1$, and $1+\alpha_{\eta}>5\alpha_{\dagger}$, we have that the above terms are all $o(N^{-\frac{1}{2}})$ a.s. 
 Then we only need to look at $\frac{1}{N}\sum_{k=1}^NR_{J_k, J_N}^{\top}\Gamma_{J_k}^{-1} \sum_{l=1}^5 \zeta_{l,k}$. We look at each $\zeta_{l,k}$ separately. 

For $l=1$, put $d_{i,k}=Z_{0,i,k}-\check Z_{i,k}=-X_{i,k}^{\top}\Delta\check\theta_{k-1}$.  Taylor's theorem, applied once to $F_0$ and once to $\Psi_{J_k}$, gives the  expansion
\[
 \bigl(F_0(Z_{0,i,k})-F_0(\check Z_{i,k})\bigr)
 \Psi_{J_k}(\check Z_{i,k})
 =\nabla_zF_0(Z_{0,i,k})d_{i,k}\Psi_{J_k}(Z_{0,i,k})
   +\mathcal R_{i,k},
 \qquad
 \Vert\mathcal R_{i,k}\Vert_2\leq CJ_k^{3/2}d_{i,k}^2.
\]
The remainder bound follows from the bounded first two derivatives of
$F_0$ and the sieve derivative envelopes.  Hence
\begin{align*}
    \frac{1}{N}\sum_{k=1}^NR_{J_k, J_N}^{\top}\Gamma_{J_k}^{-1}  \zeta_{1 ,k} & = \frac{1}{N}\sum_{k=1}^NR_{J_k, J_N}^{\top}\Gamma_{J_k}^{-1} \frac{1}{B}\sum_{i=1}^B\left(F_0\left(Z_{0,i,k}\right) - F_0\left(\check{Z}_{i,k}\right)\right)\Psi_{J_{k}} \left(\check{Z}_{i,k}\right)\\
    & = -\frac{1}{N}\sum_{k=1}^NR_{J_k, J_N}^{\top}\Gamma_{J_k}^{-1} \frac{1}{B}\sum_{i=1}^B\left(\nabla_z F_0\left(Z_{0,i,k}\right)\right)\Psi_{J_k}(Z_{0,i,k})X_{i,k}^{\top}\Delta\check\theta_{k-1}\\
    & + O\left(\frac{1}{N}\sum_{k=1}^N\sum_{i=1}^BJ_k^{\frac{3}{2}}\Vert X_{i,k}\Vert_2^2\Vert \Delta\check\theta_{k-1}\Vert_2^2\right).
\end{align*}
Let $T_k=B^{-1}\sum_{i=1}^B\Vert X_{i,k}\Vert_2^2$; then
$\mathbb E_{k-1}T_k^2\leq C$.   The exact power-sum bound and
 \autoref{as converence of sums} give
\[
 \frac1N\sum_{k=1}^NJ_k^{3/2}T_k
       \Vert\Delta\check\theta_{k-1}\Vert_2^2
 =O\left(N^{3\alpha_\dagger/2-2\alpha_\theta}\text{PolyLog}\right), \qquad \text{a.s.}
\]
The strict inequality
$\alpha_\theta>1/4+3\alpha_\dagger/4$ makes the above 
$o_{\rm a.s.}(N^{-1/2})$.
Further more, define $Q_{k} =  \frac{1}{B}\sum_{i=1}^B\left(\nabla_z F_0\left(Z_{0,i,k}\right)\right)\Psi_{J_k}(Z_{0,i,k})X_{i,k}^{\top} - \mathbb{E}[\left(\nabla_z F_0\left(Z_{0}\right)\right)\Psi_{J_k}(Z_{0})X^{\top}]$, we have that 
\begin{align*}
    & \frac{1}{N}\sum_{k=1}^NR_{J_k, J_N}^{\top}\Gamma_{J_k}^{-1} \frac{1}{B}\sum_{i=1}^B\left(\nabla_z F_0\left(Z_{0,i,k}\right)\right)\Psi_{J_k}(Z_{0,i,k})X_{i,k}^{\top}\Delta\check\theta_{k-1}\\
    & = \frac{1}{N}\sum_{k=1}^NR_{J_k, J_N}^{\top}\Gamma_{J_k}^{-1} \mathbb{E}[\left(\nabla_z F_0\left(Z_{0}\right)\right)\Psi_{J_k}(Z_{0})X^{\top}] \Delta\check\theta_{k-1} + \frac{1}{N}\sum_{k=1}^NR_{J_k, J_N}^{\top}\Gamma_{J_k}^{-1} Q_{k}\Delta\check\theta_{k-1}
\end{align*}
Let $S_{N} = \sum_{k=1}^NR_{J_k, J_N}^{\top}\Gamma_{J_k}^{-1} Q_{k}\Delta\check\theta_{k-1}$, we have that $S_{N} = R_{J_{N-1},J_N}^{\top}S_{N-1} + \Gamma_{J_N}^{-1}Q_N\Delta\check\theta_{N-1}$. For each $N$, define \[
\Vert S_{N} \Vert_{\Gamma_{J_{N}}} = \left(S_{N}^{\top}\Gamma_{J_{N}}S_{N}\right)^{\frac{1}{2}}, 
\]
We have that 
\begin{align*}
\Vert S_{N} \Vert_{\Gamma_{J_{N}}}^2 & = S_{N-1}^{\top}R_{J_{N-1},J_N}\Gamma_{J_N}R_{J_{N-1},J_N}^{\top}S_{N-1} + \Delta\check\theta_{N-1}^{\top}Q_N^{\top}\Gamma_{J_N}^{-1}Q_N\Delta\check\theta_{N-1}\\
& + 2S_{N-1}^{\top}R_{J_{N-1},J_N} Q_N\Delta\check\theta_{N-1}.
\end{align*}
Since \begin{align*}
    R_{J_{N-1},J_N}\Gamma_{J_N}R_{J_{N-1},J_N}^{\top} & = \mathbb{E}\left(R_{J_{N-1},J_N}\Psi_{J_N}(Z_0)\Psi_{J_N}(Z_0)^{\top}R_{J_{N-1},J_N}^{\top}\right)\\
    & = \mathbb{E}\left( \Psi_{J_{N-1}}(Z_0)\Psi_{J_{N-1}}(Z_0)^{\top}\right) =\Gamma_{J_{N-1}},
\end{align*} we have that 
\[
\Vert S_{N} \Vert_{\Gamma_{J_{N}}}^2  = \Vert S_{N-1} \Vert_{\Gamma_{{J_{N-1}}}}^2 +  \Delta\check\theta_{N-1}^{\top}Q_N^{\top}\Gamma_{J_N}^{-1}Q_N\Delta\check\theta_{N-1} 
  + 2S_{N-1}^{\top}R_{J_{N-1},J_N}Q_N\Delta\check\theta_{N-1}.
\]
So 
\[
\mathbb{E}_{N-1}\Vert S_{N} \Vert_{\Gamma_{J_{N}}}^2  \leq \Vert S_{N-1} \Vert_{\Gamma_{{J_{N-1}}}}^2  + CJ_N\Vert \Delta\check\theta_{N-1}\Vert_2^2\leq \Vert S_{N-1} \Vert_{\Gamma_{{J_{N-1}}}}^2 + CN^{\alpha_{\dagger} - 2\alpha_{\theta}}\log^{2c_\theta}(N), \ \  \text{a.s.}
\]
This immediately leads to $\Vert S_{N} \Vert_{\Gamma_{J_{N}}}^2 = O_{\text{a.s.}}\left(N^{(1+\alpha_{\dagger}-2\alpha_{\theta})_+}\text{PolyLog}\right)$.   Since $\Gamma_{J}$ has uniformly bounded eigenvalues, we have that 
\begin{align*}
\left\Vert \frac{1}{N}\sum_{k=1}^NR_{J_k, J_N}^{\top}\Gamma_{J_k}^{-1} Q_{k}\Delta\check\theta_{k-1}\right\Vert_2 =\frac{1}{N}\left\Vert S_{N}\right\Vert_2\leq \frac{C}{N}\Vert S_{N} \Vert_{\Gamma_{J_{N}}}  = O\left(N^{\frac{(1+\alpha_{\dagger}-2\alpha_{\theta})_+}{2}-1}\text{PolyLog}\right), \qquad  \text{a.s.}
\end{align*}
which is $o(N^{-\frac{1 }{2}})$ a.s. when $\alpha_{\theta}>\frac{1}{2}\alpha_{\dagger}$.  The above analysis together shows that 
\[
\left\Vert\frac{1}{N}\sum_{k=1}^NR_{J_k, J_N}^{\top}\Gamma_{J_k}^{-1}  \zeta_{1 ,k} +  \frac{1}{N}\sum_{k=1}^NR_{J_k, J_N}^{\top}\Gamma_{J_k}^{-1} \mathbb{E}[\left(\nabla_z F_0\left(Z_{0}\right)\right)\Psi_{J_k}(Z_{0})X^{\top}] \Delta\check\theta_{k-1}\right\Vert_2 = o\left(N^{-\frac{1 }{2}}\right), \qquad \text{a.s}
\]
Using the same strategy, for $l=2$,  let $S_{N} =\sum_{k=1}^NR_{J_k, J_N}^{\top}\Gamma_{J_k}^{-1}  \zeta_{2 ,k}$, we have that 
\begin{align*}
  \mathbb{E}_{N-1}\Vert S_{N} \Vert_{\Gamma_{J_{N}}}^2  \leq \Vert S_{N-1} \Vert_{\Gamma_{{J_{N-1}}}}^2  + CJ_N^3 \Vert \Delta\check\theta_{N-1}\Vert_2^2\leq \Vert S_{N-1} \Vert_{\Gamma_{{J_{N-1}}}}^2 + CN^{3\alpha_{\dagger} - 2\alpha_{\theta}}\log^{2c_\theta}(N), \qquad \text{a.s}
\end{align*}
So $\Vert S_{N} \Vert_{\Gamma_{J_{N}}}^2 = O_{\text{a.s.}}\left(N^{(1+3\alpha_{\dagger}-2\alpha_{\theta})_+}\text{PolyLog}\right)$.  Then \begin{align*}
\left\Vert \frac{1}{N}\sum_{k=1}^NR_{J_k, J_N}^{\top}\Gamma_{J_k}^{-1} \zeta_{2,k}\right\Vert_2 =\frac{1}{N}\left\Vert S_{N}\right\Vert_2\leq \frac{C}{N}\Vert S_{N} \Vert_{\Gamma_{J_{N}}}  = O\left(N^{\frac{(1+3\alpha_{\dagger}-2\alpha_{\theta})_+}{2}-1}\text{PolyLog}\right), \qquad \text{a.s}
\end{align*}
When $\alpha_{\theta}>\frac{3}{2}\alpha_{\dagger}$, the above term is of order $o_{\text{a.s.}}(N^{-\frac{1 }{2}})$.  

  For
$l=3$, using \autoref{as converence of sums}, we have that
\begin{align*}
\left\Vert\frac{1}{N}\sum_{k=1}^NR_{J_k, J_N}^{\top}\Gamma_{J_k}^{-1}  \zeta_{3 ,k} \right\Vert_2 & \leq \frac{C}{N}\sum_{k=1}^{N}\Vert \zeta_{3,k}\Vert_2  \leq   \frac{C}{N}\sum_{k=1}^N J_{k}^{\frac{3}{2}-s } \sum_{i=1}^B \Vert X_{i,k}\Vert_2 \Vert \Delta\check{\theta}_{k-1} \Vert_2 \\
& =  O \left(N^{\left(1+\left(\frac{3}{2} -s\right)\alpha_{\dagger} - \alpha_{\theta}\right)_+ - 1}\text{PolyLog}\right), \qquad \text{a.s.}
\end{align*}
so when $\alpha_{\theta}+(s-\frac{3}{2})\alpha_{\dagger}>\frac{1}{2}$, we have that the above term is of order $o_{\text{a.s.}}(N^{-\frac{1 }{2}})$.

For $l=4$, we look at the second term $ \frac{1}{B}\sum_{i=1}^B \left( F_0({Z}_{0,i,k}) - \mathbb{P}_{J_{k}}(F_0)(Z_{0,i,k})\right)\Psi_{J_{k}} \left( Z_{0,i,k}\right)$ first. Noting that \[\mathbb{E}_{k-1}\sum_{i=1}^B \left( F_0({Z}_{0,i,k}) - \mathbb{P}_{J_{k}}(F_0)(Z_{0,i,k})\right)\Psi_{J_{k}} \left( Z_{0,i,k}\right) =0\] and \[\mathbb{E}_{k-1}\left\Vert \sum_{i=1}^B \left( F_0({Z}_{0,i,k}) - \mathbb{P}_{J_{k}}(F_0)(Z_{0,i,k})\right)\Psi_{J_{k}} \left( Z_{0,i,k}\right)\right\Vert_2^2\leq J_k^{1-2s},\] then by \autoref{as converence of sums} we have that 
\begin{align*}
& \left\Vert \frac{1}{N}\sum_{k=1}^NR_{J_k, J_N}^{\top}\Gamma_{J_k}^{-1}  \left(\frac{1}{B} \sum_{i=1}^B \left( F_0({Z}_{0,i,k}) - \mathbb{P}_{J_{k}}(F_0)(Z_{0,i,k})\right)\Psi_{J_{k}} \left( Z_{0,i,k}\right)\right) \right\Vert_2 \\
& = O \left(N^{\frac{(1 -(2s-1)\alpha_{\dagger} )_+}{2}- 1}\text{PolyLog}\right), \qquad \text{a.s.}
\end{align*}
When $2s>1$, the above term is of order $o_{\text{a.s.}}(N^{-\frac{1 }{2}})$. Then 
\[
\left\Vert\frac{1}{N}\sum_{k=1}^NR_{J_k, J_N}^{\top}\Gamma_{J_k}^{-1}  \zeta_{4 ,k} - \frac{1}{N}\sum_{k=1}^NR_{J_k, J_N}^{\top}\Gamma_{J_k}^{-1}  \frac{1}{B}\sum_{i=1}^B\varepsilon_{i,k} \Psi_{J_{k}} \left(Z_{0,i,k}\right) \right\Vert_2 = o\left(N^{-\frac{ 1}{2}}\right), \qquad  \text{a.s.}
\]
We next look at $l=5$. Note that \[\zeta_{5,k} = \left(\Gamma_{J_k} - \widehat \Gamma_{J_k, k}\right)R_{J_{k-1}, J_k}^{\top}\Delta\widehat\beta_{k-1}+\left(\widehat\Gamma_{J_k,k} - \check \Gamma_{J_k, k}\right)R_{J_{k-1}, J_k}^{\top}\Delta\widehat\beta_{k-1},\]
so we let \[
S_{1,N} = \sum_{k=1}^NR_{J_k, J_N}^{\top}\Gamma_{J_k}^{-1} \left(\Gamma_{J_k} - \widehat \Gamma_{J_k, k}\right)R_{J_{k-1}, J_k}^{\top}\Delta\widehat\beta_{k-1},
\]
and 
\[
S_{2,N} =\sum_{k=1}^NR_{J_k, J_N}^{\top}\Gamma_{J_k}^{-1} \left(\widehat\Gamma_{J_k,k} - \check \Gamma_{J_k, k}\right)R_{J_{k-1}, J_k}^{\top}\Delta\widehat\beta_{k-1}.
\]
Since
\begin{align*}
  \mathbb{E}_{N-1}\Vert S_{1,N} \Vert_{\Gamma_{J_N}}^2  \leq \Vert S_{1,N-1} \Vert_{\Gamma_{J_{N-1}}}^2  + CJ_N^2\Vert \Delta\widehat\beta_{N-1}\Vert_2^2\leq \Vert S_{1,N-1} \Vert_{\Gamma_{J_{N-1}}}^2 + CN^{3\alpha_{\dagger} - \alpha_{\eta}}\mathfrak L_N^2, \qquad \text{a.s.}
\end{align*}
using the proof method before, we can show that 
\[
\frac{1}{N}\Vert S_{1,N}\Vert_2 = O \left(N^{\frac{(1+3\alpha_{\dagger}-\alpha_{\eta})_{+} }{2}-1}\text{PolyLog}\right), \qquad \text{a.s.}
\]
when $\alpha_{\eta}>3\alpha_{\dagger}$, we have that $\frac{1}{N}\Vert S_{1,N}\Vert_2 $ is $o_{\text{a.s.}}(N^{-\frac{1 }{2}})$. For $S_{2,N}$, the mean-value theorem and the basis derivative bound give
\[
 \Vert\widehat\Gamma_{J_k,k}-\check\Gamma_{J_k,k}\Vert_{\rm op}
 \leq CJ_k^2T_k\Vert\Delta\check\theta_{k-1}\Vert_2,
 \qquad
 T_k=\frac1B\sum_{i=1}^B(1+\Vert X_{i,k}\Vert_2),
\]
where $\mathbb E_{k-1}T_k^2\leq C$.  Therefore
\[
 \Vert S_{2,N}\Vert_2
 \leq C\sum_{k=1}^N
 k^{2\alpha_\dagger-\alpha_\theta
       -\frac{\alpha_\eta-\alpha_\dagger}{2}} T_k\text{PolyLog} , \qquad \text{a.s.}
\]
The second part of  \autoref{as converence of sums}  yields
\[
 \frac1N\Vert S_{2,N}\Vert_2
 =O \!\left(N^{^{2\alpha_\dagger-\alpha_\theta
       -\frac{\alpha_\eta-\alpha_\dagger}{2}}}\text{PolyLog}\right), \qquad \text{a.s.}
\]
The strict condition
$2\alpha_\theta+\alpha_\eta>1+5\alpha_\dagger$ leads to that the above is $o_{\rm a.s.}(N^{-1/2})$; the same conclusion holds
for $N^{-1}\sum_{k\leq N}R_{J_k,J_N}^{\top}\Gamma_{J_k}^{-1}\zeta_{5,k}$. 

Together, we have shown that 
\begin{align*}
\Delta\overline{\beta}_N & = \frac{1}{N}\sum_{k=1}^N R_{J_k, J_N}^{\top}\Gamma_{J_k}^{-1}\frac{1}{B}\sum_{i=1}^B \varepsilon_{i,k}\Psi_{J_k}(Z_{0,i,k}) -\frac{1}{N}\sum_{k=1}^N R_{J_k, J_N}^{\top}\Gamma_{J_k}^{-1}\mathbb{E}[(\nabla_z F_0(Z_0))\Psi_{J_k}(Z_0)X^{\top}]\Delta\check\theta_{k-1} \\
&  + \frac{1}{N}\sum_{k=1}^N\left(R_{J_k, J_N}^{\top} \beta_{J_k, 0} - \beta_{J_N,0}\right) + r_N,
\qquad \Vert r_N\Vert_2=o_{\rm a.s.}(N^{-1/2}).
\end{align*}
where the second last term describes the impacts of changing pseudo true coefficients. 
This proves the first result.

To prove the second result, recall that we define  
$ \overline{F}_N(z) = \Psi_{J_N}^{\top}(z)\overline{\beta}_N.
$
We immediately have that 
\begin{align*}
  \overline{F}_N(z) - F_0(z)  
& = (\Psi_{J_N}(z)^{\top}\overline{\beta}_N - \Psi_{J_N}(z)^{\top}\beta_{J_N,0}) + (\mathbb{P}_{J_N}(F_0)(z) - F_0(z))\\
&= \left(\mathbb{P}_{J_N}(F_0)(z) - F_0(z)\right)  + \frac{1}{N}\sum_{k=1}^N \left(\left(\mathbb{P}_{J_k}(F_0)(z) - \mathbb{P}_{J_N}(F_0)(z)\right)\right) \\
&+  \frac{1}{N}\sum_{k=1}^N \Psi_{J_k}(z)^{\top}\Gamma_{J_k}^{-1}\frac{1}{B}\sum_{i=1}^B \varepsilon_{i,k}\Psi_{J_k}(Z_{0,i,k})\\
& -\frac{1}{N}\sum_{k=1}^N \Psi_{J_k}(z)^{\top}\Gamma_{J_k}^{-1}\mathbb{E}[(\nabla_z F_0(Z_0))\Psi_{J_k}(Z_0)X^{\top}]\Delta\check\theta_{k-1} + r_N^F(z),
\end{align*}
where $r_N^F(z):=\Psi_{J_N}(z)^{\top}r_N$.  Hence, by
$\sup_z\Vert\Psi_{J_N}(z)\Vert_2\lesssim J_N^{1/2}$,
\[
\Vert r_N^F(\cdot )\Vert_{\infty} 
 =o \!\left(N^{-(1-\alpha_\dagger)/2}\right), \qquad {\rm a.s.}
\]
The two bias terms combine exactly as
\[
 \mathbb P_{J_N}F_0-F_0+\frac1N\sum_{k=1}^N
 (\mathbb P_{J_k}F_0-\mathbb P_{J_N}F_0)
 =\frac1N\sum_{k=1}^N(\mathbb P_{J_k}F_0-F_0).
\]
Consequently their uniform norm is at most
$CN^{-1}\sum_{k=1}^NJ_k^{-s}$.  Recall that
$\mu_{0}(\theta_0, z) = \mathbb{E}(X|Z_0 = z)$, and put
$h(z)=(\nabla_zF_0(z))\mu_0(\theta_0,z)\in\mathbb R^p$ and
$E_J(z)=\mathbb P_Jh(z)-h(z)$.  Then
\[
\Psi_{J}(z)^{\top}\Gamma_{J}^{-1}\mathbb E\left(  \Psi_{J}(Z_{0})  \nabla_z F_0(Z_{0})X^{\top}\right) = \left\{(\mathbb{P}_{J}h)(z)\right\}^{\top},
\]
Since $\Vert E_J\Vert_\infty\leq CJ^{-s_\mu}$ under
\autoref{condition6}, the generated-regressor contribution in the preceding
expansion, including its sign, is exactly
\begin{align*}
&-h(z)^{\top}\frac1N\sum_{k=1}^N\Delta\check\theta_{k-1}
-\frac1N\sum_{k=1}^NE_{J_k}(z)^{\top}\Delta\check\theta_{k-1}.
\end{align*}
Uniformly in $z$, its norm is bounded by
\[
 C\left\Vert\frac1N\sum_{k=1}^N\Delta\check\theta_{k-1}\right\Vert_2
 +\frac C N\sum_{k=1}^NJ_k^{-s_\mu}
       \Vert\Delta\check\theta_{k-1}\Vert_2
 =C\left(\left\Vert\frac1N\sum_{k=1}^N
       \Delta\check\theta_{k-1}\right\Vert_2+A_{\mu,N}\right).
\]

Finally, to study the behavior of $\frac{1}{N}\sum_{k=1}^N \Psi_{J_k}(z)^{\top}\Gamma_{J_k}^{-1}\frac{1}{B}\sum_{i=1}^B \varepsilon_{i,k}\Psi_{J_k}(Z_{0,i,k})$, we define $\varepsilon_{1,N,i,k}$ and $\varepsilon_{2,N,i,k}$ as in (\ref{truncated e}) with $M_N$ to be chosen later. For any $\varrho>0$ and $M_N$, we have that 
\[
\mathrm{Pr}\left(\sup_{z\in\mathcal Z_0}\left\Vert \frac{1}{N}\sum_{k=1}^N \Psi_{J_k}(z)^{\top}\Gamma_{J_k}^{-1}\frac{1}{B}\sum_{i=1}^B \varepsilon_{2,N,i,k}\Psi_{J_k}(Z_{0,i,k})\right\Vert_2>\varrho \right)\leq C\varrho^{-1}J_NM_N^{1-\kappa}.
\]
For any positive integer $L$,  \autoref{condition7}(iii) provides
$z_1,\ldots,z_L\in\mathcal Z_0$ such that
\[\sup_{z\in\mathcal Z_0}\min_{1\leq l\leq L}
\Vert \Psi_{J_N}(z)-\Psi_{J_N}(z_l)\Vert_2
\leq CJ_N^{3/2}L^{-1}.\]
By nestedness,
$\Psi_{J_k}(z)=R_{J_k,J_N}\Psi_{J_N}(z)$ for every $J_k\leq J_N$, and
$\sup_{J_1\leq J_2}\Vert R_{J_1,J_2}\Vert_{\rm op}<\infty$ under the
normalized-basis construction.  Hence this single terminal net controls all
lower-dimensional rows appearing in the sum, up to a fixed multiplicative
constant.
As a result, 
\begin{align*}
 &  \sup_{z\in\mathcal Z_0}\left\Vert \frac{1}{N}\sum_{k=1}^N \Psi_{J_k}(z)^{\top}\Gamma_{J_k}^{-1}\frac{1}{B}\sum_{i=1}^B \varepsilon_{1,N,i,k}\Psi_{J_k}(Z_{0,i,k})\right\Vert_2\\
 & \leq  \max_{1\leq l \leq L}\left\Vert \frac{1}{N}\sum_{k=1}^N \Psi_{J_k}(z_l)^{\top}\Gamma_{J_k}^{-1}\frac{1}{B}\sum_{i=1}^B \varepsilon_{1,N,i,k}\Psi_{J_k}(Z_{0,i,k})\right\Vert_2 + \frac{CJ_N^{2}}{NL }\sum_{k=1}^N \frac{1}{B}\sum_{i=1}^B\left|\varepsilon_{1,N,i,k}\right|
\end{align*}
Note that $\frac{CJ_N^{2}}{NL }\sum_{k=1}^N \frac{1}{B}\sum_{i=1}^B\left|\varepsilon_{1,N,i,k}\right| \leq  CL^{-1}N^{2\alpha_{\dagger}}M_N$. And 
\begin{align*}
    & \mathrm{Pr}\left(\max_{1\leq l \leq L}\left\Vert \frac{1}{N}\sum_{k=1}^N \Psi_{J_k}(z_l)^{\top}\Gamma_{J_k}^{-1}\frac{1}{B}\sum_{i=1}^B \varepsilon_{1,N,i,k}\Psi_{J_k}(Z_{0,i,k})\right\Vert_2>\varrho\right)\\
    &\leq \sum_{1\leq l \leq L} \mathrm{Pr} \left(\left\Vert \frac{1}{N}\sum_{k=1}^N \Psi_{J_k}(z_l)^{\top}\Gamma_{J_k}^{-1}\frac{1}{B}\sum_{i=1}^B \varepsilon_{1,N,i,k}\Psi_{J_k}(Z_{0,i,k})\right\Vert_2>\varrho\right)\\
    & \leq C \exp\left( \log(L)-C_1N^{1-\alpha_{\dagger}}\varrho^2 \right) + C\exp\left( \log(L) - C_2 N^{1-\alpha_{\dagger}}M_N^{-1}\varrho\right)
\end{align*}
If we choose 
\[\varrho = \sqrt {C_1^{-1}N^{-(1-\alpha_{\dagger})}(\log(L) + 2\log(N))},\]
\[
M_N= \sqrt{\frac{C_2^2N^{1-\alpha_{\dagger}}}{C_1(\log(L) + 2\log(N))}}, 
\]
and 
\[
L = \left\lceil\frac{N^{1+\alpha_{\dagger}}}{\log N}\right\rceil,
\]
we have that 
\[
\sup_{z\in\mathcal Z_0}\left\Vert\frac{1}{N}\sum_{k=1}^N \Psi_{J_k}(z)^{\top}\Gamma_{J_k}^{-1}\frac{1}{B}\sum_{i=1}^B \varepsilon_{1,N,i,k}\Psi_{J_k}(Z_{0,i,k})\right\Vert_2 = O\left( N^{-\frac{1 - \alpha_{\dagger}}{2}}\log^{\frac{1}{2}}(N)\right), a.s.,
\]
and if $\kappa>2+\frac{2(1+\alpha_{\dagger})}{1-\alpha_{\dagger}}$, we have that 
\[
\sum_{N=1}^{\infty} P\left(\sup_{z\in\mathcal Z_0}\left\Vert  \frac{1}{N}\sum_{k=1}^N \Psi_{J_k}(z)^{\top}\Gamma_{J_k}^{-1}\frac{1}{B}\sum_{i=1}^B \varepsilon_{2,N,i,k}\Psi_{J_k}(Z_{0,i,k})\right\Vert_2>\varrho \right) <\infty
\]
So together we have that 
\[
\sup_{z\in\mathcal Z_0}\left|\frac{1}{N}\sum_{k=1}^N \Psi_{J_k}(z)^{\top}\Gamma_{J_k}^{-1}\frac{1}{B}\sum_{i=1}^B \varepsilon_{ i,k}\Psi_{J_k}(Z_{0,i,k})\right| =  O\left( N^{-\frac{1 - \alpha_{\dagger}}{2}}\log^{\frac{1}{2}}(N)\right), a.s.
\]
Combining this display with the exact bias and generated-regressor
decompositions gives
\[
 \Vert\overline F_N-F_0\Vert_{\infty}
 =O_{\mathrm{a.s.}}\!\left(
 \sqrt{\frac{J_N\log N}{N}}+
 \frac1N\sum_{k=1}^NJ_k^{-s}+
 \left\Vert\frac1N\sum_{k=1}^N\Delta\check\theta_{k-1}\right\Vert_2+
 A_{\mu,N}\right),
\]
which is \eqref{eq:C1-supnorm-corrected} and finishes the proof.

\section{Phase II Joint Local Refinement} \label{appendixD}

\subsection{Pseudo Code}
The pseudo code for the second phase local refinement is given in Algorithm \ref{alg:local-online}.

\begin{algorithm}
\caption{Joint Refinement Learner with Plug-in AME Average (Phase II)}
\label{alg:local-online}
\textbf{Require:} streaming batches $\mathcal{W}_k$, $k = 1, \ldots, N$; sieve dimensions $\{J_k\}$; learning rates $\{\xi_k\}$; initial estimate $\widetilde\omega_0 = (\widetilde\theta_0^{\top}, \widetilde\beta_0^{\top})^{\top}$; predictable conditioning matrices $\{\widehat G_k\}$; predictable guardrail centers $\{\widehat{\omega}_k\}$; constant $C_{\Omega} > 0$.
\algline{1}{$\doubleoverline{\theta}_0 \gets \widetilde\theta_0$, \ $\doubleoverline{\beta}_0 \gets \widetilde\beta_0$}
\algline{2}{\textbf{for} $k = 1, 2, \ldots, N$ \textbf{do}}
\algline{3}{\quad Construct
    $\widehat\Omega_k = \left\{\omega \in \mathbb{R}^{p + J_k}: \Vert \omega - \widehat{\omega}_k \Vert_2 \leq \tfrac{C_{\Omega}}{2}\, J_k^{-2}\right\}$}
\algline{4}{\quad Re-embed the pre-update state:
    $\omega_{k-1}^{(k)}\gets
    \mathcal R_{J_{k-1},J_k}^{\top}\widetilde\omega_{k-1}$}
\algline{5}{\quad \textbf{Projected joint gradient update}:
    \[
        \widetilde\omega_k \gets \Pi_{\widehat\Omega_k}\!\left[\omega_{k-1}^{(k)}
        - \frac{\xi_k\widehat G_k}{B}\sum_{i=1}^B
        \nabla_{\omega}\mathcal{L}_{J_k}\!\left(\omega_{k-1}^{(k)},W_{i,k}\right)\right].
    \]
}
\algline{6}{\quad \textbf{Update PR averages:}
    \[
        \doubleoverline{\theta}_k \gets \frac{1}{k}\widetilde\theta_k + \frac{k-1}{k}\doubleoverline{\theta}_{k-1}, \qquad
        \doubleoverline{\beta}_k \gets \frac{1}{k}\widetilde\beta_k + \frac{k-1}{k} R_{J_{k-1}, J_k}^{\top}\,\doubleoverline{\beta}_{k-1}
    \]}
\algline{7}{\textbf{return} $\doubleoverline{\theta}_N$ and  $\doubleoverline{F}_N(\cdot) = \Psi_{J_N}(\cdot)^{\top}\doubleoverline{\beta}_N$}
\end{algorithm}

\subsection{Sufficient Conditions and Discussions}

\begin{condition}\label{cond11}
 (i) $\Vert \nabla_{zz} \Psi_J\Vert_{\infty}\leq CJ^{5/2}$ and $\Vert \nabla_{zzz} \Psi_J\Vert_{\infty}\leq CJ^{7/2}$; (ii) $\mathbb{E}[(1+|\varepsilon_{i,k}|^\kappa)(1+\Vert X_{i,k}\Vert_2^{\kappa})]<\infty$ for some $\kappa> 2 +  \max\{\frac{2}{\alpha_{\xi}}, \frac{2(1+\alpha_{\dagger})}{1-\alpha_{\dagger}}\}$; (iii) $\Vert \nabla_{zz} [F_0(\cdot) - \mathbb{P}_J(F_0)(\cdot)]\Vert _{\infty}\leq C_{F_0}J^{ 2 -s}$, $\Vert \nabla_{zzz} [F_0(\cdot) - \mathbb{P}_J(F_0)(\cdot)]\Vert _{\infty}\leq C_{F_0}J^{ 3 -s}$, and $\Vert \nabla_{zzz} F_0\Vert_{\infty}<\infty$, where $C_{F_0}$ is the constant defined in \autoref{condition5}(ii).
  \end{condition}

\begin{condition}\label{eigenvalue}
There holds
\[
0<\underline\lambda_{\mathcal L}
\leq\inf_J\lambda_{\min}\!\left(
 \nabla_{\omega\omega}\mathcal L_J(\omega_{J,0})\right)
\leq\sup_J\lambda_{\max}\!\left(
 \nabla_{\omega\omega}\mathcal L_J(\omega_{J,0})\right)
\leq\overline\lambda_{\mathcal L}<\infty.
\]
\end{condition}

\begin{condition}\label{cond14}
There is a fixed $\delta_\Omega\in(0,1)$ such that, eventually almost surely,
\[
 \Vert\widehat\omega_k-\omega_{J_k,0}\Vert_2
 \leq(1-\delta_\Omega)\frac{C_\Omega}{2}J_k^{-2}.
\]
Moreover,
$\sup_{\omega\in\widehat\Omega_k}\Vert\nabla_{\omega\omega}\mathcal L_{J_k}(\omega)-
\nabla_{\omega\omega}\mathcal L_{J_k}(\omega_{J_k,0})\Vert_{\mathrm{op}}
\leq\underline\lambda_{\mathcal L}/2$, and, for every fixed $C_\omega>0$,
$\sup_{\Vert\omega-\omega_{J,0}\Vert_2\leq C_\omega J^{-2}}
\Vert\nabla_{\omega\omega}\mathcal L_J(\omega)\Vert_{\mathrm{op}}
\leq C_\omega'$ with $C_\omega'$ independent of $J$.
\end{condition}

For later reference, let the deterministic Phase~II basin be
\begin{equation}\label{deterministic_phaseII_basin}
 \Omega_J:=\left\{\omega\in\mathbb R^{p+J}:
 \Vert\omega-\omega_{J,0}\Vert_2\leq C_\Omega J^{-2}\right\}.
\end{equation}

\begin{condition}[Estimated conditioning matrix]\label{cond_gain}
(i)  For every $k$,
$\widehat G_k\in\mathbb R^{(p+J_k)\times(p+J_k)}$ is
$\mathcal F_{k-1}$-measurable, symmetric, and positive definite; (ii)  There are
deterministic symmetric positive-definite matrices $G_{J,0}$ and constants
$0<\underline g\leq\overline g<\infty$ such that
\[
 \underline g\leq\inf_J\lambda_{\min}(G_{J,0})
 \leq\sup_J\lambda_{\max}(G_{J,0})\leq\overline g.
\]
Suppose also
that, for some $c_G>0$,
\begin{equation}\label{gain_hessian_compatibility}
 \inf_J\inf_{\omega\in\Omega_J}
 \lambda_{\min}\!\left(
 \frac{G_{J,0}\nabla_{\omega\omega}\mathcal L_J(\omega)
       +\nabla_{\omega\omega}\mathcal L_J(\omega)G_{J,0}}{2}
 \right)\geq c_G;
\end{equation}
(iii) Writing
$d_{G,k}:=\Vert\widehat G_k-G_{J_k,0}\Vert_{\rm op}$, suppose that, for some
$\alpha_G>0$,
\begin{equation}\label{gain_rate}
 d_{G,k}\leq \Vert\widehat G_k-G_{J_k,0}\Vert_F
 =O_{\rm a.s.}(k^{-\alpha_G}),
 \qquad
 \alpha_G>\frac{1-\alpha_\xi+\alpha_\dagger}{2}.
\end{equation}
Thus $\widehat G_k$ also has eigenvalues uniformly bounded away from zero and
infinity eventually almost surely. 
\end{condition}

\begin{remark}\label{remark_gain_rate}
The consistency $d_{G,k}=o_{\rm a.s.}(1)$, the uniform eigenvalue bounds, and
\eqref{gain_hessian_compatibility} are already sufficient for localization,
eventual inactivity of the projection, and the raw-iterate rates.  The
compatibility condition is automatic for the inverse-Hessian target
$G_{J,0}=\{\nabla_{\omega\omega}\mathcal L_J(\omega_{J,0})\}^{-1}$ after
$C_\Omega$ is chosen so that the local Hessian-variation bound underlying
 \autoref{cond14} holds throughout $\Omega_J$: the matrix in
\eqref{gain_hessian_compatibility} is then a uniformly small perturbation of
the identity.  For a general positive-definite
target, bounded eigenvalues alone do not ensure Euclidean one-step
contraction; alternatively one may use projection in the
$G_{J,0}^{-1}$-metric, in which case \eqref{gain_hessian_compatibility} is not
needed.

On the other side, a rate is needed only for the strong $o_{\rm a.s.}(N^{-1/2})$ remainder in the
PR representation.  The polynomial condition \eqref{gain_rate} is a
transparent sufficient condition.  The proof below uses the two primitive
requirements
\[
 \sum_{k=1}^{\infty}\frac{J_kd_{G,k}^2}{k}<\infty,
 \qquad
 \frac1{\sqrt N}\sum_{k=1}^N
 d_{G,k}k^{-(\alpha_\xi-\alpha_\dagger)/2}\rm{PolyLog}  \rightarrow0,
 \quad\text{a.s.},
\]
together with $d_{G,k}=o_{\rm a.s.}(1)$ and the analogous, smaller secondary
state term.  The first condition makes the predictably gain-weighted score a
common-Hilbert-space martingale with summable normalized quadratic variation;
the second controls the linear state component.  Condition~\eqref{gain_rate}
implies both because $\alpha_\xi<1-3\alpha_\dagger$ and
$\alpha_G>(1-\alpha_\xi+\alpha_\dagger)/2$.  These primitive conditions may be
imposed directly instead of the polynomial rate.
\end{remark}

\begin{condition}\label{new_rate}
Let $s>\frac{7}{2}$ be the approximation error rate  of $F_0$ defined in \autoref{condition5} and
$s_{\mu}$ be defined in \autoref{condition6}(ii), and define
\begin{equation}\label{def_qstar}
q_{\star} := s+\min(s_{\mu},\, s-1).
\end{equation}
The learning rate   $\alpha_{\xi}$ and
sieve growth exponent $\alpha_{\dagger}$ satisfy
$\max\left\{\tfrac{1}{2}+3\alpha_{\dagger},\ 1+(4-2s)\alpha_{\dagger}\right\}
< \alpha_{\xi} <
\min\left\{1-3\alpha_{\dagger},\ \tfrac{1}{2}+ (s-\tfrac{1}{2})\alpha_{\dagger}\right\},
$
and moreover, $q_{\star}\alpha_{\dagger} > \frac{1}{2}$.
\end{condition}

\begin{remark}
    \autoref{cond11} further regulates the norm of second- and third-order derivatives of spline basis functions. Such regulation is necessary when we consider NLS criterion. Note that for normalized spline basis functions, such regulation is automatically valid. \autoref{eigenvalue} regulates the curvature of the NLS loss function at the pseudo true parameter $\omega_{J,0}$. In the following proofs, we show that $\Vert  \mathbb{M}_J - \nabla_{\omega\omega} \mathcal{L}_J(\omega_{J,0}) \Vert_F\leq CJ^{\frac{3}{2} - s},$  where 
    \[
 \mathbb{M}_{J} = \mathbb{E}\begin{pmatrix}
    (\nabla_z F_0(Z_0))^2XX^{\top} &  \nabla_zF_0(Z_0) X\Psi_J(Z_0)^{\top}\\
    \nabla_{z} F_0(Z_0) \Psi_J(Z_0)X^{\top} & \Psi_J(Z_0)\Psi_J(Z_0)^{\top}
\end{pmatrix}.
\]
When $s>\frac{3}{2}$, \autoref{eigenvalue} is equivalent to assuming uniformly
(with respect to $J$) bounded eigenvalues from both below and above for
$\mathbb M_J$.   \autoref{cond14} imposes a strict interior
condition that is used to show eventual inactivity of projection.   The
Phase~I rates in  \autoref{theorem1} and \autoref{theorem6}, with tuning chosen
strictly inside \autoref{rate}, are sufficient to furnish this margin after an
almost surely finite time.
  
\end{remark}

\subsection{ Discussion on \autoref{cond14}}\label{discussion_on_basin}

This subsection shows that under mild conditions, the local basin with radius of order $J^{-2}$  preserves the curvature. To show this, note that the Hessian matrix of the NLS loss function is given by
\[
\nabla_{\omega\omega}\mathcal{L}_J(\omega, W) = \begin{pmatrix}
    \nabla_{\theta\theta}\mathcal{L}_J(\omega, W) &  \nabla_{\theta\beta}\mathcal{L}_J(\omega, W) \\
    \nabla_{\beta\theta}\mathcal{L}_J(\omega, W) & \nabla_{\beta\beta}\mathcal{L}_J(\omega, W)
\end{pmatrix},
\]
where by simple calculation \[\nabla_{\theta\theta}\mathcal{L}_J(\omega, W)  = ((\nabla_z \Psi_J(x_0 + X^{\top}\theta))^{\top}\beta)^2XX^{\top} - (Y - \Psi_J(x_0+X^{\top}\theta)^{\top}\beta)(\nabla_{zz}\Psi_J(x_0+X^{\top}\theta))^{\top}\beta XX^{\top},\] \[\nabla_{\theta\beta}\mathcal{L}_J(\omega, W)  = (\nabla_z \Psi_J(x_0 + X^{\top}\theta))^{\top}\beta X\Psi_J(x_0+X^{\top}\theta)^{\top} - (Y - \Psi_J(x_0+X^{\top}\theta)^{\top}\beta)X (\nabla_{z}\Psi_J(x_0+X^{\top}\theta))^{\top},\] and \[\nabla_{\beta\beta}\mathcal{L}_J(\omega, W) = \Psi_J(x_0+X^{\top}\theta)\Psi_J(x_0+X^{\top}\theta)^{\top}.\]
 We will calculate the bound of the derivative for each of the Hessian block with the local basin.  In particular, we prove the following result. 

 \begin{lemma}\label{local_basin}
Suppose that, for $m=0,1,2,3$,
\[
 J^{-(2m+1)/2}\Vert\nabla_z^m\Psi_J\Vert_\infty\le C,
 \qquad \Vert\nabla_z^mF_0\Vert_\infty\le C,
 \qquad
 \Vert\nabla_z^m(F_0-\mathbb P_JF_0)\Vert_\infty\le CJ^{m-s},
\]
with $s>3$, and suppose $\mathbb E\Vert X\Vert_2^4<\infty$.
For every fixed $C_1<\infty$ there is $C_2(C_1)<\infty$, independent of
$J$, such that, whenever
$\Vert\omega_i-\omega_{J,0}\Vert_2\le C_1J^{-2}$, $i=1,2$,
\[
 \Vert\nabla_{\omega\omega}\mathcal L_J(\omega_1)
       -\nabla_{\omega\omega}\mathcal L_J(\omega_2)\Vert_F
 \le C_2J^2\Vert\omega_1-\omega_2\Vert_2 .
\]
\end{lemma}

\begin{proof}
Write $t_\theta=x_0+X^\top\theta$,
$f_\beta(t)=\Psi_J(t)^\top\beta$, and
$r_\omega=F_0(Z_0)-f_\beta(t_\theta)$.  On the stated ball, the inverse
inequalities and the approximation assumptions imply, uniformly in $t$,
\begin{equation}\label{local_basin_envelopes}
 |f_\beta(t)|+|f_\beta'(t)|\le C,
 \qquad |f_\beta''(t)|\le CJ^{1/2},
 \qquad |f_\beta^{(3)}(t)|\le CJ^{3/2},
\end{equation}
and
\begin{equation}\label{local_basin_residual}
 |r_\omega|
 \le C\{J^{-3/2}+J^{-2}\Vert X\Vert_2\}.
\end{equation}
Indeed, the coefficient displacement contributes
$J^{(2m+1)/2}J^{-2}$ to the $m$th derivative, while
$f_{\beta_{J,0}}=\mathbb P_JF_0$; the index displacement contributes
$C\Vert X\Vert_2J^{-2}$ to the residual.

It is enough to bound the directional derivative of the population Hessian.
Let $h=(u^\top,v^\top)^\top$ be a unit vector and put
$\dot t=X^\top u$.  Along $\omega+qh$,
\[
 \dot f=f_\beta'\dot t+\Psi_J^\top v,
 \quad
 \dot f'=f_\beta''\dot t+(\nabla_z\Psi_J)^\top v,
 \quad
 \dot f''=f_\beta^{(3)}\dot t+(\nabla_{zz}\Psi_J)^\top v,
 \quad \dot r=-\dot f .
\]
Taking expectation of the sample Hessian replaces $Y$ by $F_0(Z_0)$, so the three population Hessian blocks are
\begin{align*}
 H_{\beta\beta}(\omega)
 &=\mathbb E[\Psi_J(t_\theta)\Psi_J(t_\theta)^\top],\\
 H_{\theta\theta}(\omega)
 &=\mathbb E[\{(f_\beta')^2-r_\omega f_\beta''\}XX^\top],\\
 H_{\theta\beta}(\omega)
 &=\mathbb E[X\{f_\beta'\Psi_J(t_\theta)^\top
                    -r_\omega\nabla_z\Psi_J(t_\theta)^\top\}].
\end{align*}
For the first block,
\[
 \Vert DH_{\beta\beta}(\omega)[h]\Vert_F
 \le 2\,\mathbb E\!\left[
      \Vert\Psi_J(t_\theta)\Vert_2
      \Vert\nabla_z\Psi_J(t_\theta)\Vert_2
      |\dot t|\right]
 \le CJ^2.
\]
For the scalar multiplying $XX^\top$ in $H_{\theta\theta}$,
\[
 D\{(f_\beta')^2-r_\omega f_\beta''\}[h]
 =2f_\beta'\dot f'+\dot f\,f_\beta''-r_\omega\dot f''.
\]
Using \eqref{local_basin_envelopes}--\eqref{local_basin_residual}, its
absolute value is at most
$C\{J^{3/2}+J^{1/2}\Vert X\Vert_2
       +J^{-1/2}\Vert X\Vert_2^2\}$.
Consequently, by $\mathbb E\Vert X\Vert_2^4<\infty$,
\[
 \Vert DH_{\theta\theta}(\omega)[h]\Vert_F\le CJ^{3/2}.
\]
Finally,
\begin{align*}
 D\{f_\beta'\Psi_J^\top-r_\omega(\nabla_z\Psi_J)^\top\}[h]
 ={}&\dot f'\Psi_J^\top+f_\beta'(\nabla_z\Psi_J)^\top\dot t\\
 &+\dot f(\nabla_z\Psi_J)^\top
   -r_\omega(\nabla_{zz}\Psi_J)^\top\dot t .
\end{align*}
The same envelopes, followed by multiplication by $X$ and Hölder's
inequality, give
$\Vert DH_{\theta\beta}(\omega)[h]\Vert_F\le CJ^2$.
The transpose block has the same bound.  Integrating these directional
bounds along the line segment joining $\omega_1$ and $\omega_2$, which lies
in the convex local ball, proves the claim.
\end{proof}

\subsection{Proofs of  \autoref{theorem_NLS_update}}

   Similar to the proof of \autoref{theorem6}, the proof is divided into three steps.  To keep the already lengthy score decomposition readable, the three baseline steps are first written for an auxiliary identity-gain path; throughout those steps, $\widetilde\omega_k$, $D_k$, and $T_k$ refer to that auxiliary path.  Subsection~\ref{proof_estimated_gain} then repeats the needed stability bounds for the actual conditioned path, replaces the identity by the population matrices $G_{J,0}$, and finally controls the predictable estimates $\widehat G_k$.  That subsection verifies both the gain-dependent raw recursion and the cancellation that leaves the PR influence function unchanged. Before starting, we record two auxiliary lemmas: a \emph{centered} bound for the sieve score at the pseudo-true parameter, and a \emph{telescoping} bound for the sieve-target jumps generated by nestedness of the sieve spaces. Throughout, recall $q_{\star}=s+\min(s_{\mu},s-1)$ from \eqref{def_qstar}, $\mathcal{J}(m) = \min\{k: J_k\geq m\}$, and $h(z) := \left(\nabla_z F_0(z)\right)\mu_0(\theta_0, z)$.

\begin{lemma}[Centered sieve score]\label{lemma_centered_score}
Let \autoref{condition5}--\ref{condition7} and \autoref{cond11} hold with $s>2$. Define
\[
b_J := \mathbb{E}\left[\nabla_{\omega}\mathcal{L}_J(\omega_{J,0}, W)\right] = \left(b_{\theta,J}^{\top}, b_{\beta,J}^{\top}\right)^{\top}.
\]
Then (i) $b_{\beta,J} = 0$; (ii) $\Vert b_{\theta, J}\Vert_2\leq CJ^{-q_{\star}}$; (iii) $\mathbb{E}\Vert \nabla_{\omega}\mathcal{L}_J(\omega_{J,0}, W) - b_J \Vert_2^2\leq CJ$.
\end{lemma}

\begin{proof}
Recall $\nabla_{\beta}\mathcal{L}_J(\omega_{J,0}, W) = - (Y - \Psi_J(Z_0)^{\top}\beta_{J,0}) \Psi_J(Z_0)$. Since $\mathbb{E}(\varepsilon|x_0, X)=0$,
\[
b_{\beta,J} = -\mathbb{E}\left[\left(F_0(Z_0) - \mathbb{P}_{J}(F_0)(Z_0)\right)\Psi_{J}(Z_0)\right] = 0
\]
by the definition of the $L_2$-projection \eqref{L2-projection}. This proves (i). For (ii), $\nabla_{\theta}\mathcal{L}_J(\omega_{J,0}, W) = - (Y - \Psi_J(Z_0)^{\top}\beta_{J,0})(\nabla_{z}\Psi_J(Z_0))^{\top}\beta_{J,0} X$, so conditioning on $Z_0$,
\begin{align*}
b_{\theta,J} & = -\mathbb{E}\left[\left(F_0 - \mathbb{P}_{J}(F_0)\right)(Z_0)\, \nabla_z\mathbb{P}_{J}(F_0)(Z_0)\,\mu_0(\theta_0, Z_0)\right]\\
& = -\mathbb{E}\left[\left(F_0 - \mathbb{P}_{J}(F_0)\right)(Z_0)\, h(Z_0)\right]
- \mathbb{E}\left[\left(F_0 - \mathbb{P}_{J}(F_0)\right)(Z_0)
\left(\nabla_z\mathbb{P}_{J}(F_0) - \nabla_z F_0\right)(Z_0)\,\mu_0(\theta_0, Z_0)\right].
\end{align*}
By the $L_2$-projection property, $\mathbb{E}[(F_0 - \mathbb{P}_{J}(F_0))(Z_0)g(Z_0)]=0$ for every $g\in\mathcal{S}_J$ componentwise; applying this with $g = \mathbb{P}_J(h)$ and using \autoref{condition5} and \autoref{condition6}(ii),
\[
\left\Vert \mathbb{E}\left[\left(F_0 - \mathbb{P}_{J}(F_0)\right)(Z_0) h(Z_0)\right]\right\Vert_2 = \left\Vert \mathbb{E}\left[\left(F_0 - \mathbb{P}_{J}(F_0)\right)(Z_0)\left(h - \mathbb{P}_{J}(h)\right)(Z_0)\right]\right\Vert_2\leq CJ^{-s}\cdot C_{\mu}J^{-s_{\mu}}.
\]
For the second term, \autoref{condition5}(ii) gives $\Vert F_0 - \mathbb{P}_J(F_0)\Vert_{\infty}\leq CJ^{-s}$ and $\Vert \nabla_z[\mathbb{P}_J(F_0) - F_0]\Vert_{\infty}\leq CJ^{1-s}$, while $\mathbb{E}\Vert \mu_0(\theta_0, Z_0)\Vert_2\leq (\mathbb{E}\Vert X\Vert_2^2)^{1/2}<\infty$; hence this term is $O(J^{-s}\cdot J^{1-s}) = O(J^{-(2s-1)})$. Since $q_{\star} = s + \min(s_{\mu}, s-1) \leq \min\{s+s_{\mu},\, 2s-1\}$, (ii) follows. For (iii), $\nabla_z\mathbb{P}_J(F_0)$ is uniformly bounded for $s>1$, so $\mathbb{E}\Vert\nabla_{\theta}\mathcal{L}_J(\omega_{J,0},W)\Vert_2^2\leq C\mathbb{E}[(1+\varepsilon^2)\Vert X\Vert_2^2]\leq C$ by \autoref{condition6}; and
\[
\mathbb{E}\Vert\nabla_{\beta}\mathcal{L}_J(\omega_{J,0},W)\Vert_2^2\leq C\,\mathbb{E}\Vert \Psi_J(Z_0)\Vert_2^2 = C\,\mathrm{tr}(\Gamma_J)\leq CJ
\]
by \autoref{condition7}(ii). Subtracting the mean $b_J$ does not increase the order.
\end{proof}

\begin{lemma}[Telescoping bound for sieve-target jumps]\label{lemma_telescoping}
Let \autoref{condition5} and \autoref{condition7} hold, and recall from \autoref{sec2.2} that the sieve spaces are nested with the exact refinement relation \eqref{eq:reembedding}. For $m\geq J_0+1$ define the target jump
\[
\delta_m := \left\Vert \beta_{m,0} - R^{\top}_{m-1,m}\beta_{m-1,0} \right\Vert_2 .
\]
Then there is a constant $C_{\delta}$ independent of $m$ such that:
\begin{enumerate}
    \item[(i)] $\delta_m\leq C_{\delta}m^{-s}$;
    \item[(ii)] $\sum_{m>M}\delta_m^2\leq C_{\delta}^2M^{-2s}$ for every $M\geq J_0$;
    \item[(iii)] for every $\nu\geq 0$ there is $C(\nu)<\infty$ such that for all $M\geq 2J_0$,
    \[
    \sum_{m=J_0+1}^{M} m^{\nu}\delta_m\leq C(\nu)\left(M^{\nu + \frac{1}{2} - s} + \log M\right), \qquad
    \sum_{m=J_0+1}^{M} m^{\nu}\delta_m^2\leq C(\nu)\left(M^{\nu - 2s} + \log M\right).
    \]
\end{enumerate}
\end{lemma}

\begin{proof}
By the exact refinement relation $\Psi_{m-1}(z) = R_{m-1,m}\Psi_m(z)$, the vector $\beta_{m,0} - R^{\top}_{m-1,m}\beta_{m-1,0}$ collects the coefficients, in the basis $\Psi_m$, of the \emph{detail} function
\[
\Delta_m(z) := \mathbb{P}_m(F_0)(z) - \mathbb{P}_{m-1}(F_0)(z).
\]
Since $\mathbb{P}_{m}$ is the orthogonal projection of $L_2(Z_0)$ onto $\mathcal{S}_{m}$ and $\mathcal{S}_{m-1}\subseteq\mathcal{S}_{m}$, we have $\Delta_m\perp \mathcal{S}_{m-1}$ in $L_2(Z_0)$, while $\Delta_{m'}\in\mathcal{S}_{m'}\subseteq \mathcal{S}_{m-1}$ for every $m'<m$. Hence 
\begin{align*}
& \sum_{m=M+1}^{M'}\Vert \Delta_m\Vert_{L_2(Z_0)}^2 = \left\Vert \mathbb{P}_{M'}(F_0) - \mathbb{P}_{M}(F_0)\right\Vert_{L_2(Z_0)}^2\\
& = \left\Vert \mathbb{P}_{M'}\left(F_0 - \mathbb{P}_{M}(F_0)\right)\right\Vert_{L_2(Z_0)}^2 \leq \left\Vert F_0 - \mathbb{P}_{M}(F_0)\right\Vert_{\infty}^2\leq C^2M^{-2s},
\end{align*}
where the last step uses \autoref{condition5}. Letting $M'\rightarrow\infty$ and using $\delta_m^2\leq \lambda_{\mathrm{min}}(\Gamma_m)^{-1}\Vert \Delta_m\Vert^2_{L_2(Z_0)}\leq C\Vert \Delta_m\Vert^2_{L_2(Z_0)}$ by \autoref{condition7}(ii) proves (ii); taking $M = m-1$ and keeping the single summand $\delta_m^2$ proves (i). For (iii), decompose $\{J_0+1,\ldots,M\}$ into dyadic blocks $B_{\ell} = \{m: 2^{\ell}\leq m<2^{\ell+1}\}$. By the Cauchy--Schwarz inequality and (ii),
\[
\sum_{m\in B_{\ell}}m^{\nu}\delta_m\leq 2^{(\ell+1)\nu}\left(2^{\ell}\right)^{\frac{1}{2}}\left(\sum_{m\geq 2^{\ell}}\delta_m^2\right)^{\frac{1}{2}}\leq C2^{\ell\left(\nu+\frac{1}{2}-s\right)},
\qquad
\sum_{m\in B_{\ell}}m^{\nu}\delta_m^2\leq C2^{\ell\left(\nu-2s\right)}.
\]
Summing the geometric series over $\ell\leq \log_2 M$ gives the displayed bounds, where the term $\log M$ covers the cases in which the exponent is nonpositive.
\end{proof}

\begin{remark}\label{remark_multijump}
Since $J_k = [J_0k^{\alpha_{\dagger}}]$, we have $J_k - J_{k-1}\in\{0,1\}$ for all $k$ sufficiently large, so eventually each sieve expansion inserts exactly one knot. For the finitely many early steps with $J_k-J_{k-1}>1$, telescoping through intermediate dimensions and $\sup_{J_1\leq J_2}\Vert R_{J_1,J_2}\Vert_{\mathrm{op}}<\infty$ give
$\Vert R^{\top}_{J_{k-1},J_k}\beta_{J_{k-1},0} - \beta_{J_k,0}\Vert_2\leq C\sum_{m=J_{k-1}+1}^{J_k}\delta_m$,
so these steps only affect constants. In particular, \autoref{lemma_telescoping}(i) recovers the bound $\Vert R^{\top}_{J_{k-1},J_k}\beta_{J_{k-1},0} - \beta_{J_k,0}\Vert_2\leq CJ_k^{-s}$ used in the proof of \autoref{theorem6}, while \autoref{lemma_telescoping}(ii)--(iii) are strictly sharper \emph{on average} across expansions: the multiresolution (nested-projection) structure makes the jump sizes square-summable at rate $M^{-2s}$, which is the property exploited in \autoref{lemma5.2part3} below.
\end{remark}

\begin{lemma}[Common-space martingale averaging]\label{lem:common_space_mds}
Let $U_k=(U_{\theta,k}^{\top},U_{\beta,k}^{\top})^{\top}\in
\mathbb R^{p+J_k}$ be an $\mathcal F_k$-adapted square-integrable martingale
difference array.  Suppose
\begin{equation}\label{common_space_qv}
 \sum_{k=1}^{\infty}\frac1k\,
 \mathbb E_{k-1}\Vert U_k\Vert_2^2<\infty
 \qquad\text{a.s.}
\end{equation}
Then
\begin{equation}\label{common_space_endpoint}
 \frac1{\sqrt N}\left\Vert
 \sum_{k=1}^N\mathcal R_{J_k,J_N}^{\top}U_k
 \right\Vert_2\longrightarrow0\qquad\text{a.s.}
\end{equation}
Moreover, if $C_k\in\mathbb R^{q\times(p+J_k)}$ is
$\mathcal F_{k-1}$-measurable for a fixed $q$ and
$\sup_k\Vert C_k\Vert_{\rm op}<\infty$ almost surely, then
\begin{equation}\label{common_space_rows}
 \max_{1\leq n\leq N}\frac1{\sqrt N}
 \left\Vert\sum_{k=1}^n C_kU_k\right\Vert_2
 \longrightarrow0\qquad\text{a.s.}
\end{equation}
The conclusions remain true after multiplying $U_k$ by any predictable scalar
$a_k$ with $|a_k|\leq1$.
\end{lemma}

\begin{proof}
Put $\mathcal H=\mathbb R^p\oplus L_2(P_Z)$ and define the deterministic
common-coordinate embedding
\[
 \iota_J(v_\theta,v_\beta)
   =(v_\theta,\Psi_J^{\top}v_\beta).
\]
By Condition~\ref{condition7},
\[
 c\Vert v\Vert_2^2\leq\Vert\iota_Jv\Vert_{\mathcal H}^2
 =\Vert v_\theta\Vert_2^2+v_\beta^{\top}\Gamma_Jv_\beta
 \leq C\Vert v\Vert_2^2
\]
uniformly in $J$.  Exact re-embedding gives, for $J_1\leq J_2$,
\begin{equation}\label{common_space_embedding}
 \iota_{J_2}(\mathcal R_{J_1,J_2}^{\top}v)=\iota_{J_1}(v),
\end{equation}
because $\Psi_{J_1}=R_{J_1,J_2}\Psi_{J_2}$.
After localization at the stopping times at which the series in
\eqref{common_space_qv} first exceeds an integer, the Hilbert-space martingale
$\sum_k k^{-1/2}\iota_{J_k}(U_k)$ has summable conditional quadratic
variation and therefore converges almost surely.  Kronecker's lemma in
$\mathcal H$ yields
$N^{-1/2}\sum_{k\leq N}\iota_{J_k}(U_k)\to0$.  Equation
\eqref{common_space_embedding} and norm equivalence prove
\eqref{common_space_endpoint}.

For the row statement, $C_kU_k$ is an $\mathbb R^q$-valued martingale
difference and its normalized conditional quadratic variation is summable by
\eqref{common_space_qv}; the same argument gives
$n^{-1/2}\sum_{k\leq n}C_kU_k\to0$.  Given $\epsilon>0$, this endpoint
convergence bounds all sufficiently late partial sums by $\epsilon\sqrt n$,
while the finitely many early sums divided by $\sqrt N$ vanish.  This proves
\eqref{common_space_rows}.  A predictable bounded scalar multiplier preserves
all arguments.
\end{proof}

\subsubsection{ Obtaining Initial Convergence Rate}\label{lemma5.2part1}

\vspace{0.2cm}
\textit{Summary: This part establishes the following results for $\widetilde\omega_k$: for any $r>1$, }
 \begin{align*}
    \Vert \Delta \widetilde\omega_{k}  \Vert_2 = O\left(k^{\frac{1 +\alpha_{\dagger} -  2\alpha_{\xi} }{2}}\log^{\frac{r}{2}}(k) + k^{\frac{1-\alpha_{\xi}}{2} - q_{\star}\alpha_{\dagger}}\log^{\frac{r}{2}}(k) + k^{-\left(s - \frac{1}{2}\right)\alpha_{\dagger}}\log^{\frac{r}{2}}(k)\right), \qquad \text{a.s.}
\end{align*}
\textit{Under \autoref{new_rate}, each of the three exponents is strictly smaller than $-2\alpha_{\dagger}$, so the projection is inactive almost surely for $k$ large.}

\vspace{0.3cm}

Define $\Delta\widetilde\omega_{k} =  \widetilde\omega_{k} - \omega_{J_k, 0}$, where recall that $\omega_{J_k,0} = (\theta_0^{\top}, \beta_{J_k,0}^{\top})^{\top}$.  Note that 
\begin{align*}
    & \mathcal{R}_{J_{k-1}, J_k}^{\top}\widetilde\omega_{k-1} - \frac{\xi_k}{B}\sum_{i=1}^B \nabla_{\omega}  \mathcal{L}_{J_k}\left(\mathcal{R}_{J_{k-1}, J_k}^{\top}\widetilde\omega_{k-1}, W_{i,k}\right) - \omega_{J_{k}, 0}\\
& =  \mathcal{R}_{J_{k-1}, J_k}^{\top}\Delta\widetilde\omega_{k-1} - \frac{\xi_{k}}{B}\sum_{i=1}^B\left( \nabla_{\omega}\mathcal{L}_{J_{k}} (\mathcal{R}_{J_{k-1}, J_k}^{\top}\widetilde\omega_{k-1}, W_{i,k}) - \nabla_{\omega}\mathcal{L}_{J_{k}} (\mathcal{R}_{J_{k-1}, J_k}^{\top}\omega_{J_{k-1},0}, W_{i,k})\right)\\
& - \frac{\xi_{k}}{B}\sum_{i=1}^B \left( \nabla_{\omega}\mathcal{L}_{J_{k}} (\mathcal{R}_{J_{k-1}, J_k}^{\top}\omega_{J_{k-1},0}, W_{i,k}) - \nabla_{\omega}\mathcal{L}_{J_{k}} (\omega_{J_{k},0}, W_{i,k})\right) \\
&- \frac{\xi_k}{B}\sum_{i=1}^B\nabla_{\omega}\mathcal{L}_{J_{k}} (\omega_{J_{k},0}, W_{i,k}) + \mathcal{R}_{J_{k-1}, J_k}^{\top}\omega_{J_{k-1},0} - \omega_{J_k, 0}\\
& = \int_{0}^1 \left(\mathbb{I}_{p+J_k} - \xi_k \nabla_{\omega\omega} \mathcal{L}_{J_k}\left(\mathcal{R}^{\top}_{J_{k-1}, J_k}\left(\omega_{J_{k-1},0} + \tau\Delta\widetilde{\omega}_{k-1}\right)\right)\right)d\tau\mathcal{R}^{\top}_{J_{k-1}, J_k}\Delta\widetilde\omega_{k-1}\\
& -\frac{\xi_k}{B}\sum_{i=1}^B \int_{0}^1\left(\nabla_{\omega\omega} \mathcal{L}_{J_k}\left(\mathcal{R}^{\top}_{J_{k-1}, J_k}\left(\omega_{J_{k-1},0} + \tau\Delta\widetilde{\omega}_{k-1}\right), W_{i,k}\right) \right.\\
&\left. \ \ \ \ \ \ \ \ \ \ \ \ \ \ \ \ -\nabla_{\omega\omega} \mathcal{L}_{J_k}\left(\mathcal{R}^{\top}_{J_{k-1}, J_k}\left(\omega_{J_{k-1},0} + \tau\Delta\widetilde{\omega}_{k-1}\right)\right)\right)d\tau\mathcal{R}^{\top}_{J_{k-1}, J_k}\Delta\widetilde\omega_{k-1}\\
& - \frac{\xi_{k}}{B}\sum_{i=1}^B \left( \nabla_{\omega}\mathcal{L}_{J_{k}} (\mathcal{R}_{J_{k-1}, J_k}^{\top}\omega_{J_{k-1},0}, W_{i,k}) - \nabla_{\omega}\mathcal{L}_{J_{k}} (\omega_{J_{k},0}, W_{i,k})\right)\\
&- \frac{\xi_k}{B}\sum_{i=1}^B\left(\nabla_{\omega}\mathcal{L}_{J_{k}} (\omega_{J_{k},0}, W_{i,k})  +\varphi_{J_k}(W_{i,k})\right) + \frac{\xi_k}{B}\sum_{i=1}^B \varphi_{J_k}(W_{i,k}) + \mathcal{R}_{J_{k-1}, J_k}^{\top}\omega_{J_{k-1},0} - \omega_{J_k, 0}.
\end{align*} 
where 
\[
\varphi_J(W) =  \begin{pmatrix}
    \left(Y - F_0(Z_0)\right)\nabla_z F_0(Z_0) X\\
    \left(Y - F_0(Z_0)\right) \Psi_J(Z_0)
\end{pmatrix}.
\]
Now we analyze the above terms one by one. When $J_k = J_{k-1}$, $\mathcal{R}_{J_{k-1},J_k}^{\top}\widetilde\omega_{k-1} = \widetilde\omega_{k-1}\in\widehat\Omega_{k-1}$. Since $\widehat\Omega_{k-1}$ is convex by construction, \autoref{eigenvalue} and \autoref{cond14} together lead to that the smallest eigenvalue of $\nabla_{\omega \omega}\mathcal{L}_{J_{k}} (\omega_{J_{k-1},0} + \tau\Delta\widetilde{\omega}_{k-1}) = \nabla_{\omega \omega}\mathcal{L}_{J_{k-1}} (\omega_{J_{k-1},0} + \tau\Delta\widetilde{\omega}_{k-1}) $ is uniformly lower bounded from 0 for all $\tau\in[0,1]$. When $J_k>J_{k-1}$, note that \[\Vert \mathcal{R}^{\top}_{J_{k-1}, J_k}\left( \omega_{J_{k-1},0} + \tau\Delta\widetilde{\omega}_{k-1}\right) - \omega_{J_k, 0}\Vert_2 \leq \Vert \mathcal{R}^{\top}_{J_{k-1}, J_k} \omega_{J_{k-1},0}   - \omega_{J_k, 0}\Vert_2 + C\Vert  \Delta\widetilde\omega_{k-1}\Vert_2  \leq CJ_k^{-2},\]
where the last inequality comes from the fact that $\widehat\Omega_{k}$ has radius of order $J_k^{-2}$, \autoref{cond14},   $s>2$,  and \[\Vert \mathcal{R}^{\top}_{J_{k-1}, J_k} \omega_{J_{k-1},0}   - \omega_{J_k, 0}\Vert_2 = \Vert R^{\top}_{J_{k-1}, J_k} \beta_{J_{k-1},0}   - \beta_{J_k, 0}\Vert_2\leq CJ_k^{-s},\] according to the proof of \autoref{theorem6}. 
So  $\nabla_{\omega\omega} \mathcal{L}_{J_k}(\mathcal{R}^{\top}_{J_{k-1}, J_k}(\omega_{J_{k-1},0} + \tau\Delta\widetilde{\omega}_{k-1}))$ has eigenvalues bounded  (in absolute value) by some constant independent of $k$ and $\tau$ following \autoref{cond14}.  Together there holds 
\begin{align*}
& \mathbb{E}_{k-1}\left\Vert \int_{0}^1 \left(\mathbb{I}_{p+J_k} - \xi_k \nabla_{\omega\omega} \mathcal{L}_{J_k}\left(\mathcal{R}^{\top}_{J_{k-1}, J_k}\left(\omega_{J_{k-1},0} + \tau\Delta\widetilde{\omega}_{k-1}\right)\right)\right)d\tau\mathcal{R}^{\top}_{J_{k-1}, J_k}\Delta\widetilde\omega_{k-1} \right\Vert_2 ^2 \\
& \leq \left(1 - C\xi_k + C\boldsymbol{1}_k\right)\Vert \Delta\widetilde\omega_{k-1}\Vert_2^2,
\end{align*}
regardless of whether $J_k = J_{k-1}$ or $J_k>J_{k-1}$, where recall that $\boldsymbol{1}_k=\boldsymbol{1}(J_k>J_{k-1})$. 

Note that the Hessian matrix of the NLS loss function is given by
\[
\nabla_{\omega\omega}\mathcal{L}_J(\omega, W) = \begin{pmatrix}
    \nabla_{\theta\theta}\mathcal{L}_J(\omega, W) &  \nabla_{\theta\beta}\mathcal{L}_J(\omega, W) \\
    \nabla_{\beta\theta}\mathcal{L}_J(\omega, W) & \nabla_{\beta\beta}\mathcal{L}_J(\omega, W)
\end{pmatrix},
\]
where by simple calculation \[\nabla_{\theta\theta}\mathcal{L}_J(\omega, W)  = ((\nabla_z \Psi_J(x_0 + X^{\top}\theta))^{\top}\beta)^2XX^{\top} - (Y - \Psi_J(x_0+X^{\top}\theta)^{\top}\beta)(\nabla_{zz}\Psi_J(x_0+X^{\top}\theta))^{\top}\beta XX^{\top},\] \[\nabla_{\theta\beta}\mathcal{L}_J(\omega, W)  = (\nabla_z \Psi_J(x_0 + X^{\top}\theta))^{\top}\beta X\Psi_J(x_0+X^{\top}\theta)^{\top} - (Y - \Psi_J(x_0+X^{\top}\theta)^{\top}\beta)X (\nabla_{z}\Psi_J(x_0+X^{\top}\theta))^{\top},\] and \[\nabla_{\beta\beta}\mathcal{L}_J(\omega, W) = \Psi_J(x_0+X^{\top}\theta)\Psi_J(x_0+X^{\top}\theta)^{\top}.\]
When $\omega\in\Omega_J$, inverse inequalities and the $O(J^{-2})$
coefficient radius give
\[
\begin{split}
\left\Vert\Psi_J^\top\beta-F_0\right\Vert_\infty
 &=O(J^{-3/2}+J^{-s}),\\
\left\Vert(\nabla_z\Psi_J)^\top\beta-\nabla_zF_0\right\Vert_\infty
 &=O(J^{-1/2}+J^{1-s}),\\
\left\Vert(\nabla_{zz}\Psi_J)^\top\beta-\nabla_{zz}F_0\right\Vert_\infty
 &=O(J^{1/2}+J^{2-s}).
\end{split}
\]
The last quantity is not uniformly bounded in $J$.  Using the bounded true
derivatives gives the  envelope
\[
\Vert \nabla_{\omega\omega}\mathcal{L}_J(\omega, W) \Vert_F
\leq C J^{1/2}(1+|Y|)\Vert X\Vert_2^2
   +C J^{3/2}(1+|Y|)\Vert X\Vert_2+CJ .
\]
\autoref{cond11}(ii), which in particular has $\kappa>4$, makes the
squared expectation of this envelope $O(J^3)$.  Therefore
\begin{align*}
  \mathbb{E}_{k-1} & \left\Vert \frac{\xi_k}{B}\sum_{i=1}^B \int_{0}^1\left(\nabla_{\omega\omega} \mathcal{L}_{J_k}\left(\mathcal{R}^{\top}_{J_{k-1}, J_k}\left(\omega_{J_{k-1},0} + \tau\Delta\widetilde{\omega}_{k-1}\right), W_{i,k}\right) \right.\right.\\
   &\left.\left.- \nabla_{\omega\omega} \mathcal{L}_{J_k}\left(\mathcal{R}^{\top}_{J_{k-1}, J_k}\left(\omega_{J_{k-1},0} + \tau\Delta\widetilde{\omega}_{k-1}\right)\right)\right)d\tau\mathcal{R}^{\top}_{J_{k-1}, J_k}\Delta\widetilde\omega_{k-1}\right\Vert^2_2 \leq C\xi_k^2J_{k-1}^3\Vert \Delta\widetilde\omega_{k-1}\Vert_2^2.
\end{align*}
This also leads to  
\begin{align*}
    &\mathbb{E}_{k-1} \left\Vert\frac{\xi_{k}}{B}\sum_{i=1}^B \left( \nabla_{\omega}\mathcal{L}_{J_{k}} (\mathcal{R}_{J_{k-1}, J_k}^{\top}\omega_{J_{k-1},0}, W_{i,k}) - \nabla_{\omega}\mathcal{L}_{J_{k}} (\omega_{J_{k},0}, W_{i,k})\right)\right\Vert^2_2 \\
    & \leq C\xi_k^2 J_k^{3}  \Vert \mathcal{R}_{J_{k-1}, J_k}^{\top}\omega_{J_{k-1},0} - \omega_{J_{k},0} \Vert^2_2 \leq C\xi_k^2 J_k^{3- 2s},  
\end{align*}
Moreover, note that  \[\nabla_{\theta}\mathcal{L}_J(\omega, W) = - (Y - \Psi_J(x_0 +X^{\top}\theta)^{\top}\beta)(\nabla_{z}\Psi_J(x_0+X^{\top}\theta))^{\top}\beta X\] and \[\nabla_{\beta}\mathcal{L}_J(\omega, W) = - (Y - \Psi_J(x_0 +X^{\top}\theta)^{\top}\beta) \Psi_J(x_0+X^{\top}\theta).\]  
Note that due to $J^{-s}$ and $J^{1-s}$ global approximation error of $F_0$ and $\nabla_z F_0$, we have that 
\begin{align*}
& \left\Vert  \nabla_{\omega}\mathcal{L}_{J} (\omega_{J,0}, W) +  \varphi_J(W) \right\Vert_2
\leq C(1+|Y|)(1+\Vert X\Vert_2)J^{\frac{3}{2}-s},
\end{align*}
which leads  to   
\[
\mathbb{E}_{k-1}\left\Vert \frac{\xi_{k}}{B}\sum_{i=1}^B  \left(\nabla_{\omega}\mathcal{L}_{J_{k}} (\omega_{J_{k},0}, W_{i,k})+\varphi_{J_k}(W_{i,k})\right)\right\Vert^2_2 \leq C \xi_k^2 J_k^{3-2s}.
\]

Also, we have that 
\[
\mathbb{E}_{k-1}\Vert \varphi_{J_k}(W_{i,k})\Vert_2^2 \leq C J_k.
\]

Given the above results, we now expand the squared norm of the pre-projection update.  Write
\[
g_k := \frac{1}{B}\sum_{i=1}^B\nabla_{\omega}\mathcal{L}_{J_{k}} (\omega_{J_{k},0}, W_{i,k}), \qquad \mathbb{E}_{k-1}\,g_k = b_{J_k},
\]
where $b_{J_k}$ is the mean sieve score of \autoref{lemma_centered_score}, so that $\Vert b_{J_k}\Vert_2\leq CJ_k^{-q_{\star}}$ and $\mathbb{E}_{k-1}\Vert g_k - b_{J_k}\Vert_2^2\leq CJ_k$. All remaining stochastic terms in the decomposition displayed at the beginning of this subsection  are handled exactly as bounded above. Collecting terms, the conditional mean of the pre-projection increment consists of (a) the contraction of $\mathcal{R}^{\top}_{J_{k-1},J_k}\Delta\widetilde\omega_{k-1}$ through the integrated population Hessian, (b) the deterministic mean score $-\xi_kb_{J_k}$, and, at expansion times only, (c) the deterministic terms $\mathcal{R}_{J_{k-1}, J_k}^{\top}\omega_{J_{k-1},0} - \omega_{J_k, 0}$ and \[- \mathbb{E}\frac{\xi_k}{B}\sum_{i=1}^B[
    \nabla_\omega\mathcal L_{J_k}
       (\mathcal R_{J_{k-1},J_k}^{\top}\omega_{J_{k-1},0},W_{i,k})
    -\nabla_\omega\mathcal L_{J_k}(\omega_{J_k,0},W_{i,k})].\]  All of these norms are at most $CJ_k^{-s}$ by \autoref{lemma_telescoping}(i), \autoref{remark_multijump}, and \autoref{cond14}. Using $|ab|\leq \frac{ca^2}{2}+\frac{b^2}{2c}$ (any $c>0$) on the cross term between the contracted error and (b),
\[
2\xi_k\left|\left\langle \left(\mathbb{I} - \xi_k\overline{\mathbb{H}}_k\right)\mathcal{R}^{\top}_{J_{k-1},J_k}\Delta\widetilde\omega_{k-1},\; b_{J_k}\right\rangle\right| \leq c\,\xi_k\Vert \Delta\widetilde\omega_{k-1}\Vert_2^2 + C\xi_k\Vert b_{J_k}\Vert_2^2\leq c\,\xi_k\Vert \Delta\widetilde\omega_{k-1}\Vert_2^2 + C\xi_kJ_k^{-2q_{\star}},
\]
where $\overline{\mathbb{H}}_k$ denotes the integrated population Hessian appearing in the contraction term, and $c$ is chosen small relative to the contraction constant. The variance of the centered score contributes $\xi_k^2\,\mathbb{E}_{k-1}\Vert g_k - b_{J_k}\Vert_2^2\leq C\xi_k^2J_k$, the variance of the $\frac{\xi_k}{B}\sum_{i=1}^B\Bigl[
    \nabla_\omega\mathcal L_{J_k}(\omega_{J_k,0},W_{i,k})
    +\varphi_{J_k}(W_{i,k})\Bigr]$-type remainder contributes $C\xi_k^2J_k^{3-2s}\leq C\xi_k^2 J_k$ since $s>1$, and the cross terms at expansion times are bounded by $C\boldsymbol{1}_k(\Vert\Delta\widetilde\omega_{k-1}\Vert_2^2 + J_k^{-2s})$. So when $J_k = J_{k-1}$, there holds
\begin{align*}
& \mathbb{E}_{k-1}\Vert \mathcal{R}_{J_{k-1}, J_k}^{\top}\widetilde\omega_{k-1} - \frac{\xi_k}{B}\sum_{i=1}^B \nabla_{\omega}  \mathcal{L}_{J_k}\left(\mathcal{R}_{J_{k-1}, J_k}^{\top}\widetilde\omega_{k-1}, W_{i,k}\right) - \omega_{J_{k}, 0} \Vert_2^2 \\
&\leq \left(1 - C\xi_k +C\xi_k^2J_{k}^3 \right)\Vert \Delta\widetilde\omega_{k-1}\Vert_2^2 + C\xi_k^2J_k + C\xi_k J_k^{-2q_{\star}},
\end{align*}
and when $J_{k}>J_{k-1}$, we have that
\begin{align*}
& \mathbb{E}_{k-1}\Vert \mathcal{R}_{J_{k-1}, J_k}^{\top}\widetilde\omega_{k-1} - \frac{\xi_k}{B}\sum_{i=1}^B \nabla_{\omega}  \mathcal{L}_{J_k}\left(\mathcal{R}_{J_{k-1}, J_k}^{\top}\widetilde\omega_{k-1}, W_{i,k}\right) - \omega_{J_{k}, 0} \Vert_2^2 \\
&\leq C\left(1 - C\xi_k +  C\xi_k^2J_k^3  +C\xi_k J_k^{3-2s}\right)\Vert \Delta\widetilde\omega_{k-1}\Vert_2^2 + C\xi_k^2J_k + C\xi_k J_k^{-2q_{\star}} + CJ_k^{-2s}.
\end{align*}
Then use the fact that $\omega_{J_k,0}\in\widehat\Omega_{k}$ according to \autoref{cond14}, we have that
\begin{align}\label{general_dynamics_nls}
    \mathbb{E}_{k-1}\Vert \Delta \widetilde\omega_{k}  \Vert_2^2  \leq \left( 1 - C\xi_k  + C\boldsymbol{1}_k  \right)\Vert \Delta\widetilde\omega_{k-1}\Vert_2^2 + C\xi_k^2 J_{k} + C\xi_k J_k^{-2q_{\star}}   + C\boldsymbol{1}_kJ_k^{-2s},
\end{align}
when $ \alpha_{\xi}>3\alpha_{\dagger}$ and $s>\frac{5}{2}$.

\autoref{cond14} gives
$\omega_{J_k,0}\in\widehat\Omega_k$ eventually almost surely.  On that
probability-one tail, nonexpansiveness of projection gives
\[\Vert \Delta\widetilde \omega_k\Vert_2^2\leq \left\Vert \mathcal{R}_{J_{k-1}, J_k}^{\top}\widetilde\omega_{k-1} - \frac{\xi_k}{B}\sum_{i=1}^B \nabla_{\omega}  \mathcal{L}_{J_k}\left(\mathcal{R}_{J_{k-1}, J_k}^{\top}\widetilde\omega_{k-1}, W_{i,k}\right) - \omega_{J_{k}, 0} \right\Vert_2^2\]
This leads to \[
    \mathbb{E}_{k-1}\Vert \Delta \widetilde\omega_{k}  \Vert_2^2  \leq \left( 1 - C\xi_k  + C\boldsymbol{1}_k  \right)\Vert \Delta\widetilde\omega_{k-1}\Vert_2^2 + C\xi_k^2 J_{k} + C\xi_kJ_k^{-2q_{\star}}   + C\boldsymbol{1}_kJ_k^{-2s}.
\]
 The same blockwise Robbins--Siegmund argument as in \autoref{C1.part1} yields
\begin{align*}
    \Vert \Delta \widetilde\omega_{k}  \Vert_2 = O\left(k^{\frac{1 +\alpha_{\dagger} -  2\alpha_{\xi} }{2}}\log^{\frac{r}{2}}(k) + k^{\frac{1-\alpha_{\xi}}{2}-q_{\star}\alpha_{\dagger}}\log^{\frac{r}{2}}(k) + k^{-\left(s-\frac{1}{2}\right)\alpha_{\dagger}}\log^{\frac{r}{2}}(k)\right), \qquad \text{a.s.}
\end{align*}
We finally verify that under \autoref{new_rate} each exponent exceeds
$2\alpha_{\dagger}$, so that the \emph{projected} error satisfies
$\Vert\Delta\widetilde\omega_k\Vert_2=o(J_k^{-2})$ almost surely. First,
$\frac{2\alpha_{\xi}-\alpha_{\dagger}-1}{2}>2\alpha_{\dagger}$ iff
$\alpha_{\xi}>\frac{1}{2}+\frac{5}{2}\alpha_{\dagger}$, which is implied by
$\alpha_{\xi}>\frac{1}{2}+3\alpha_{\dagger}$. Second, using
$\alpha_{\xi}>\frac{1}{2}+3\alpha_{\dagger}$ we get
$\frac{1-\alpha_{\xi}}{2}<\frac{1}{4}-\frac{3}{2}\alpha_{\dagger}$, so
$q_{\star}\alpha_{\dagger}-\frac{1-\alpha_{\xi}}{2}>
\frac{1}{2}-\frac{1}{4}+\frac{3}{2}\alpha_{\dagger}>2\alpha_{\dagger}$
whenever $q_{\star}\alpha_{\dagger}>\frac{1}{2}$. Third,
$(s-\frac{1}{2})\alpha_{\dagger}>2\alpha_{\dagger}$ iff $s>\frac52$.

Given such result, we easily have that 
\begin{align*}
    \Vert \widetilde\omega_{k} - \widehat\omega_k\Vert_2 & \leq \Vert \widetilde\omega_{k}  - \omega_{J_k,0}\Vert_2 + \Vert \widehat\omega_{k}  - \omega_{J_k,0}\Vert_2\\
    & \leq o(J_k^{-2}) + \frac{(1-\delta_{\Omega})C_{\Omega}J_k^{-2}}{2}<\frac{C_{\Omega}J_k^{-2}}{2}, \qquad \text{a.s.}
\end{align*}
This finally validates the projection inactivity.

\subsubsection{Obtain the Refined Rate}\label{lemma5.2part2}

\vspace{0.3cm}

\textit{Summary: This part decomposes the raw iterate error into a \emph{regular} component $D_k$ and a deterministic \emph{sieve-expansion transient} $T_k$,
\[
\Delta\widetilde\omega_k = D_k + T_k,
\]
and builds the refined rates.  In particular, we have that }
\begin{equation*}
    \Vert D_{k}\Vert_2 = O\left(k^{-\frac{\alpha_{\xi}-\alpha_{\dagger}}{2}}\text{PolyLog} + k^{-(2s-2)\alpha_{\dagger}}\right), \qquad \Vert T_k\Vert_2\leq Ck^{-s\alpha_{\dagger}}, \qquad \text{a.s.},
\end{equation*}
\textit{so that \[\Vert \Delta\widetilde\omega_{k}\Vert_2 = O (k^{-\frac{\alpha_{\xi}-\alpha_{\dagger}}{2}}\text{PolyLog} + k^{-s\alpha_{\dagger}}), \qquad \text{a.s.}\]}

\vspace{0.3cm}

In \autoref{lemma5.2part1}, we have shown that the projection is inactive almost surely for $k$ large. In this case, almost surely the following holds for $k$ sufficiently large \begin{align} \label{local joint update - no projection}
\widetilde\omega_{k} =  \mathcal{R}_{J_{k-1}, J_k}^{\top}\widetilde\omega_{k-1} - \frac{\xi_k}{B}\sum_{i=1}^B\nabla_{\omega}\mathcal{L}_{J_k}\left(\mathcal{R}_{J_{k-1}, J_k}^{\top}\widetilde\omega_{k-1}, W_{i,k}\right).
\end{align}
To simplify our following exposition, without loss of generality we assume that (\ref{local joint update - no projection}) holds true for all $k$; otherwise we can start from a random integer $k_0$, whose impacts will be exponentially decaying so are negligible.  Define 
\[
 \mathbb{M}_{J} = \mathbb{E}\begin{pmatrix}
    (\nabla_z F_0(Z_0))^2XX^{\top} &  \nabla_zF_0(Z_0) X\Psi_J(Z_0)^{\top}\\
    \nabla_{z} F_0(Z_0) \Psi_J(Z_0)X^{\top} & \Psi_J(Z_0)\Psi_J(Z_0)^{\top}
\end{pmatrix}.
\]
At the sieve reference point,
\begin{align*}
\nabla_{\theta\theta}\mathcal L_J(\omega_{J,0})
&=\mathbb E\!\left[
 \left\{(\nabla_z\mathbb P_JF_0(Z_0))^2
 -(F_0-\mathbb P_JF_0)(Z_0)
   \nabla_{zz}\mathbb P_JF_0(Z_0)\right\}XX^\top\right],\\
\nabla_{\theta\beta}\mathcal L_J(\omega_{J,0})
&=\mathbb E\!\left[
 \nabla_z\mathbb P_JF_0(Z_0)X\Psi_J(Z_0)^\top
 -(F_0-\mathbb P_JF_0)(Z_0)
   X\nabla_z\Psi_J(Z_0)^\top\right].
\end{align*}
By \autoref{condition6} and \autoref{cond11} and the fact that $s>2$, we have that   \[\Vert  \mathbb{M}_J - \nabla_{\omega\omega} \mathcal{L}_J(\omega_{J,0}) \Vert_F\leq CJ^{\frac{3}{2} - s},\]  
where $C$ is a constant independent of $J$. 
From the proof of the first step, let
\[
 e_k^-:=\mathcal R_{J_{k-1},J_k}^{\top}\Delta\widetilde\omega_{k-1},
 \qquad
 A_k:=\bigl(\mathbb I_{p+J_k}-\xi_k\mathbb M_{J_k}\bigr)
       \mathcal R_{J_{k-1},J_k}^{\top}.
\]
Then the unprojected recursion has the exact decomposition
\begin{equation}\label{eq:phaseII-zeta-decomposition}
 \Delta\widetilde\omega_k
 =A_k\Delta\widetilde\omega_{k-1}
  -\xi_k\sum_{\ell=1}^{6}\zeta_{\ell,k}+\zeta_{7,k},
\end{equation}
where
\begin{align*}
 \zeta_{1,k}
 &:=\{\nabla_{\omega\omega}\mathcal L_{J_k}(\omega_{J_k,0})
          -\mathbb M_{J_k}\}e_k^-,\\
 \zeta_{2,k}
 &:=\int_0^1\!\Bigl[
    \nabla_{\omega\omega}\mathcal L_{J_k}
       \{\mathcal R_{J_{k-1},J_k}^{\top}
          (\omega_{J_{k-1},0}+\tau\Delta\widetilde\omega_{k-1})\}
    -\nabla_{\omega\omega}\mathcal L_{J_k}(\omega_{J_k,0})
    \Bigr]d\tau\,e_k^-,\\
 \zeta_{3,k}
 &:=\frac1B\sum_{i=1}^B\int_0^1\!\Bigl[
    \nabla_{\omega\omega}\mathcal L_{J_k}
       \{\mathcal R_{J_{k-1},J_k}^{\top}
          (\omega_{J_{k-1},0}+\tau\Delta\widetilde\omega_{k-1}),W_{i,k}\}
    -\nabla_{\omega\omega}\mathcal L_{J_k}
       \{\mathcal R_{J_{k-1},J_k}^{\top}
          (\omega_{J_{k-1},0}+\tau\Delta\widetilde\omega_{k-1})\}
    \Bigr]d\tau\,e_k^-,\\
 \zeta_{4,k}
 &:=\frac1B\sum_{i=1}^B\Bigl[
    \nabla_\omega\mathcal L_{J_k}
       (\mathcal R_{J_{k-1},J_k}^{\top}\omega_{J_{k-1},0},W_{i,k})
    -\nabla_\omega\mathcal L_{J_k}(\omega_{J_k,0},W_{i,k})\Bigr],\\
 \zeta_{5,k}
 &:=\frac1B\sum_{i=1}^B\Bigl[
    \nabla_\omega\mathcal L_{J_k}(\omega_{J_k,0},W_{i,k})
    +\varphi_{J_k}(W_{i,k})\Bigr],\\
 \zeta_{6,k}
 &:=-\frac1B\sum_{i=1}^B\varphi_{J_k}(W_{i,k}),\\
 \zeta_{7,k}
 &:=\mathcal R_{J_{k-1},J_k}^{\top}\omega_{J_{k-1},0}-\omega_{J_k,0}.
\end{align*}
Thus \eqref{eq:phaseII-zeta-decomposition} yields
\begin{align*}
\Delta\widetilde\omega_k
=
\left(\prod_{j=1}^k A_j\right)\Delta\widetilde\omega_0
-
\sum_{j=1}^k\xi_j
\left(\prod_{m=j+1}^k A_m\right)\sum_{l=1}^6 \zeta_{l,j} + \sum_{j=1}^k
\left(\prod_{m=j+1}^k A_m\right)  \zeta_{7,j},
\end{align*}
where following the previous convention,  $\prod_{j=1}^k A_j = A_k A_{k-1}\cdots A_1$, $\prod_{m=j+1}^k A_m = A_k A_{k-1}\cdots A_{j+1}$ if $j\leq k-1$ and  $\prod_{m=j+1}^k A_m =\mathbb{I}_{p + J_k}$ if $j = k$.

Before analyzing these terms, we now formally introduce the transient/regular decomposition announced in the summary. Decompose the score-difference term at expansion times into its (deterministic) mean and its martingale part:
$\zeta_{4,k} = \mathbb{E}\zeta_{4,k} + \zeta_{4,k}^{\circ},
$
noting that $\zeta_{4,k}$ is evaluated at the deterministic points $\mathcal{R}^{\top}_{J_{k-1},J_k}\omega_{J_{k-1},0}$ and $\omega_{J_k,0}$, so $\mathbb{E}_{k-1}\zeta_{4,k} = \mathbb{E}\zeta_{4,k}$ is nonrandom. Define the deterministic \emph{sieve-expansion transient}
\begin{equation}\label{transient_def}
T_k := \sum_{j=1}^k\left(\prod_{m=j+1}^kA_m\right)\tau_j, \qquad \tau_j := \zeta_{7,j} - \xi_j\,\mathbb{E}\zeta_{4,j}, \qquad D_k := \Delta\widetilde\omega_k - T_k,
\end{equation}
so that $T_k = A_kT_{k-1}+\tau_k$ with $T_0=0$, and
\begin{equation}\label{regular_recursion}
D_k = A_kD_{k-1} - \xi_k\left(\zeta_{1,k}+\zeta_{2,k}+\zeta_{3,k}+\zeta_{4,k}^{\circ}+\zeta_{5,k}+\zeta_{6,k}\right), \qquad D_0 = \Delta\widetilde\omega_0.
\end{equation}
Note that $\tau_j\neq 0$ only at expansion times, and by \autoref{cond14} (boundedness of the population Hessian on the local basin, which contains the segment between $\mathcal{R}^{\top}_{J_{j-1},J_j}\omega_{J_{j-1},0}$ and $\omega_{J_j,0}$ since $s>2$), a mean-value expansion gives $\Vert\mathbb{E}\zeta_{4,j}\Vert_2\leq C\Vert \mathcal{R}^{\top}_{J_{j-1},J_j}\omega_{J_{j-1},0}-\omega_{J_j,0}\Vert_2$. Hence, by \autoref{lemma_telescoping} and \autoref{remark_multijump},
\begin{equation}\label{tau_bound}
\Vert \tau_j\Vert_2\leq C\boldsymbol{1}_j\sum_{m=J_{j-1}+1}^{J_j}\delta_m\leq C\boldsymbol{1}_jJ_j^{-s}.
\end{equation}
The next lemma collects the properties of $T_k$ that will be used below and in \autoref{lemma5.2part3}.

\begin{lemma}[Checkpoint transient]\label{lemma_transient}
Let the conditions of \autoref{theorem_NLS_update} hold. Then there are constants $c, C>0$ such that:
\begin{enumerate}
\item[(i)] $\left\Vert \prod_{m=j+1}^kA_m\right\Vert_{\mathrm{op}}\leq C\exp\left(-c\left(k^{1-\alpha_{\xi}} - j^{1-\alpha_{\xi}}\right)\right)$ for all $j\leq k$;
\item[(ii)] $\Vert T_k\Vert_2\leq Ck^{-s\alpha_{\dagger}}$ for all $k$, and $\Vert T_k\Vert_2^2\leq C\sum_{m\leq J_k}\exp\left(-2c\left(k^{1-\alpha_{\xi}}-\mathcal{J}(m)^{1-\alpha_{\xi}}\right)\right)\delta_m^2$;
\item[(iii)] for every $\nu\geq 0$ there is $C(\nu)$ such that for all $N$,
\[
\sum_{k=1}^Nk^{\nu}\Vert T_k\Vert_2\leq C(\nu)\left(N^{\nu+\alpha_{\xi}+\left(\frac{1}{2}-s\right)\alpha_{\dagger}} + \log N\right), \ \ \ 
\sum_{k=1}^Nk^{\nu}\Vert T_k\Vert_2^2\leq C(\nu)\left(N^{\nu+\alpha_{\xi}+\left(1-2s\right)\alpha_{\dagger}} + \log N\right);
\]
\item[(iv)] for every fixed $r>0$ and all $m$ sufficiently large, $\Vert T_{\mathcal{J}(m)-1}\Vert_2\leq \mathcal{J}(m)^{-r}$; that is, the transient is superpolynomially small at the end of each constant-dimension block.
\end{enumerate}
\end{lemma}

\begin{proof}
(i) is obviously from our previous proofs. For (ii), by \eqref{tau_bound} and (i), and merging \autoref{remark_multijump} into constants,
\[
\Vert T_k\Vert_2\leq C\sum_{m\leq J_k}\exp\left(-c\left(k^{1-\alpha_{\xi}}-\mathcal{J}(m)^{1-\alpha_{\xi}}\right)\right)\delta_m =: C\sum_{m\leq J_k}e_m(k)\,\delta_m.
\]
For consecutive $m$, the gap between expansion times satisfies $\mathcal{J}(m+1)-\mathcal{J}(m)\asymp \mathcal{J}(m)^{1-\alpha_{\dagger}}$, so
\[
\frac{e_{m}(k)}{e_{m+1}(k)} = \exp\left(-c\left(\mathcal{J}(m+1)^{1-\alpha_{\xi}}-\mathcal{J}(m)^{1-\alpha_{\xi}}\right)\right)\leq \exp\left(-c^{\prime}\mathcal{J}(m)^{1-\alpha_{\xi}-\alpha_{\dagger}}\right)\leq \frac{1}{2}
\]
for all $m\geq m_0$, since $1-\alpha_{\xi}-\alpha_{\dagger}>0$. Hence the weights $e_m(k)$ decay super-geometrically as $m$ decreases below $J_k$, $\sum_{m\leq J_k}e_m(k)\leq C$, and $\sum_{m\leq J_k}e_m(k)\delta_m\leq C\max\{\delta_m: J_{k}/2\leq m\leq J_k\} + C\exp(-c^{\prime\prime}k^{1-\alpha_{\xi}})\leq CJ_k^{-s}\leq Ck^{-s\alpha_{\dagger}}$ by \autoref{lemma_telescoping}(i). The second claim in (ii) follows directly from Cauchy--Schwarz:
$\Vert T_k\Vert_2^2\leq(\sum_me_m(k))(\sum_me_m(k)\delta_m^2)
\leq C\sum_me_m(k)\delta_m^2$.  Relabelling the positive decay constant
absorbs the harmless distinction between $e_m(k)$ and $e_m(k)^2$. For (iii),  interchanging the order of summation and using, for $\varrho := \alpha_{\xi}<1$,
\[
\sum_{k\geq j}k^{\nu}\exp\left(-c\left(k^{1-\varrho}-j^{1-\varrho}\right)\right)\leq Cj^{\nu+\varrho},
\]
which follows by splitting $k\leq 2j$ (where $k^{1-\varrho}-j^{1-\varrho}\geq (1-\varrho)(k-j)(2j)^{-\varrho}$, giving a geometric sum of size $Cj^{\varrho}$ with $k^{\nu}\leq (2j)^{\nu}$) from $k>2j$ (where $k^{1-\varrho}-j^{1-\varrho}\geq (1-2^{-(1-\varrho)})k^{1-\varrho}$, giving an $O(1)$ tail), we obtain
\[
\sum_{k=1}^Nk^{\nu}\Vert T_k\Vert_2\leq C\sum_{m\leq J_N}\delta_m\,\mathcal{J}(m)^{\nu+\alpha_{\xi}}\leq C\sum_{m\leq J_N}m^{\frac{\nu+\alpha_{\xi}}{\alpha_{\dagger}}}\delta_m,
\]
and \autoref{lemma_telescoping}(iii) with $J_N\leq CN^{\alpha_{\dagger}}$ gives the first display; the second is identical with $\delta_m^2$ in place of $\delta_m$. (iv) The most recent expansion before time $\mathcal{J}(m)-1$ occurs at $\mathcal{J}(m-1)$, and by (i)--(ii),
\[
\Vert T_{\mathcal{J}(m)-1}\Vert_2\leq C\exp\left(-c\left((\mathcal{J}(m)-1)^{1-\alpha_{\xi}}-\mathcal{J}(m-1)^{1-\alpha_{\xi}}\right)\right)\leq C\exp\left(-c^{\prime}\mathcal{J}(m)^{1-\alpha_{\xi}-\alpha_{\dagger}}\right),
\]
which is superpolynomially small since $1-\alpha_{\xi}-\alpha_{\dagger}>0$.
\end{proof}

\autoref{lemma_transient} provides the bound for $\Vert T_k\Vert_2$. It remains to bound $\Vert D_k\Vert_2$. In view of \eqref{regular_recursion}, the regular component satisfies $D_k = (\prod_{j=1}^kA_j)\Delta\widetilde\omega_0 - \sum_{l\in\{1,2,3,5,6\}}Q_{l,k} - Q_{4,k}^{\circ}$, where $
Q_{l,k} = A_kQ_{l,k-1} + \xi_k\zeta_{l,k}.
$ and $Q^{\circ}_{4,k}$ is defined with $\zeta^{\circ}_{4,k}$ in place of $\zeta_{4,k}$.
For $(\prod_{j=1}^k A_j)\Delta\widetilde\omega_0
$, similar to the proof in \autoref{C1.part2},  when $J_{k-1}= J_k$, $\mathcal{R}_{J_{k-1},J_k} = \mathbb{I}_{p+J_k}$, so $A_k = \mathbb{I}_{p + J_k} - \xi_k \mathbb{M}_{J_k}$ and $\Vert A_k \Vert_{\mathrm{op}}\leq (1 - C\xi_k)$; on the other side, when $J_k>J_{k-1}$, we have that $\Vert A_k\Vert_{\mathrm{op}}\leq C$. Together, we have that 
\begin{align*}
\left\Vert \left(\prod_{j=1}^k A_j\right)\Delta\widetilde\omega_0 \right\Vert_2&  \leq C^{J_k}\prod_{j=1}^k \left(1 - C\xi_j\right)\Vert \Delta\widetilde\omega_0\Vert_2\leq C\exp\left(-Ck^{1-\alpha_{\xi} }  + Ck^{\alpha_{\dagger}}\right)\\
& \leq C\exp\left(-Ck^{1-\alpha_{\xi} }\right)
\end{align*}
when $1 - \alpha_{\xi}>\alpha_{\dagger}$. 
It then remains to look at  $Q_{l,k}$ one by one.

We first show that  $Q_{1,k}$ does not determine the convergence rate of $\Delta\widetilde\omega_{k}$ under the given conditions. In particular, let $\left\Vert \Delta\widetilde\omega_{k}\right\Vert_{2} = O(k^{-\alpha_{\omega}})$  a.s. holds, where $\alpha_{\omega}$ is some preliminary positive rate. Such rate exists according to our first-step proof. Then
\begin{align*}
\Vert Q_{1,k} \Vert_2  &  \leq \Vert A_k \Vert_{\mathrm{op}}\Vert Q_{1,k-1}\Vert_2 + C\xi_k J_k^{\frac{3}{2} -s} \Vert \Delta\widetilde\omega_{k-1} \Vert_2\\
&\leq \left(1 - Ck^{-\alpha_{\xi}} + C\boldsymbol{1}_k\right)\Vert Q_{1,k-1}\Vert_2 + Ck^{-\alpha_{\xi} - \alpha_{\omega}+(\frac{3}{2} -s)\alpha_{\dagger}}, \qquad   \text{a.s.}
\end{align*}
where the exponent $(\frac{3}{2} -s)\alpha_{\dagger}$ comes from the fact that $\Vert  \mathbb{M}_J - \nabla_{\omega\omega} \mathcal{L}_J(\omega_{J,0}) \Vert_F\leq CJ^{\frac{3}{2} - s}$. 
Without loss of generality, we can assume the above holds for all $k$ (otherwise we can consider sufficiently large $k$ only). Iterations gives \begin{align*}
\Vert Q_{1,k} \Vert_2 
&\leq C\exp\left(
    -C\sum_{j=1}^k j^{-\alpha_{\xi}} + Ck^{\alpha_{\dagger}}
\right)\Vert Q_{1,0} \Vert_2  \\
&\quad + C\sum_{j=1}^k
\exp\left(
    -C\sum_{m=j+1}^k m^{-\alpha_{\xi}} + C(k^{\alpha_{\dagger}}-j^{\alpha_{\dagger}}+1)
\right)
j^{-\alpha_{\xi} - \alpha_{\omega}+(\frac{3}{2} -s)\alpha_{\dagger}}\\
& \leq C\exp\left(- Ck^{1-\alpha_{\xi}} + Ck^{\alpha_{\dagger}}\right) \Vert Q_{1,0} \Vert_2  \\
& \quad + C\exp\left(- C^*k^{1-\alpha_{\xi}} + Ck^{\alpha_{\dagger}}\right)\sum_{j=1}^k \exp\left(C^*j^{1-\alpha_{\xi}} - Cj^{\alpha_{\dagger}} \right)j^{-\alpha_{\xi} -  \alpha_{\omega} + (\frac{3}{2} -s)\alpha_{\dagger}}.
\end{align*}
So when $1 - \alpha_{\xi}>\alpha_{\dagger}$, by \autoref{lemma4},  there holds  $\Vert Q_{1,k} \Vert_2 = O_{\text{a.s.}}(k^{- \alpha_{\omega} + (\frac{3}{2} -s)\alpha_{\dagger}})$. When $s>\frac{3}{2}$, we have that  $\Vert Q_{1,k}\Vert_2 = o(k^{-\alpha_{\omega}})$ a.s., so it does not affect the bound of the convergence rate of $\Delta\widetilde\omega_{k}$. 

For $Q_{2,k}$, note that the first Hessian argument in $\zeta_{2,k}$ differs from $\omega_{J_k,0}$
by $\zeta_{7,k}+\tau e_k^-$.  Hence
 \autoref{local_basin}, bounded re-embedding, and Young's inequality give
\begin{equation}\label{zeta2_corrected_bound}
\begin{split}
 \Vert\zeta_{2,k}\Vert_2
 &\leq CJ_k^2\left\{
  \Vert\Delta\widetilde\omega_{k-1}\Vert_2^2
  +\boldsymbol1_kJ_k^{-s}
   \Vert\Delta\widetilde\omega_{k-1}\Vert_2\right\}\\
 &\leq CJ_k^2\Vert\Delta\widetilde\omega_{k-1}\Vert_2^2
  +C\boldsymbol1_kJ_k^{2-2s}.
\end{split}
\end{equation}
Thus
\begin{align*}
\Vert Q_{2,k}\Vert_2
&\leq(1-Ck^{-\alpha_\xi}+C\boldsymbol1_k)
 \Vert Q_{2,k-1}\Vert_2
 +Ck^{-\alpha_\xi-2\alpha_\omega+2\alpha_\dagger}
 +C\boldsymbol1_k k^{-\alpha_\xi-(2s-2)\alpha_\dagger}, \qquad \text{a.s.}
\end{align*}
\autoref{lemma_transient}(i) therefore yields
\[
 \Vert Q_{2,k}\Vert_2
 =O\!\left(
 k^{-2\alpha_\omega+2\alpha_\dagger}
 +k^{-(2s-2)\alpha_\dagger}\right), \qquad \text{a.s.}
\]
When $\alpha_{\omega}>2\alpha_{\dagger}$, which holds given the proof of first part, the first term does not affect the convergence rate.

For $Q_{3,k}$, we obviously have that 
\[
\mathbb{E}_{k-1}\Vert Q_{3,k}\Vert_2^2\leq (1-C\xi_k +C\boldsymbol{1}_k)\Vert Q_{3,k-1}\Vert_2^2 + C\xi_k^2J_k^3 \Vert \Delta\widetilde\omega_{k-1}\Vert_2^2.
\]
This leads to that
\[
\Vert Q_{3,k}\Vert_2 = O\left(k^{\frac{1+3\alpha_{\dagger}-2\alpha_{\xi}}{2} - \alpha_{\omega}}\log^{\frac{r}{2}}(k)\right), \qquad \text{a.s.}
\]
When $\alpha_{\xi}>\frac{1+3\alpha_{\dagger}}{2}$, we have that $Q_{3,k}$
does not affect the convergence rate of $\Delta\widetilde\omega_{k-1}$.

For $Q_{4,k}^{\circ}$, recall that the deterministic mean $\mathbb{E}\zeta_{4,k}$ has been moved into the transient $T_k$ via \eqref{transient_def}, so only the martingale part $\zeta_{4,k}^{\circ}$ remains, and $\zeta_{4,k}^{\circ}=0$ for any $k$ with $J_k =J_{k-1}$. By the Hessian bound used for $\zeta_{4,k}$ in \autoref{lemma5.2part1} and \autoref{lemma_telescoping}(i),
\[
\mathbb{E}_{k-1}\Vert \zeta^{\circ}_{4,k}\Vert^2_2\leq \mathbb{E}\Vert \zeta_{4,k}\Vert_2^2\leq CJ_k^3\Vert \mathcal{R}^{\top}_{J_{k-1},J_k}\omega_{J_{k-1},0} - \omega_{J_k,0}\Vert_2^2\leq C\boldsymbol{1}_kJ_k^{3-2s},
\]
so that
\begin{align*}
    \mathbb{E}_{k-1}\Vert Q^{\circ}_{4,k}\Vert_2^2 \leq \left(1 - Ck^{-\alpha_{\xi}}+C\boldsymbol{1}_k\right)\Vert Q^{\circ}_{4,k-1}\Vert_2^2 + C\boldsymbol{1}_k k^{-2\alpha_{\xi} + (3-2s)\alpha_{\dagger}}.
\end{align*}
Multiplying by the weight $k^{\alpha_{\xi}-\alpha_{\dagger}}\log^{-r}(k)$, the forcing term is summable:
\[
\sum_k\boldsymbol{1}_kk^{\,\alpha_{\xi}-\alpha_{\dagger}-2\alpha_{\xi}+(3-2s)\alpha_{\dagger}}\log^{-r}(k)\asymp \sum_m \mathcal{J}(m)^{-\alpha_{\xi}+(2-2s)\alpha_{\dagger}}\log^{-r}(\mathcal{J}(m))<\infty,
\]
since $-\alpha_{\xi}+(2-2s)\alpha_{\dagger}<-\alpha_{\dagger}$ holds trivially for $s>\frac{3}{2}$. Following the same argument used for $Q_{4,k}$-type terms previously (supermartingale convergence along constant-dimension blocks), we obtain
\[
\Vert Q^{\circ}_{4,k}\Vert_2 = O\left(k^{-\frac{\alpha_{\xi} - \alpha_{\dagger}}{2}}\log^{\frac{r}{2}}(k)\right), \qquad \text{a.s.}
\]

For $Q_{5,k}$, we first note that since $\mathbb{E}[\varphi_{J}(W)]=0$, the mean of $\zeta_{5,k}$ is precisely the mean sieve score of \autoref{lemma_centered_score}:
\[
\mathbb{E}\zeta_{5,k} = \mathbb{E}\left[\nabla_{\omega}\mathcal{L}_{J_k}(\omega_{J_k,0},W)\right] = b_{J_k}, \qquad \Vert b_{J_k}\Vert_2\leq CJ_k^{-q_{\star}} = CJ_k^{-s-\min(s_{\mu},\,s-1)}.
\]
Moreover, $\mathbb{E}\Vert \zeta_{5,k}\Vert_2^2\leq CJ^{3-2s}\mathbb{E}[(1+|Y|^2)(1+\Vert X\Vert_2^2)]\leq CJ^{3-2s}$.

Given the above two bounds, we have that 
\[
Q_{5,k} = \xi_k \mathbb{E}_{k-1}\zeta_{5,k} +  \xi_k(\zeta_{5,k} - \mathbb{E}_{k-1}\zeta_{5,k})+ A_k Q_{5,k-1}
\]
So we only need to bound two sequences $Q_{5,k,1}$ and $Q_{5,k,2}$ defined by $Q_{5,k,1} = \xi_k\mathbb{E}_{k-1}\zeta_{5,k} + A_kQ_{5,k-1,1}$ and $Q_{5,k,2} = \xi_k(\zeta_{5,k} - \mathbb{E}_{k-1}\zeta_{5,k}) + A_kQ_{5,k-1,2}$. 
Obviously the first term is of order $O_{\text{a.s.}}(k^{-(s+\min(s_{\mu},\, s-1))\alpha_{\dagger} })$  by \autoref{lemma4}, while the second term, by the previous proof method for $Q_{4,k}$, is of order $O_{\text{a.s.}}(k^{-\frac{2\alpha_{\xi} +(2s-3)\alpha_{\dagger} -1}{2}}\text{PolyLog})$.  This
shows that 
\[
\Vert Q_{5,k}\Vert_2 = O\left(k^{-(s+\min(s_{\mu},\, s-1))\alpha_{\dagger} } + k^{-\frac{2\alpha_{\xi} +(2s-3)\alpha_{\dagger} -1}{2}}\text{PolyLog} \right), \qquad \text{a.s.}
\]
For $Q_{6,k}$, applying the proof of \autoref{theorem6} to $\varepsilon X$ together shows that 
\[
\Vert Q_{6,k}\Vert_2 = O\left(k^{-\frac{\alpha_{\xi}-\alpha_{\dagger}}{2}}\log^{\frac{1}{2}}(k)\right). 
\]
Finally, the terms driven by $\zeta_{7,j}$ and by $\mathbb{E}\zeta_{4,j}$ constitute exactly the transient $T_k$ of \eqref{transient_def}, and \autoref{lemma_transient}(ii) leads to 
$
\Vert T_{k}\Vert_2 \leq Ck^{- s\alpha_{\dagger}}.
$

Consequently, we have that 
\begin{equation}\label{refined_D_rate}
    \Vert D_{k}\Vert_2 = O\left(k^{-\frac{\alpha_{\xi}-\alpha_{\dagger}}{2}}  \text{PolyLog} + k^{-(2s-2)\alpha_{\dagger}}\right), \ \ \ 
    \Vert \Delta\widetilde\omega_{k}\Vert_2 = O\left(k^{-\frac{\alpha_{\xi}- \alpha_{\dagger}}{2}}\text{PolyLog}  + k^{-s\alpha_{\dagger}}\right), \ \ \ \text{a.s.}
\end{equation}
Here the component $k^{-(2s-2)\alpha_{\dagger}}$ of $\Vert D_k\Vert_2$ collects the contributions of $Q_{1,k}$ and $Q_{2,k}$ driven by the transient part of $\Delta\widetilde\omega$ (namely $O(k^{-(2s-\frac{3}{2})\alpha_{\dagger}})$ and $O(k^{-(2s-2)\alpha_{\dagger}})$ respectively), while the $Q_{5,k}$ mean component $O(k^{-q_{\star}\alpha_{\dagger}})$ is absorbed into the leading term because $q_{\star}\alpha_{\dagger}>\frac{1}{2}>\frac{\alpha_{\xi}-\alpha_{\dagger}}{2}$, and the $Q_{5,k}$ variance component is absorbed because $\frac{2\alpha_{\xi}+(2s-3)\alpha_{\dagger}-1}{2}\geq \frac{\alpha_{\xi}-\alpha_{\dagger}}{2}$ iff $\alpha_{\xi}+(2s-2)\alpha_{\dagger}\geq 1$, which follows from $\alpha_{\xi}>1+(4-2s)\alpha_{\dagger}$. The second display in \eqref{refined_D_rate} follows from the first, \autoref{lemma_transient}(ii), and $(2s-2)\alpha_{\dagger}\geq s\alpha_{\dagger}$ for $s\geq 2$. This gives the refined convergence rate of $\Vert \Delta\widetilde\omega_{k}\Vert_2$

\subsubsection{PR Averaging and the Asymptotic Linear Representation}\label{lemma5.2part3}

\vspace{0.3cm}

\textit{Summary: This part will prove the main results without conditioning.}

\vspace{0.3cm}

From the proof of \autoref{lemma5.2part2}, we have the decomposition $\Delta\widetilde\omega_k = D_k + T_k$ from \eqref{transient_def}--\eqref{regular_recursion}. We first dispose of the transient component of the PR average directly --- this is the checkpoint-aware step. By \autoref{lemma_transient}(iii) with $\nu = 0$, and $\Vert\mathcal{R}_{J_k,J_N}\Vert_{\mathrm{op}}\leq C$,
\begin{equation}\label{transient_average}
\left\Vert \frac{1}{N}\sum_{k=1}^N\mathcal{R}^{\top}_{J_k,J_N}T_k\right\Vert_2\leq \frac{C}{N}\sum_{k=1}^N\Vert T_k\Vert_2\leq \frac{C}{N}\left(N^{\alpha_{\xi}+\left(\frac{1}{2}-s\right)\alpha_{\dagger}}+\log N\right) = o\left(N^{-\frac{1}{2}}\right),
\end{equation}
where the final step holds because $\alpha_{\xi}<\frac{1}{2}+(s-\frac{1}{2})\alpha_{\dagger}$ under \autoref{new_rate}. 

It remains to analyze the regular component. For $\mathcal{J}(m)+1\leq k\leq \mathcal{J}(m+1)-1$, sieve dimension holds constant, so $A_k = \mathbb{I}_{p+m} - \xi_k \mathbb{M}_{m}$, $\tau_k = 0$, and $\zeta^{\circ}_{4,k} = 0$, whence by \eqref{regular_recursion},
\[
D_{k}   = \left(\mathbb{I}_{p+m} - \xi_k  \mathbb{M}_m\right)D_{k-1} - \xi_k\sum_{l\in\{1,2,3,5,6\}}\zeta_{l,k}, \ \  \mathcal{J}(m)+1\leq k\leq \mathcal{J}(m+1)-1.
\]
This implies that \[D_{k-1} = \xi_k^{-1} \mathbb{M}_m^{-1}\left( D_{k-1}  - D_{k} \right) -  \mathbb{M}_m^{-1}\sum_{l\in\{1,2,3,5,6\}} \zeta_{l,k}, \ \mathcal{J}(m)+1\leq k\leq \mathcal{J}(m+1)-1\] and
\begin{align*}
    &\sum_{k=\mathcal{J}(m)}^{\mathcal{J}(m+1)-2}D_{k} =  \sum_{k=\mathcal{J}(m)+1}^{\mathcal{J}(m+1)-1}D_{k-1}\\
    & =  \mathbb{M}_m^{-1} \left[\sum_{k=\mathcal{J}(m)+1}^{\mathcal{J}(m+1)-2}\left(\xi_{k+1}^{-1} - \xi_k^{-1}\right)D_k + \xi_{\mathcal{J}(m)+1}^{-1}D_{\mathcal{J}(m)} -\xi_{\mathcal{J}(m+1)-1}^{-1}D_{\mathcal{J}(m+1)-1}\right]\\
    & - \mathbb{M}_m^{-1}\sum_{k=\mathcal{J}(m)+1}^{\mathcal{J}(m+1)-1} \sum_{l\in\{1,2,3,5,6\}} \zeta_{l,k}.
\end{align*}
Now we look at $\doubleoverline{\omega}_N = \frac{1}{N}\sum_{k=1}^N \mathcal{R}_{J_k, J_N}^{\top}\widetilde\omega_{k}$. Using $\Delta\widetilde\omega_k = D_k+T_k$, \eqref{transient_average}, and the same block decomposition as before applied to $D_k$,
\begin{align*}
    \Delta\doubleoverline{\omega}_N & =  \frac{1}{N}\sum_{m=J_0}^{J_N - 1} \mathcal{R}_{m, J_N}^{\top} \sum_{k=\mathcal{J}(m)}^{\mathcal{J}(m+1)-2} D_{k}+   \frac{1}{N}\sum_{k = \mathcal{J}(J_N)}^{N}D_{k}
    + \frac{1}{N}\sum_{m=J_0}^{J_N-1} \mathcal{R}_{m, J_N}^{\top}D_{\mathcal{J}(m+1)-1}\\
    &+ \frac{1}{N}\sum_{m=J_0}^{J_N-1}\sum_{k=\mathcal{J}(m)}^{\mathcal{J}(m+1)-1}\left(\mathcal{R}_{m, J_N}^{\top}\omega_{m,0}- \omega_{J_N, 0}\right) + o\left(N^{-\frac{1}{2}}\right)\\
    & = \frac{1}{N}\sum_{m=J_0}^{J_N - 1}\mathcal{R}_{m, J_N}^{\top} \mathbb{M}_m^{-1} \sum_{k=\mathcal{J}(m)+1}^{\mathcal{J}(m+1)-2}\left(\xi_{k+1}^{-1} - \xi_k^{-1}\right)D_k + \frac{1}{N} \mathbb{M}_{J_N}^{-1} \sum_{k=\mathcal{J}(J_N)+1}^{N-1}\left(\xi_{k+1}^{-1} - \xi_k^{-1}\right)D_k \\
    & + \frac{1}{N}\sum_{m=J_0}^{J_N - 1}\mathcal{R}_{m, J_N}^{\top} \mathbb{M}_m^{-1}\left(\xi_{\mathcal{J}(m)+1}^{-1}D_{\mathcal{J}(m)} - \xi_{\mathcal{J}(m+1)-1}^{-1}D_{\mathcal{J}(m+1)-1}\right)\\
    & + \frac{1}{N} \mathbb{M}_{J_N}^{-1}\left(\xi_{\mathcal{J}(J_N)+1}^{-1}D_{\mathcal{J}(J_N)} - \xi_{N}^{-1}D_{N}\right)\\
    & - \frac{1}{N}\sum_{m=J_0}^{J_N -1}\mathcal{R}_{m, J_N}^{\top} \mathbb{M}^{-1}_m\sum_{k=\mathcal{J}(m)+1}^{\mathcal{J}(m+1)-1}\sum_{l\in\{1,2,3,5,6\}} \zeta_{l,k} - \frac{1}{N} \mathbb{M}^{-1}_{J_N}\sum_{k=\mathcal{J}(J_N)+1}^{N}\sum_{l\in\{1,2,3,5,6\}} \zeta_{l,k}\\
    & + \frac{1}{N}\sum_{m=J_0}^{J_N-1} \mathcal{R}_{m, J_N}^{\top}D_{\mathcal{J}(m+1)-1}
    + \frac{1}{N}\sum_{m=J_0}^{J_N-1}\sum_{k=\mathcal{J}(m)}^{\mathcal{J}(m+1)-1}\left(\mathcal{R}_{m, J_N}^{\top}\omega_{m,0}- \omega_{J_N, 0}\right) + o\left(N^{-\frac{1}{2}}\right).
\end{align*}
We analyze the above equation term by term, using the refined rate \eqref{refined_D_rate} for $D_k$; note that each bound below is stated for the two components of \eqref{refined_D_rate} separately. Due to the boundedness of eigenvalues of $ \mathbb{M}_m$ and $\Vert \mathcal{R}_{m, J_N}\Vert_{\mathrm{op}}$, we have that
\begin{align*}
    & \left\Vert \frac{1}{N}\sum_{m=J_0}^{J_N - 1}\mathcal{R}_{m, J_N}^{\top} \mathbb{M}_m^{-1} \sum_{k=\mathcal{J}(m)+1}^{\mathcal{J}(m+1)-2}\left(\xi_{k+1}^{-1} - \xi_k^{-1}\right)D_k + \frac{1}{N} \mathbb{M}_{J_N}^{-1} \sum_{k=\mathcal{J}(J_N)+1}^{N-1}\left(\xi_{k+1}^{-1} - \xi_k^{-1}\right)D_k \right\Vert_2\\
    & \leq \frac{1}{N}\sum_{k=1}^N \left(\xi_{k+1}^{-1} - \xi_k^{-1}\right)\left\Vert D_k \right\Vert_2\leq \frac{1}{N}\sum_{k=1}^N \left(k^{\alpha_{\xi}-1 - \frac{\alpha_{\xi}-\alpha_{\dagger}}{2}}\text{PolyLog} + k^{\alpha_{\xi}-1-(2s-2)\alpha_{\dagger}}\right)\\
    &\leq CN^{\frac{\alpha_{\xi}+\alpha_{\dagger}}{2}-1}\text{PolyLog} + CN^{\alpha_{\xi} - (2s-2)\alpha_{\dagger}-1} + CN^{-1}\log(N),  \qquad \text{a.s.}
\end{align*}
\begin{align*}
   & \left\Vert\frac{1}{N}\sum_{m=J_0}^{J_N - 1}\mathcal{R}_{m, J_N}^{\top} \mathbb{M}_m^{-1}\left(\xi_{\mathcal{J}(m)+1}^{-1}D_{\mathcal{J}(m)} - \xi_{\mathcal{J}(m+1)-1}^{-1}D_{\mathcal{J}(m+1)-1}\right)
    + \frac{1}{N} \mathbb{M}_{J_N}^{-1}\left(\xi_{\mathcal{J}(J_N)+1}^{-1}D_{\mathcal{J}(J_N)} - \xi_{N}^{-1}D_{N}\right)\right\Vert_2\\
    &\leq \frac{C}{N}\left(\sum_{m=J_0}^{J_N } \left(m^{\frac{\alpha_{\xi}}{\alpha_{\dagger}} - \frac{\alpha_{\xi} - \alpha_{\dagger}}{2\alpha_{\dagger}}}\text{PolyLog} + m^{\frac{\alpha_{\xi} - (2s-2)\alpha_{\dagger}}{\alpha_{\dagger}}}\right) + N^{\alpha_{\xi}- \frac{\alpha_{\xi} - \alpha_{\dagger}}{2}}\text{PolyLog} + N^{\alpha_{\xi}-(2s-2)\alpha_{\dagger}}\right)\\
    &\leq CN^{\frac{\alpha_{\xi} +3\alpha_{\dagger}}{2} - 1}\text{PolyLog} + CN^{\alpha_{\xi}+\alpha_{\dagger}-(2s-2)\alpha_{\dagger}-1} + CN^{-1}\log(N),  \qquad \text{a.s.}
\end{align*}
\begin{align*}
    \left\Vert \frac{1}{N}\sum_{m=J_0}^{J_N-1} \mathcal{R}_{m, J_N}^{\top}D_{\mathcal{J}(m+1)-1}  \right\Vert_2 &  \leq \frac{C}{N}\sum_{m=J_0}^{J_N}\left(m^{-\frac{\alpha_{\xi}-\alpha_{\dagger}}{2\alpha_{\dagger}}}\text{PolyLog} + m^{-\frac{(2s-2)\alpha_{\dagger}}{\alpha_{\dagger}}}\right)\leq \frac{C}{N}\left(N^{\frac{3\alpha_{\dagger}-\alpha_{\xi}}{2}}\text{PolyLog} + \log N\right), \qquad \text{a.s.}
\end{align*}
Note that the above terms are all $o(N^{-\frac{1}{2}})$ a.s.: the leading components require $\alpha_{\xi}+3\alpha_{\dagger}<1$, exactly as before, while the components generated by the secondary rate $k^{-(2s-2)\alpha_{\dagger}}$ in \eqref{refined_D_rate} are $o(N^{-\frac{1}{2}})$ whenever $\alpha_{\xi}<\frac{1}{2}+(2s-3)\alpha_{\dagger}$, which is implied by $\alpha_{\xi}<\frac{1}{2}+(s-\frac{1}{2})\alpha_{\dagger}$ since $2s-3\geq s-\frac{1}{2}$ for $s\geq\frac{5}{2}$.

It only remains to look at \[\frac{1}{N}\sum_{m=J_0}^{J_N -1}\mathcal{R}_{m, J_N}^{\top} \mathbb{M}^{-1}_m\sum_{k=\mathcal{J}(m)+1}^{\mathcal{J}(m+1)-1}\sum_{l\in\{1,2,3,5,6\}} \zeta_{l,k} + \frac{1}{N} \mathbb{M}^{-1}_{J_N}\sum_{k=\mathcal{J}(J_N)+1}^{N}\sum_{l\in\{1,2,3,5,6\}} \zeta_{l,k}.\]
Similar to the proof in \autoref{lemma5.2part2}, we look at each $l$ separately. Note that $\zeta_{1,k}$, $\zeta_{2,k}$ and $\zeta_{3,k}$ depend on the \emph{full} iterate error $\Delta\widetilde\omega_{k-1} = D_{k-1}+T_{k-1}$, so we bound their sums using both components: the refined rate \eqref{refined_D_rate} for $D$, and the checkpoint-aware sums of \autoref{lemma_transient}(iii) for $T$.

For $l=1$, obviously we have that $\Vert \zeta_{1,k}\Vert_2\leq J_k^{\frac{3}{2} -s} \Vert \Delta\widetilde\omega_{k-1} \Vert_2\leq J_k^{\frac{3}{2}-s}\left(\Vert D_{k-1}\Vert_2 + \Vert T_{k-1}\Vert_2\right)$, so
\begin{align*}
  &   \left\Vert \frac{1}{N}\sum_{m=J_0}^{J_N -1}\mathcal{R}_{m, J_N}^{\top} \mathbb{M}^{-1}_m\sum_{k=\mathcal{J}(m)+1}^{\mathcal{J}(m+1)-1}\zeta_{1,k} + \frac{1}{N} \mathbb{M}^{-1}_{J_N}\sum_{k=\mathcal{J}(J_N)+1}^{N} \zeta_{1,k}\right\Vert_2 \\
  & \leq \frac{C}{N}\sum_{k=1}^N \Vert \zeta_{1,k}\Vert_2 \leq \frac{C}{N}\sum_{k=1}^N k^{-\frac{\alpha_{\xi}-\alpha_{\dagger}}{2} + (\frac{3}{2}-s)\alpha_{\dagger}}\text{PolyLog}+ \frac{C}{N}\sum_{k=1}^Nk^{(\frac{3}{2}-s)\alpha_{\dagger}}\left(k^{-(2s-2)\alpha_{\dagger}} + \Vert T_{k-1}\Vert_2\right)\\
  &\leq CN^{-\frac{\alpha_{\xi}-\alpha_{\dagger}}{2} + (\frac{3}{2}-s)\alpha_{\dagger}}\text{PolyLog} + CN^{(\frac{7}{2}-3s)\alpha_{\dagger}} + \frac{C}{N}\left(N^{\alpha_{\xi}+(2-2s)\alpha_{\dagger}}+\log N\right), \qquad \text{a.s.},
\end{align*}
where the sum involving $\Vert T_{k-1}\Vert_2$ is bounded by \autoref{lemma_transient}(iii) (the factor $k^{(\frac{3}{2}-s)\alpha_{\dagger}}$ is decreasing, so for the checkpoint at $\mathcal{J}(m)$ it may be replaced by $\mathcal{J}(m)^{(\frac{3}{2}-s)\alpha_{\dagger}}$ inside the checkpoint sums, giving $\sum_mm^{\frac{\alpha_{\xi}+(\frac{3}{2}-s)\alpha_{\dagger}}{\alpha_{\dagger}}}\delta_m$-type sums, which \autoref{lemma_telescoping}(iii) bounds by $C(N^{\alpha_{\xi}+(\frac{3}{2}-s)\alpha_{\dagger}+(\frac{1}{2}-s)\alpha_{\dagger}}+\log N) = C(N^{\alpha_{\xi}+(2-2s)\alpha_{\dagger}}+\log N)$). When $\alpha_{\xi}+(2s-4)\alpha_{\dagger}>1$, the first term is $o(N^{-\frac{1}{2}})$; the second is $o(N^{-\frac{1}{2}})$ since $(3s-\frac{7}{2})\alpha_{\dagger}\geq (2s-1)\alpha_{\dagger}\geq q_{\star}\alpha_{\dagger}>\frac{1}{2}$ for $s\geq\frac{5}{2}$; and the third is $o(N^{-\frac{1}{2}})$ since $\alpha_{\xi}+(2-2s)\alpha_{\dagger}<\frac{1}{2}$ is implied by $\alpha_{\xi}<\frac{1}{2}+(s-\frac{1}{2})\alpha_{\dagger}$, using $2s-2\geq s-\frac{1}{2}$ for $s\geq\frac{3}{2}$. So the $l=1$ term is $o_{\text{a.s.}}(N^{-\frac{1}{2}})$.

 The block normal equations
contain only no-jump indices.  Hence \eqref{zeta2_corrected_bound} gives
\begin{align*}
  &   \left\Vert \frac{1}{N}\sum_{m=J_0}^{J_N -1}\mathcal{R}_{m, J_N}^{\top} \mathbb{M}^{-1}_m\sum_{k=\mathcal{J}(m)+1}^{\mathcal{J}(m+1)-1}\zeta_{2,k} + \frac{1}{N} \mathbb{M}^{-1}_{J_N}\sum_{k=\mathcal{J}(J_N)+1}^{N} \zeta_{2,k}\right\Vert_2 \\
  & \leq \frac{C}{N}\sum_{k=1}^N
 J_k^{2}\Vert \Delta \widetilde\omega_{k-1}\Vert_2^2 \leq \frac{C}{N}\sum_{k=1}^N J_k^{2}\Vert D_{k}\Vert_2^2
  + \frac{C}{N}\sum_{k=1}^N J_k^{2}\Vert T_{k}\Vert_2^2.
\end{align*}
The index shifts are harmless because $J_k/J_{k-1}$ is uniformly bounded.
For the $D$-part, \eqref{refined_D_rate} gives
\[
\frac{C}{N}\sum_{k=1}^N J_k^{2}\Vert D_{k}\Vert_2^2\leq CN^{-\alpha_{\xi}+ 3\alpha_{\dagger}}\text{PolyLog} + CN^{(6-4s)\alpha_{\dagger}}\text{PolyLog}, \qquad \text{a.s.},
\]
which is $o(N^{-\frac{1}{2}})$ when $\alpha_{\xi} >\frac{1}{2} + 3\alpha_{\dagger}$, using also $(4s-6)\alpha_{\dagger}\geq (2s-1)\alpha_{\dagger}\geq q_{\star}\alpha_{\dagger}>\frac{1}{2}$ for $s\geq\frac{5}{2}$. For the $T$-part, \autoref{lemma_transient}(iii) (with the decreasing-weight device described under $l=1$) gives
\[
\frac{C}{N}\sum_{k=1}^N J_k^{2}\Vert T_{k}\Vert_2^2\leq \frac{C}{N}\left(N^{\alpha_{\xi}+(3-2s)\alpha_{\dagger}}+\log N\right) = o\left(N^{-\frac{1}{2}}\right),
\]
since $\alpha_{\xi}+(3-2s)\alpha_{\dagger}<\frac{1}{2}$ is implied by $\alpha_{\xi}<\frac{1}{2}+(s-\frac{1}{2})\alpha_{\dagger}$, using $2s-3\geq s-\frac{1}{2}$ for $s\geq\frac{5}{2}$. So the $l=2$ term is $o_{\text{a.s.}}(N^{-\frac{1}{2}})$.

For later gain and AME arguments, we also record an all-index checkpoint
bound.  Let $t_m=\mathcal J(m)$ and $R=(\alpha_\xi-\alpha_\dagger)/2$.
Equations~\eqref{refined_D_rate} and Lemma~\ref{lemma_transient}(iv) imply
\[
 \Vert\Delta\widetilde\omega_{\mathcal J(m)-1}\Vert_2
 =O \!\left(m^{-\frac{\alpha_{\xi}-\alpha_{\dagger}}{2\alpha{_\dagger}}}\text{PolyLog}(m)+m^{-s}\right), \qquad \text{a.s.}
\]
Because \autoref{new_rate} gives
$\alpha_\dagger<1/12$ and $\alpha_\xi>9\alpha_\dagger$, we have
$R/\alpha_\dagger>4$.  Therefore \eqref{zeta2_corrected_bound} yields
\begin{align}
 \sum_{\mathcal{J}(m)\leq N}\Vert\zeta_{2,\mathcal{J}(m)}\Vert_2
 &\leq C\sum_{m\leq J_N}m^2\left[
  \{m^{-\frac{\alpha_{\xi}-\alpha_{\dagger}}{2\alpha_\dagger}}\text{PolyLog} +m^{-s}\}^2\right.\nonumber\\
 &\hspace{34mm}\left.
  +m^{-s}\{m^{-\frac{\alpha_{\xi}-\alpha_{\dagger}}{2\alpha_\dagger}}\text{PolyLog} +m^{-s}\}
 \right]
 =O_{\rm a.s.}(1).\label{zeta2_checkpoint_absolute}
\end{align}
Together with the no-jump calculation above, this gives
\begin{equation}\label{zeta2_all_index_absolute}
 \sum_{k\leq N}\Vert\zeta_{2,k}\Vert_2=o(\sqrt N), \qquad \text{a.s.}
\end{equation}

For $l=3$, let $\chi_k=\boldsymbol 1\{J_k=J_{k-1}\}$ (and set
$\chi_1=0$).  Up to finitely many initialization indices, the block expression
in \eqref{regular_recursion} is exactly equal to $
 \frac1N\sum_{k=1}^N\mathcal R_{J_k,J_N}^{\top}
 \mathbb M_{J_k}^{-1}\chi_k\zeta_{3,k}.
$ 
Define 
$U_{3,k}=\chi_k\mathbb M_{J_k}^{-1}\zeta_{3,k}$, $U_{3,k}$ is a martingale
difference.  The empirical-Hessian bound and the uniform inverse bound give
$
 \mathbb E_{k-1}\Vert U_{3,k}\Vert_2^2
 \leq CJ_k^3\Vert\Delta\widetilde\omega_{k-1}\Vert_2^2.
$ 
Put $a=\alpha_\dagger$ and $\rho=\alpha_\xi$.  By
\eqref{refined_D_rate},  \autoref{lemma_transient}(ii), and
$\Vert D+T\Vert^2\leq2\Vert D\Vert^2+2\Vert T\Vert^2$, we have that 
\begin{align*}
 \sum_{k=2}^{\infty}\frac1k
 \mathbb E_{k-1}\Vert U_{3,k}\Vert_2^2
 &\leq C\sum_k k^{-1-\alpha_\xi+4\alpha_\dagger}\text{PolyLog}
   +C\sum_k k^{-1+(7-4s)\alpha_\dagger}\text{PolyLog}\\
   & 
   +C\sum_k k^{-1+(3-2s)\alpha_{\dagger}}<\infty
 \qquad\text{a.s.}
\end{align*}
Here $\alpha_\xi>4\alpha_\dagger$ follows from
$\alpha_\xi>\tfrac12+3\alpha_\dagger$ and $\alpha_\dagger<\tfrac12$, while $s>7/2$ makes the other two
exponents strictly below $-1$.   \autoref{lem:common_space_mds} therefore
implies that the $l=3$ block sum is $o_{\rm a.s.}(N^{-1/2})$.

For $l=4$, recall that the block sums involve $\zeta_{4,k}^{\circ}$ only through indices $\mathcal{J}(m)+1\leq k\leq \mathcal{J}(m+1)-1$ at which the sieve dimension does not change, so $\zeta^{\circ}_{4,k}=0$ there and this term vanishes identically (the deterministic mean $\mathbb{E}\zeta_{4,k}$ at expansion times is part of the transient $T_k$, already handled in \eqref{transient_average}).

For $l=5$, decompose
$\zeta_{5,k}=b_{J_k}+\zeta_{5,k}^{\circ}$, where
$\mathbb E_{k-1}\zeta_{5,k}^{\circ}=0$,
$\Vert b_{J_k}\Vert_2\leq CJ_k^{-q_\star}$, and
$\mathbb E_{k-1}\Vert\zeta_{5,k}^{\circ}\Vert_2^2\leq CJ_k^{3-2s}$.
The deterministic part satisfies the exact power-sum bound $\frac C N\sum_{k=1}^N J_k^{-q_\star} = o(N^{-1/2})$ 
because $q_\star\alpha_\dagger>1/2$.  For the centered part set
$U_{5,k}=\chi_k\mathbb M_{J_k}^{-1}\zeta_{5,k}^{\circ}$.  It is a martingale
difference and
\[
 \sum_{k=2}^{\infty}\frac1k
 \mathbb E_{k-1}\Vert U_{5,k}\Vert_2^2
 \leq C\sum_{k=2}^{\infty}k^{-1+(3-2s)\alpha_\dagger}<\infty
\]
because $s>3/2$.   \autoref{lem:common_space_mds} gives an
$o_{\rm a.s.}(N^{-1/2})$ averaged centered term.  Hence the entire $l=5$
contribution is $o_{\rm a.s.}(N^{-1/2})$.

The above implies that the following a.s. holds  
\begin{align*}
    \Delta\doubleoverline{\omega}_N  & =   
     \frac{1}{N}\sum_{k=1}^{N}\mathcal{R}_{J_k, J_N}^{\top} \mathbb{M}^{-1}_{J_k}\frac{1}{B}\sum_{i=1}^B \varepsilon_{i,k}\begin{pmatrix}
   \nabla_z  F_0(Z_{i,k}(\theta_0)) X_{i,k}\\
    \Psi_{J_{k}}(Z_{i,k}(\theta_0))
\end{pmatrix}\\
&  \ \ \ \ \ \ \ \ \ \ \ \ \ \  \ \ \ \ \ \  \ \ \ \ \ \  \ \ \ \ \ \   +\frac{1}{N}\sum_{k=1}^{N}\left(\mathcal{R}_{J_k, J_N}^{\top}\omega_{J_k,0}- \omega_{J_N, 0}\right)  +o\left(N^{-\frac{1}{2}}\right).
\end{align*}
Note that the first summation on the right-hand side  omits the expansion
indices, but define 
$\chi_k^{\rm  }=\boldsymbol1\{J_k>J_{k-1}\}$ and
$U_{6,k} =\chi_k \mathbb M_{J_k}^{-1}\zeta_{6,k}$.
Then
$
 \mathbb E_{k-1}\Vert U_{6,k} \Vert_2^2
 \leq C\chi_k J_k.$ 
We have that 
\[
 \sum_{k=1}^{\infty}\frac1k
 \mathbb E_{k-1}\Vert U_{6,k} \Vert_2^2
 \leq C\sum_m\frac{m}{\mathcal J(m)}
 \asymp\sum_m m^{1-1/\alpha_\dagger}<\infty,
\]
since $\alpha_\dagger<1/2$.   \autoref{lem:common_space_mds} shows that the
restored checkpoint score is $o_{\rm a.s.}(\sqrt N)$ before division by $N$.

To further simplify notations, define \[ \mathbb{M}_{J,11} = \mathbb{E}\left[(\nabla_zF_0(Z_0))^2 XX^{\top}\right], \ \  \mathbb{M}_{J,12} = \mathbb{E}\left[ \nabla_zF_0(Z_0)  X\Psi_J(Z_0)^{\top}\right], \ \ \mathbb{M}_{J,22} = \Gamma_J.\] We have 
\begin{align*}
  &    \mathbb{M}_{J}^{-1}\begin{pmatrix}
   \varepsilon\nabla_z  F_0(Z_0) X\\
   \varepsilon\Psi_{J}(Z_0)
\end{pmatrix}\\
& =\varepsilon\begin{pmatrix}
    \left(  \mathbb{M}_{J,11} -  \mathbb{M}_{J,12} \mathbb{M}_{J,22}^{-1} \mathbb{M}_{J,12}^{\top}\right)^{-1}\left( \nabla_z F_0(Z_0) X -  \mathbb{M}_{J,12} \mathbb{M}_{J,22}^{-1}\Psi_{J}(Z_0) \right)\\
    \left( \mathbb{M}_{J,22} -  \mathbb{M}_{J,12}^{\top} \mathbb{M}_{J,11}^{-1} \mathbb{M}_{J,12}\right)^{-1}\left( \Psi_{J}(Z_0) -  \mathbb{M}_{J, 12}^{\top} \mathbb{M}_{J,11}^{-1}\nabla_z  F_0(Z_0) X\right)
\end{pmatrix} 
\end{align*}
We note that by \autoref{condition6}, 
\begin{align*}
 & \left\Vert \Psi_J(\cdot)^{\top}\Gamma_J^{-1}\mathbb{E}\left(\nabla_zF_0(Z_0)  \Psi_J(Z_0)X^{\top}\right) - \nabla_z F_0(\cdot)\mu_0(\theta_0, \cdot)^{\top}\right\Vert_{\infty}\\
 & = \left\Vert \Psi_J(\cdot)^{\top}\Gamma_J^{-1}\mathbb{E}\left(\nabla_zF_0(Z_0)  \Psi_J(Z_0)\mu_0(\theta_0,Z_0)^{\top}\right) - \nabla_z F_0(\cdot)\mu_0(\theta_0, \cdot)^{\top}\right\Vert_{\infty}\leq CJ^{-s_{\mu}}.
\end{align*}
So \begin{align*}
& \left\Vert  \mathbb{M}_{J,12} \mathbb{M}_{J,22}^{-1} \mathbb{M}_{J,12}^{\top} - \mathbb{E}\left((\nabla_z F_0(Z_0))^2 \mu_0(\theta_0, Z_0)\mu_0(\theta_0, Z_0)^{\top}\right)\right\Vert_F\\
& \leq  \mathbb{E}[\Vert\nabla_z F_0(Z_0)X\Vert_2]\left\Vert \Psi_J(\cdot)^{\top}\Gamma_J^{-1}\mathbb{E}\left(\nabla_zF_0(Z_0)  \Psi_J(Z_0)X^{\top}\right) - \nabla_z F_0(\cdot)\mu_0(\theta_0, \cdot)^{\top}\right\Vert_{\infty}\leq CJ^{-s_{\mu}}.
\end{align*}
This, due to uniformly lower bounded eigenvalues of $\mathbb{M}_{\theta}$ according to \autoref{eigenvalue} and $\Vert  \mathbb{M}_J - \nabla_{\omega\omega} \mathcal{L}_J(\omega_{J,0}) \Vert_F\leq CJ^{\frac{3}{2} - s}$, implies that 
\begin{align*}
  &  \left\Vert  \left(  \mathbb{M}_{J,11} -  \mathbb{M}_{J,12} \mathbb{M}_{J,22}^{-1} \mathbb{M}_{J,12}^{\top}\right)^{-1} - \mathbb{M}_{\theta}^{-1} \right\Vert_F \leq CJ^{-s_{\mu}} \rightarrow 0, J\rightarrow \infty,
\end{align*}
where recall that $\mathbb{M}_{\theta}  = \mathbb{E}[\left(\nabla_z F_0(Z_0)\right)^2(X- \mu_0(\theta_0, Z_0))(X- \mu_0(\theta_0, Z_0))^{\top}]$. 
Similarly, we have that 
\begin{align*}
     \Vert   \mathbb{M}_{J,12} \mathbb{M}_{J,22}^{-1}\Psi_J(Z_0)-\nabla_z F_0(Z_0)\mu_0(\theta_0, Z_0)\Vert_2 \leq CJ^{-s_{\mu}}
\end{align*}
uniformly for all $Z_0$. Together, we have that 
\begin{align*}
   & \varepsilon\left(  \mathbb{M}_{J,11} -  \mathbb{M}_{J,12} \mathbb{M}_{J,22}^{-1} \mathbb{M}_{J,12}^{\top}\right)^{-1}\left( \nabla_z F_0(Z_0) X -  \mathbb{M}_{J,12} \mathbb{M}_{J,22}^{-1}\Psi_{J}(Z_0) \right)\\
   &  = \varepsilon \mathbb{M}_{\theta}^{-1} \nabla_z F_0(Z_0)\left(X - \mu_0(\theta_0, Z_0)\right) + \varepsilon\delta_{J,\theta}(W), \\
   & \qquad \qquad \Vert \delta_{J,\theta}(W) \Vert_2\leq CJ^{-s_{\mu}}(1 + \Vert X\Vert_2 + \Vert \mu_0(\theta_0, Z_0) \Vert_2) \text{ for all $W$}. 
\end{align*}
The above  implies that 
\begin{align}
    \Delta\doubleoverline{\theta}_N & = \frac{1}{N}\sum_{k=1}^N \frac{1}{B}\sum_{i=1}^B \varepsilon_{i,k} \mathbb{M}_{\theta}^{-1} \nabla_z F_0(Z_{0,i,k})\left(X_{i,k} - \mu_0(\theta_0, Z_{0,i,k})\right)\nonumber\\
    & + \frac{1}{N}\sum_{k=1}^N\frac{1}{B}\sum_{i=1}^B\varepsilon_{i,k}\delta_{J_k,\theta}(W_{i,k}) + o\left(N^{-\frac{1}{2}}\right), \qquad \text{a.s.}
\end{align}
To finish the first result, set
\[
 U_{\theta,k}=\frac1B\sum_{i=1}^B
 \varepsilon_{i,k}\delta_{J_k,\theta}(W_{i,k}).
\]
This is a fixed-dimensional martingale difference and the displayed envelope
for $\delta_{J,\theta}$, together with \autoref{cond11}, gives
$\mathbb E_{k-1}\Vert U_{\theta,k}\Vert_2^2\leq CJ_k^{-2s_\mu}$.
Consequently
$\sum_k k^{-1}\mathbb E_{k-1}\Vert U_{\theta,k}\Vert_2^2<\infty$.
The row part of \autoref{lem:common_space_mds} yields
$N^{-1}\sum_{k\leq N}U_{\theta,k}=o_{\rm a.s.}(N^{-1/2})$.

Also, for   $\Delta\doubleoverline{\beta}_N$, using the same reasoning, we have that
\begin{align*}
    \Delta\doubleoverline{\beta}_N  & =   
     \frac{1}{N}\sum_{k=1}^{N}R_{J_k, J_N}^{\top}\mathbb{M}^{-1}_{\beta, J_k}\frac{1}{B}\sum_{i=1}^B \varepsilon_{i,k}\left(\Psi_{J_k}\left(Z_{0,i,k}\right) - P_{J_k}\nabla_{z}F_0(Z_{0,i,k})X_{i,k}\right)\\
     & + \frac{1}{N}\sum_{k=1}^{N}\left(R_{J_k, J_N}^{\top}\beta_{J_k,0}- \beta_{J_N, 0}\right) +o\left(N^{-\frac{1}{2}}\right), \qquad \text{a.s.}  
\end{align*}
where recall that $ P_{J} = \left(\mathbb{E}\left[\Psi_J(Z_0)\nabla_z F_0(Z_0) X^{\top}\right]\right)\left(\mathbb{E}\left[\left(\nabla_z F_0(Z_0)\right)^2XX^{\top}\right]\right)^{-1}$.

\subsubsection{Extension to estimated conditioning matrices}\label{proof_estimated_gain}
We now pass from the auxiliary identity-gain calculation to the  update adjusted by conditioning matrix $\widehat G_k$.  For simplicity, define 
\[
 G_k^0=G_{J_k,0},\qquad E_k=\widehat G_k-G_k^0,\qquad
 B_k=(G_k^0)^{-1}E_k,\qquad d_k=\Vert E_k\Vert_{\rm op},
\]
and
\[
 \overline g_k(\omega)=\frac1B\sum_{i=1}^B
 \nabla_\omega\mathcal L_{J_k}(\omega,W_{i,k}).
\]
All the analysis are based on the probability-one event on which the gain,
localization, and rate conditions hold.

\paragraph{Stability, localization, and raw rates.}
On a no-jump step let $H_k$ be the integrated population Hessian on the local
segment.  \autoref{cond_gain} gives  $d_k=o(1)$, and the local Hessian bound
imply, eventually,
\[
 \lambda_{\min}\!\left\{\frac{\widehat G_kH_k+H_k\widehat G_k}{2}\right\}
 \geq c_G/2,
 \qquad \Vert\widehat G_kH_k\Vert_{\rm op}\leq C.
\]
Hence
\begin{equation}\label{audited_gain_contraction}
 \Vert(I-\xi_k\widehat G_kH_k)v\Vert_2^2
 \leq(1-c\xi_k)\Vert v\Vert_2^2
\end{equation}
for all sufficiently large $k$.   At a sieve expansion, bounded re-embedding contributes the
same $C\boldsymbol1_k$ factor as before.  Premultiplication by the bounded
predictable $\widehat G_k$ changes only constants in the score and empirical
Hessian bounds.  Repeating the centered-score expansion leading to
\eqref{general_dynamics_nls} therefore gives the explicit recurrence
\begin{equation}\label{audited_gain_dynamics}
 \mathbb E_{k-1}\Vert\Delta\widetilde\omega_k^G\Vert_2^2
 \leq(1-c\xi_k+C\boldsymbol1_k)
       \Vert\Delta\widetilde\omega_{k-1}^G\Vert_2^2
   +C\xi_k^2J_k+C\xi_kJ_k^{-2q_\star}
   +C\boldsymbol1_kJ_k^{-2s}.
\end{equation}
Thus the initial rate and eventual inactivity of the projection follow by the
same blockwise Robbins--Siegmund argument.

The refined decomposition also survives.  The extra linear
term generated by $E_k$ is absorbed into the contraction in
\eqref{audited_gain_contraction}; its centered score has conditional covariance
bounded by $Cd_k^2J_k\leq CJ_k$; and all nonlinear terms are unchanged up to a
constant.  At a sieve increase, the local-Hessian remainder also contains the
second term in \eqref{zeta2_corrected_bound}, with the actual gain-path error
in place of $\Delta\widetilde\omega_{k-1}$.  Bounded predictable
premultiplication changes only its constant, and Young's inequality produces
the same $k^{-(2s-2)\alpha_\dagger}$ deterministic floor.  Thus it does not
alter either the raw-rate bootstrap or the componentwise theta rate.  Let $D_k^G$ and $T_k^G$ denote the counterparts of $D_k$ and $T_k$ with preconditioning, the above implies that 
\begin{equation}\label{audited_gain_raw_rates}
 \Vert D_k^G\Vert_2
 =O_{\rm a.s.}\left(k^{-\frac{\alpha_\xi-\alpha_\dagger}{2}}\text{PolyLog}+k^{-(2s-2)\alpha_\dagger}\right),
 \qquad
 \Vert T_k^G\Vert_2\leq Ck^{-s\alpha_\dagger},
\end{equation}
with the weighted transient bounds in  \autoref{lemma_transient}.

\paragraph{Exact population-gain cancellation.}
Using the same score remainders as in \eqref{regular_recursion}, evaluated on
the conditioned path, the no-projection recursion is
\begin{equation}\label{gain_regular_recursion}
 \Delta\widetilde\omega_k^G
 =A_k^G\Delta\widetilde\omega_{k-1}^G
 -\xi_kG_k^0\left(\sum_{\ell=1}^6\zeta_{\ell,k}+\zeta_{G,k}\right)
 +\zeta_{7,k},
\end{equation}
where
\[
 A_k^G=(I-\xi_kG_k^0\mathbb M_{J_k})
       \mathcal R_{J_{k-1},J_k}^{\top},
 \qquad
 \zeta_{G,k}=B_k\,
 \overline g_k(\mathcal R_{J_{k-1},J_k}^{\top}
                         \widetilde\omega_{k-1}^G).
\]
Define the checkpoint transient with
$\tau_k^G=\zeta_{7,k}-\xi_kG_k^0\mathbb E_{k-1}\zeta_{4,k}$.  On a constant-$m$ block,
\[
 D_k^G=(I-\xi_kG_{m,0}\mathbb M_m)D_{k-1}^G
       -\xi_kG_{m,0}Z_k^G,
\]
where
\[
 Z_k^G=\zeta_{1,k}+\zeta_{2,k}+\zeta_{3,k}+\zeta_{4,k}^{\circ}
       +\zeta_{5,k}+\zeta_{6,k}+\zeta_{G,k}.
\]
The transition is uniformly stable because
\[
 I-\xi G_{m,0}\mathbb M_m
 =G_{m,0}^{1/2}
  \{I-\xi G_{m,0}^{1/2}\mathbb M_mG_{m,0}^{1/2}\}
  G_{m,0}^{-1/2}.
\]
Most importantly, the exact normal equation is
\begin{equation}\label{gain_normal_equation}
 D_{k-1}^G
 =\xi_k^{-1}\mathbb M_m^{-1}G_{m,0}^{-1}
   (D_{k-1}^G-D_k^G)-\mathbb M_m^{-1}Z_k^G.
\end{equation}
After blockwise Abel summation, $G_{m,0}^{-1}$ appears only in boundary and
step-size-difference terms.  Those rows are uniformly bounded, so the boundary
estimates in Subsection~\ref{lemma5.2part3} are unchanged.  The gain cancels
exactly from the score term in \eqref{gain_normal_equation}; consequently the
leading PR score $-\mathbb M_m^{-1}\zeta_{6,k}$ remains.

\begin{lemma}[Estimated-gain averaging remainder]\label{lem:estimated_gain_remainder}
Under the primitive conditions in  \autoref{remark_gain_rate}, for every
predictable scalar $a_k$ with $|a_k|\leq1$,
\begin{equation}\label{estimated_gain_average_goal}
 \frac1{\sqrt N}\left\Vert\sum_{k=1}^N
 a_k\mathcal R_{J_k,J_N}^{\top}\mathbb M_{J_k}^{-1}\zeta_{G,k}
 \right\Vert_2\longrightarrow0\qquad\text{a.s.}
\end{equation}
If $C_k\in\mathbb R^{q\times(p+J_k)}$ is predictable, $q$ is fixed, and
$\sup_k\Vert C_k\Vert_{\rm op}<\infty$, then also
\begin{equation}\label{estimated_gain_row_goal}
 \max_{1\leq n\leq N}\frac1{\sqrt N}
 \left\Vert\sum_{k=1}^n a_kC_k\mathbb M_{J_k}^{-1}\zeta_{G,k}
 \right\Vert_2\longrightarrow0\qquad\text{a.s.}
\end{equation}
\end{lemma}

\begin{proof}
The exact score decomposition is
\begin{equation}\label{audited_gain_score_decomposition}
 \overline g_k(\mathcal R_{J_{k-1},J_k}^{\top}\widetilde\omega_{k-1}^G)
 =\mathbb M_{J_k}\mathcal R_{J_{k-1},J_k}^{\top}
     (D_{k-1}^G+T_{k-1}^G)
  +\sum_{\ell=1}^6\zeta_{\ell,k}.
\end{equation}
For the regular state part, uniform bounds on $\mathbb M_J^{\pm1}$ and
$(G_J^0)^{-1}$ give
\[
 \frac1{\sqrt N}\sum_{k\leq N}
 d_k\Vert D_{k-1}^G\Vert_2\longrightarrow0
\]
by the second primitive condition in \autoref{remark_gain_rate} and
\eqref{audited_gain_raw_rates}; the secondary state exponent
$(2s-2)\alpha_\dagger$ is larger than $R$.  Since $d_k\leq1$ eventually,
 \autoref{lemma_transient}(iii) gives the same conclusion for $T_{k-1}^G$.
These absolute bounds are uniform over partial sums and remain valid after any
bounded row premultiplication.

Now we analyze the $\zeta$ terms in the PR average. The $\zeta_1$ sum and the no-jump part of the $\zeta_2$ sum are absolutely
$o_{\rm a.s.}(\sqrt N)$ by the $l=1,2$ calculations in
 \autoref{lemma5.2part3}.  The checkpoint calculation
\eqref{zeta2_checkpoint_absolute} also applies to the actual gain path by
\eqref{audited_gain_raw_rates}, so the expansion-round $\zeta_2$ sum is
$O_{\rm a.s.}(1)$.  Multiplication by
$\mathbb M_{J_k}^{-1}B_k$, whose norm is $O(d_k)$ and eventually bounded,
does not enlarge them.  For $\zeta_3$, define
$U_{G_3,k}=\mathbb M_{J_k}^{-1}B_k\zeta_{3,k}$.  It is a martingale difference
and
$
 \mathbb E_{k-1}\Vert U_{G_3,k}\Vert_2^2
 \leq Cd_k^2J_k^3\Vert\Delta\widetilde\omega_{k-1}^G\Vert_2^2.
$
The normalized quadratic variation is summable by the calculation for
$U_{3,k}$ above, since $d_k\leq1$ eventually.  \autoref{lem:common_space_mds}
handles both \eqref{estimated_gain_average_goal} and the row-uniform version.
For $\zeta_4$ term, again decompose $\zeta_{4,k}=\mathbb E_{k-1}\zeta_{4,k}+\zeta_{4,k}^{\circ}$.
For the centered part,
$
 \mathbb E_{k-1}\Vert
 \mathbb M_{J_k}^{-1}B_k\zeta_{4,k}^{\circ}\Vert_2^2
 \leq Cd_k^2\boldsymbol1_kJ_k^{3-2s},
$ 
whose normalized sum is finite.  The mean is supported on checkpoints and
\[
 \sum_{k=1}^N \Vert\mathbb M_{J_k}^{-1}B_k
       \mathbb E_{k-1}\zeta_{4,k}\Vert_2
 \leq C\sum_{m\leq J_N}\delta_m=O(\log N)=o(\sqrt N)
\]
by  \autoref{lemma_telescoping}(iii).
Next write $\zeta_{5,k}=b_{J_k}+\zeta_{5,k}^{\circ}$.  Its mean contributes
$O(\sum_{k\leq N}J_k^{-q_\star})=o(\sqrt N)$.  The centered part is a
martingale difference with conditional second moment at most
$Cd_k^2J_k^{3-2s}$, so  \autoref{lem:common_space_mds} applies.  Finally, define $
 U_{G_6,k}=\mathbb M_{J_k}^{-1}B_k\zeta_{6,k}.$ $U_{G_6,k}$
is a martingale difference with
$\mathbb E_{k-1}\Vert U_{G_6,k}\Vert_2^2\leq Cd_k^2J_k$.  The first primitive
condition in \autoref{remark_gain_rate} is exactly the normalized
quadratic-variation condition for this term.  Summing the seven pieces in
\eqref{audited_gain_score_decomposition} proves both conclusions.  We also point that the
predictability of $\widehat G_k$ is used at every martingale step.
\end{proof}

Apply \autoref{lem:estimated_gain_remainder} with the no-jump indicator in
the block normal equations.  Together with \eqref{gain_normal_equation}, the
identity-gain remainder calculations, and the transient average, it yields the
same joint PR representation.  This finishes the entire of Lemma 5.2 in the main text. 

\subsection{Proof of \autoref{corollary1}}
Let
\[
 \psi_{\theta,k}=\frac1B\sum_{i=1}^B
 \varepsilon_{i,k}\mathbb M_\theta^{-1}F_0'(Z_{0,i,k})
 \{X_{i,k}-\mu_0(\theta_0,Z_{0,i,k})\}.
\]
These are independent, mean-zero $p$-vectors with covariance $\Omega_\theta$
and the moment required by the classical multivariate LIL.  Its cluster-set
form gives that the limit points of
\[
 (2N\log\log N)^{-1/2}\Omega_\theta^{-1/2}
 \sum_{k=1}^N\psi_{\theta,k}
\]
form the closed unit ball.  The almost-sure
$o(N^{-1/2})$ representation in \autoref{theorem_NLS_update} is negligible on
the LIL scale, proving the first part of the theorem including the stated directional equivalent.

To prove the second part, use the second expansion in
\autoref{theorem_NLS_update} and retain its deterministic approximation part
exactly as $
 \frac1N\sum_{k=1}^N
 \{\mathbb P_{J_k}F_0(z)-F_0(z)\}.$ 
The basis envelope $\sup_{z\in\mathcal Z_0}\Vert\Psi_{J_N}(z)\Vert_2\leq CJ_N^{1/2}$
and the common-space martingale maximal bound give the stochastic term
$O_{\rm a.s.}(J_N^{1/2}N^{-1/2}(\log N)^{1/2})$.  Uniform sieve
approximation gives
\[
 \sup_{z\in\mathcal Z_0}\left|\frac1N\sum_{k=1}^N
 \{\mathbb P_{J_k}F_0(z)-F_0(z)\}\right|
 \leq \frac CN\sum_{k=1}^N J_k^{-s}.
\]
This proves the second part of the theorem.  

\subsection{Proof of \autoref{corollary2}}
The batch influences $\psi_{\theta,k}$ defined in the preceding proof are
iid, mean zero, have covariance $\Omega_\theta$, and have a finite
$2+\delta$ moment under \autoref{cond11}.  The multivariate invariance
principle therefore gives
$N^{-1/2}\sum_{k\leq[Nr]}\psi_{\theta,k}\Rightarrow
\Omega_\theta^{1/2}\mathbb W_p(r)$ for $r\in[0,1]$.  It remains only to prove the
uniform equivalence
    \[
    \max_{0\leq r\leq 1} \left\Vert  \frac{1}{\sqrt N}\sum_{k=1}^{[Nr]} \frac{1}{B}\sum_{i=1}^B \varepsilon_{i,k} \mathbb{M}_{\theta}^{-1} \nabla_z F_0(Z_{0,i,k})\left(X_{i,k} - \mu_0(\theta_0, Z_{0,i,k})\right) -\frac{[Nr]}{\sqrt{N}}\Delta\doubleoverline{\theta}_{[Nr]} \right\Vert_{2} \rightarrow_{a.s.} 0
    \]

For simplicity, define $e_N =  \Delta\doubleoverline{\theta}_N  -\frac{1}{N}\sum_{k=1}^N \frac{1}{B}\sum_{i=1}^B \varepsilon_{i,k} \mathbb{M}_{\theta}^{-1} \nabla_z F_0(Z_{0,i,k})\left(X_{i,k} - \mu_0(\theta_0, Z_{0,i,k})\right)$, then it remains to show that 
\[
\max_{0\leq r\leq 1} \left\Vert\frac{[Nr]}{\sqrt{N}}  e_{[Nr]} \right\Vert_{2} \rightarrow_{a.s.} 0.
\]
According to \autoref{theorem_NLS_update}, we have that $\Vert e_N\Vert_2 =  o\left(N^{-\frac{1}{2}}\right)$. So for any $\epsilon>0$, we can find  path-dependent  $N_1$ such that for any $N>N_1$, there holds $\Vert e_N\Vert_2\leq \frac{\epsilon}{\sqrt{N}}$. Also, for such $\epsilon$, we can also find path-dependent $N_2$ such that 
$\max_{1\leq k\leq N_1} \Vert \frac{ke_{k}}{\sqrt{N}}\Vert_2\leq \epsilon$. So together we have that for $N\geq N_2$, we have 
\[
\max_{0\leq r\leq 1} \left\Vert\frac{[Nr]}{\sqrt{N}}  e_{[Nr]} \right\Vert_{2} \leq \max\left\{\max_{1\leq k\leq N_1}  \left\Vert\frac{k e_{k}}{\sqrt{N}}  \right\Vert_{2}, \max_{N_1<k\leq N}  \left\Vert\frac{ke_{k}}{\sqrt{N}}   \right\Vert_{2}\right\}\leq \max\left\{\epsilon, \max_{N_1<k\leq N}\frac{\epsilon\sqrt{k}}{\sqrt{N}}\right\}\leq \epsilon.
\]
Since $\epsilon$ can be arbitrarily small, this shows the result.

\section{Average Marginal Effects}\label{appendixE}

\subsection{Additional Conditions for \autoref{FCLT_average_marginal_effect}} 
\begin{condition}\label{cond_ame} 
(i) $\mathbb{E}\Vert \nabla_z \mu_0(\theta_0, Z_0)\Vert^2_2<\infty$; 
(ii) Let $f_Z(z)$ denote the density of $Z_0$, which is continuously differentiable. There holds 
$\lim_{z\rightarrow\pm\infty}\mu_0(\theta_0,z) f_Z(z)=0$; define 
$\ell(z)=\nabla_z f_Z(z)/f_Z(z)$. Then $\mathbb{E}\ell^2(Z_0)<\infty$, and for all $J$ there holds
\[
\mathbb{E}\left|\ell(Z_0)-\mathbb{P}_J(\ell)(Z_0)\right|^2
\leq C_fJ^{-2s_f}
\]
for some $s_f>0$. Moreover, $(s+s_f)\alpha_\dagger>\frac12$; 
(iii) Define $\Upsilon_k=B^{-1}\sum_{i=1}^B\Upsilon_{i,k}$. There holds
\[
\mathbb{E}\Upsilon_k\Upsilon_k^\top\longrightarrow V_\Upsilon
\]
as $k\rightarrow\infty$ for some positive definite matrix $V_\Upsilon$.
\end{condition}

\subsection{Proof of \autoref{FCLT_average_marginal_effect}}
\label{proof_marginal_effect}

Note that
\[
\widetilde\tau_k
=
\frac1B\sum_{i=1}^B
\left(
\nabla_z\Psi_{J_{k-1}}
\left(Z_{i,k}(\widetilde\theta_{k-1})\right)
\right)^\top
\widetilde\beta_{k-1}\widetilde\theta_{k-1}
\]
can be decomposed as follows:
\begin{align*}
\widetilde\tau_k
&=
\frac1B\sum_{i=1}^B
\left(
\nabla_z\Psi_{J_{k-1}}
\left(Z_{i,k}(\widetilde\theta_{k-1})\right)
\right)^\top
\widetilde\beta_{k-1}\widetilde\theta_{k-1}\\
&=
\underset{(i)}{\underbrace{
\frac1B\sum_{i=1}^B
\left(
\nabla_z\Psi_{J_{k-1}}(Z_{0,i,k})
\right)^\top
\widetilde\beta_{k-1}\widetilde\theta_{k-1}
}}\\
&\quad+
\underset{(ii)}{\underbrace{
\frac1B\sum_{i=1}^B
\widetilde\beta_{k-1}^\top
\left(
\nabla_{zz}\Psi_{J_{k-1}}(Z_{0,i,k})
\right)
X_{i,k}^\top
\Delta\widetilde\theta_{k-1}
\widetilde\theta_{k-1}
}}\\
&\quad+
\underset{(iii)}{\underbrace{
\frac1B\sum_{i=1}^B
\int_0^1(1-\tau)\,
\widetilde\beta_{k-1}^\top
\nabla_{zzz}\Psi_{J_{k-1}}
\left(
Z_{0,i,k}
+\tau X_{i,k}^\top\Delta\widetilde\theta_{k-1}
\right)
\left(
X_{i,k}^\top\Delta\widetilde\theta_{k-1}
\right)^2
\widetilde\theta_{k-1}\,d\tau
}}.
\end{align*}

We first look at (iii). By \autoref{condition6}, \autoref{cond11}, the decomposition
$\Delta\widetilde\omega_{k-1}=D_{k-1}+T_{k-1}$, and the proof of
\autoref{lemma5.2part2},
\[
\Vert(iii)\Vert_2
\leq
\frac{C}{B}\sum_{i=1}^B\Vert X_{i,k}\Vert_2^2
\left\{
\Vert D_{k-1}\Vert_2^2
+J_{k-1}^{\frac72}\Vert D_{k-1}\Vert_2^3
+\Vert T_{k-1}\Vert_2^2
+J_{k-1}^{\frac72}\Vert T_{k-1}\Vert_2^3
\right\},
\qquad\text{a.s.}
\]
Since $s>\frac72$, $J_{k-1}^{7/2}\Vert T_{k-1}\Vert_2\leq C$ eventually. Therefore, by
\eqref{refined_D_rate}, \autoref{lemma_transient}(ii)--(iii), and the same weighted strong-law argument used above,
\begin{align*}
\left\Vert
\frac1N\sum_{k=1}^N(iii)
\right\Vert_2
&=
O\left(
N^{\left(
\frac{7\alpha_\dagger}{2}
-\frac{\alpha_\xi-\alpha_\dagger}{2}
\right)_+
-\alpha_\xi+\alpha_\dagger}
\text{PolyLog}
+
N^{-(4s-4)\alpha_\dagger}\text{PolyLog}
\right.\\
&\hspace{22mm}\left.
+
N^{(\frac{19}{2}-6s)\alpha_\dagger}\text{PolyLog}
+
N^{-1}\text{PolyLog}
\right)\\
&\quad+
O\left(
\frac1N
\left\{
N^{\alpha_\xi+(1-2s)\alpha_\dagger}
+\log N
\right\}
\right)\\
&=
o\left(N^{-\frac12}\right),
\qquad\text{a.s.}
\end{align*}
The leading term is $o(N^{-1/2})$ under \autoref{new_rate}. When the positive part is active, it is enough that $
3\alpha_\xi>1+10\alpha_\dagger,
$ 
which follows from
$\alpha_\xi>\frac12+3\alpha_\dagger$ and $\alpha_\dagger<\frac12$; when it is inactive,
$\alpha_\xi>\frac12+\alpha_\dagger$ suffices. The remaining terms are smaller under the same conditions.
So we will only look at (i) and (ii).

We first look at the second term. For (ii), we can see that
\begin{align*}
(ii)
&=
\frac1B\sum_{i=1}^B
\beta_{J_{k-1},0}^\top
\left(
\nabla_{zz}\Psi_{J_{k-1}}(Z_{0,i,k})
\right)
\widetilde\theta_{k-1}
X_{i,k}^\top\Delta\widetilde\theta_{k-1}
\qquad (iv)\\
&\quad+
\frac1B\sum_{i=1}^B
\Delta\widetilde\beta_{k-1}^\top
\left(
\nabla_{zz}\Psi_{J_{k-1}}(Z_{0,i,k})
\right)
\widetilde\theta_{k-1}
X_{i,k}^\top\Delta\widetilde\theta_{k-1}
\qquad (v).
\end{align*}

We first look at (iv). The exact decomposition is
\begin{align*}
&(iv)
-\frac1B\sum_{i=1}^B
\nabla_{zz}F_0(Z_{0,i,k})
\theta_0X_{i,k}^\top
\Delta\widetilde\theta_{k-1}\\
&=
\frac1B\sum_{i=1}^B
\nabla_{zz}
\left\{
\mathbb P_{J_{k-1}}(F_0)-F_0
\right\}(Z_{0,i,k})
\theta_0X_{i,k}^\top
\Delta\widetilde\theta_{k-1}\\
&\quad+
\frac1B\sum_{i=1}^B
\nabla_{zz}\mathbb P_{J_{k-1}}(F_0)(Z_{0,i,k})
\Delta\widetilde\theta_{k-1}
X_{i,k}^\top
\Delta\widetilde\theta_{k-1}.
\end{align*}

For the first term, integration by parts gives
\begin{align*}
&\left\Vert
\mathbb E\left[
\nabla_{zz}
\left\{
\mathbb P_J(F_0)-F_0
\right\}(Z_0)X^\top
\right]
\right\Vert_2\\
&=
\left\Vert
-\int
\nabla_z
\left\{
\mathbb P_J(F_0)-F_0
\right\}(z)
\left\{
\nabla_z\mu_0(\theta_0,z)^\top f_Z(z)
+
\mu_0(\theta_0,z)^\top\nabla_zf_Z(z)
\right\}dz
\right\Vert_2\\
&\leq
CJ^{1-s}
\left[
\left\{
\mathbb E
\Vert\nabla_z\mu_0(\theta_0,Z_0)\Vert_2^2
\right\}^{1/2}
+
\left\{
\mathbb E
\Vert\mu_0(\theta_0,Z_0)\Vert_2^2
\right\}^{1/2}
\left\{
\mathbb E\ell^2(Z_0)
\right\}^{1/2}
\right]\leq CJ^{1-s}.
\end{align*}
Hence, using
$\Delta\widetilde\omega_{k-1}=D_{k-1}+T_{k-1}$,
\begin{align*}
\frac1N\sum_{k=1}^N
J_{k-1}^{1-s}
\Vert\Delta\widetilde\theta_{k-1}\Vert_2&\leq
CN^{(1-s)\alpha_\dagger-\frac{\alpha_\xi-\alpha_\dagger}{2}}
\text{PolyLog}
+
CN^{(3-3s)\alpha_\dagger}\text{PolyLog}
+
CN^{-1}\text{PolyLog}\\
&\quad+
\frac{C}{N}
\left\{
N^{\alpha_\xi+(\frac12-s)\alpha_\dagger}
+\log N
\right\}=
o(N^{-\frac12}),
\qquad\text{a.s.}
\end{align*}
After subtracting its conditional mean, the same martingale-difference argument gives an
$o_{\rm a.s.}(N^{-1/2})$ running average. Indeed, using
$\Delta\widetilde\omega_{k-1}=D_{k-1}+T_{k-1}$,
\eqref{refined_D_rate}, and \autoref{lemma_transient}(ii), its conditional second moments divided by $k$ are summable.

For the second term in the exact decomposition of (iv), boundedness of
$\nabla_{zz}\mathbb P_{J_{k-1}}(F_0)$ gives
\begin{align*}
&\frac1N\sum_{k=1}^N
\left\Vert
\frac1B\sum_{i=1}^B
\nabla_{zz}\mathbb P_{J_{k-1}}(F_0)(Z_{0,i,k})
\Delta\widetilde\theta_{k-1}
X_{i,k}^\top
\Delta\widetilde\theta_{k-1}
\right\Vert_2\\
&\leq
CN^{-\alpha_\xi+\alpha_\dagger}\text{PolyLog}
+
CN^{-(4s-4)\alpha_\dagger}\text{PolyLog}
+
CN^{-1}\text{PolyLog}\\
&\quad+
\frac{C}{N}
\left\{
N^{\alpha_\xi+(1-2s)\alpha_\dagger}
+\log N
\right\}=o(N^{-\frac12}),
\qquad\text{a.s.}
\end{align*}
Consequently,
\[
\frac1N\sum_{k=1}^N(iv)
=
\frac1N\sum_{k=1}^N
\frac1B\sum_{i=1}^B
\nabla_{zz}F_0(Z_{0,i,k})
\theta_0X_{i,k}^\top
\Delta\widetilde\theta_{k-1}
+
o(N^{-\frac12}),
\qquad\text{a.s.}
\]

Now we look at
$\left\Vert N^{-1}\sum_{k=1}^N(v)\right\Vert_2$.
Similar to the previous strategy, we decompose each summand into its conditional mean and martingale-difference part. Note that
\begin{align*}
\mathbb E
\left[
\nabla_{zz}\Psi_J(Z_0)X^\top
\right]
&=
\int
\nabla_{zz}\Psi_J(z)
\mu_0(\theta_0,z)^\top
f_Z(z)\,dz\\
&=
-\mathbb E\left[
\nabla_z\Psi_J(Z_0)
\left\{
\mu_0(\theta_0,Z_0)^\top\ell(Z_0)
+
\nabla_z\mu_0(\theta_0,Z_0)^\top
\right\}
\right].
\end{align*}
Therefore, by \autoref{cond_ame}(i)--(ii),
\[
\left\Vert
\mathbb E
\left[
\nabla_{zz}\Psi_J(Z_0)X^\top
\right]
\right\Vert_F
\leq CJ^{\frac32}.
\]
It follows that
\begin{align*}
\left\Vert
\frac1N\sum_{k=1}^N
\mathbb E_{k-1}(v)
\right\Vert_2
&\leq
\frac{C}{N}\sum_{k=1}^N
J_k^{\frac32}
\left(
\Vert D_{k-1}\Vert_2^2
+
\Vert T_{k-1}\Vert_2^2
\right)\\
&\leq
CN^{\frac{5\alpha_\dagger}{2}-\alpha_\xi}
\text{PolyLog}
+
CN^{(\frac{11}{2}-4s)\alpha_\dagger}
\text{PolyLog}
+
CN^{-1}\text{PolyLog}\\
&\quad+
\frac{C}{N}
\left\{
N^{\alpha_\xi+(\frac52-2s)\alpha_\dagger}
+\log N
\right\} =
o(N^{-\frac12}),
\qquad\text{a.s.}
\end{align*}
On the other side,
\begin{align*}
\sum_{k=2}^{\infty}
\frac{
\mathbb E_{k-1}
\Vert(v)-\mathbb E_{k-1}(v)\Vert_2^2
}{k}
&\leq
C\sum_{k=2}^{\infty}
\left[
k^{-1+7\alpha_\dagger-2\alpha_\xi}
\text{PolyLog}
+
k^{-1+(13-8s)\alpha_\dagger}
\text{PolyLog}
\right.\\
&\hspace{35mm}\left.
+
k^{-1+(5-4s)\alpha_\dagger}
\right]
<\infty,
\qquad\text{a.s.}
\end{align*}
so
\[
\left\Vert
\frac1N\sum_{k=1}^N
\left\{
(v)-\mathbb E_{k-1}(v)
\right\}
\right\Vert_2
=
o(N^{-\frac12}),
\qquad\text{a.s.}
\]
Finally, again using
$\Delta\widetilde\omega_{k-1}=D_{k-1}+T_{k-1}$,
\[
\sum_{k=2}^{\infty}
\frac1k
\mathbb E_{k-1}
\left\Vert
\left\{
\frac1B\sum_{i=1}^B
\nabla_{zz}F_0(Z_{0,i,k})
\theta_0X_{i,k}^\top
-
\mathbb E
\left[
\nabla_{zz}F_0(Z_0)\theta_0X^\top
\right]
\right\}
\Delta\widetilde\theta_{k-1}
\right\Vert_2^2
<\infty,
\qquad\text{a.s.}
\]
and hence its running average is
$o_{\rm a.s.}(N^{-1/2})$. This shows that
\[
\frac1N\sum_{k=1}^N(ii)
=
\mathbb E
\left[
\nabla_{zz}F_0(Z_0)\theta_0X^\top
\right]
\frac1N\sum_{k=1}^N
\Delta\widetilde\theta_{k-1}
+
o(N^{-\frac12}),
\qquad\text{a.s.}
\]

We finally look at (i). Expansion leads to
\begin{align*}
&(i)-\mathbb E[\nabla_zF_0(Z_0)]\theta_0\\
&=
\mathbb E
\left[
\nabla_z\Psi_{J_{k-1}}(Z_0)^\top
\right]
\beta_{J_{k-1},0}\theta_0
-
\mathbb E[\nabla_zF_0(Z_0)]\theta_0
\qquad (i.1)\\
&\quad+
\left(
\frac1B\sum_{i=1}^B
\nabla_z\Psi_{J_{k-1}}(Z_{0,i,k})^\top
-
\mathbb E
\left[
\nabla_z\Psi_{J_{k-1}}(Z_0)^\top
\right]
\right)
\beta_{J_{k-1},0}\theta_0
\qquad (i.2)\\
&\quad+
\mathbb E
\left[
\nabla_z\Psi_{J_{k-1}}(Z_0)^\top
\right]
\Delta\widetilde\beta_{k-1}\theta_0
\qquad (i.3)\\
&\quad+
\mathbb E
\left[
\nabla_z\Psi_{J_{k-1}}(Z_0)^\top
\right]
\beta_{J_{k-1},0}
\Delta\widetilde\theta_{k-1}
\qquad (i.4)\\
&\quad+
\mathbb E
\left[
\nabla_z\Psi_{J_{k-1}}(Z_0)^\top
\right]
\Delta\widetilde\beta_{k-1}
\Delta\widetilde\theta_{k-1}
\qquad (i.5)\\
&\quad+
\left(
\frac1B\sum_{i=1}^B
\nabla_z\Psi_{J_{k-1}}(Z_{0,i,k})^\top
-
\mathbb E
\left[
\nabla_z\Psi_{J_{k-1}}(Z_0)^\top
\right]
\right)
\Delta\widetilde\beta_{k-1}\theta_0
\qquad (i.6)\\
&\quad+
\left(
\frac1B\sum_{i=1}^B
\nabla_z\Psi_{J_{k-1}}(Z_{0,i,k})^\top
-
\mathbb E
\left[
\nabla_z\Psi_{J_{k-1}}(Z_0)^\top
\right]
\right)
\beta_{J_{k-1},0}
\Delta\widetilde\theta_{k-1}
\qquad (i.7)\\
&\quad+
\left(
\frac1B\sum_{i=1}^B
\nabla_z\Psi_{J_{k-1}}(Z_{0,i,k})^\top
-
\mathbb E
\left[
\nabla_z\Psi_{J_{k-1}}(Z_0)^\top
\right]
\right)
\Delta\widetilde\beta_{k-1}
\Delta\widetilde\theta_{k-1}
\qquad (i.8).
\end{align*}

We look at (i.1)--(i.8) separately. For (i.1), by
\autoref{cond_ame}(ii) and integration by parts,
\begin{align*}
\mathbb E
\left[
\nabla_z
\left\{
\mathbb P_J(F_0)(Z_0)-F_0(Z_0)
\right\}
\right]
&=
\int
\nabla_z
\left\{
\mathbb P_J(F_0)(z)-F_0(z)
\right\}
f_Z(z)\,dz\\
&=
\left.
\left\{
\mathbb P_J(F_0)(z)-F_0(z)
\right\}f_Z(z)
\right|_{z=-\infty}^{z=+\infty}\\
&\quad-
\int
\left\{
\mathbb P_J(F_0)(z)-F_0(z)
\right\}
\nabla_zf_Z(z)\,dz\\
&=
-\mathbb E
\left[
\left\{
\mathbb P_J(F_0)(Z_0)-F_0(Z_0)
\right\}
\ell(Z_0)
\right]\\
&=
-\mathbb E
\left[
\left\{
\mathbb P_J(F_0)(Z_0)-F_0(Z_0)
\right\}
\left\{
\ell(Z_0)-\mathbb P_J(\ell)(Z_0)
\right\}
\right],
\end{align*}
because $\mathbb P_J(\ell)\in\mathcal S_J$ and the projection residual is orthogonal to
$\mathcal S_J$. Thus, under \autoref{cond_ame}(ii),
\[
\Vert(i.1)\Vert_2
\leq
C
\Vert\mathbb P_J(F_0)-F_0\Vert_\infty
\left\{
\mathbb E
\left[
\ell(Z_0)-\mathbb P_J(\ell)(Z_0)
\right]^2
\right\}^{1/2}
\leq
CJ_{k-1}^{-s-s_f}.
\]
When $(s+s_f)\alpha_\dagger>1/2$,
\[
\left\Vert
\frac1N\sum_{k=1}^N(i.1)
\right\Vert_2
=
o(N^{-1/2}).
\]

For (i.2), we note that
\begin{align*}
&(i.2)
-
\frac1B\sum_{i=1}^B
\left\{
\nabla_zF_0(Z_{0,i,k})
-
\mathbb E[\nabla_zF_0(Z_0)]
\right\}\theta_0\\
&=
\frac1B\sum_{i=1}^B
\left[
\nabla_z
\left\{
\mathbb P_{J_{k-1}}(F_0)(Z_{0,i,k})
-
F_0(Z_{0,i,k})
\right\}
\right.\\
&\hspace{36mm}\left.
-
\mathbb E
\nabla_z
\left\{
\mathbb P_{J_{k-1}}(F_0)(Z_0)
-
F_0(Z_0)
\right\}
\right]\theta_0.
\end{align*}
Using the proof strategy for sums of martingale-difference sequences, when
$s>1$, as guaranteed by \autoref{new_rate},
\begin{align*}
&\left\Vert
\frac1N\sum_{k=1}^N
\frac1B\sum_{i=1}^B
\left[
\nabla_z
\left\{
\mathbb P_{J_{k-1}}(F_0)(Z_{0,i,k})
-
F_0(Z_{0,i,k})
\right\}
\right.\right.\\
&\hspace{32mm}\left.\left.
-
\mathbb E
\nabla_z
\left\{
\mathbb P_{J_{k-1}}(F_0)(Z_0)
-
F_0(Z_0)
\right\}
\right]\theta_0
\right\Vert_2 =
o(N^{-1/2}),
\qquad\text{a.s.}
\end{align*}
As a result,
\[
\frac1N\sum_{k=1}^N(i.2)
=
\frac1N\sum_{k=1}^N
\frac1B\sum_{i=1}^B
\left\{
\nabla_zF_0(Z_{0,i,k})
-
\mathbb E[\nabla_zF_0(Z_0)]
\right\}\theta_0
+
o(N^{-1/2}),
\qquad\text{a.s.}
\]

Next, for (i.3), note that
\begin{align*}
&\frac1N\sum_{k=1}^N
\mathbb E
\left[
\nabla_z\Psi_{J_{k-1}}(Z_0)^\top
\right]
\Delta\widetilde\beta_{k-1}\theta_0\\
&=
\frac{N-1}{N}
\mathbb E
\left[
\nabla_z\Psi_{J_{N-1}}(Z_0)^\top
\right]
\Delta\doubleoverline\beta_{N-1}\theta_0  +
\frac1N\sum_{k=1}^N
\mathbb E
\left[
\nabla_z
\left\{
\mathbb P_{J_{N-1}}(F_0)(Z_0)
-
\mathbb P_{J_{k-1}}(F_0)(Z_0)
\right\}
\right]\theta_0.
\end{align*}
So we only need to look at
\[
\mathbb E
\left[
\nabla_z\Psi_{J_N}(Z_0)^\top
\right]
\Delta\doubleoverline\beta_N.
\]
According to \autoref{theorem_NLS_update},
\begin{align*}
\Delta\doubleoverline\beta_N
&=
\frac1N\sum_{k=1}^N
R_{J_k,J_N}^\top
\mathbb M_{\beta,J_k}^{-1}
\frac1B\sum_{i=1}^B
\varepsilon_{i,k}
\left\{
\Psi_{J_k}(Z_{0,i,k})
-
P_{J_k}\nabla_zF_0(Z_{0,i,k})X_{i,k}
\right\}\\
&\quad+
\frac1N\sum_{k=1}^N
\left(
R_{J_k,J_N}^\top\beta_{J_k,0}
-
\beta_{J_N,0}
\right)
+
o(N^{-1/2}),
\qquad\text{a.s.}
\end{align*}
We have that
\begin{align*}
&\mathbb E
\left[
\nabla_z\Psi_{J_N}(Z_0)^\top
\right]
\Delta\doubleoverline\beta_N\\
&=
\frac1N\sum_{k=1}^N
\frac1B\sum_{i=1}^B
\varepsilon_{i,k}
\mathbb E
\left[
\nabla_z\Psi_{J_k}(Z_0)^\top
\right]
\mathbb M_{\beta,J_k}^{-1}
\left\{
\Psi_{J_k}(Z_{0,i,k})
-
P_{J_k}\nabla_zF_0(Z_{0,i,k})X_{i,k}
\right\}\\
&\quad+
\frac1N\sum_{k=1}^N
\mathbb E
\left[
\nabla_z
\left\{
\mathbb P_{J_k}(F_0)(Z_0)-F_0(Z_0)
\right\}
\right] -
\mathbb E
\left[
\nabla_z
\left\{
\mathbb P_{J_N}(F_0)(Z_0)-F_0(Z_0)
\right\}
\right]
+
o(N^{-1/2}),
\qquad\text{a.s.}
\end{align*}
where the last $o(N^{-1/2})$ follows from
\[
\sup_J
\left\Vert
\mathbb E
\left[
\nabla_z\Psi_J(Z_0)
\right]
\right\Vert_2
<\infty.
\]
Indeed, integration by parts and Cauchy--Schwarz give
\[
\left\Vert
\mathbb E
\left[
\nabla_z\Psi_J(Z_0)
\right]
\right\Vert_2
=
\left\Vert
\mathbb E
\left[
\Psi_J(Z_0)\ell(Z_0)
\right]
\right\Vert_2
\leq
\overline\lambda(\Gamma_J)^{1/2}
\left\{
\mathbb E|\ell(Z_0)|^2
\right\}^{1/2}
\leq C
\]
under \autoref{condition7} and \autoref{cond_ame}(ii). Moreover, according to the proof of (i.1), the second and third lines above are bounded by
\[
\frac{C}{N}\sum_{k=1}^N J_k^{-s-s_f}
+
CJ_N^{-s-s_f}
=
O\left(
N^{-(s+s_f)\alpha_\dagger}
\right)
=
o(N^{-1/2}).
\]
Therefore,
\begin{align*}
&\mathbb E
\left[
\nabla_z\Psi_{J_N}(Z_0)^\top
\right]
\Delta\doubleoverline\beta_N\\
&=
\frac1N\sum_{k=1}^N
\frac1B\sum_{i=1}^B
\varepsilon_{i,k}
\mathbb E
\left[
\nabla_z\Psi_{J_k}(Z_0)^\top
\right]
\mathbb M_{\beta,J_k}^{-1}
\left\{
\Psi_{J_k}(Z_{0,i,k})
-
P_{J_k}\nabla_zF_0(Z_{0,i,k})X_{i,k}
\right\}\\
&\quad+
o(N^{-1/2}),
\qquad\text{a.s.}
\end{align*}

For (i.4), we have that
\[
\left\Vert
(i.4)
-
\mathbb E[\nabla_zF_0(Z_0)]
\Delta\widetilde\theta_{k-1}
\right\Vert_2
\leq
CJ_{k-1}^{1-s}
\left(
\Vert D_{k-1}\Vert_2
+
\Vert T_{k-1}\Vert_2
\right),
\qquad\text{a.s.}
\]
Consequently,
\begin{align*}
&\frac1N\sum_{k=1}^N
\left\Vert
(i.4)
-
\mathbb E[\nabla_zF_0(Z_0)]
\Delta\widetilde\theta_{k-1}
\right\Vert_2\\
&\leq
CN^{-(s-1)\alpha_\dagger-\frac{\alpha_\xi-\alpha_\dagger}{2}}
\text{PolyLog}
+
CN^{(3-3s)\alpha_\dagger}\text{PolyLog}
+
CN^{-1}\text{PolyLog}\\
&\quad+
\frac{C}{N}
\left\{
N^{\alpha_\xi+(\frac12-s)\alpha_\dagger}
+
\log N
\right\} =
o(N^{-1/2}),
\qquad\text{a.s.}
\end{align*}
Using the PR representation of
$\Delta\doubleoverline\theta_N$, we obtain
\begin{align*}
\frac1N\sum_{k=1}^N(i.4)
&=
\frac1N\sum_{k=1}^N
\frac1B\sum_{i=1}^B
\varepsilon_{i,k}
\mathbb E[\nabla_zF_0(Z_0)]
\mathbb M_\theta^{-1}
\nabla_zF_0(Z_{0,i,k})\\
&\qquad\qquad\times
\left\{
X_{i,k}
-
\mu_0(\theta_0,Z_{0,i,k})
\right\}
+
o(N^{-1/2}),
\qquad\text{a.s.}
\end{align*}

For (i.5),
\begin{align*}
\frac1N\sum_{k=1}^N
\Vert(i.5)\Vert_2
&\leq
\frac{C}{N}\sum_{k=1}^N
\left(
\Vert D_{k-1}\Vert_2^2
+
\Vert T_{k-1}\Vert_2^2
\right)\\
&\leq
CN^{-\alpha_\xi+\alpha_\dagger}\text{PolyLog}
+
CN^{-(4s-4)\alpha_\dagger}\text{PolyLog}
+
CN^{-1}\text{PolyLog}\\
&\quad+
\frac{C}{N}
\left\{
N^{\alpha_\xi+(1-2s)\alpha_\dagger}
+
\log N
\right\} =
o(N^{-1/2}),
\qquad\text{a.s.}
\end{align*}

For (i.6)--(i.8), note that
\[
\mathbb E
\left\Vert
\nabla_z\Psi_J(Z_0)
-
\mathbb E\nabla_z\Psi_J(Z_0)
\right\Vert_2^2
\leq CJ^3.
\]
Using
$\Delta\widetilde\omega_{k-1}=D_{k-1}+T_{k-1}$,
\eqref{refined_D_rate}, and
\autoref{lemma_transient}(ii), for each of these terms we have
\[
\sum_{k=2}^{\infty}
\frac1k
\mathbb E_{k-1}
\Vert(i.l)\Vert_2^2
\leq
C\sum_{k=2}^{\infty}
\left[
k^{-1-\alpha_\xi+4\alpha_\dagger}
\text{PolyLog}
+
k^{-1+(7-4s)\alpha_\dagger}
\text{PolyLog}
+
k^{-1+(3-2s)\alpha_\dagger}
\right]
<\infty,
\ l=6,7,8.
\]
Thus their running averages are all
$o_{\rm a.s.}(N^{-1/2})$.

Combining the above analysis, we have that
\begin{align*}
\Delta\doubleoverline\tau_N
&=
\frac1N\sum_{k=1}^N
\frac1B\sum_{i=1}^B
\varepsilon_{i,k}
\mathbb E
\left[
\nabla_{zz}F_0(Z_0)\theta_0X^\top
\right]
\mathbb M_\theta^{-1}
\nabla_zF_0(Z_{0,i,k})
\left\{
X_{i,k}
-
\mu_0(\theta_0,Z_{0,i,k})
\right\}\\
&\quad+
\frac1N\sum_{k=1}^N
\frac1B\sum_{i=1}^B
\left\{
\nabla_zF_0(Z_{0,i,k})
-
\mathbb E[\nabla_zF_0(Z_0)]
\right\}\theta_0\\
&\quad+
\frac1N\sum_{k=1}^N
\frac1B\sum_{i=1}^B
\varepsilon_{i,k}\theta_0
\mathbb E
\left[
\nabla_z\Psi_{J_k}(Z_0)^\top
\right]
\mathbb M_{\beta,J_k}^{-1}
\left\{
\Psi_{J_k}(Z_{0,i,k})
-
P_{J_k}\nabla_zF_0(Z_{0,i,k})X_{i,k}
\right\}\\
&\quad+
\frac1N\sum_{k=1}^N
\frac1B\sum_{i=1}^B
\varepsilon_{i,k}
\mathbb E[\nabla_zF_0(Z_0)]
\mathbb M_\theta^{-1}
\nabla_zF_0(Z_{0,i,k})
\left\{
X_{i,k}
-
\mu_0(\theta_0,Z_{0,i,k})
\right\}\\
&\quad+
o(N^{-1/2}),
\qquad\text{a.s.}
\end{align*}
Therefore,
\[
\Delta\doubleoverline\tau_N
=
\frac1N\sum_{k=1}^N\Upsilon_k
+
o(N^{-1/2}),
\qquad\text{a.s.}
\]

Under \autoref{cond_ame}(iii), Ces\`aro summation gives
\[
\frac1N\sum_{k=1}^{[Nr]}
\mathbb E
\left[
\Upsilon_k\Upsilon_k^\top
\right]
\longrightarrow
rV_\Upsilon,
\qquad r\in[0,1].
\]
We verify the Lindeberg condition via Lyapunov's criterion. The first, second, and fourth components of
$\Upsilon_{i,k}$ have finite $(2+\delta)$-th moments for
$\delta\in(0,\kappa-2)$ by \autoref{condition6} and \autoref{cond11}. For the third component, since
\[
\left\Vert
\mathbb E
\left[
\nabla_z\Psi_J(Z_0)
\right]
\right\Vert_2
\leq C,
\]
the eigenvalues of $\mathbb M_{\beta,J}$ are uniformly bounded away from zero,
$\Vert\Psi_J\Vert_\infty\leq CJ^{1/2}$, and
$\Vert P_J\Vert_F\leq CJ^{1/2}$, its norm is bounded by
\[
C|\varepsilon_{i,k}|J_k^{1/2}
\left(
1+\Vert X_{i,k}\Vert_2
\right).
\]
Hence
\[
\mathbb E
\Vert\Upsilon_k\Vert_2^{2+\delta}
\leq
CJ_k^{1+\delta/2},
\]
and
\[
N^{-1-\frac\delta2}
\sum_{k=1}^N
\mathbb E
\Vert\Upsilon_k\Vert_2^{2+\delta}
\leq
CN^{-\frac\delta2+
\left(1+\frac\delta2\right)\alpha_\dagger}
\longrightarrow0
\]
whenever
\[
\delta>
\frac{2\alpha_\dagger}{1-\alpha_\dagger}.
\]
Such
\[
\delta\in
\left(
\frac{2\alpha_\dagger}{1-\alpha_\dagger},
\kappa-2
\right)
\]
exists under \autoref{cond11}. The multivariate Lyapunov FCLT therefore gives
\[
\frac1{\sqrt N}
\sum_{k=1}^{[Nr]}
\Upsilon_k
\Longrightarrow
V_\Upsilon^{1/2}
\mathbb W_p(r),
\qquad r\in[0,1].
\]
Combining this with the almost-sure
$o(N^{-1/2})$ remainder and the argument in the proof of
\autoref{corollary2}, we obtain
\[
(NV_\Upsilon)^{-1/2}
[Nr]\,
\Delta\doubleoverline\tau_{[Nr]}
\Longrightarrow
\mathbb W_p(r),
\qquad r\in[0,1].
\]

\subsection{Average Marginal Effects of Binary Regressors}\label{ame_binary}

For the binary-regressor marginal effect defined in the main text, an additional counterfactual overlap condition is necessary. Without introducing further notation, it is enough to assume directly that, for the corresponding $j$,
\[
\sup_J
\left\Vert
\mathbb E
\left[
\Psi_J(Z_{0,-j}+\theta_{0,j})
-
\Psi_J(Z_{0,-j})
\right]
\right\Vert_2
<\infty
\]
and
\[
\frac1N\sum_{k=1}^N
\left|
\mathbb E
\left[
\left\{
\mathbb P_{J_{k-1}}(F_0)-F_0
\right\}
(Z_{0,-j}+\theta_{0,j})
-
\left\{
\mathbb P_{J_{k-1}}(F_0)-F_0
\right\}
(Z_{0,-j})
\right]
\right|
=
o(N^{-1/2}).
\]
The second display is the direct counterfactual analogue of the small-bias condition used for the continuous average marginal effect. The centered difference between the sieve and true counterfactual signals has conditional second moment bounded by
$CJ_{k-1}^{-2s}$, so its running average is
$o_{\rm a.s.}(N^{-1/2})$ by the same martingale argument used for (i.2).

A first-order expansion of the two counterfactual evaluations gives
\begin{align*}
&\frac1N\sum_{k=1}^N
\left(
\widetilde\tau_{k,j}
-
\tau_{0,j}
\right)\\
&=
\frac1N\sum_{k=1}^N
\frac1B\sum_{i=1}^B
\left[
F_0
\left(
x_{0,i,k}
+
\sum_{l\ne j}\theta_{0,l}x_{l,i,k}
+
\theta_{0,j}
\right)
-
F_0
\left(
x_{0,i,k}
+
\sum_{l\ne j}\theta_{0,l}x_{l,i,k}
\right)
-
\tau_{0,j}
\right]\\
&\quad+
\mathbb E
\left[
\left.
\nabla_\theta
\left\{
F_0
\left(
x_0+\sum_{l\ne j}\theta_lx_l+\theta_j
\right)
-
F_0
\left(
x_0+\sum_{l\ne j}\theta_lx_l
\right)
\right\}
\right|_{\theta=\theta_0}
\right]
\frac1N\sum_{k=1}^N
\Delta\widetilde\theta_{k-1}\\
&\quad+
\frac1N\sum_{k=1}^N
\mathbb E
\left[
\Psi_{J_{k-1}}(Z_{0,-j}+\theta_{0,j})
-
\Psi_{J_{k-1}}(Z_{0,-j})
\right]^\top
\Delta\widetilde\beta_{k-1}+
o(N^{-1/2}),
\qquad\text{a.s.}
\end{align*}
Indeed, the second-order and centered empirical remainders are bounded by the same quantities used above for
(iii), (v), and (i.5)--(i.8). In particular, using
\[
\Delta\widetilde\omega_{k-1}
=
D_{k-1}
+
T_{k-1},
\]
their $D_{k-1}$-parts follow from
\eqref{refined_D_rate}, while their $T_{k-1}$-parts follow from
\autoref{lemma_transient}(ii)--(iii), and all are
$o_{\rm a.s.}(N^{-1/2})$ after averaging.

Using the PR representations of
$\Delta\doubleoverline\theta_N$ and
$\Delta\doubleoverline\beta_N$, together with the preceding counterfactual small-bias condition, we obtain
\[
\doubleoverline\tau_{N,j}
-
\tau_{0,j}
=
\frac1N\sum_{k=1}^N
\Upsilon_k
+
o(N^{-1/2}),
\qquad\text{a.s.},
\]
where, for this binary-regressor marginal effect,
\begin{align*}
\Upsilon_{i,k}
&=
F_0
\left(
x_{0,i,k}
+
\sum_{l\ne j}\theta_{0,l}x_{l,i,k}
+
\theta_{0,j}
\right)
-
F_0
\left(
x_{0,i,k}
+
\sum_{l\ne j}\theta_{0,l}x_{l,i,k}
\right)
-
\tau_{0,j}\\
&\quad+
\varepsilon_{i,k}
\mathbb E
\left[
\left.
\nabla_\theta
\left\{
F_0
\left(
x_0+\sum_{l\ne j}\theta_lx_l+\theta_j
\right)
-
F_0
\left(
x_0+\sum_{l\ne j}\theta_lx_l
\right)
\right\}
\right|_{\theta=\theta_0}
\right]\\
&\qquad\qquad\times
\mathbb M_\theta^{-1}
\nabla_zF_0(Z_{0,i,k})
\left\{
X_{i,k}
-
\mu_0(\theta_0,Z_{0,i,k})
\right\}\\
&\quad+
\varepsilon_{i,k}
\mathbb E
\left[
\Psi_{J_k}(Z_{0,-j}+\theta_{0,j})
-
\Psi_{J_k}(Z_{0,-j})
\right]^\top
\mathbb M_{\beta,J_k}^{-1}\\
&\qquad\qquad\times
\left\{
\Psi_{J_k}(Z_{0,i,k})
-
P_{J_k}\nabla_zF_0(Z_{0,i,k})X_{i,k}
\right\}.
\end{align*}
Thus, when \autoref{cond_ame}(iii) is imposed for
\[
\Upsilon_k
=
\frac1B\sum_{i=1}^B
\Upsilon_{i,k},
\]
the same Lyapunov argument yields
\[
(NV_\Upsilon)^{-1/2}
[Nr]
\left\{
\doubleoverline\tau_{[Nr],j}
-
\tau_{0,j}
\right\}
\Longrightarrow
\mathbb W(r),
\qquad r\in[0,1].
\]

\section{Proof of \autoref{theorem1} in \autoref{appendixB}}\label{appendixF}

This Appendix provides details for the proof of  \autoref{theorem1}. In particular, \autoref{appendix_b1} provides some preliminary results; \autoref{discrete_regressors} extends the results to the case where some regressors are discrete. \autoref{appendix_b2} provides the asymptotic properties of the raw updates $\widehat\theta_N$. Finally, the last section concludes with the proof of \autoref{theorem1}. 

\subsection{Preliminaries}\label{appendix_b1}
 
  Define
\begin{equation}\label{empirical score}\Phi_k\left(\theta, \mathcal{W}_k\right) =\frac{1}{h_k\cdot B(B-1)}\cdot\sum_{i_1\neq i_2}^B \mathcal{K}\left(\frac{Z_{i_1,k}(\theta) - Z_{i_2,k}(\theta)}{h_k}\right)\left(Y_{i_1,k} - Y_{i_2,k}\right)\left(X_{i_1,k} - X_{i_2,k}\right),\end{equation}
\begin{equation}\label{theoretical score}
\Phi_k\left(\theta\right) = \mathbb E\left[\Phi_k\left(\theta, \mathcal{W}_k\right)\right],
\end{equation}
and 
\begin{equation}\label{limit score}
\Phi\left(\theta\right) = \lim_{k\rightarrow\infty}\Phi_k\left(\theta\right).
\end{equation}

Given \autoref{condition1} and \autoref{condition2}, we have the following lemma that describes the properties of $\Phi_k(\theta,\mathcal{W}_k), \Phi_k(\theta),$ and $\Phi(\theta)$. 

\begin{lemma}\label{lemma1}
    Let \autoref{condition1} and \autoref{condition2} hold, for any choice of $\{h_k\}_{k=1}^{\infty}$ with $h_k\downarrow 0$, the following results hold:
    \begin{enumerate}
        \item For any $\theta\in\mathbb{R}^p$ we have:
        \begin{align*}
            \Phi(\theta) & = \int \left[ F_0\left(z - X_1^{\top}\Delta\theta\right)-F_0\left(z - X_2^{\top}\Delta\theta\right)\right]\left(X_1 - X_2\right)\times\\
            & \ \ \ \ \ \ \ \ \ \ \ \ \ \ f\left(z - X_1^{\top}\theta, X_1\right)f\left(z - X_2^{\top}\theta, X_2\right)dzdX_1dX_2;
        \end{align*} 
        \item  There is a finite constant $C_{\Phi, 1}>0$ such that \[\sup_{\theta\in\mathbb{R}^p}\left\Vert\Phi_k(\theta) - \Phi(\theta)\right\Vert_2\leq C_{\Phi, 1}h_k^{s_{\mathcal K}};\]
        \item There holds 
        \[\mathbb E\left[\left(\Phi_k(\theta, \mathcal{W}_k) - \Phi_k(\theta)\right)\left(\Phi_k(\theta, \mathcal{W}_k) - \Phi_k(\theta)\right)^{\top} \right] =  \frac{2\mathcal{V}_{\mathcal{K}}(\theta)}{B(B-1)h_k} + O(1),\]
        where 
        \begin{align*}
        \mathcal{V}_{\mathcal{K}}(\theta) = & \int \mathcal{K}^2(t)dt \int \left(\sigma^2(z - X_1^{\top}\theta, X_1) + \sigma^2(z - X_2^{\top}\theta, X_2) + \left(F_0(z - X_1^{\top}\Delta\theta) - F_0(z - X_2^{\top}\Delta\theta)\right)^2\right)  \\
        & \ \ \ \ \ \ \ \ \ \ \ \ \ \ \ \ \ \ \ \ \ \ \ \ \times (X_1-X_2)(X_1-X_2)^{\top} f(z- X_1^{\top}\theta, X_1)f(z- X_2^{\top}\theta, X_2)dzdX_1dX_2;
    \end{align*}  
        Consequently, there are finite constants $C_{\Phi,2}$ such that
        \[
        \sup_{\theta\in\mathbb R^p}\mathbb E\left\Vert
        \Phi_k(\theta,W_k)-\Phi_k(\theta)\right\Vert_2^2
        \leq  C_{\Phi,2}h_k^{-1}.
        \]

        \item There exist positive constants $C_{\Phi,3}$ such that
        \[
        \sup_{\theta\in\mathbb R^p}
        \mathbb E\left\Vert\nabla_{\theta}\Phi_k(\theta,W_k)
        -\nabla_{\theta}\Phi_k(\theta)\right\Vert_2^2
        \leq C_{\Phi,3}h_k^{-3}.
        \]
        In particular, for the fixed batch size maintained in this paper this is
        $O(h_k^{-3})$.

        \item There is a
        finite constant $C_{\Phi,4}$ such that
        \[
        \sup_{\theta\in\mathbb R^p}
        \mathbb E\left\Vert\Phi_k(\theta,W_k)-\Phi_k(\theta)\right\Vert_2^4
        \leq C_{\Phi,4}h_k^{-3}.
        \]

    \end{enumerate}
\end{lemma}

\begin{proof}[Proof of \autoref{lemma1}]
    We first prove \autoref{lemma1}(1) and \autoref{lemma1}(2). Obviously, under \autoref{condition1}, the expression in \autoref{lemma1}(1) is well defined. For arbitrary $x_0, X$, define $Z(\theta) = x_0 + X^{\top}\theta$. Then \begin{equation}\mathbb E(Y|x_{0}, X) = F_0(x_{0} + X^{\top}\theta_0) = F_0(Z(\theta) - X^{\top}\Delta\theta).\end{equation}  So  
    \begin{align}
        \Phi_k(\theta)  &= \mathbb E\left[h_k^{-1}\mathcal{K}\left(h_k^{-1}\left(Z_{1,k}(\theta) - Z_{2,k}(\theta)\right)\right)\left(Y_{1,k} - Y_{2,k}\right)\left(X_{1,k} - X_{2,k}\right)\right]\nonumber \\
        & = \int h_k^{-1}\mathcal{K}\left(h_k^{-1}\left(z_1 - z_2\right)\right)\left( F_0(z_1 - X_1^{\top}\Delta\theta) - F_0(z_2 - X_2^{\top}\Delta\theta)\right)\times\nonumber\\
        & \ \ \ \ \ \ \ \ \ \ \ \ \ \left(X_{1} - X_{2}\right)v(z_1,X_1|\theta)v(z_2, X_2|\theta)dz_1dz_2dX_1dX_2,
    \end{align}
    where $v(z, X|\theta)$ is the joint density of $Z(\theta)$ and $X$ given $\theta$. Now we derive the expression of $v(z, X|\theta)$.  Since \[\mathrm{Pr}(Z_{1,k}(\theta)\leq z, X_{1,k}\leq X) = \int_{X_{1,k}\leq X} \int_{x_{0,1,k} \leq z - X_{1,k}^{\top}\theta} f(x_{0,1, k}, X_{1,k})dx_{0,1,k}dX_{1,k},\]
    we have that 
    \[
    \frac{\partial \mathrm{Pr}(Z_{1,k}(\theta)\leq z, X_{1,k}\leq X)}{\partial z} = \int_{X_{1,k}\leq X} f(z - X_{1,k}^{\top}\theta, X_{1,k})dX_{1,k}
    \]
    and 
    \[
    \frac{\partial^{p+1}P(Z_{1,k}(\theta)\leq z, X_{1,k}\leq X)}{\partial z\partial X} =   f(z - X^{\top}\theta, X).
    \]
    This implies that   $v(z, X|\theta) = f(z - X^{\top}\theta, X)$. So 
   \begin{align*}
        \Phi_k(\theta) & = \int \mathcal{K}\left(t\right)\left( F_0(z + th_k - X_1^{\top}\Delta\theta) - F_0(z - X_2^{\top}\Delta\theta)\right)\left(X_{1} - X_{2}\right)\times\\
        & \ \ \ \ \ \ \ \ f(z + th_k - X_1^{\top}\theta ,X_1)f(z     - X_2^{\top}\theta ,X_2)dtdzdX_1dX_2\\
        & = \int \mathcal{K}\left(t\right)\left(  F_0(z  - X_1^{\top}\Delta\theta) - F_0(z - X_2^{\top}\Delta\theta)+ \sum_{j=1}^{s_{\mathcal K}-1} \frac{\partial^j F_0(z - X_1^{\top}\Delta\theta)}{j!}t^jh_k^j \right.\\
        & \left. \ \ \ \ \ \ \ \ \ \ \ \ \ \ \ \ + \frac{\partial^{s_{\mathcal K}} F_0(\zeta_1)}{s_{\mathcal K}!}t^{s_{\mathcal K}}h_k^{s_{\mathcal K}} \right) \left(f\left(z  - X_1^{\top}\theta ,X_1\right) + \right. \\
        & \ \ \ \ \ \ \ \ \ \ \ \left. \sum_{j=1}^{s_{\mathcal K}-1} \frac{1}{j!}\frac{\partial^j f(z - X_1^{\top}\theta, X_1)}{\partial z^j}t^jh_k^j + \frac{1}{s_{\mathcal K}!}\frac{\partial^{s_{\mathcal K}} f(\zeta_2, X_1)}{\partial z^{s_{\mathcal K}}}t^{s_{\mathcal K}}h_k^{s_{\mathcal K}}\right)\times\\
        & \ \ \ \ \ \ \left(X_{1} - X_{2}\right)f(z   - X_2^{\top}\theta ,X_2)dtdzdX_1dX_2,
    \end{align*}
    where $\zeta_1$ lies  somewhere between $z- X_1^{\top}\Delta\theta$ and $z+ th_k- X_1^{\top}\Delta\theta$, and $\zeta_2$ lies somewhere between $z- X_1^{\top} \theta$ and $z+ th_k- X_1^{\top} \theta$. Since $\int \mathcal{K}(t)dt = 1$, we have that 
    \begin{align*}
        \Phi_k(\theta) & = \Phi(\theta) + \int \mathcal{K}(t) \left(  F_0(z  - X_1^{\top}\Delta\theta) - F_0(z - X_2^{\top}\Delta\theta)\right) (X_1 - X_2)f(z - X_2^{\top}\theta, X_2)\times  \\
        & \ \ \ \ \ \ \  \ \  \ \  \ \  \ \  \left(\sum_{j=1}^{s_{\mathcal K} - 1} \frac{1}{j!}\frac{\partial^j f(z - X_1^{\top}\theta, X_1)}{\partial z^j}t^jh_k^j + \frac{1}{s_{\mathcal K}!}\frac{\partial^{s_{\mathcal K}} f(\zeta_2, X_1)}{\partial z^{s_{\mathcal K}}}t^{s_{\mathcal K}}h_k^{s_{\mathcal K}}\right)dtdzdX_1dX_2\\
        & \ \ \ \ \ \  \ \ \ \ +  \int \mathcal{K}(t) (X_1 - X_2) f(z - X_1^{\top}\theta, X_1) f(z - X_2^{\top}\theta, X_2)\times\\
        &  \ \ \ \ \ \ \  \ \  \ \  \ \  \ \ \left(\sum_{j=1}^{s_{\mathcal K}-1} \frac{\partial^j F_0(z - X_1^{\top}\Delta\theta)}{j!}t^jh_k^j + \frac{\partial^{s_{\mathcal K}} F_0(\zeta_1)}{s_{\mathcal K}!}t^{s_{\mathcal K}}h_k^{s_{\mathcal K}}\right) dtdzdX_1dX_2\\
        & \ \ \ \ \ \  \ \ \ \ + \int \mathcal{K}(t)(X_1 - X_2)f(z - X_2^{\top}\theta, X_2)\left(\sum_{j=1}^{s_{\mathcal K}-1} \frac{\partial^j F_0(z - X_1^{\top}\Delta\theta)}{j!}t^jh_k^j + \frac{\partial^{s_{\mathcal K}} F_0(\zeta_1)}{s_{\mathcal K}!}t^{s_{\mathcal K}}h_k^{s_{\mathcal K}}\right)\times\\
        &  \ \ \ \ \ \ \  \ \  \ \  \ \  \ \ \left(\sum_{j=1}^{s_{\mathcal K}-1} \frac{1}{j!}\frac{\partial^j f(z - X_1^{\top}\theta, X_1)}{\partial z^j}t^jh_k^j + \frac{1}{s_{\mathcal K}!}\frac{\partial^{s_{\mathcal K}} f(\zeta_2, X_1)}{\partial z^{s_{\mathcal K}}}t^{s_{\mathcal K}}h_k^{s_{\mathcal K}}\right)dtdzdX_1dX_2.
    \end{align*}
    \autoref{condition1} and \autoref{condition2} immediately lead to \autoref{lemma1}(2) and   \autoref{lemma1}(1) is proved because $h_k\downarrow 0$.

   To prove \autoref{lemma1}(3), let $\mathbb{K}_{k,\theta}(W_1,W_2) = \mathcal{K}(h_k^{-1}(Z_1(\theta)-Z_2(\theta))(Y_1-Y_2)(X_1 - X_2)$ denote the
symmetric order-two kernel in \eqref{empirical score}.  The exact Hoeffding covariance identity is
\begin{align*}
 \operatorname{Cov}\{\Phi_k(\theta,W_k)\}
& =\frac{4(B-2)}{B(B-1)}
       \operatorname{Cov}\{\mathbb{E}[\mathbb{K}_{k,\theta}(W_1,W_2)|W_1] - \mathbb{E}[\mathbb{K}_{k,\theta}(W_1,W_2)]\}\\
 &+\frac{2}{B(B-1)}
       \operatorname{Cov}\{\mathbb{K}_{k,\theta}(W_1,W_2)\}.
\end{align*}
 In the calculations below,
$\mathbb E_1$ denotes integration w.r.t. the first observation.
Note that 
\begin{align*}
   &  \mathbb{E}_1 \left[\frac{1}{h_k}\mathcal{K}\left(\frac{Z_{1,k}(\theta) - Z_{2,k}(\theta)}{h_k}\right)\left(Y_{1,k} - Y_{2,k}\right)\left(X_{1,k} - X_{2,k}\right)\right]\\
   & = \int \frac{1}{h_k} \mathcal{K}\left(\frac{z - Z_{2,k}(\theta)}{h_k}\right)\left[F_0(z - X_1^{\top}\Delta\theta) - Y_{2,k}\right]\left(X_1 - X_{2,k}\right)f(z - X_1^{\top}\theta, X_1)dzdX_1\\
   & = \int \mathcal{K}(t)\left(F_0\left( Z_{2,k}(\theta) +th_k - X_1^{\top}\Delta\theta\right) - Y_{2,k}\right)(X_{1} - X_{2,k})f(Z_{2,k}(\theta) + th_k - X_1^{\top}\theta, X_1)dtdX_1\\
   & = \int \left(F_0\left( Z_{2,k}(\theta)   - X_1^{\top}\Delta\theta\right) - Y_{2,k}\right)(X_1 - X_{2,k})f(Z_{2,k}(\theta) - X_1^{\top}\theta, X_1)dX_1\\
   & + O\left( h_k^{s_{\mathcal K}}  (1 + |Y_{2,k}|)(1 + \Vert X_{2,k}\Vert_2) \right ).
\end{align*}
So the above term is bounded by $(1 + |Y_{2,k}|)(1 + \Vert X_{2,k}\Vert_2)$ up to some constant. This implies that 
\begin{align*}
& \mathrm{var}\left(\mathbb{E}_1 \left[\frac{1}{h_k}\cdot \mathcal{K}\left(\frac{Z_{1,k}(\theta) - Z_{2,k}(\theta)}{h_k}\right)\left(Y_{1,k} - Y_{2,k}\right)\left(X_{1,k} - X_{2,k}\right)\right]\right)  \\
& \leq C\mathbb{E}\left[ (1 + |Y_{2,k}|)^2(1 + \Vert X_{2,k}\Vert_2)^2 \right]<\infty
\end{align*}
due to the fact that $\sup_{x_0, X}\mathbb{E}(Y^2|x_0, X) \leq  \sup_{x_0, X}F_0^2(x_0+X^{\top}\theta_0) + \sup_{x_0, X}\sigma^2(x_0, X)<\infty$ and $\mathbb{E}\Vert X\Vert^2_2<\infty$. 
 To show the remaining parts, we note that $\Phi_k(\theta) = \Phi(\theta) + O(h_k^{s_{\mathcal K}})$ and $\Phi(\theta)$ is uniformly bounded. So 
\begin{align*}
     & \mathrm{var}\left( \frac{1}{h_k }\cdot  \mathcal{K}\left(\frac{Z_{1,k}(\theta) - Z_{2,k}(\theta)}{h_k}\right)\left(Y_{1,k} - Y_{2,k}\right)\left(X_{1,k} - X_{2,k}\right)\right)\\
     & = h_k^{-2}\mathbb{E}\left( \mathcal{K}^2\left(\frac{Z_{1,k}(\theta) - Z_{2,k}(\theta)}{h_k}\right)\left(Y_{1,k} - Y_{2,k}\right)^2\left(X_{1,k} - X_{2,k}\right)\left(X_{1,k} - X_{2,k}\right)^{\top}\right) + O(1).
\end{align*}
Note that 
\begin{align*}
   & h_k^{-2}\mathbb{E}\left( \mathcal{K}^2\left(\frac{Z_{1,k}(\theta) - Z_{2,k}(\theta)}{h_k}\right)\left(Y_{1,k} - Y_{2,k}\right)^2\left(X_{1,k} - X_{2,k}\right)\left(X_{1,k} - X_{2,k}\right)^{\top}\right)\\
   & = h_k^{-2}\mathbb{E}\left( \mathcal{K}^2\left(\frac{Z_{1,k}(\theta) - Z_{2,k}(\theta)}{h_k}\right)\left(F_0(Z_{1,k}(\theta) - X_{1,k}^{\top}\Delta\theta) -  F_0(Z_{2,k}(\theta) - X_{2,k}^{\top}\Delta\theta)\right)^2\right.\\
   &\left. \qquad \qquad \qquad \qquad \times \left(X_{1,k} - X_{2,k}\right)\left(X_{1,k} - X_{2,k}\right)^{\top}\right)\\
   & + h_k^{-2}\mathbb{E}\left( \mathcal{K}^2\left(\frac{Z_{1,k}(\theta) - Z_{2,k}(\theta)}{h_k}\right)\left(\sigma^2(Z_{1,k}(\theta) - X_{1,k}^{\top}\theta,X_{1,k}) +  \sigma^2(Z_{2,k}(\theta) - X_{2,k}^{\top}\theta,X_{2,k})\right) \right. \\
   & \left. \qquad \qquad \qquad \times \left(X_{1,k} - X_{2,k}\right) \left(X_{1,k} - X_{2,k}\right)^{\top}\right)\\
  & = h_k^{-2}\int  \mathcal{K}^2\left(\frac{z_1 - z_2}{h_k}\right)\left(F_0(z_1 - X_{1}^{\top}\Delta\theta) -  F_0(z_{2} - X_{2}^{\top}\Delta\theta)\right)^2\left(X_{1} - X_{2}\right)\left(X_{1} - X_{2}\right)^{\top}\times \\
  & \ \ \ \ \ \ \ \ \ \ \ \ \ \ \ \ \ \ \ \ \ \ \ \ \  \ \ \ \ \ \ \ \ \ \ \ \ \ \ \ \ \ \ \ \ f(z_1 - X_1^{\top} \theta, X_1) f(z_2 - X_2^{\top} \theta, X_2)dz_1dz_2dX_1dX_2 \\
  & + h_k^{-2}\int  \mathcal{K}^2\left(\frac{z_1 - z_2}{h_k}\right)\left(\sigma^2(z_1 - X_{1}^{\top} \theta, X_1) +  \sigma^2(z_{2} - X_{2}^{\top}\theta, X_2)\right)\left(X_{1} - X_{2}\right)\left(X_{1} - X_{2}\right)^{\top}\times \\
  & \ \ \ \ \ \ \ \ \ \ \ \ \ \ \ \ \ \ \ \ \ \ \ \ \  \ \ \ \ \ \ \ \ \ \ \ \ \ \ \ \ \ \ \ \ f(z_1 - X_1^{\top} \theta, X_1) f(z_2 - X_2^{\top} \theta, X_2)dz_1dz_2dX_1dX_2
\end{align*}
Following the previous proofs, we can show that the   terms on the RHS of the  last equality can be written as 
\begin{align*}
& h_k^{-1}\int\mathcal{K}^2(t)dt \int \left(F_0(z - X_{1}^{\top}\Delta\theta) -  F_0(z - X_{2}^{\top}\Delta\theta)\right)^2\left(X_{1} - X_{2}\right)\left(X_{1} - X_{2}\right)^{\top}\times \\
& \ \ \ \ \ \ \ \ \ \ \ \ \ \ \ \ \ \ \ \ \ \ \ \ \  \ \ \ \ \ \ \ \ \ \ \ \ \ \ \ \ \ \ \ \ f(z - X_1^{\top} \theta, X_1) f(z - X_2^{\top}\theta, X_2 )dzdX_1dX_2\\
& + h_k^{-1}\int\mathcal{K}^2(t)dt \int \left(\sigma^2(z   - X_{1}^{\top} \theta, X_1) +  \sigma^2(z  - X_{2}^{\top}\theta, X_2)\right)\left(X_{1} - X_{2}\right)\left(X_{1} - X_{2}\right)^{\top}\times \\
& \ \ \ \ \ \ \ \ \ \ \ \ \ \ \ \ \ \ \ \ \ \ \ \ \  \ \ \ \ \ \ \ \ \ \ \ \ \ \ \ \ \ \ \ \ f(z - X_1^{\top} \theta, X_1) f(z - X_2^{\top}\theta, X_2 )dzdX_1dX_2 + O(1).
\end{align*}
This proves the result. 

For \autoref{lemma1}(4), note that \[\nabla_{\theta}\Phi_k\left(\theta, \mathcal{W}_k\right) = \frac{1}{h_k^{2}B(B-1)}\sum_{i\neq j}^B\nabla\mathcal{K}\left(h_k^{-1}\left(Z_{i,k}(\theta) - Z_{j,k}(\theta)\right)\right)\left(Y_{i,k} - Y_{j,k}\right)\left(X_{i,k} - X_{j,k}\right)\left(X_{i,k} - X_{j,k}\right)^{\top}\]
Again using the Hoeffding decomposition for a symmetric second-order
U-statistic, the first projection and the degenerate remainder contribute,
respectively,
\begin{align*}
 &\mathbb E\left\Vert\nabla_{\theta}\Phi_k(\theta,\mathcal{W}_k)
       -\nabla_{\theta}\Phi_k(\theta)\right\Vert_2^2\\
 &\qquad\leq \frac{C}{B}
 +\frac{C}{B(B-1)h_k^3}
 \int (\nabla \mathcal K(t))^2
 \left[\left\{F_0(z+th_k-X_1^\top\Delta\theta)
              -F_0(z-X_2^\top\Delta\theta)\right\}^2\right.\\
 &\hspace{5.5cm}\left.
 +\sigma^2(z+th_k-X_1^\top\theta,X_1)
 +\sigma^2(z-X_2^\top\theta,X_2)\right]
 \Vert X_1-X_2\Vert_2^4\\
 &\hspace{5.5cm}\times
 f(z+th_k-X_1^\top\theta,X_1)
 f(z-X_2^\top\theta,X_2)\,dt\,dz\,dX_1\,dX_2\\
 &\qquad\leq \frac{C}{B}+\frac{C}{B(B-1)h_k^3}.
\end{align*}
The last integral is uniformly bounded by  \autoref{condition1} and \autoref{condition2}.

For part~(5), write the centered U-statistic as a finite sum of centered
ordered-pair kernels.  Since $B$ is fixed, the inequality
$\Vert\sum_{r=1}^{B(B-1)}a_r\Vert_2^4\leq C_B\sum_r\Vert a_r\Vert_2^4$
reduces the claim to one ordered pair.  Conditional fourth moments of
$Y_1-Y_2$ are uniformly bounded by  \autoref{condition1}(iv), and the
fourth moment of $X_1-X_2$ is integrable by  \autoref{condition1}(iii).
Consequently, uniformly in $\theta$,
\begin{align*}
 &\mathbb E\left[
 h_k^{-4}\mathcal K^4\!\left(
   \frac{Z_1(\theta)-Z_2(\theta)}{h_k}\right)
 |Y_1-Y_2|^4\Vert X_1-X_2\Vert_2^4\right]\\
 &\qquad\leq C h_k^{-4}\int
 \mathcal K^4\!\left(\frac{z_1-z_2}{h_k}\right)
 (\Vert X_1\Vert_2^4+\Vert X_2\Vert_2^4)
 f(z_1-X_1^\top\theta,X_1)f(z_2-X_2^\top\theta,X_2)
 \,dz_1dz_2dX_1dX_2\\
 &\qquad\leq C h_k^{-3}\int \mathcal K^4(t)\,dt.
\end{align*}
 This proves part~(5) and the
lemma.
\end{proof}

Given \autoref{lemma1}, we can briefly discuss the intuition of our algorithm. Note that (\ref{first-step algorithm}) leads to 
\begin{align}\label{intuition}
    \widehat\theta_k &  = \widehat\theta_{k-1} +  \gamma_k \cdot \Phi_k(\widehat\theta_{k-1},\mathcal{W}_k)\nonumber \\
    & = \underset{(I)}{\underbrace{\widehat\theta_{k-1} + \gamma_k \cdot \Phi(\widehat\theta_{k-1})}} - \underset{(II)}{\underbrace{ \gamma_k \cdot \left( \Phi(\widehat\theta_{k-1}) - \Phi_k(\widehat\theta_{k-1})\right)}} - \underset{(III)}{\underbrace{\gamma_k\cdot \left( \Phi_k(\widehat\theta_{k-1}) - \Phi_k(\widehat\theta_{k-1}, \mathcal{W}_k)\right)}}.
\end{align}
In the above update, term (III) has zero expectation conditioned on observations up to period $k-1$, and term (II) will vanish as $k$ increases if we choose $h_k\downarrow 0$ according to \autoref{lemma1}. So we can regard the update (\ref{first-step algorithm}) as being mainly driven by term (I).   Define 
  \begin{align}
  \mathbb{H}(\theta, \tau) \equiv & \int \nabla_z F_0(z - X_2^{\top}\Delta\theta - \tau(X_1 - X_2)^{\top}\Delta\theta)f(z - X_1^{\top}\theta, X_1)f(z-X_2^{\top}\theta, X_2)\nonumber \\
  & \ \ \ \ \ \ \ \ \ \ \ \ \ \ \ \ \ \ \ \ \ \ \  \times\left(X_1 - X_2\right)\left(X_1 - X_2\right)^{\top}dzdX_1dX_2,
  \end{align}
  then for any $\theta$, we obviously have that \begin{equation}\Phi(\theta) =- \left[\int_{0}^{1}  \mathbb{H}(\theta,\tau)d\tau \right]\Delta\theta\end{equation}  and \begin{equation}\lambda_{\mathrm{min}}\left(\int_{0}^1\mathbb{H}(\theta, \tau)d\tau\right)\geq 0.\end{equation} 
  If additional conditions are imposed on the data generating process so that $\lambda_{\mathrm{min}}(\int_{0}^1\mathbb{H}(\theta, \tau)d\tau)> 0$ for any $\theta$,  (I)  is a contraction mapping and  the same applies to the kernel-based warm-start update as $k\rightarrow \infty$. Consequently,  our method guarantees global stability, meaning that we do not require any assumptions over the initial starting point of the update. Such strict lower boundedness of the smallest eigenvalue is formally stated in the following lemma. Note that the following lemma relies on the fact that all regressors are continuous; we relax this requirement in the following \autoref{discrete_regressors}.

  \begin{lemma}\label{lemma2}
      If \autoref{condition1}  holds, then there exists $C_{\Phi, 4}>0$ such that 
       \[
\sup_{\theta\in\mathbb{R}^p}\sup_{\tau\in[0,1]}\lambda_{\text{max}}\left(\mathbb{H}(\theta,\tau)\right)\leq C_{\Phi, 4}.
 \]
 If \autoref{condition3} holds, then  
 \[
\inf_{\tau\in[0,1]}\lambda_{\mathrm{min}}\left(\mathbb{H}(\theta,\tau)\right)\geq c_{\mathbb{H}} \cdot \left(\frac{\overline{z} - \underline{z}}{12(1 + \Vert\Delta\theta\Vert_2 + \Vert\theta_0\Vert_2)}\wedge r_X\right)^{2p+2},
 \]
 where $c_{\mathbb{H}}$ is a positive constant that does not depend on the choice of $\theta$. 
 \end{lemma}

\begin{proof}[Proof of \autoref{lemma2}]
\autoref{lemma2}(i) is obvious if we note that \begin{align*}\lambda_{\text{max}}\left(\mathbb{H}(\theta,\tau)\right)\leq \Vert\mathbb{H}(\theta,\tau)\Vert_F &\leq \Vert\nabla_z F_0\Vert_{\infty}\int  \Vert X_1 - X_2\Vert_2^2f(z- X_1^{\top}\theta, X_1) f(z- X_2^{\top}\theta, X_2)dzdX_1dX_2\\
& \leq 4\Vert\nabla_z F_0\Vert_{\infty}\int\Vert  X_1\Vert_2^2 f(z- X_1^{\top}\theta, X_1) f(z- X_2^{\top}\theta, X_2)dzdX_1dX_2\\
& \leq 4\Vert\nabla_z F_0\Vert_{\infty} \int\Vert  X_1\Vert_2^2g_0(X_1)dX_1 \int f(z- X_2^{\top}\theta, X_2)dzdX_2 \end{align*} which is uniformly bounded according to \autoref{condition1}.

To prove \autoref{lemma2}(ii), consider the choices of $z$, $X_1$, and $X_2$  such that 
\[
\frac{3\underline{z} + \overline{z}}{4}\leq z\leq  \frac{\underline{z} + 3\overline{z}}{4}, \  \Vert X_1\Vert_2\leq \frac{\overline{z} - \underline{z}}{12(1 + \Vert\Delta\theta\Vert_2 + \Vert\theta_0\Vert_2)}\wedge r_X, \ \Vert X_2\Vert_2\leq \frac{\overline{z} - \underline{z}}{12(1 + \Vert\Delta\theta\Vert_2 + \Vert\theta_0\Vert_2)}\wedge r_X.
\]
In this case, 
\begin{align*}
z - X_{2}^{\top} \Delta  \theta - \tau \left(X_{1} - X_{2}\right)^{\top} \Delta \theta - \underline{z} \geq  \frac{\overline{z}-\underline{z}}{4} - 2\Vert X_2\Vert_2\Vert\Delta\theta\Vert_2 -\Vert X_1\Vert_2\Vert\Delta\theta\Vert_2\geq 0,  
\end{align*}
and 
\begin{align*}
z - X_{2}^{\top} \Delta  \theta - \tau \left(X_{1} - X_{2}\right)^{\top} \Delta \theta - \overline{z} \leq  -\frac{\overline{z}-\underline{z}}{4} + 2\Vert X_2\Vert_2\Vert\Delta\theta\Vert_2 + \Vert X_1\Vert_2\Vert\Delta\theta\Vert_2\leq 0,  
\end{align*}
so $\nabla F_0\left(z - X_{2}^{\top} \Delta  \theta - \tau \left(X_{1} - X_{2}\right)^{\top}\Delta\theta\right)\geq \underline{c}_F$ according to \autoref{condition3}. On the other side, for $z, X_1, X_2$ satisfying the above requirement, we also have that for $i=1,2$,  $z - X_i^{\top}\theta\geq z - \Vert X_i\Vert_2\Vert\theta\Vert_2\geq z - \Vert X_i\Vert_2(\Vert\theta_0\Vert_2 + \Vert\Delta\theta\Vert_2 )\geq \underline{z}$ and similarly  $z - X_i^{\top}\theta\leq \overline{z}$. This implies that $f(z- X_i^{\top}\theta, X_i)\geq \underline{c}_f$ according again to \autoref{condition3}. 

Then define area $\Omega = \{(z, X_1,X_2): \frac{3\underline{z} + \overline{z}}{4}\leq z\leq  \frac{\underline{z} + 3\overline{z}}{4}, 
\Vert X_1\Vert_2\leq \frac{\overline{z} - \underline{z}}{12(1 + \Vert\Delta\theta\Vert_2 + \Vert\theta_0\Vert_2)}\wedge r_X, \Vert X_2\Vert_2\leq \frac{\overline{z} - \underline{z}}{12(1 + \Vert\Delta\theta\Vert_2 + \Vert\theta_0\Vert_2)}\wedge r_X\}$.  We have that
\begin{align*}
& \lambda_{\mathrm{min}}(\mathbb{H}(\theta,\tau))\\
& \geq \lambda_{\mathrm{min}}\left(\int_{\Omega}\partial F_0\left(z - X_{2}^{\top} \Delta  \theta - \tau \left(X_{1} - X_{2}\right)^{\top}\Delta\theta\right)f(z - X_1^{\top}\theta, X_1)f(z-X_2^{\top}\theta, X_2)\right.\\
& \left. \qquad \qquad \qquad \qquad \times \left(X_1 - X_2\right)\left(X_1 - X_2\right)^{\top}dzdX_1dX_2\right)\\
& \geq \underline{c}_F\underline{c}_f^2\lambda_{\mathrm{min}}\left( \int_{\Omega}\left(X_1 - X_2\right)\left(X_1 - X_2\right)^{\top}dzdX_1dX_2\right)\\
&= (\overline{z} - \underline{z})\underline{c}_F\underline{c}_f^2 c_p \cdot \left(\frac{\overline{z} - \underline{z}}{12(1 + \Vert\Delta\theta\Vert_2 + \Vert\theta_0\Vert_2)}\wedge r_X\right)^{p}\cdot \lambda_{\mathrm{min}}\left( \int_{
\Vert X\Vert_2\leq \frac{\overline{z} - \underline{z}}{12(1 + \Vert\Delta\theta\Vert_2 + \Vert\theta_0\Vert_2)}\wedge r_X}XX^{\top}dX\right)\\
& = (\overline{z} - \underline{z})\underline{c}_F\underline{c}_f^2 c_p \cdot \left(\frac{\overline{z} - \underline{z}}{12(1 + \Vert\Delta\theta\Vert_2 + \Vert\theta_0\Vert_2)}\wedge r_X\right)^{p}\cdot  \int_{
\Vert X\Vert_2\leq \frac{\overline{z} - \underline{z}}{12(1 + \Vert\Delta\theta\Vert_2 + \Vert\theta_0\Vert_2)}\wedge r_X}x_1^2dX \\
&=  \frac{c_p^2(\overline{z}-\underline{z})\underline{c}_F\underline{c}_f^2}{p+2} \left(\frac{\overline{z} - \underline{z}}{12(1 + \Vert\Delta\theta\Vert_2 + \Vert\theta_0\Vert_2)}\wedge r_X\right)^{2p+2},
\end{align*}
where $c_p$ is the volume of $p$-dimensional unit ball. This shows the result.
\end{proof}

The lower bound of the smallest eigenvalue as a function of $\theta$ in
\autoref{lemma2} is crucial because, for all sufficiently large $k$, it makes
term~(I) in \eqref{intuition} a strict contraction.  The same result also
guarantees point identification of $\theta_0$.

\subsection{The Case with Discrete Regressors}\label{discrete_regressors}

This section verifies the curvature and identification ingredient needed to
extend the preceding global-stability argument when $X$ contains both
continuous and discrete components.  The stochastic convergence conclusions
also require the corresponding moment, noise, kernel, and step-size assumptions
from \autoref{condition1}, \autoref{condition2}, and \autoref{condition4}; the
lemma below does not replace those assumptions.  We require
throughout that $x_0$ be continuously distributed. Decompose $X = (X_c^{\top}, 
X_d^{\top})^{\top}$, where $X_c \in \mathbb{R}^{p_c}$ collects the continuous regressors and 
$X_d \in \mathcal{X}_d \subseteq \mathbb{R}^{p_d}$ the discrete ones, $p = p_c + p_d$, with 
$\mathcal{X}_d$ countable. Let $p(x_d) = \mathbb{P}(X_d = x_d)$, and let $f(\cdot, \cdot 
\mid x_d)$ denote the conditional joint density of $(Z_0, X_c^{\top})^{\top}$ given $X_d = 
x_d$, where $Z_0 = x_0 + X^{\top}\theta_0$ is the true index.

Define 
\begin{align}\label{eq:H_discrete}
\mathbb{H}(\theta, \tau) &  \;=\; \sum_{X_{d,1}\in\mathcal{X}_d}\sum_{X_{d,2}\in\mathcal{X}_d} 
p(X_{d,1})\,p(X_{d,2}) \int \nabla_z F_0\big(z - X_2^{\top}\Delta\theta - \tau(X_1 - 
X_2)^{\top}\Delta\theta\big)\,  \nonumber \\
& \ \ \ \ \ \ \ \ \ \ \ \ \ \times\, f(z - X_1^{\top}\Delta\theta, X_{c,1}\mid X_{d,1})\, f(z - X_2^{\top}\Delta\theta, 
X_{c,2}\mid X_{d,2})\\ \nonumber
& \ \ \ \ \ \  \ \ \ \ \ \ \  \times (X_1 - X_2)(X_1 - X_2)^{\top}\,dz\,dX_{c,1}\,dX_{c,2},
\end{align}
indexed by $\tau \in [0,1]$. The object of this section is a strictly positive lower bound 
on the minimum eigenvalue of $\mathbb{H}(\theta,\tau)$, which substitutes for the 
corresponding step in the continuous-regressor proofs.

\medskip
\noindent\textbf{Condition 1$'$.} \textit{(i) $F_0$ is bounded and has uniformly bounded derivatives up to
order $s_{\mathcal K}$; (ii) for each $x_d \in \mathcal{X}_d$, $f(\cdot, X_c \mid x_d)$ has partial 
derivatives with respect to its first argument up to order $s_{\mathcal K}$, and these are uniformly 
bounded over $x_d \in \mathcal{X}_d$; (iii) there exist nonnegative functions $g_0, \ldots, 
g_{s_{\mathcal K}+1}$ on $\mathbb{R}^{p_c}$ such that, uniformly over $x_d \in \mathcal{X}_d$ and $0 
\leq j \leq s_{\mathcal K}+1$, $\sup_{z\in\mathbb{R}} \left|\partial^{j} f(z, X_c \mid x_d)/\partial 
z^{j}\right| \leq g_j(X_c)$ and $\int \left(1 + \Vert X_c\Vert_2^4\right) g_j(X_c)\, dX_c < 
\infty$; moreover, $\mathbb{E}\Vert X\Vert_2^4 < \infty$; (iv) $\sup_{x_0,X}\mathbb E(\varepsilon^4\mid x_0,X)<\infty$ is uniformly bounded and $\sigma^2(x_0, X)$ has a uniformly bounded 
derivative with respect to $x_0$.}

\medskip
\noindent\textbf{Condition 3$'$.} \textit{(i) $\nabla_z F_0(z) > 0$ for all $z \in 
\mathbb{R}$; (ii) there exist finitely many pairs of support points $(x_d^{(l)}, 
\tilde{x}_d^{(l)}) \in \mathcal{X}_d \times \mathcal{X}_d$, $l = 1, \ldots, L$, each with 
$p(x_d^{(l)}) > 0$ and $p(\tilde{x}_d^{(l)}) > 0$, such that
\[
\mathrm{span}\left\{\,x_d^{(1)} - \tilde{x}_d^{(1)}, \;\ldots,\; x_d^{(L)} - 
\tilde{x}_d^{(L)}\,\right\} \;=\; \mathbb{R}^{p_d};
\]
(iii) there exists $r_X > 0$ such that, for every support point $x_d$ appearing in the 
pairs of part (ii), the map $(z, X_c) \mapsto f(z, X_c \mid x_d)$ is continuous and, for 
every compact $I \subset \mathbb{R}$,
\[
\inf_{z \in I,\; \Vert X_c\Vert_2 \leq r_X} f(z, X_c \mid x_d) \;>\; 0.
\]}

\begin{remark}\label{rem:spanning}
Condition 3$'$(ii) is the natural discrete-regressor analogue of the full-rank requirement 
in the continuous case, and is in the spirit of the identification conditions for 
maximum-rank-correlation estimators \citep{han1987non}. It accommodates mutually 
exclusive dummy variables: if $X_d$ consists of group indicators with an omitted reference 
category (as with the race and stop-reason dummies in Stanford Open Policing Project application in the main text), the pairs 
$(e_j, \boldsymbol{0}_{p_d})$, $j = 1, \ldots, p_d$ --- ``group $j$ versus the reference 
group'' --- easily satisfy the spanning requirement whenever each category has positive 
probability. 
\end{remark}

\begin{lemma}\label{lem:H_discrete}
Suppose Conditions~1$'$ and 3$'$ hold and $x_0$ is continuously distributed given $X$. Then 
for every $\theta \in \mathbb{R}^p$,
\[
\inf_{\tau \in [0,1]} \lambda_{\mathrm{min}}\big(\mathbb{H}(\theta, \tau)\big) \;>\; 0.
\]
Moreover, for every compact $\Theta \subset \mathbb{R}^p$, $\inf_{\theta\in\Theta} 
\inf_{\tau\in[0,1]} \lambda_{\mathrm{min}}(\mathbb{H}(\theta,\tau)) > 0$.
\end{lemma}

\begin{proof}
  We first verify that $\mathbb{H}(\theta,\tau)$ 
is finite and jointly continuous in $(\theta, \tau)$. Each entry of the integrand in 
\eqref{eq:H_discrete} is bounded in absolute value by $\Vert\nabla_z F_0\Vert_{\infty}\, 
\Vert X_1 - X_2\Vert_2^2\, f_1 f_2$ with $f_l := f(z - X_l^{\top}\Delta\theta, X_{c,l} \mid 
X_{d,l})$. Integrating $z$ first against $f_2$ for fixed $(X_{c,2}, X_{d,2})$ gives the 
marginal conditional density value $f_{X_c \mid X_d}(X_{c,2} \mid X_{d,2})$, while $f_1 
\leq g_0(X_{c,1})$ by Condition~1$'$(iii). Hence, by $\Vert X_1 - X_2\Vert_2^2 \leq 2\Vert 
X_1\Vert_2^2 + 2\Vert X_2\Vert_2^2$, the total mass is bounded by
\[
C \Vert\nabla_z F_0\Vert_{\infty} \left[\int (1 + \Vert X_c\Vert_2^2)\, g_0(X_c)\, dX_c + 
\mathbb{E}\Vert X\Vert_2^2 + \sum_{x_d} p(x_d)\Vert x_d\Vert_2^2 \right] < \infty,
\]
using Condition~1$'$(iii) and $\mathbb{E}\Vert X\Vert_2^4 < \infty$ (which implies 
$\sum_{x_d} p(x_d)\Vert x_d\Vert_2^2 < \infty$). The integrand is continuous in $(\theta, 
\tau)$ for almost every $(z, X_{c,1}, X_{c,2})$ and every $(X_{d,1}, X_{d,2})$, by 
continuity of $\nabla_z F_0$ (Condition~1$'$(i)) and of $f$ in its first argument 
(Condition~1$'$(ii)); the display above provides an integrable envelope that is locally 
uniform in $\theta$. Dominated convergence therefore yields joint continuity of $(\theta, 
\tau) \mapsto \mathbb{H}(\theta, \tau)$.

Now fix $\theta$ and $\tau 
\in [0,1]$, and let $v = (v_c^{\top}, v_d^{\top})^{\top} \in \mathbb{R}^p$ with $\Vert 
v\Vert_2 = 1$. Since $\nabla_z F_0 > 0$ and all densities and weights are nonnegative, 
every term of the double sum in $v^{\top}\mathbb{H}(\theta,\tau)v$ is nonnegative, so
\begin{equation}\label{eq:oneterm}
v^{\top}\mathbb{H}(\theta,\tau)v \;\geq\; p(x_d)\,p(\tilde{x}_d) \int_{\mathcal{R}} 
\nabla_z F_0(\zeta)\, f_1 f_2\, \big[(X_1 - X_2)^{\top}v\big]^2\, dz\, dX_{c,1}\, dX_{c,2}
\end{equation}
for any fixed pair $(x_d, \tilde{x}_d)$ of support points and any measurable region 
$\mathcal{R}$, where $\zeta := z - X_2^{\top}\Delta\theta - \tau(X_1 - 
X_2)^{\top}\Delta\theta$, $X_1 = (X_{c,1}^{\top}, x_d^{\top})^{\top}$, $X_2 = 
(X_{c,2}^{\top}, \tilde{x}_d^{\top})^{\top}$. Take $\mathcal{R} = \mathcal{R}(r) := \{(z, 
X_{c,1}, X_{c,2}) : z \in [-1,1],\; \Vert X_{c,1}\Vert_2 \leq r,\; \Vert X_{c,2}\Vert_2 
\leq r\}$ for an $r \in (0, r_X]$ to be chosen. On $\mathcal{R}(r)$, the argument $\zeta$ 
lies in the compact interval
\[
\mathcal{A}_{\theta} := \left[-1 - (r + M_d)\Vert\Delta\theta\Vert_2 - 2(r + 
M_d)\Vert\Delta\theta\Vert_2,\; 1 + (r + M_d)\Vert\Delta\theta\Vert_2 + 2(r + 
M_d)\Vert\Delta\theta\Vert_2\right]
\]
for all $\tau \in [0,1]$, where $M_d := \max_l \max\{\Vert x_d^{(l)}\Vert_2, \Vert 
\tilde{x}_d^{(l)}\Vert_2\}$. By Condition~3$'$(i) and continuity of $\nabla_z F_0$,
\[
c_0(\theta) := \inf_{u \in \mathcal{A}_{\theta}} \nabla_z F_0(u) > 0,
\]
and this bound holds uniformly over $\tau \in [0,1]$ because $\zeta \in \mathcal{A}_{\theta}$ for 
every $\tau$. Similarly, the arguments $z - X_l^{\top}\Delta\theta$ of $f_1, f_2$ lie in a 
compact interval $I_{\theta}$ on $\mathcal{R}(r)$, so Condition~3$'$(iii) gives
\[
\delta(\theta) := \min_{x_d \in \{x_d^{(l)}, \tilde{x}_d^{(l)}: l \leq L\}}\; \inf_{z \in 
I_{\theta},\, \Vert X_c\Vert_2 \leq r_X} f(z, X_c \mid x_d) > 0.
\]
Now consider $v$ with $v_d \neq 0$. By Condition~3$'$(ii), there exists $l$ with $a := 
(x_d^{(l)} - \tilde{x}_d^{(l)})^{\top} v_d \neq 0$; take $(x_d, \tilde{x}_d) = (x_d^{(l)}, 
\tilde{x}_d^{(l)})$ in \eqref{eq:oneterm} and set $r = \min\{r_X, |a|/4\}$. On 
$\mathcal{R}(r)$,
\[
\left|(X_1 - X_2)^{\top}v\right| \;=\; \left|(X_{c,1} - X_{c,2})^{\top}v_c + a\right| 
\;\geq\; |a| - 2r\Vert v_c\Vert_2 \;\geq\; |a| - 2r \;\geq\; \frac{|a|}{2},
\]
so $[(X_1 - X_2)^{\top}v]^2 \geq a^2/4$ pointwise. Consequently, we have that 
\[
v^{\top}\mathbb{H}(\theta,\tau)v \;\geq\; p(x_d^{(l)})\, p(\tilde{x}_d^{(l)})\, 
c_0(\theta)\, \delta(\theta)^2\, \frac{a^2}{4}\, \cdot 2\, \mathrm{vol}(B_r)^2 \;>\; 0,
\]
where $\mathrm{vol}(B_r)$ is the volume of the $r$-ball in $\mathbb{R}^{p_c}$ and the 
factor $2$ is the length of $[-1,1]$.

We finally consider the case of $v$ with $v_d = 0$, $v_c \neq 0$, so $\Vert v_c\Vert_2 = 1$.  Take any single 
support point from the pairs, say $x_d = \tilde{x}_d = x_d^{(1)}$, and $r = r_X$ in 
\eqref{eq:oneterm}. Then $(X_1 - X_2)^{\top}v = (X_{c,1} - X_{c,2})^{\top}v_c$, and by 
symmetry of $B_{r}\times B_{r}$ (the cross term vanishes and the ball is isotropic),
\[
\int_{B_{r}}\!\int_{B_{r}} \big[(X_{c,1} - X_{c,2})^{\top}v_c\big]^2\, dX_{c,1}\, dX_{c,2} 
\;=\; 2\,\mathrm{vol}(B_{r}) \int_{B_{r}} (u^{\top}v_c)^2\, du \;=\; 
2\,\mathrm{vol}(B_{r})\, \kappa(r)\,\Vert v_c\Vert_2^2 \;=:\; c_1(r) > 0,
\]
with $\kappa(r) = \int_{B_r}(u^{\top}e)^2 du > 0$ independent of the unit vector $e$. Hence
\[
v^{\top}\mathbb{H}(\theta,\tau)v \;\geq\; p(x_d^{(1)})^2\, c_0(\theta)\, 
\delta(\theta)^2\, \cdot 2\, c_1(r_X) \;>\; 0.
\]
The above discussion shows that $v^{\top}\mathbb{H}(\theta,\tau)v > 0$ for 
every unit $v$, i.e., $\mathbb{H}(\theta,\tau)$ is positive definite; since it is a finite 
symmetric matrix, $\lambda_{\mathrm{min}}(\mathbb{H}(\theta,\tau)) > 0$.   Finally, to show the uniform boundedness, we only need to note that    $(\theta,\tau) \mapsto 
\lambda_{\mathrm{min}}(\mathbb{H}(\theta,\tau))$ is continuous on any compact set.
\end{proof}

\begin{remark}\label{rem:global_vs_local}
As in the continuous case, Condition~3$'$(i) and the full-support requirement in
Condition~3$'$(iii) (positivity of $f(\cdot, X_c \mid x_d)$ for all
$z\in\mathbb R$) supply the curvature condition used for \emph{global}
stability and are stronger than what point identification requires.  Together
with the remaining stochastic assumptions stated above, they permit the
continuous-regressor stability proof to be repeated conditional on the
discrete support points.
\end{remark}

\subsection{Convergence Rate of Raw Iterate $\widehat\theta_N$}
\label{appendix_b2}
This section gives the sharp rate and distribution of $\widehat\theta_N$, which will be frequently used in the proof of \autoref{theorem1}. 
Consider the following decomposition of $\mathbb{H}_0$: 
\[
 \mathbb H_0=\mathcal P_{\mathbb H_0}^{\top}
 \varLambda_{\mathbb H_0}\mathcal P_{\mathbb H_0},
 \qquad
 \mathcal P_{\mathbb H_0}\mathbb H_0
 \mathcal P_{\mathbb H_0}^{\top}=\varLambda_{\mathbb H_0}.
\]
 Choose a deterministic
$k_0$ so large that
$0<\gamma_k\lambda_{\max}(\mathbb H_0)<1$ for $k\geq k_0$, and define
\[
 A_{k_0-1}=\mathbb I_p,
 \qquad
 A_k=\prod_{\ell=k_0}^{k}(\mathbb I_p-\gamma_\ell\mathbb H_0)^{-1},
 \qquad
 \zeta_{0,k}=\Phi_k(\theta_0,\mathcal{W}_k)-\Phi_k(\theta_0).
\]
Define
\[
 \Sigma_N=\operatorname{Var}\!\left(
   \sum_{k=k_0}^{N}\gamma_k\mathcal P_{\mathbb H_0}A_k\zeta_{0,k}
 \right),
 \qquad
 s_{N,j}^2=e_j^{\top}\Sigma_Ne_j,
 \qquad D_N=\operatorname{Diag}(s_{N,1},\ldots,s_{N,p}).
\]
 We have the following results. 
\begin{theorem}\label{theorem2}
Let \autoref{condition1}--\ref{condition4} hold.  Then
for every fixed coordinate $j$,
\begin{equation}\label{lil_theta}
 \limsup_{N\to\infty}
 \frac{e_j^{\top}\mathcal P_{\mathbb H_0}A_N
       \Delta\widehat\theta_N}
      {\sqrt{2s_{N,j}^2\log\log(s_{N,j}^2)}}=1,
 \qquad
 \liminf_{N\to\infty}
 \frac{e_j^{\top}\mathcal P_{\mathbb H_0}A_N
       \Delta\widehat\theta_N}
      {\sqrt{2s_{N,j}^2\log\log(s_{N,j}^2)}}=-1
 \quad\text{a.s.}
\end{equation}
\end{theorem}

\autoref{theorem2} provides the sharp rate and coordinate-wise LIL under
Conditions~\ref{condition1}--\ref{condition4}.  In particular, \autoref{lemma4} gives, for every
fixed $j$,
\[
 s_{N,j}^2\asymp
 N^{-\alpha_\gamma+\alpha_h}
 \exp\!\left(\frac{2\gamma_0\lambda_j}{1-\alpha_\gamma}
 N^{1-\alpha_\gamma}\right),
 \qquad
 e_j^{\top}\mathcal P_{\mathbb H_0}A_N
 \mathcal P_{\mathbb H_0}^{\top}e_j
 \asymp
 \exp\!\left(\frac{\gamma_0\lambda_j}{1-\alpha_\gamma}
 N^{1-\alpha_\gamma}\right).
\]
Thus \eqref{lil_theta} implies
$\Vert\Delta\widehat\theta_N\Vert_2=
 O_{\mathrm{a.s.}}(N^{-(\alpha_\gamma-\alpha_h)/2}
 \sqrt{\log N})$.  Detailed supporting lemmas follow in
Section~\ref{appendix_b3}.

\subsubsection{Development of \autoref{theorem2}}\label{appendix_b3}

This section provides detailed development of \autoref{theorem2}. Define
\begin{align*}
        V_N = & \gamma_N (\Phi_N(\widehat\theta_{N-1}, \mathcal W_N) - \Phi_N(\widehat\theta_{N-1})) + \gamma_{N-1} (\mathbb I_p - \gamma_{N}\mathbb{H}_0)(\Phi_{N-1}(\widehat\theta_{N-2}, \mathcal W_{N-1}) - \Phi_{N-1}(\widehat\theta_{N-2}))\\
        & +\cdots +  \gamma_1 \left(\mathbb I_p - \gamma_{N}\mathbb{H}_0\right)\cdots \left(\mathbb I_p - \gamma_{2}\mathbb{H}_0\right)(\Phi_1(\widehat\theta_0, \mathcal W_1) - \Phi_1(\widehat\theta_0)).
        \end{align*}
    The following lemma describes a preliminary a.s. convergence rate of $\Vert V_N\Vert^2_2$.
\begin{lemma}\label{lemma3}
    Let \autoref{condition1}--\autoref{condition4} hold, then for any $r>1$ we have that 
        \[
       \Vert V_N \Vert_2 ^2 = o\left(N^{-2\alpha_{\gamma}+\alpha_h + 1}\log^r(N)\right), \  a.s.
        \]
\end{lemma}

\begin{proof}[Proof of \autoref{lemma3}] 
   According to \autoref{lemma2}, we have that $\lambda_{\mathrm{min}}(\mathbb{H}_0) >0$.  Note that 
    \[
    V_{N} = \gamma_{N}(\Phi_{N}(\widehat\theta_{N-1}, \mathcal W_{N}) - \Phi_{N}(\widehat\theta_{N-1})) + (\mathbb I_{p} - \gamma_{N}\mathbb{H}_0)V_{N-1}, 
    \]
    and $\mathbb E_{N-1} [V_{N-1}^{\top}(\Phi_{N}(\widehat\theta_{N-1}, \mathcal W_{N}) - \Phi_{N}(\widehat\theta_{N-1})) ]=0 $. 
    Then according to \autoref{lemma1}(3), we have that 
    \begin{align*}
   \mathbb  E_{N-1}\left\Vert V_N\right\Vert^2_2 & = \Vert\left(\mathbb I_p - \gamma_{N}\mathbb{H}_0\right)V_{N-1} \Vert^2_2 + \gamma_N^2 \mathbb E_{N-1}\Vert\Phi_{N}(\widehat\theta_{N-1}, W_{N}) - \Phi_{N}(\widehat\theta_{N-1})\Vert^2_2\\
    &\leq  \left(1 - C\gamma_N\right)\left\Vert V_{N-1}\right\Vert^2_2 + C\gamma_N^2h_N^{-1} \leq \left(1 - CN^{-\alpha_{\gamma}}\right)\left\Vert V_{N-1}\right\Vert^2_2 + CN^{-2\alpha_{\gamma} + \alpha_{h}}.
    \end{align*}
    Define $a_N = N^{2\alpha_{\gamma}-\alpha_h - 1}\log^{-r}(N)$ for any $r>1$, we have 
    \begin{align*}
    & \mathbb E_{N-1}\left[a_N\left\Vert V_N\right\Vert^2_2\right]  \leq \frac{a_N}{a_{N-1}} \left( 1 - CN^{-\alpha_{\gamma}}\right)\left[a_{N-1}\left\Vert V_{N-1}\right\Vert^2_2\right]  + C(N\log^r(N))^{-1}. 
    \end{align*}
  When $2\alpha_{\gamma} - \alpha_h - 1>0 $,   for $N$ sufficiently large we have  \[\frac{N^{2\alpha_{\gamma} - \alpha_h - 1}}{(N-1)^{2\alpha_{\gamma} - \alpha_h - 1}} = \left(1+ \frac{1}{N-1}\right)^{2\alpha_{\gamma}-\alpha_h - 1}\leq 1+\frac{C}{N},\] and \[\frac{\log^r(N)}{\log^{r}(N-1)}\leq \left(1 + \frac{1}{(N-1)\log(N-1)}\right)^r\leq 1+\frac{C}{N}.\] So 
  \[
  \frac{a_N}{a_{N-1}}  \left( 1- CN^{-\alpha_{\gamma}}\right) \leq \left( 1 +CN^{-1}\right)\left(1 - CN^{-\alpha_{\gamma}}\right) \leq 1 - CN^{-\alpha_{\gamma}}
  \]
  for $N$ sufficiently large, then we have that for $k$ sufficiently large, 
  \[
  \mathbb E_{N-1}\left[a_N\left\Vert V_N\right\Vert^2_2\right]  \leq  \left( 1 - CN^{-\alpha_{\gamma}}\right)\left[a_{N-1}\left\Vert V_{N-1}\right\Vert^2_2\right] +C(N\log^r(N))^{-1}.
  \]
  Since $r>1$,  $\sum_{N=1}^{\infty}(N\log^r(N))^{-1}<\infty$, thus we have that $a_N\left\Vert V_N\right\Vert^2_2$ a.s. converges to a finite random variable by Theorem 1 of \citet{robbins1971convergence}. Furthermore, we have that 
 $\sum_{N=1}^{\infty} N^{-\alpha_{\gamma}}\left[a_{N-1}\left\Vert V_{N-1}\right\Vert^2_2\right] < \infty $ a.s. holds.
  This shows that $a_{N-1}\left\Vert V_{N-1}\right\Vert^2_2  \rightarrow0$ a.s. holds, which proves the result. 
\end{proof}

As will be seen later, the convergence rate provided in \autoref{lemma3} is quite rough, but  it will be used to provide a preliminary convergence rate for $ \widehat\theta_N$. To do this, we first provide a technical lemma. 

\begin{lemma}\label{lemma4} Let $\frac{1}{2}<\alpha<1$ and $0<C_0<1$, then 
(i) the following exists \[\lim_{j\rightarrow \infty}\frac{\prod_{l=1}^j (1 - C_0l^{-\alpha})}{ \exp\left(-\frac{C_0}{1-\alpha}j^{1-\alpha}\right)};\] (ii)  
    further, let $\varsigma>0$, and let $b(x)$ be a strictly positive,
    continuously differentiable function for all sufficiently large $x$ such that
    $b(x)\exp(\frac{\varsigma C_0}{1-\alpha}x^{1-\alpha})$ is increasing for
    $x\geq x^*$, $\lim_{x\rightarrow \infty}b^{\prime}(x)x^{\alpha}/b(x)=0$,
    and $\int_{1}^{k} (x^{\alpha-1}b(x) + x^{\alpha}b^{\prime}(x))
    \exp(\frac{\varsigma C_0}{1-\alpha}x^{1-\alpha})dx\uparrow \infty$,
    then the following limit exists:
    \[
\lim_{k\rightarrow \infty} \frac{\sum_{j=1}^k b(j)\left(\prod_{l=1}^j (1 - C_0l^{-\alpha})\right)^{-\varsigma}}{ k^{\alpha}b(k)\exp\left(\frac{\varsigma C_0}{1-\alpha}k^{1-\alpha}\right)}. 
    \]
\end{lemma}

\begin{proof}[Proof of \autoref{lemma4}]
To show (i), note that  for $x\in(0,1)$, $\log(1-x)\leq -x$. So sequence 
\[
\log\left(\prod_{l=1}^j(1-C_0l^{-\alpha})\right) +  C_0\sum_{l=1}^j l^{-\alpha} = \sum_{l=1}^j \left(\log(1- C_0l^{-\alpha}) + C_0l^{-\alpha}\right)
\]
is non-increasing with respect to $j$. On the other side, Taylor expansion leads to 
\[
\left|\log\left(\prod_{l=1}^j(1-C_0l^{-\alpha})\right) +  C_0\sum_{l=1}^j l^{-\alpha} \right|\leq C\sum_{l=1}^j l^{-2\alpha},
\]
where the RHS is upper bounded uniformly w.r.t. $j$ because $2\alpha>1$. 
This implies that \[
\lim_{j\rightarrow \infty}\left[\log\left(\prod_{l=1}^j(1-C_0l^{-\alpha})\right) +  C_0\sum_{l=1}^j l^{-\alpha}\right]
\]
exists. Also note that $\sum_{l=1}^j l^{-\alpha} - \frac{1}{1-\alpha}j^{1-\alpha}$ is bounded and decreasing, so the limit also exists. This leads to that \[
\lim_{j\rightarrow \infty}\left[\log\left(\prod_{l=1}^j(1-C_0l^{-\alpha})\right) +  \frac{C_0}{1-\alpha}j^{1-\alpha}\right]
\]exists.  Denote such limit as $C_{1}$, we have that  
\[
 \lim_{j\rightarrow \infty}\exp\left(\log\left(\prod_{l=1}^j(1-C_0l^{-\alpha})\right) +  \frac{C_0}{1-\alpha}j^{1-\alpha}\right) = \exp(C_1).
\]
This proves the result.
 
 To prove (ii), note that from (i), for any fixed $\varsigma$, 
 \[\delta_j\equiv \frac{ \left(\prod_{l=1}^j (1 - C_0l^{-\alpha})\right)^{-\varsigma}}{  \exp\left(\frac{\varsigma C_0}{1-\alpha}j^{1-\alpha}\right)}\rightarrow \delta_{\infty}\]
 so 
 \[
 \sum_{j=1}^k b(j)\left(\prod_{l=1}^j (1 - C_0l^{-\alpha})\right)^{-\varsigma} = \delta_{\infty} \sum_{j=1}^k b(j)\exp\left(\frac{\varsigma C_0}{1-\alpha}j^{1-\alpha}\right) + \sum_{j=1}^k (\delta_j-\delta_{\infty})b(j)\exp\left(\frac{\varsigma C_0}{1-\alpha}j^{1-\alpha}\right).
 \]
 Since $b(x)\exp(\frac{\varsigma C_0}{1-\alpha}x^{1-\alpha})$ eventually increases for sufficiently large $x$, we have that 
 \[
 \frac{\sum_{j=1}^k (\delta_j-\delta_{\infty})b(j)\exp\left(\frac{\varsigma C_0}{1-\alpha}j^{1-\alpha}\right)}{\sum_{j=1}^k b(j)\exp\left(\frac{\varsigma C_0}{1-\alpha}j^{1-\alpha}\right)}\rightarrow 0,
 \]
and we only need to look at the ratio
\[\frac{\sum_{j=1}^k b(j)\exp\left(\frac{\varsigma C_0}{1-\alpha}j^{1-\alpha}\right)}{k^{\alpha}b(k)\exp(\frac{\varsigma C_0}{1-\alpha}k^{1-\alpha})}.\]  Note that for $k\geq x^*$, $b(x)\exp\left(\frac{\varsigma C_0}{1-\alpha}x^{1-\alpha}\right)$ is increasing, so 
\begin{align*}
\int_{[x^*]}^{k} b(x)\exp\left(\frac{\varsigma C_0}{1-\alpha}x^{1-\alpha}\right)dx & \leq \sum_{j=[x^*]+1}^k b(j)\exp\left(\frac{\varsigma C_0}{1-\alpha}j^{1-\alpha}\right)\\
& \leq \int_{[x^*]}^{k} b(x)\exp\left(\frac{\varsigma C_0}{1-\alpha}x^{1-\alpha}\right)dx + b(k)\exp\left(\frac{\varsigma C_0}{1-\alpha}k^{1-\alpha}\right).
\end{align*}
For $\int_{[x^*]}^{k} b(x)\exp(\frac{\varsigma C_0}{1-\alpha}x^{1-\alpha})dx$, integration by part leads to 
\begin{align*}
\int_{[x^*]}^{k} b(x)\exp\left(\frac{\varsigma C_0}{1-\alpha}x^{1-\alpha}\right)dx & = \int_{[x^*]}^{k} \frac{x^{\alpha}b(x)}{\varsigma C_0}\varsigma C_0 x^{-\alpha}\exp\left(\frac{\varsigma C_0}{1-\alpha}x^{1-\alpha}\right)dx\\
& = \frac{k^{\alpha}b(k)}{\varsigma C_0}\exp\left(\frac{\varsigma C_0}{1-\alpha}k^{1-\alpha}\right) - C(x^*)\\
&- \frac{1}{\varsigma C_0}\int_{[x^*]}^{k} (\alpha x^{\alpha-1}b(x) + x^{\alpha}b^{\prime}(x))\exp\left(\frac{\varsigma C_0}{1-\alpha}x^{1-\alpha}\right)dx,
\end{align*}
where $C(x^*)$ is a constant depending only on $x^*$. Using L'Hospital's rule, we have that 
\begin{align*}
& \lim_{k\rightarrow \infty} \frac{\int_{[x^*]}^{k} (\alpha x^{\alpha-1}b(x) + x^{\alpha}b^{\prime}(x))\exp\left(\frac{\varsigma C_0}{1-\alpha}x^{1-\alpha}\right)dx}{k^{\alpha}b(k)\exp(\frac{\varsigma C_0}{1-\alpha}k^{1-\alpha})} = \lim_{k\rightarrow \infty}\frac{ \alpha k^{\alpha-1}b(k) + k^{\alpha}b^{\prime}(k)}{ \alpha k^{\alpha-1}b(k) + k^{\alpha}b^{\prime}(k) + \varsigma C_0q(k) } = 0.
\end{align*}
So the following exists 
\[
\lim_{k\rightarrow \infty}\frac{\int_{[x^*]}^{k} b(x)\exp\left(\frac{\varsigma C_0}{1-\alpha}x^{1-\alpha}\right)dx}{ k^{\alpha}b(k)\exp\left(\frac{\varsigma C_0}{1-\alpha}k^{1-\alpha}\right)}.
\]
This shows the result.
\end{proof}

Combine \autoref{lemma3} and \autoref{lemma4}, we   can  immediately provide an a.s. convergence rate  for $\Vert\Delta\widehat\theta_N\Vert^2_2$. But to highlight how the convergence rate of $\Vert\Delta\widehat\theta_N\Vert^2_2$ can be accelerated as we will do later, we first provide the following general results. 

\begin{lemma}\label{lemma5}
    Let \autoref{condition1}--\autoref{condition4}  hold. Then for any $b(x)\downarrow 0$    satisfying all the requirements in \autoref{lemma4} with $\alpha = \alpha_{\gamma}$, if $\Vert V_N\Vert^2_2 = O(b(N))$ a.s. holds, there holds  
 \[ \Vert\Delta\widehat\theta_N \Vert^2_2 = O(b(N)\vee N^{-2s_{\mathcal K}\alpha_h}), \ a.s.\]
 \end{lemma}

 \begin{proof}[Proof of \autoref{lemma5}] 
Note that we can decompose the dynamics of $\widehat\theta_N$ as 
\[
\Delta\widehat\theta_{N} -V_N = \left(\mathbb{I}_p - \gamma_{N}\mathbb{H}_0\right)(\Delta\widehat\theta_{N-1} - V_{N-1}) + \gamma_N \delta_{1,N} + \gamma_N \delta_{2,N}
\]
where $\delta_{1,N} = \Phi_N(\widehat\theta_{N-1}) - \Phi(\widehat\theta_{N-1})$, and $\delta_{1,N} = \int_0^1(\mathbb{H}(\theta_0, \tau) - \mathbb{H}(\widehat\theta_{N-1}, \tau))\Delta\widehat\theta_{N-1}d\tau$. Note that $\gamma_N\Vert \delta_{1,N}\Vert_2 \leq CN^{-\alpha_{\gamma} - s_{\mathcal K}\alpha_h}$. Under \autoref{condition1},  we have that $\mathbb{H}(\theta, \tau)$ is Lipschitz with respect to $\theta$  uniformly for all $\tau$, that is, there exists a positive constant $C$ such that for any $\theta_1$ and $\theta_2$, there holds $\sup_{\tau\in[0,1]}\Vert \mathbb{H}(\theta_1, \tau) -\mathbb{H}(\theta_2, \tau) \Vert_F \leq C\Vert \theta_1 - \theta_2\Vert_2 $. So  $\Vert\delta_{2,N}\Vert_2\leq C\Vert \Delta\widehat\theta_{N-1}\Vert^2_2$. Define \begin{equation}
    \mathcal{A}_N = \Delta\widehat\theta_N - V_{N},
\end{equation} then $\Vert\delta_{2,N}\Vert_2\leq C \Vert \mathcal{A}_{N-1} + V_{N-1}\Vert^2_2\leq 2C\Vert \mathcal{A}_{N-1}\Vert^2_2  + 2C\Vert V_{N-1}\Vert^2_2$. Since $\Delta\widehat\theta_N$ and $V_N$ are both $o(1)$ a.s.,
$\mathcal A_N=o(1)$ a.s.  Let $c_0>0$ be the eventual linear contraction
constant of $I-\gamma_N\mathbb H_0$.  On each probability-one path choose
$N_0$ so large that
$2C\Vert\mathcal A_{N-1}\Vert_2\leq c_0/2$ for $N\geq N_0$.
Then the quadratic term satisfies
$2C\gamma_N\Vert\mathcal A_{N-1}\Vert_2^2
\leq(c_0/2)\gamma_N\Vert\mathcal A_{N-1}\Vert_2$ and is absorbed by the
linear contraction.  Thus, for $N$ sufficiently large on that path, 
\[
\Vert \mathcal{A}_{N}\Vert_2 \leq (1 - CN^{-\alpha_{\gamma}})\Vert \mathcal{A}_{N-1}\Vert_2 +  CN^{-\alpha_{\gamma} - s_{\mathcal K}\alpha_h} + CN^{-\alpha_{\gamma}} b(N)   
\]
because we assume that $\Vert V_{N-1}\Vert_2^2 = O(b(N))$ a.s. holds. 
Then for almost all paths, there exists a (path-specific) $N_0$, such that 
\begin{align*}
\Vert \mathcal{A}_N\Vert_2 &  \leq \left(\prod_{k=N_0}^N (1 - Ck^{-\alpha_{\gamma}})\right)C\sum_{k=N_0}^N k^{-\alpha_{\gamma}}(k^{-s_{\mathcal K}\alpha_{h}} +b(k))\left(\prod_{j=N_0}^k (1 - Cj^{-\alpha_{\gamma}})\right)^{-1}\\
& + \left(\prod_{k=N_0+1}^N (1 - Ck^{-\alpha_{\gamma}})\right)\Vert \mathcal{A}_{N_0}\Vert_2. 
\end{align*}
So $\Vert \mathcal{A}_N\Vert_2 = O(N^{- s_{\mathcal K}\alpha_h}\vee b(N))$ a.s. holds according to \autoref{lemma4}. Finally, note that $\Vert\Delta\widehat\theta_N\Vert _2\leq \Vert \mathcal{A}_N\Vert_2  +\Vert V_{N} \Vert_2 = O(N^{-s_{\mathcal K}\alpha_h}\vee \sqrt{b(N)})$ a.s. holds. This proves the result. 
\end{proof}

Obviously, if we specify $b(N) = N^{-2\alpha_{\gamma} + \alpha_h +1}\log^r(N)$, then  \autoref{lemma3} and \autoref{lemma5} immediately lead to that $\Vert\Delta\widehat\theta_N \Vert^2_2 = O(N^{-2\alpha_{\gamma}+\alpha_h + 1 }\log^r(N))$ a.s. for any $r>1$. Again,  we point out that this rate is slow.  However, as we have demonstrated before, such preliminary rate can be used to build the sharp rate for $\Delta\widehat\theta_N$. To achieve sharp rate, we  need a more delicate decomposition for $V_{N}$. In particular, decompose $V_{N} = V_{1,N} + V_{2,N}$, where 
    \begin{align*}
        V_{1,N} = & \gamma_N \left(\Phi_N(\theta_{0}, \mathcal W_N) - \Phi_N(\theta_{0})\right) + \gamma_{N-1} \left(\mathbb I_p - \gamma_{N}\mathbb{H}_0\right)\left(\Phi_{N-1}(\theta_{0}, \mathcal W_{N-1}) - \Phi_{N-1}(\theta_{0})\right)+\\
        & \cdots +  \gamma_1 \left(\mathbb{I}_p - \gamma_{N}\mathbb{H}_0\right)\cdots \left(\mathbb I_p - \gamma_{2}\mathbb{H}_0\right)\left(\Phi_1(\theta_{0}, \mathcal W_1) - \Phi_1(\theta_{0})\right).
        \end{align*}
and 
\begin{align*}
        V_{2,N} = & \gamma_N\int_{0}^{1} \partial_{\theta}\left(\Phi_N(\theta_{0} + \tau\Delta\widehat\theta_{N-1}, \mathcal  W_N) - \Phi_N(\theta_{0} + \tau\Delta\widehat\theta_{N-1})\right)d\tau \Delta\widehat\theta_{N-1}\\
        & + \gamma_{N-1} \left(\mathbb I_p - \gamma_{N}\mathbb{H}_0\right)\int_{0}^{1} \partial_{\theta}\left(\Phi_{N-1}(\theta_{0} + \tau\Delta\widehat\theta_{N-2}, \mathcal  W_{N-1}) - \Phi_{N-1}(\theta_{0} + \tau\Delta\widehat\theta_{N-2})\right)d\tau \Delta\widehat\theta_{N-2}\\
        & +\cdots +  \gamma_1 \left(\mathbb I_p - \gamma_{N}\mathbb{H}_0\right)\cdots \left(\mathbb I_p - \gamma_{2}\mathbb{H}_0\right) \int_{0}^{1} \partial_{\theta}\left(\Phi_1(\theta_{0} + \tau\Delta\widehat\theta_0, \mathcal  W_1) - \Phi_1(\theta_{0} + \tau\Delta\widehat\theta_0)\right)d\tau \Delta\widehat\theta_0.
        \end{align*}
We present the asymptotic behaviors of $V_{1,N}$ and $V_{2,N}$.  

\begin{lemma}\label{lemma6}
 Let \autoref{condition1}--\autoref{condition4} hold.  With $A_N$,
 $\Sigma_N$, and $s_{N,j}^2$ defined above, for every fixed
 $j=1,\ldots,p$,
 \begin{equation*}
  \limsup_{N\to\infty}
  \frac{e_j^{\top}\mathcal P_{\mathbb H_0}A_NV_{1,N}}
       {\sqrt{2s_{N,j}^2\log\log(s_{N,j}^2)}}=1,
  \qquad
  \liminf_{N\to\infty}
  \frac{e_j^{\top}\mathcal P_{\mathbb H_0}A_NV_{1,N}}
       {\sqrt{2s_{N,j}^2\log\log(s_{N,j}^2)}}=-1
  \quad\text{a.s.}
 \end{equation*}
\end{lemma}

\begin{proof}[Proof of \autoref{lemma6}]
Choose $k_0$ as above.  According to \autoref{lemma1},
\[ \mathbb E\left[\mathcal{P}_{\mathbb{H}_0}\left(\Phi_N(\theta_{0}, W_N) - \Phi_N(\theta_{0})\right)\left(\Phi_N(\theta_{0}, W_N) - \Phi_N(\theta_{0})\right)^{\top}\mathcal{P}_{\mathbb{H}_0}^{\top}\right] = \frac{2\mathcal{P}_{\mathbb{H}_0}\mathcal{V}_{\mathcal{K},0}\mathcal{P}_{\mathbb{H}_0}^{\top}}{B(B-1)h_N} + O(1).\] 
Recall that
$\mathbb H_0=\mathcal P_{\mathbb H_0}^{\top}
\varLambda_{\mathbb H_0}\mathcal P_{\mathbb H_0}$.
For coordinate $j$, let $\lambda_j$ be the corresponding diagonal element,
let $\widetilde e_{k,j}$ be coordinate $j$ of
$\mathcal P_{\mathbb H_0}\zeta_{0,k}$, and put
\[
 a_{k,j}=\prod_{\ell=k_0}^{k}(1-\lambda_j\gamma_\ell)^{-1},
 \qquad e_{k,j}=\gamma_k a_{k,j}\widetilde e_{k,j}.
\]
Cancellation of the tail factors gives the exact identity
\[
 e_j^\top\mathcal P_{\mathbb H_0}A_NV_{1,N}
 =C_{j,k_0-1}+\sum_{k=k_0}^{N}e_{k,j},
\]
where $C_{j,k_0-1}$ is the fixed contribution of the finitely many innovations
before $k_0$.  
We give the scalar calculation for $j=1$ and write $e_k=e_{k,1}$ and
$s_N^2=\sum_{k=k_0}^N\mathbb E e_k^2$.  Then
\[\mathbb Ee_k^2 = \frac{\gamma_k^2}{\left[\prod_{j=k_0}^k (1- \lambda_{\mathbb{H}_0,1}\gamma_j)\right]^{2}} \left\{ \frac{2[\mathcal{P}_{\mathbb{H}_0}\mathcal{V}_{\mathcal{K},0}\mathcal{P}_{\mathbb{H}_0}^{\top}](1,1)}{B(B-1)h_k}+  O(1)\right\}\]
 where $[\mathcal{P}_{\mathbb{H}_0}\mathcal{V}_{\mathcal{K},0}\mathcal{P}_{\mathbb{H}_0}^{\top}](1,1)$ refers to the element in first column and first row of $\mathcal{P}_{\mathbb{H}_0}\mathcal{V}_{\mathcal{K},0}\mathcal{P}_{\mathbb{H}_0}^{\top}$.  This leads to that for any $N$, 
\[
\left|s_N^2 -\sum_{k=k_0}^N \frac{ 2[\mathcal{P}_{\mathbb{H}_0}\mathcal{V}_{\mathcal{K},0}\mathcal{P}_{\mathbb{H}_0}^{\top}](1,1)\gamma_0^2h_0^{-1} k^{-2\alpha_{\gamma}+ \alpha_{h}}}{B(B-1)\left[\prod_{j=k_0}^k (1-\gamma_0\lambda_{\mathbb{H}_0,1}j^{-\alpha_{\gamma}})\right]^2}\right|\leq C\sum_{k=k_0}^N \frac{  k^{-2\alpha_{\gamma}}}{\left[\prod_{j=k_0}^k (1-\gamma_0\lambda_{\mathbb{H}_0,1}j^{-\alpha_{\gamma}})\right]^2}.
\]
Using \autoref{lemma4},  we can show that for $N$ sufficiently large, we have that \[
s_N^{-2} \sum_{k=k_0}^N \frac{ 2[\mathcal{P}_{\mathbb{H}_0}\mathcal{V}_{\mathcal{K},0}\mathcal{P}_{\mathbb{H}_0}^{\top}](1,1)\gamma_0^2h_0^{-1} k^{-2\alpha_{\gamma}+ \alpha_{h}}}{B(B-1)\left[\prod_{j=k_0}^k (1-\gamma_0\lambda_{\mathbb{H}_0,1}j^{-\alpha_{\gamma}})\right]^2}\rightarrow 1
\]
and
\[
0<C_1\leq \frac{s_N^2}{ N^{-\alpha_{\gamma}+ \alpha_{h}}\exp(C^*N^{1 - \alpha_{\gamma}}) } \leq C_2<
\infty\] for some positive constant $C^*$. So given the choice of $\alpha_{\gamma}$ and $\alpha_{h}$, we have that $s_N^2 \rightarrow \infty$. 

Moreover, we can similarly show that 
\[
\frac{\log^2(\log(s_N^2))\mathbb{E}e_N^4}{s_N^4}\leq \frac{C\log^2(N)N^{-4\alpha_{\gamma}+3\alpha_h}\exp(2C^{*}N^{1-\alpha_{\gamma}})}{N^{-2\alpha_{\gamma}+2\alpha_{h}}\exp(2C^{*}N^{1 - \alpha_{\gamma}})}= C\log^2(N)N^{-2\alpha_{\gamma}+\alpha_h}.
\]
This implies that 
$
\sum_{N\geq \max\{k_0,3\}}s_N^{-4}\log^2(\log(s_N^2))\mathbb Ee_N^4<\infty
$ because $2\alpha_{\gamma} - \alpha_h > 1$, 
so \[\left|s_N^{-1}\sqrt{\log(\log(s_N^2))}e_N\right|<C, \ a.s. \]holds for any $C>0$, implying that $s_N^{-1}\sqrt{\log(\log(s_N^2))}e_N\rightarrow 0$ a.s. holds. Using the LIL from Theorem 1.1 of  \citet{tomkins1983lindeberg}, we have that 
\[
H_0\leq \overline{\lim}_{N\rightarrow \infty}\frac{\sum_{k=k_0}^N e_k }{\sqrt{2s_N^2\log(\log(s_N^2))}}\leq H_1, \  a.s.,
\]
and
\[
-H_1^-\leq \underline{\lim}_{N\rightarrow \infty}\frac{\sum_{k=k_0}^N e_k }{\sqrt{2s_N^2\log(\log(s_N^2))}}\leq -H_0^-, \  a.s.,
\]
for some $0\leq H_0\leq H_1\leq1$ and $0\leq H_0^-\leq H_1^-\leq 1$.

We next show that $H_0 = H_1 = 1$; the proof that
$H_0^- = H_1^- = 1$ is identical. To show the results, following \citet{tomkins1983lindeberg} we define
$
H_N(x) = s_N^{-2}\sum_{k=k_0}^N \mathbb E\left(e_k^2 \boldsymbol{1}(|e_k|\leq xs_kt_k^{-1})\right)
$
with $t_k = \sqrt{2\log(\log(s_k^2))}$. For any $x$, we have that 
$
H_N(x) = 1 - s_N^{-2}\sum_{k=k_0}^N \mathbb E\left(e_k^2 \boldsymbol{1}(|e_k|>xs_kt_k^{-1})\right),
$
where \[\mathbb E\left(e_k^2 \boldsymbol{1}(|e_k|>xs_kt_k^{-1})\right)\leq \sqrt{\mathbb Ee_k^4}\sqrt{\mathbb E\boldsymbol{1}(|e_k|>xs_kt_k^{-1})}.\]  According to our previous analysis, we know that \[\sqrt{\mathbb Ee_k^4}\leq Ck^{-2\alpha_{\gamma}+3/2\alpha_h}\exp(C^*k^{1-\alpha_{\gamma}}),\] \[\sqrt{\mathbb E\boldsymbol{1}(|e_k|>xs_kt_k^{-1})}\leq C\sqrt{\frac{t_k^4\mathbb E e_k^4}{x^4s_k^4}}\leq Cx^{-2}\log (k)k^{-\alpha_{\gamma}+\alpha_h/2}\]
So $\sqrt{\mathbb Ee_k^4}\sqrt{\mathbb E\boldsymbol{1}(|e_k|>xs_kt_k^{-1})}\leq Cx^{-2}\log(k)k^{-3\alpha_{\gamma} + 2\alpha_h}\exp(C^*k^{1-\alpha_{\gamma}})$. Then
\begin{align*}
\sum_{k=k_0}^N \mathbb E\left(e_k^2 \boldsymbol{1}(|e_k|>xs_kt_k^{-1})\right) & \leq Cx^{-2}\sum_{k=k_0}^N \log(k)k^{-3\alpha_{\gamma} + 2\alpha_h}\exp(C^*k^{1-\alpha_{\gamma}})\\
& \leq Cx^{-2}\log(N)N^{-2\alpha_{\gamma} + 2\alpha_h}\exp(C^*N^{1-\alpha_{\gamma}}).
\end{align*}
Since $s_N^2 \geq C N^{-\alpha_{\gamma} + \alpha_h}\exp(C^*N^{1-\alpha_{\gamma}})$, we have that
$s_N^{-2}\sum_{k=k_0}^N \mathbb E\{e_k^2
\boldsymbol{1}(|e_k|>xs_kt_k^{-1})\}\rightarrow 0$, and consequently,
$
\lim_{N\rightarrow \infty} H_N(x) = 1. 
$
According to \citet{tomkins1983lindeberg}, we have that 
\[
H_0 = \lim\inf_{N\rightarrow\infty} H_N(x), \ \ H_1 = \lim\sup_{N\rightarrow\infty} H_N(x),
\]
This leads to $H_0 = H_1 = 1$, and similarly $H_0^- = H_1^- = 1$. Based on the above analysis, we have that
\[
\limsup_{N\to\infty}\frac{\sum_{k=k_0}^N e_k}
 {\sqrt{2s_N^2\log\log(s_N^2)}}=1,
 \qquad
\liminf_{N\to\infty}\frac{\sum_{k=k_0}^N e_k}
 {\sqrt{2s_N^2\log\log(s_N^2)}}=-1,
 \quad\text{a.s.}
\]
The same argument applies to every coordinate and proves the coordinatewise
LIL using only \autoref{condition1}--\autoref{condition4}.  
\end{proof}

Next we introduce the lemma that clearly demonstrates how the convergence of $\Vert \Delta\widehat\theta_N\Vert ^2$  accelerates the convergence rate of $\Vert V_{2,N}\Vert^2$, which, in turn,  accelerates the convergence rate of $\Vert \Delta\widehat\theta_N\Vert ^2$ itself. 

\begin{lemma}\label{lemma7}
    Let \autoref{condition1}--\autoref{condition4} hold.  If
    $\Vert \Delta\widehat\theta_N\Vert^2 = O(b(N))$ a.s. for some
    $b(N)\downarrow0$ with $b(N-1)/b(N)\leq1+C/N$, then, for every $r>1$,
    \[
    \Vert  V_{2,N} \Vert^2_2 = o\left(b(N)N^{-2\alpha_{\gamma}+3\alpha_h +1}\log^r(N)\right), \ a.s. 
    \]
\end{lemma}

\begin{proof}[Proof of \autoref{lemma7}] 
    Note that 
    \[
    V_{2,N} = \gamma_N\int_{0}^{1} \partial_{\theta}\left(\Phi_N(\theta_{0} + \tau\Delta\widehat\theta_{N-1}, W_N) - \Phi_N(\theta_{0} + \tau\Delta\widehat\theta_{N-1})\right)d\tau \Delta\widehat\theta_{N-1} + \left(\mathbb I_{p} - \gamma_N \mathbb{H}_0\right)V_{2, N-1}. 
    \]
    Then according to \autoref{lemma1}(4), 
    \[
    \mathbb E_{N-1}\Vert V_{2,N}\Vert ^2_2 \leq \left(1 - C\gamma_N\right)\Vert V_{2,N-1}\Vert ^2 _2+ C\gamma_N^2 h_N^{-3}\Vert\Delta\widehat\theta_{N-1}\Vert^2_2,
    \]
    and 
     \begin{align*}
   & \mathbb E_{N-1}\left[b^{-1}(N)N^{ 2\alpha_{\gamma} - 3\alpha_h - 1}\log^{-r}(N)\Vert V_{2,N}\Vert ^2_2\right] \\
   & \leq \left(1 - C N^{-\alpha_{\gamma}}\right)\frac{b(N-1)\log^r(N-1)}{b(N)\log^r(N)}\left[b^{-1}(N-1)(N-1)^{2\alpha_{\gamma} - 3\alpha_h - 1}\log^{-r}(N-1)\Vert V_{2,N-1}\Vert ^2_2\right] \\
    & + Cb(N-1) b^{-1}(N)(N\log^r(N))^{-1} b^{-1}(N-1)\Vert\Delta\widehat\theta_{N-1}\Vert^2_2\\
    & \leq \left(1 - CN^{-\alpha_{\gamma}}\right)\left[b ^{-1}(N-1)(N-1)^{2\alpha_{\gamma} - 3\alpha_h - 1}\log^{-r}(N-1)\Vert V_{2,N-1}\Vert ^2_2\right] \\
    &  + Cb(N-1) b^{-1}(N)(N\log^r(N))^{-1} b^{-1}(N-1)\Vert\Delta\widehat\theta_{N-1}\Vert^2_2.
    \end{align*}
  Since $\Vert \Delta\widehat\theta_N\Vert^2_2 = O(b(N))$ a.s. and $b(N-1)/b(N)\leq C$, we have that 
  \[
  \sum_{N=1}^{\infty} b(N-1) b^{-1}(N)(N\log^r(N))^{-1} b^{-1}(N-1)\Vert\Delta\widehat\theta_{N-1}\Vert^2_2 <\infty, \ a.s.
  \]
Using the method that we use in the proof of \autoref{lemma3}, we have that 
\[
b^{-1}(N)N^{ 2\alpha_{\gamma} - 3\alpha_h - 1}\log^{-r}(N)\Vert V_{2,N}\Vert_2 ^2 =o(1), \ a.s.
\]
So $\Vert V_{2,N}\Vert^2 _2= o\left( b(N)N^{-2\alpha_{\gamma} + 3\alpha_h + 1}\log^r(N)\right)$ a.s. holds. 
\end{proof}

\subsubsection{Proof of \autoref{theorem2}}

Define 
$
 d=2\alpha_\gamma-3\alpha_h-1>0,
  \delta_0=2\alpha_\gamma-\alpha_h-1,
  \delta_\star=\alpha_\gamma-\alpha_h,
  \delta_{m+1}=\min\{\delta_\star,\delta_m+d\}.
$ 
The preliminary bounds in \autoref{lemma3} and \autoref{lemma5} give the
exponent $\delta_0$.  If at some stage
$\Vert\Delta\widehat\theta_N\Vert_2^2=O_{\mathrm{a.s.}}(N^{-\delta_m}
\operatorname{polylog}N)$, \autoref{lemma7} improves the nonlinear martingale
by $d$, while \autoref{lemma5} transfers that improvement back to the iterate.
Because $d>0$, after finitely many repetitions the recursion reaches
$\delta_\star$.  If the final exponent step hits $\delta_\star$ exactly, the
bound can still carry an arbitrary fixed polylogarithmic factor.  Apply
\autoref{lemma7} once more to that bound.  Its factor $N^{-d}$ makes the
nonlinear martingale polynomially smaller than $N^{-\delta_\star}$, so
\autoref{lemma6}, \autoref{lemma5}, and the deterministic bias bound yield
\[
 \Vert\Delta\widehat\theta_N\Vert_2^2=O_{\rm a.s.}(b_N),
 \qquad
 b_N=N^{-\delta_\star}\log N
     =N^{-\alpha_\gamma+\alpha_h}\log N.
\]
A final application of \autoref{lemma7} yields
$\Vert V_{2,N}\Vert_2^2
=o_{\mathrm{a.s.}}(b_NN^{-d}(\log N)^r)$ for every $r>1$; in particular,
$\Vert V_{2,N}\Vert_2=o_{\mathrm{a.s.}}(N^{-\delta_\star/2})$ because
$d>0$.  The deterministic
linearization remainder from \autoref{lemma5} satisfies
\[
 \Vert\mathcal A_N\Vert_2
 =O_{\mathrm{a.s.}}(b_N\vee N^{-s_{\mathcal K}\alpha_h})
 =o_{\mathrm{a.s.}}(N^{-\delta_\star/2}),
\]
where the last equality follows from \autoref{condition4}.  Hence
\[
 \Delta\widehat\theta_N=V_{1,N}
 +o_{\mathrm{a.s.}}(N^{-\delta_\star/2}).
\]
By the coordinatewise LIL in  \autoref{lemma6} and the fact that $p$ is fixed,
\[
 \Vert V_{1,N}\Vert_2
 =O_{\mathrm{a.s.}}\!\left(
 N^{-\delta_\star/2}\sqrt{\log N}\right).
\]
Together with the preceding display, this proves the displayed raw-iterate
rate.

For the linear recursion, commutativity of the factors gives
\[
 \mathcal P_{\mathbb H_0}A_NV_{1,N}
 =C_{k_0-1}+\sum_{k=k_0}^{N}
 \gamma_k\mathcal P_{\mathbb H_0}A_k\zeta_{0,k},
\]
where $C_{k_0-1}$ is a fixed finite prefix.   \autoref{lemma6} proves from
 \autoref{condition1}--\autoref{condition4} that each $s_{N,j}$ diverges
and that the independent triangular array satisfies the scalar LIL.  Thus the
prefix is negligible coordinatewise and \eqref{lil_theta} follows.  The same
scalar variance estimates give
\[
 \frac{\|e_j^{\top}\mathcal P_{\mathbb H_0}A_N\|_2}{s_{N,j}}
 =O(N^{\delta_\star/2}),
\]
so the $o_{\rm a.s.}(N^{-\delta_\star/2})$ nonlinear and deterministic
remainders are negligible under every coordinate LIL normalization.

\subsection{Proof of \autoref{theorem1}}

The raw kernel-based updates lead to $\widehat\theta_{k} = \widehat\theta_{k-1} + \gamma_k \Phi_k(\widehat\theta_{k-1}, \mathcal W_k)$, then we can decompose the update as follows
\[
\widehat\theta_{k} = \widehat\theta_{k-1} + \gamma_k\Phi(\widehat\theta_{k-1}) + \gamma_k(\Phi_k(\widehat\theta_{k-1}) - \Phi(\widehat\theta_{k-1})) + \gamma_k(\Phi_k(\widehat\theta_{k-1}, \mathcal  W_k)- \Phi_k(\widehat\theta_{k-1})). 
\]
Simple calculation leads to 
\begin{align}
\Vert\Delta\widehat\theta_{k}\Vert^2_2 & = \Vert\Delta\widehat\theta_{k-1} + \gamma_k\Phi(\widehat\theta_{k-1})\Vert^2_2   + \gamma_k^2 \Vert \Phi_k(\widehat\theta_{k-1}) - \Phi(\widehat\theta_{k-1})\Vert^2_2+ \gamma_k^2 \Vert \Phi_k(\widehat\theta_{k-1}, \mathcal W_k)- \Phi_k(\widehat\theta_{k-1})\Vert_2 ^2\nonumber\\
&  + 2\gamma_k (\Delta\widehat\theta_{k-1} + \gamma_k\Phi(\widehat\theta_{k-1}))^{\top}(\Phi_k(\widehat\theta_{k-1}) - \Phi(\widehat\theta_{k-1}))\nonumber \\
& + 2\gamma_k (\Delta\widehat\theta_{k-1} + \gamma_k\Phi(\widehat\theta_{k-1}))^{\top}(\Phi_k(\widehat\theta_{k-1}, \mathcal  W_k)- \Phi_k(\widehat\theta_{k-1}))\nonumber\\
&  + 2\gamma_k^2(\Phi_k(\widehat\theta_{k-1}) - \Phi(\widehat\theta_{k-1}))^{\top}(\Phi_k(\widehat\theta_{k-1},\mathcal  W_k)- \Phi_k(\widehat\theta_{k-1})) \label{dynamics}. 
\end{align}
Recall that we use $\mathbb E_{k-1}$ to denote the expectation conditioned on the first $k-1$ batches (and let $\mathbb E_0$ denote the unconditional expectation). We verify the conditional expectation for all the terms on the right-hand side of (\ref{dynamics}).  The population Hessian is uniformly bounded above under \autoref{condition1}; hence choose a deterministic $k_0$ such that
$\gamma_k\sup_{\theta,\tau}\lambda_{\max}\{\mathbb H(\theta,\tau)\}\leq1$
for $k\geq k_0$.  All contraction displays below are asserted for
$k\geq k_0$; the finite prefix changes only constants. Since $\Phi(\theta) = -\int_{0}^{1} \mathbb{H}(\theta,\tau)d\tau \Delta\theta$ for any $\theta$, we have that 
$
\Phi(\widehat\theta_{k-1}) = -\int_{0}^1 \mathbb{H}(\widehat\theta_{k-1}, \tau)d\tau \Delta\widehat\theta_{k-1}
    $. So \autoref{lemma2} leads to 
    \begin{align*}
        \mathbb{E}_{k-1}\Vert\Delta\widehat\theta_{k-1} + \gamma_k\Phi(\widehat\theta_{k-1})\Vert^2_2 & = \Delta\widehat\theta_{k-1}^{\top}\left(\mathbb I_p - \gamma_k \int_{0}^1 \mathbb{H}(\widehat\theta_{k-1}, \tau)d\tau\right)^2 \Delta\widehat\theta_{k-1} \\
        & \leq \left(1 -  C\gamma_k\left(\frac{\overline{z} - \underline{z}}{12(1 + \Vert\Delta\widehat\theta_{k-1}\Vert_2 + \Vert\theta_0\Vert_2)}\wedge r_X\right)^{2p+2} \right)\Vert \Delta\widehat\theta_{k-1}\Vert_2^2.  
    \end{align*}
 \autoref{lemma1}(2) leads to $\mathbb E_{k-1} \gamma_k^2 \Vert\Phi_k(\widehat\theta_{k-1}) -\Phi(\widehat\theta_{k-1}) \Vert_2^2\leq C\gamma_k^2h_k^{2s_{\mathcal K}}$, and  \autoref{lemma1}(3) leads to  $\gamma_k^2 \mathbb E_{k-1} \Vert \Phi_k(\widehat\theta_{k-1}, W_k)- \Phi_k(\widehat\theta_{k-1})\Vert _2^2\leq C\gamma_k^2h_k^{-1}$. Moreover, we have that 
\begin{align*}
& \left| \mathbb{E}_{k-1}\left[2\gamma_k (\Delta\widehat\theta_{k-1} + \gamma_k\Phi(\widehat\theta_{k-1}))^{\top}(\Phi_k(\widehat\theta_{k-1}) - \Phi(\widehat\theta_{k-1}))\right]\right|\\
& \leq  Ch_k^{2s_{\mathcal K}} \left(1 -  C\gamma_k\left(\frac{\overline{z} - \underline{z}}{12(1 + \Vert\Delta\widehat\theta_{k-1}\Vert_2 + \Vert\theta_0\Vert_2)}\wedge r_X\right)^{2p+2} \right)\Vert \Delta\widehat\theta_{k-1}\Vert^2_2 + C\gamma_k^2,
\end{align*}
\[
\mathbb{E}_{k-1}\left[2\gamma_k (\Delta\widehat\theta_{k-1} + \gamma_k\Phi(\widehat\theta_{k-1}))^{\top}(\Phi_k(\widehat\theta_{k-1}, W_k)- \Phi_k(\widehat\theta_{k-1}))\right] = 0,
\]
and 
\begin{align*}
\mathbb{E}_{k-1}\left[2\gamma_k^2(\Phi_k(\widehat\theta_{k-1}) - \Phi(\widehat\theta_{k-1}))^{\top}(\Phi_k(\widehat\theta_{k-1}, W_k)- \Phi_k(\widehat\theta_{k-1}))\right] = 0.
\end{align*}
  The above together leads to 
  \begin{align*}
 \mathbb E_{k-1} \Vert\Delta\widehat\theta_{k} \Vert^2_2 & \leq \left(1 + Ch_k^{2s_{\mathcal K}}\right)\left\Vert\Delta\widehat\theta_{k-1}\right\Vert^2_2  + C\gamma_k^2h_k^{-1} 
  \\
  & -C\gamma_k  \left(\frac{\overline{z} - \underline{z}}{12(1 + \Vert\Delta\widehat\theta_{k-1}\Vert_2 + \Vert\theta_0\Vert_2)}\wedge r_X\right)^{2p+2}\Vert\Delta\widehat\theta_{k-1}\Vert^2_2.
  \end{align*}
  Then since $\sum_{k=1}^{\infty}h_k^{2s_{\mathcal K}}<\infty$ and $\sum_{k=1}^{\infty} \gamma_{k}^{2}h_k^{-1}<\infty$, 
  according to Theorem 1 of \citet{robbins1971convergence}, we have that $\Vert\Delta\widehat\theta_{k-1}\Vert^2_2$ converges a.s. to a finite random variable, and that 
  \[
  \sum_{k=1}^{\infty}\gamma_k \left(\frac{\overline{z} - \underline{z}}{12(1 + \Vert\Delta\widehat\theta_{k-1}\Vert + \Vert\theta_0\Vert)}\wedge r_X\right)^{2p+2}\Vert\Delta\widehat\theta_{k-1}\Vert^2<\infty, \qquad \text{a.s.} 
  \]
  Since $\sum_{k=1}^{\infty}\gamma_k = \infty$, if $\Vert\Delta\widehat\theta_{k-1}\Vert^2_2$ converges but not to zero, then \[
  \sum_{k=1}^{\infty}\gamma_k \left(\frac{\overline{z} - \underline{z}}{12(1 + \Vert\Delta\widehat\theta_{k-1}\Vert_2 + \Vert\theta_0\Vert_2)}\wedge r_X\right)^{2p+2}\Vert\Delta\widehat\theta_{k-1}\Vert^2_2  = + \infty,
  \]
  which leads to a contradiction. This implies that $\Vert\Delta\widehat\theta_k\Vert^2_2\rightarrow0$ must a.s. hold. This proves the convergence of $\widehat\theta_k$ and $\overline{\theta}_k$.

We further decompose the PR average. Note that 
\begin{align*}
\Delta\widehat\theta_N & = \Delta\widehat\theta_{N-1} +  \gamma_N \Phi_N(\widehat\theta_{N-1}, \mathcal W_N)\\
& =  \Delta\widehat\theta_{N-1} - \gamma_N\mathbb{H}_0 \Delta\widehat\theta_{N-1} + \gamma_N\left[\delta_{1,N} +\delta_{2,N} + \Phi_N(\widehat\theta_{N-1}, \mathcal W_N)-\Phi_N(\widehat\theta_{N-1})\right],
\end{align*}
so
\[
\mathbb{H}_0 \Delta\widehat\theta_{N-1} = \gamma_N^{-1}(\Delta\widehat\theta_{N-1} - \Delta\widehat\theta_N)  +  \left[ \delta_{1,N} + \delta_{2,N} + \Phi_{N}(\widehat\theta_{N-1},\mathcal W_N) -\Phi_{N}(\widehat\theta_{N-1})\right], 
\]
where recall that $\delta_{1,N} = \Phi_{N}(\widehat\theta_{N-1}) -\Phi(\widehat\theta_{N-1})$  and $\delta_{2,N} = \int_0^1(\mathbb{H}(\theta_0, \tau) - \mathbb{H}(\widehat\theta_{N-1}, \tau))d\tau\Delta\widehat\theta_{N-1}$. 
Then 
\[
\sum_{k=1}^N \mathbb{H}_0 \Delta\widehat\theta_{k-1} = \sum_{k=1}^N \gamma_{k}^{-1}(\Delta\widehat\theta_{k-1} - \Delta\widehat\theta_k)  + \sum_{k=1}^N  \left[\delta_{1,N} + \delta_{2,N} +\Phi_{k}(\widehat\theta_{k-1}, \mathcal W_k) -\Phi_{k}(\widehat\theta_{k-1}) \right]
\]
Because $\overline\theta_N=N^{-1}\sum_{k=1}^N\widehat\theta_k$, 
\[
 N\mathbb H_0\Delta\overline\theta_N
 =\sum_{k=1}^N\mathbb H_0\Delta\widehat\theta_{k-1}
   +\mathbb H_0(\Delta\widehat\theta_N-\Delta\widehat\theta_0).
\]
The last boundary term is $O(1)$ a.s. after consistency and hence is
$o(N^{(1+\alpha_h)/2})$ a.s.
For the first term, we have that 
\begin{align*}
\sum_{k=1}^N \gamma_{k}^{-1}(\Delta\widehat\theta_{k-1} - \Delta\widehat\theta_k) & = \sum_{k=1}^N \gamma_{k}^{-1}\Delta\widehat\theta_{k-1} - \sum_{k=1}^N \gamma_{k}^{-1}\Delta\widehat\theta_k= \sum_{k=0}^{N-1} \gamma_{k+1}^{-1}\Delta\widehat\theta_{k} - \sum_{k=1}^N \gamma_{k}^{-1}\Delta\widehat\theta_k\\
&  = \gamma_1^{-1}\Delta\widehat\theta_0 + \sum_{k=1}^{N-1}( \gamma_{k+1}^{-1}-\gamma_k^{-1})\Delta\widehat\theta_{k} - \gamma_N^{-1}\Delta\widehat\theta_N  
\end{align*}
Under \autoref{condition4}, we have that $\gamma_{k+1}^{-1} - \gamma_k^{-1} = \gamma_0^{-1}((k+1)^{\alpha_{\gamma}} -k^{\alpha_{\gamma}})\leq Ck^{\alpha_{\gamma} - 1}$.  So 
\[
\left\Vert \sum_{k=1}^N \gamma_{k}^{-1}(\Delta\widehat\theta_{k-1} - \Delta\widehat\theta_k) \right\Vert_2 \leq C + C\sum_{k=1}^{N-1}k^{\alpha_{\gamma}-1}\left\Vert \Delta\widehat\theta_{k} \right\Vert_2 + \gamma_0^{-1}N^{\alpha_{\gamma}}\left\Vert \Delta\widehat\theta_{N} \right\Vert_2
\]
According to \autoref{theorem2}, we have that \[\sum_{k=1}^{N-1}k^{\alpha_{\gamma}-1}\left\Vert \Delta\widehat\theta_{k} \right\Vert_2\leq C\sum_{k=1}^{N-1}k^{\frac{\alpha_{\gamma} + \alpha_{h}}{2}-1}\sqrt{\log k}, \ a.s.\]  and that \[\gamma_0^{-1}N^{\alpha_{\gamma}}\left\Vert \Delta\widehat\theta_{N}\right\Vert \leq CN^{\frac{\alpha_{\gamma}+\alpha_h}{2}}\sqrt{\log N}, \qquad  \text{a.s.}\] This implies that 
\[
\left\Vert \sum_{k=1}^N \gamma_{k}^{-1}(\Delta\widehat\theta_{k-1} - \Delta\widehat\theta_k) \right\Vert_2 = O\left( N^{\frac{\alpha_{\gamma} + \alpha_{h}}{2}}\sqrt{\log N}\right), \qquad \text{a.s.},
\]
and as a result, 
\[
\left( N^{1+\alpha_{h}} \right)^{-\frac{1}{2}}\left\Vert \sum_{k=1}^N \gamma_{k}^{-1}(\Delta\widehat\theta_{k-1} - \Delta\widehat\theta_k) \right\Vert _2\rightarrow 0, \qquad \text{a.s.}
\]
 On the other side, we have that $\Vert \delta_{1,N} \Vert_2 \leq CN^{-s_{\mathcal K}\alpha_h}$ and $\Vert \delta_{2,N}\Vert \leq C\Vert\Delta\widehat\theta_{N-1}\Vert^2_2$ . So $\Vert \sum_{k=1}^N \delta_{2,N} \Vert_2 = O(N^{-\alpha_{\gamma}+\alpha_h +1}\log N)$ a.s. and, by the exact power-sum bound,
\[
 \left\Vert\sum_{k=1}^N \delta_{1,N} \right\Vert_2
 \leq C\sum_{k=1}^Nk^{-s_{\mathcal K}\alpha_h}
 =O\!\left(1+N^{(1-s_{\mathcal K}\alpha_h)_+}\log N\right).
\]
The latter is $o(N^{(1+\alpha_h)/2})$ because
$s_{\mathcal K}\alpha_h>1/2$ under  \autoref{condition4}; the former has
that order because $\alpha_\gamma>(1+\alpha_h)/2$.  Thus both terms are
$o(\sqrt{N^{1+\alpha_h}})$ a.s.\@

We finally look at $\sum_{k=1}^N (\Phi_{k}(\widehat\theta_{k-1}, \mathcal W_k) -\Phi_{k}(\widehat\theta_{k-1}))$. Obviously, according to our previous decomposition, we have that 
\begin{align*}
\sum_{k=1}^N \left(\Phi_{k}(\widehat\theta_{k-1}, \mathcal W_k) -\Phi_{k}(\widehat\theta_{k-1})\right) &  = \sum_{k=1}^N \left(\Phi_k(\theta_{0}, \mathcal W_k) - \Phi_k(\theta_{0})\right)  \\
& +  \sum_{k=1}^N\int_{0}^{1} \nabla_{\theta}\left(\Phi_k(\theta_{0} + \tau\Delta\widehat\theta_{k-1}, \mathcal  W_k) - \Phi_k(\theta_{0} + \tau\Delta\widehat\theta_{k-1})\right)d\tau \Delta\widehat\theta_{k-1}. 
\end{align*}
We next show that 
\[
\sum_{k=1}^N\int_{0}^{1} \nabla_{\theta}\left(\Phi_k(\theta_{0} + \tau\Delta\widehat\theta_{k-1}, \mathcal  W_k) - \Phi_k(\theta_{0} + \tau\Delta\widehat\theta_{k-1})\right)d\tau \Delta\widehat\theta_{k-1} = o(\sqrt{ N^{1+\alpha_{h}}}),  \qquad  \text{a.s.}
\]
Let
\[
 M_N=\sum_{k=1}^N\int_0^1
 \nabla_\theta\!
 \left\{\Phi_k(\theta_0+\tau\Delta\widehat\theta_{k-1}, \mathcal W_k)
       -\Phi_k(\theta_0+\tau\Delta\widehat\theta_{k-1})\right\}d\tau
 \Delta\widehat\theta_{k-1}.
\]
We have that 
\[
 \mathbb E_{N-1}\Vert  M_N\Vert_2^2
 \leq  \Vert  M_{N-1}\Vert_2^2 +  Ch_N^{-3}\Vert\Delta\widehat\theta_{N-1}\Vert_2^2 \leq \Vert  M_{N-1}\Vert_2^2  
 +CN^{-\alpha_\gamma+4\alpha_h}\log N, \qquad \text{a.s.}
\]
where the last equality uses the raw-iterate rate from
 \autoref{theorem2}. Then using the previous proofs leads to $\Vert M_N\Vert_2 = O(\frac{1+4\alpha_h - \alpha_{\gamma}}{2}\sqrt{\log^{1+r} (N)})$ a.s. for any $r>1$. 
This is $o_{\rm a.s.}(N^{(1+\alpha_h)/2})$ because
$\alpha_\gamma>3\alpha_h$, which follows from
$2\alpha_\gamma-3\alpha_h>1$ and $\alpha_\gamma<1$.

The above analysis implies that 
\[
\left\Vert N\left(\frac{2\mathcal{V}_{\mathcal{K},0}}{B(B-1)}\right)^{-\frac{1}{2}}\mathbb{H}_0\Delta\overline\theta_N  - \left(\frac{2\mathcal{V}_{\mathcal{K},0}}{B(B-1)}\right)^{-\frac{1}{2}}\sum_{k=1}^N \left(\Phi_k(\theta_{0}, \mathcal W_k) - \Phi_k(\theta_{0})\right)\right\Vert_2 = o\left(\sqrt{N^{1+\alpha_{h}}}\right), \ a.s.
\]
Note that 
\[
\text{var}\left[\left(\frac{2\mathcal{V}_{\mathcal{K},0}}{B(B-1)}\right)^{-\frac{1}{2}}(\Phi_k(\theta_0, \mathcal W_k) - \Phi_k(\theta_0))\right] = h_k^{-1}\mathbb{I}_p + O(1) = h_0^{-1}k^{\alpha_h}\mathbb{I}_p + O(1),  
\]
and 
\[
\sum_{k=1}^N \text{var}\left[\left(\frac{2\mathcal{V}_{\mathcal{K},0}}{B(B-1)}\right)^{-\frac{1}{2}}(\Phi_k(\theta_0, \mathcal W_k) - \Phi_k(\theta_0))\right] = \frac{N^{1+\alpha_h} \mathbb{I}_p}{h_0( 1 + \alpha_h)} + o(N^{1+\alpha_h}).
\]
For a fixed coordinate $j$, write
\[
  \zeta_{k,j}=e_j^{\top}\left(\frac{2\mathcal V_{\mathcal K,0}}
 {B(B-1)}\right)^{-1/2}
 \{\Phi_k(\theta_0,\mathcal W_k)-\Phi_k(\theta_0)\},
 \qquad s_{N,j}^2=\sum_{k\leq N}\mathbb E \zeta_{k,j}^2.
\]
The preceding variance calculation gives
$s_{N,j}^2\asymp N^{1+\alpha_h}$, while  \autoref{lemma1}(5) gives
$\mathbb E|\zeta_{k,j}|^4\leq C h_k^{-3}\leq Ck^{3\alpha_h}$.  Therefore
\[
 \sum_{N\geq3}
 \frac{\{\log\log(s_{N,j}^2)\}^2\mathbb E|\zeta_{N,j}|^4}
      {s_{N,j}^4}
 \leq C\sum_{N\geq3}N^{-2+\alpha_h}(\log\log N)^2<\infty.
\]
Here $\alpha_h<1/3$ follows from
$2\alpha_\gamma-3\alpha_h>1$ and $\alpha_\gamma<1$.  The scalar
Lindeberg LIL of \citet{tomkins1983lindeberg} applies and yields the two
coordinatewise limits in  \autoref{theorem1}. 

To prove the CLT result, part~(5) of  \autoref{lemma1} gives
\[\sum_{k=1}^N \mathbb{E}\left\Vert\Phi_k\left(\theta_0, W_k\right) - \Phi_k\left(\theta_0\right) \right\Vert^4_2 \leq C \sum_{k=1}^N h_k^{-3}\leq CN^{1+3\alpha_h}.\]
so 
\[
\frac{\sum_{k=1}^N \mathbb{E}\left\Vert\Phi_k\left(\theta_0, W_k\right) - \Phi_k\left(\theta_0\right) \right\Vert^4_2}{N^{2+2\alpha_h}}\rightarrow 0.
\]
The multivariate Lyapunov central limit theorem therefore gives the stated result.

\end{document}